\documentclass[prx, twocolumn, superscriptaddress, dvipsnames]{revtex4-2}
\pdfoutput=1
\usepackage[utf8]{inputenc}
\usepackage[english]{babel}
\usepackage{amsmath}
\usepackage{amssymb}
\usepackage{amsthm}
\usepackage[scr=rsfso]{mathalpha}
\usepackage{amsfonts}
\usepackage{graphicx}
\usepackage{caption}
\usepackage{bbm}
\usepackage{makecell}
\usepackage{subcaption}
\usepackage{xcolor}
\usepackage{bm}
\usepackage{here}
\usepackage{soul}
\usepackage[colorlinks,allcolors=blue]{hyperref}
\usepackage[noabbrev, capitalize, nameinlink]{cleveref}

\newtheorem{thm}{Theorem}[section]
\newtheorem{lemma}[thm]{Lemma}

\creflabelformat{equation}{#2\textup{#1}#3}
\theoremstyle{definition}
\newtheorem{defn}[thm]{Definition}
\newtheorem{exmp}[thm]{Example}
\newtheorem{prop}[thm]{Proposition}
\newtheorem{remark}[thm]{Remark}

\usepackage{tikz}
\usetikzlibrary{decorations.pathreplacing,calligraphy,decorations.markings}

\definecolor{tensorcolor}{rgb}{0.65,0.77,0.95}
\definecolor{mpdocolor}{HTML}{FCDE70}
\definecolor{whampdocolor}{HTML}{FCDE70}
\definecolor{whampdocolorw}{HTML}{EEF7FF}
\definecolor{mpdotcolor}{rgb}{1,0.98,0.94}
\definecolor{btensorcolor}{rgb}{0.65,0.50,0.69}
\definecolor{whitetensorcolor}{HTML}{F8F8F8}
\definecolor{diamondcolor}{HTML}{E6F1ED}
\definecolor{unitarycolor}{rgb}{0.8,0.5,.5}
\definecolor{lcolor}{HTML}{D9EAFD}

\newcommand\doubledx{1.6}
\newcommand\singledx{1.8}

\newcommand\whadx{0.988}
\newcommand\mthick{}
\newcommand\mthickt{very thick}

\newcommand{\GcTensor}[6]{
    \begin{scope}[shift={(#1)}]
    \ifnum#5=0
		\draw[Virtual] (-#2,0) -- (#2,0);
		\draw[] (0,#2) -- (0,-#2);
    \fi
    \ifnum#5=-1
		\draw[Virtual] (0,0) -- (#2,0);
		\draw[] (0,#2) -- (0,-#2);
    \fi
    \ifnum#5=1
		\draw[Virtual] (-#2,0) -- (0,0);
		\draw[] (0,#2) -- (0,-#2);
    \fi

    \ifnum#5=2
		\draw[Virtual] (-#2,0) -- (#2,0);
		\draw[] (0,0) -- (0,#2);
    \fi
    \ifnum#5=-2
		\draw[Virtual] (-#2,0) -- (#2,0);
		\draw[] (0,0) -- (0,-#2);
    \fi

    \ifnum#5=3
		\draw[Virtual] (-#2,0) -- (#2,0);
		\draw[] (0,-#2) -- (0,#2);
    \fi
    \ifnum#5=4
    \fi
        \draw[fill=#6, rounded corners=2pt,\mthick] (-#3,-#3) rectangle (#3,#3);
		\draw (0,0) node {\scriptsize #4};
	\end{scope}
}

\newcommand{\GTensor}[5]{
	 \GcTensor{#1}{#2}{#3}{#4}{#5}{tensorcolor};
}

\newcommand{\dmTensor}[4]{
    \begin{scope}[shift={(#1)}]
        \draw[\mthick] (-#2,0) -- (#2,0);
        \draw[\mthick] (-#2,0) -- (-#2,#2*0.3);
        \draw[\mthick] (#2,0) -- (#2,#2*0.3);
        \draw[\mthick, fill=diamondcolor] (#3,0) -- (0,#3) -- (-#3,0) -- (0,-#3) -- cycle;
        \draw (0,0) node {\scriptsize #4};
    \end{scope}
}

\newcommand{\uTensor}[6]{
    \begin{scope}[shift={(#1)}]
        \draw[\mthick](-#4,-#5) -- (-#4,#5);
        \draw[\mthick](#4,-#5) -- (#4,#5);
        \draw[\mthick, fill=lcolor, rounded corners=2pt] (-#2,-#3) rectangle (#2, #3);
        \draw (0,0) node {\scriptsize #6};
    \end{scope}
}

\newcommand{\vTensor}[6]{
    \begin{scope}[shift={(#1)}]
        \draw[\mthick](-#4,0) -- (-#4,#5);
        \draw[\mthick](#4,0) -- (#4,#5);
        \draw[\mthick](0,-#5) -- (0,0);
        \draw[\mthick, fill=whitetensorcolor, rounded corners=2pt] (-#2,-#3) rectangle (#2, #3);
        \draw (0,0) node {\scriptsize #6};
    \end{scope}
}

\newcommand{\MTensor}[4]{
\begin{scope}[shift={(#1)}]
    \draw[Virtual] (-#2,0) -- (#2,0);
    \draw[] (0,-#2) -- (0,#2);
    \filldraw[\mthick, fill=mpdocolor] (0,0) circle[radius=#3];
    \draw (0,0) node {\scriptsize #4};
\end{scope}
}

\newcommand{\etatensor}[6]{
    \begin{scope}[shift={(#1)}]
    \def\ls{1.2};
    \def\rs{0.8};
        \draw[-mid, \mthick](-#4*\ls,#3) -- (-#4*\ls,#5);
        \draw[mid-, \mthick](#4*\ls,#3) -- (#4*\ls,#5);
        \draw[-mid, \mthick](#4*\rs,#3) -- (#4*\rs,#5);
        \draw[mid-, \mthick](-#4*\rs,#3) -- (-#4*\rs,#5);
        \draw[\mthick, fill=lcolor, rounded corners=2pt] (-#2,-#3) rectangle (#2, #3);
        \draw (0,0) node {\scriptsize #6};
    \end{scope}
}

\newcommand{\etaleft}[6]{
 \begin{scope}[shift={(#1)}]
    \def\ls{1.2};
    \def\rs{0.8};
    \draw[\mthick](-#4*\ls,0) -- (-#4*\ls,#5);
    \draw[\mthick](-#4*\rs,0) -- (-#4*\rs,#5);
    \draw[-mid, \mthick](-#4*\ls,#3) -- (-#4*\ls,#5);
    \draw[mid-, \mthick](-#4*\rs,#3) -- (-#4*\rs,#5);
     \fill[lcolor, rounded corners=2pt] (-#2,-#3) rectangle (0, #3);
    \draw[black, \mthick, rounded corners=2pt] 
        (-#2, 0) -- (-#2, #3) -- (0, #3); 
    \draw[black, \mthick, rounded corners=2pt]
        (-#2, 0) -- (-#2, -#3) -- (0, -#3); 
    \draw (0,0) node {\scriptsize #6};
    \end{scope}
}

\newcommand{\etaright}[6]{
 \begin{scope}[shift={(#1)}]
    \def\ls{1.2};
    \def\rs{0.8};
    \draw[\mthick](#4*\ls,0) -- (#4*\ls,#5);
    \draw[\mthick](#4*\rs,0) -- (#4*\rs,#5);
    \draw[mid-, \mthick](#4*\ls,#3) -- (#4*\ls,#5);
        \draw[-mid, \mthick](#4*\rs,#3) -- (#4*\rs,#5);
     \fill[lcolor, rounded corners=2pt] (0,-#3) rectangle (#2, #3);
    \draw[black, \mthick, rounded corners=2pt] 
        (#2, 0) -- (#2, #3) -- (0, #3); 
    \draw[black, \mthick, rounded corners=2pt]
        (#2, 0) -- (#2, -#3) -- (0, -#3); 
    \draw (0,0) node {\scriptsize #6};
    \end{scope}
}

\newcommand{\tracetensor}[3]{
    \begin{scope}[shift={(#1)}]
    \def\ro{.1};
    \def\offset{0.08};
    \def\ls{1.2};
    \def\rs{0.8};
    \def\as{0.25};
        
		\draw[\mthickt,rounded corners=4pt](-#2*0.4,-#3) -- (-#2*0.4,#3*\ro)--(#2*0.4,#3*\ro)--(#2*0.4,-#3);

        \draw[-mid,\mthickt](-#2*0.4,-#3)--(-#2*0.4,#3*\ro-\as);
        \draw[mid-,\mthickt](#2*0.4,-#3)--(#2*0.4,#3*\ro-\as);
        
    \end{scope}
}

\newcommand{\EquationBox}[6]{
\begin{scope}[shift={(#1)}]
        \def\eqnv{1.};
        \draw[\mthick, rounded corners=4pt] (-#2, -#3) rectangle (#2, #3);
        \draw (0,0) node {\scriptsize #5};
        \draw (-#2+0.4,#3-0.4) node {#6};
        \end{scope}
}

\newcommand{\SingleDots}[2]{
	\begin{scope}[shift={(#1)}]
      \draw [thick, dotted] (-#2/2,0) to (#2/2,0);
	\end{scope}
}

\newcommand{\SideIdentityTensor}[4]{
	\begin{scope}[shift={(#1)}]
    \ifnum#4=-1
	   \draw [\mthick] (\doubledx-1,0.8) to  [bend right=90] (\doubledx-1,-0.8);
    \fi
    \ifnum#4=-2
	   \draw [\mthick] (\doubledx-1,0.8) to  [bend right=90] (\doubledx-1,-0.8);
      \draw [\mthick] (\doubledx-1,0.8) -- (\doubledx-0.5,0.8);
      \draw [\mthick] (\doubledx-1,-0.8) -- (\doubledx-0.5,-0.8);
    \fi
    \ifnum#4=-3
	   \draw [\mthick] (\doubledx-1,0.8) to  [bend right=90] (\doubledx-1,-0.8);
      \draw [\mthick] (\doubledx-1,0.8) -- (\doubledx-0.5,0.8);
      \draw [\mthick] (\doubledx-1,-0.8) -- (\doubledx-0.5,-0.8);
	\filldraw[color=black, fill=whitetensorcolor, \mthick] (\doubledx-1.4,0) circle (#3);
	\draw (\doubledx-1.4,0) node {#2};
    \fi
    \ifnum#4=1
	   \draw [\mthick] (-\doubledx+1,0.8) to  [bend left=90] (-\doubledx+1,-0.8);
    \fi
    \ifnum#4=2
	   \draw [\mthick] (-\doubledx+1,0.8) to  [bend left=90] (-\doubledx+1,-0.8);
      \draw [\mthick] (-\doubledx+1,0.8) -- (-\doubledx+0.5,0.8);
      \draw [\mthick] (-\doubledx+1,-0.8) -- (-\doubledx+0.5,-0.8);
    \fi
    \ifnum#4=3
	   \draw [\mthick] (-\doubledx+1,0.8) to  [bend left=90] (-\doubledx+1,-0.8);
      \draw [\mthick] (-\doubledx+1,0.8) -- (-\doubledx+0.5,0.8);
      \draw [\mthick] (-\doubledx+1,-0.8) -- (-\doubledx+0.5,-0.8);
	\filldraw[color=black, fill=whitetensorcolor, \mthick] (-\doubledx+1.4,0) circle (#3);
	\draw (-\doubledx+1.4,0) node {#2};
    \fi
\end{scope}
}

\newcommand{\projector}[7]{
    \begin{scope}[shift={(#1)}]
    \def\ro{.6};
    \def\ls{1.2};
    \def\rs{0.8};
    \def\offset{0.08};
    \def\hshiftt{0.3};
    \def\yshift{0.5};
    \def\ashift{0.15};
    \def\hshift{\singledx};
    \ifnum#7=0
		\draw[-mid, \mthick](-#4*\ls,0) -- (-#4*\ls,#5*3);
        \draw[mid-, \mthick](-#4*\rs,0) -- (-#4*\rs,#5*3);
        \draw[\mthick, rounded corners=2pt](#4*\rs,0) -- (#4*\rs,#5)--(#4*\ls,#5)--(#4*\ls,0);
        \draw[-mid,\mthickt](-#4*0.4,-#5)--(-#4*0.4,-#5+#3);
        \draw[mid-,\mthickt](#4*0.4,-#5)--(#4*0.4,-#5+#3);
        \draw[-mid, \mthick](#4*\rs,#3) -- (#4*\rs,#5-\ashift);
        \draw[mid-, \mthick](#4*\ls,#3) -- (#4*\ls,#5-\ashift);
    \fi
    \ifnum#7=1
        \draw[-mid,\mthick](#4*\rs,0) -- (#4*\rs,#5*3);
        \draw[mid-,\mthick](#4*\ls,0) -- (#4*\ls,#5*3);
        \draw[\mthick, rounded corners=2pt](-#4*\ls,0)--(-#4*\ls,#5)--(-#4*\rs,#5)--(-#4*\rs,0);
        \draw[-mid,\mthickt](-#4*0.4,-#5)--(-#4*0.4,-#5+#3);
        \draw[mid-,\mthickt](#4*0.4,-#5)--(#4*0.4,-#5+#3);
        \draw[mid-, \mthick](-#4*\rs,#3) -- (-#4*\rs,#5-\ashift);
        \draw[-mid, \mthick](-#4*\ls,#3) -- (-#4*\ls,#5-\ashift);
    \fi
    \ifnum#7=2
        \draw[\mthickt, rounded corners=4pt](-#4*0.4,0) -- (-#4*0.4,#5)--(#4*0.4,#5)--(#4*0.4,0);
        \draw[-mid,\mthickt](-#4*0.4,-#5)--(-#4*0.4,-#5+#3);
        \draw[mid-,\mthickt](#4*0.4,-#5)--(#4*0.4,-#5+#3);
    \fi
    \ifnum#7=3
        \draw[\mthick, rounded corners=2pt](-#4*\rs,0) -- (-#4*\rs,#5+\hshiftt*0.5)--(-#4*\rs-\hshift,#5+\hshiftt*0.5)--(-#4*\rs-\hshift,#5+2);
        \draw[\mthick, rounded corners=2pt](-#4*\ls,0)--(-#4*\ls,#5-\hshiftt*0.5)--(-#4*\ls-\hshift,#5-\hshiftt*0.5)--(-#4*\ls-\hshift,#5+2);
        \draw[\mthick, rounded corners=2pt](#4*\rs,0) -- (#4*\rs,#5)--(#4*\ls,#5)--(#4*\ls,0);
        \draw[-mid,\mthickt](-#4*0.4,-#5)--(-#4*0.4,-#5+#3);
        \draw[mid-, \mthickt](#4*0.4,-#5)--(#4*0.4,-#5+#3);
        \draw[-mid, \mthick](#4*\rs,#3) -- (#4*\rs,#5-\ashift);
        \draw[mid-, \mthick](#4*\ls,#3) -- (#4*\ls,#5-\ashift);
        \draw[-mid, \mthick](-#4*\ls-\hshift,#5+\hshiftt)--(-#4*\ls-\hshift,#5+2);
        \draw[mid-, \mthick](-#4*\rs-\hshift,#5+\hshiftt)--(-#4*\rs-\hshift,#5+2);
    \fi
    \ifnum#7=4
        \draw[\mthick,rounded corners=2pt](#4*\ls,0) -- (#4*\ls,#5-\hshiftt*0.5)--(#4*\ls+\hshift,#5-\hshiftt*0.5)--(#4*\ls+\hshift,#5+2);
        \draw[\mthick, rounded corners=2pt](-#4*\ls,0)--(-#4*\ls,#5)--(-#4*\rs,#5)--(-#4*\rs,0) ;
        \draw[\mthick, rounded corners=2pt](#4*\rs,0)--(#4*\rs,#5+\hshiftt*0.5)--(#4*\rs+\hshift,#5+\hshiftt*0.5)--(#4*\rs+\hshift,#5+2);
        \draw[-mid,\mthickt](-#4*0.4,-#5)--(-#4*0.4,-#5+#3);
        \draw[mid-,\mthickt](#4*0.4,-#5)--(#4*0.4,-#5+#3);
        \draw[mid-, \mthick](-#4*\rs,#3) -- (-#4*\rs,#5-\ashift);
        \draw[-mid, \mthick](-#4*\ls,#3) -- (-#4*\ls,#5-\ashift);
        \draw[mid-, \mthick](#4*\ls+\hshift,#5+\hshiftt)--(#4*\ls+\hshift,#5+2);
        \draw[-mid, \mthick](#4*\rs+\hshift,#5+\hshiftt)--(#4*\rs+\hshift,#5+2);
    \fi
    \ifnum#7=5
        \draw[\mthick, rounded corners=2pt](-#4*\rs,0) -- (-#4*\rs,#5-\hshiftt*0.5+\yshift)--(-#4*\rs+\hshift,#5-\hshiftt*0.5+\yshift)--(-#4*\rs+\hshift,#5+2);
        \draw[\mthick, rounded corners=2pt](-#4*\ls,0)--(-#4*\ls,#5+\hshiftt*0.5+\yshift)--(-#4*\ls+\hshift,#5+\hshiftt*0.5+\yshift)--(-#4*\ls+\hshift,#5+2);
        \draw[\mthick, rounded corners=2pt](#4*\rs,0) -- (#4*\rs,#5)--(#4*\ls,#5)--(#4*\ls,0);
        \draw[-mid, \mthick](#4*\rs,#3) -- (#4*\rs,#5-\ashift);
        \draw[mid-, \mthick](#4*\ls,#3) -- (#4*\ls,#5-\ashift);
        \draw[-mid,\mthickt](-#4*0.4,-#5)--(-#4*0.4,-#5+#3);
        \draw[mid-,\mthickt](#4*0.4,-#5)--(#4*0.4,-#5+#3);
        \draw[-mid, \mthick](-#4*\ls+\hshift,#5+\hshiftt+\yshift)--(-#4*\ls+\hshift,#5+2);
        \draw[mid-, \mthick](-#4*\rs+\hshift,#5+\hshiftt+\yshift)--(-#4*\rs+\hshift,#5+2);
    \fi
    \ifnum#7=6
        \draw[\mthick,rounded corners=2pt](#4*\ls,0) -- (#4*\ls,#5+\hshiftt*0.5+\yshift)--(#4*\ls-\hshift,#5+\hshiftt*0.5+\yshift)--(#4*\ls-\hshift,#5+2);
        \draw[\mthick, rounded corners=2pt](-#4*\ls,0)--(-#4*\ls,#5)--(-#4*\rs,#5)--(-#4*\rs,0);
        \draw[\mthick, rounded corners=2pt](#4*\rs,0)--(#4*\rs,#5-\hshiftt*0.5+\yshift)--(#4*\rs-\hshift,#5-\hshiftt*0.5+\yshift)--(#4*\rs-\hshift,#5+2);
        \draw[-mid,\mthickt](-#4*0.4,-#5)--(-#4*0.4,-#5+#3);
        \draw[mid-,\mthickt](#4*0.4,-#5)--(#4*0.4,-#5+#3);
        \draw[mid-, \mthick](-#4*\rs,#3) -- (-#4*\rs,#5-\ashift);
        \draw[-mid, \mthick](-#4*\ls,#3) -- (-#4*\ls,#5-\ashift);
        \draw[-mid, \mthick](#4*\rs-\hshift,#5+\hshiftt+\yshift)--(#4*\rs-\hshift,#5+2);
        \draw[mid-, \mthick](#4*\ls-\hshift,#5+\hshiftt+\yshift)--(#4*\ls-\hshift,#5+2);
    \fi
        \draw[\mthick, fill=whitetensorcolor, rounded corners=2pt] (-#2,-#3) rectangle (#2, #3);
        \draw (0,0) node {\scriptsize #6};
    \end{scope}
}

\newcommand{\blackdot}[5]{
\begin{scope}[shift={(#1)}]
\def\tp{0.75};
\draw[\mthick] (0,0.) -- (0, 0.5);
    \draw[\mthick] (0,-0.5) -- (0,0);
    \draw[Virtual, \mthick] (0,0.) -- (-0.5, 0.);
    \draw[Virtual, \mthick] (0.5,0.) -- (0., 0.);
    \node[symb, symb disk, symb large] at (0,0) {};
    \draw (0,\tp) node {\small #2};
    \draw (0,-\tp) node {\small #3};
    \draw (-\tp,0) node {\small #4};
    \draw (\tp,0.) node {\small #5};
\end{scope}
}

\newcommand{\haprojector}[3]{
\begin{scope}[shift={(#1)}]
\def\rd{0.25};
\def\disp{0.75};
    \filldraw[ultra thin, fill=white] (-1, -\rd) arc[start angle=90, end angle=270, x radius=\rd, y radius=-\rd] -- (1, \rd) arc[start angle=-90, end angle=90, x radius=\rd, y radius=-\rd] -- cycle;
    \ifnum#3=1
    \draw[mid-] (-\disp-\rd, -\rd) -- (-\disp-\rd, -2*\rd);
    \draw[-mid] (-\disp+\rd, -\rd) -- (-\disp+\rd, -2*\rd);
    \draw[mid-] (\disp-\rd, -\rd) -- (\disp-\rd, -2*\rd);
    \draw[-mid] (\disp+\rd, -\rd) -- (\disp+\rd, -2*\rd);
    \fi
    \draw (0,0) node {#2};
\end{scope}
}

\newcommand{\haMtensor}[2]{
\begin{scope}[shift={(#1)}]
\def\rd{0.25};
\def\vht{0.318};
 \draw[bevel, mid-] (0, -\rd) arc[start angle=-163.45, end angle=0, radius=\rd] -- (2*\rd, \rd+\vht);
 \draw[-mid] (0, \rd) -- (0, \rd+\vht);
  \filldraw[fill=whampdocolor] (0,0) circle[radius=\rd];
  \draw (0,0) node {#2};
\end{scope}
}

\newcommand{\whaMsymborg}[3]{
\begin{scope}[shift={(#1)}]
\filldraw[fill=#3] 
    (0,0) circle[radius=0.247];
    \node[anchor=center,scale=0.9] at (0,0) {#2};
\end{scope}
}

\newcommand{\whaNsymb}[1]{
\begin{scope}[shift={(#1)}]
\filldraw[fill=blue!5] 
    (0,0) circle[radius=0.247];
    \node[anchor=center,scale=0.9] at (0,0) {$\wt$};
\end{scope}
}

\newcommand{\whaMsymb}[1]{
\whaMsymborg{#1}{$M$}{whampdocolor}
}

\newcommand{\whaMorg}[4]{
\begin{scope}[shift={(#1)}]
\def\rd{0.247};
\def\rlen{0.318};
\ifnum#2=1
    \draw[-mid] (0, \rd) -- (0, \rd+\rlen);
    \draw[-mid] (0, -\rd-\rlen) -- (0, -\rd);
    \draw[Virtual, -mid] (\rd+\rlen, 0) -- (\rd, 0);
    \draw[Virtual, -mid] (-\rd, 0) -- (-\rd-\rlen, 0);
\fi
\ifnum#2=2
    \draw[-mid] (0, \rd) -- (0, \rd+\rlen);
    \draw[-mid] (0, -\rd-\rlen) -- (0, -\rd);
    \draw[Virtual] (\rd+\rlen, 0) -- (\rd, 0);
    \draw[Virtual] (-\rd, 0) -- (-\rd-\rlen, 0);
\fi
\whaMsymborg{(0,0)}{#3}{#4};
\end{scope}
}

\newcommand{\whaM}[2]{
\whaMorg{#1}{#2}{$M$}{whampdocolor};
}

\newcommand{\whaMDouble}[1]{
\begin{scope}[shift={(#1)}]
\draw[-mid] (0, 0.247) -- (0, 0.565);
\draw[bevel, mid-] (0.011, -0.246) arc[start angle=-163.45, end angle=0, radius=0.247] -- (0.494, 0.565);
\whaMsymb{(0, 0)};
\end{scope}
}

\newcommand{\whaNDouble}[1]{
\begin{scope}[shift={(#1)}]
\def\ysep{0.988};
\def\rd{0.247};
\draw[mid-] (0, -0.247) -- (0, -\ysep+\rd);
\draw[bevel, -mid] (0.011, 0.246) arc[start angle=163.45, end angle=0, radius=0.247] -- (0.494, -0.565) -- (0.494,-\ysep-0.18) arc[start angle=0, end angle=-163.45, radius=0.247] (0.011, -\ysep-0.246);
\whaNsymb{(0, 0)};
\end{scope}
}

\newcommand{\whaMtracesgl}[2]{
\begin{scope}[shift={(#1)}]
\def\rd{0.247};
\def\rlen{0.318};
\def\hro{0.25};
\draw[bevel, mid-] (0.011, -0.246) arc[start angle=-163.45, end angle=0, radius=0.247] -- (0.494, 0.18) arc[start angle=0, end angle=163.45, radius=0.247] (0.011, 0.246);
\whaMsymb{(0, 0)};
\end{scope}
}

\newcommand{\whachannelone}[2]{
\begin{scope}[shift={(#1)}]
\def\sep{#2};
\def\xone{0.};
\def\rd{0.247};
\draw[Virtual] (\xone-0.239, 0.626) arc[start angle=118.955, end angle=151.045, x radius=0.494, y radius=-0.494];
    \draw[Virtual] (\xone+0.6*\sep-\rd, 0.564) -- (\xone+\rd, 0.564);
    \draw[Virtual] (\xone+\sep-\rd, 0.564) -- (\xone+0.6*\sep+\rd, 0.564);
    \draw[Virtual, -mid] (\xone+\sep-\rd, 1.552) -- (\xone+\rd, 1.552);
    \filldraw[ultra thin, fill=white] 
    (\xone+0.6*\sep, 0.564) circle[radius=0.247];
    \node[anchor=center,scale=0.85] at (\xone+0.6*\sep, 0.564) {$\Xi$};
    \draw[Virtual] (\xone-0.239, 1.49) arc[start angle=118.955, end angle=151.045, radius=0.494];
    \filldraw[ultra thin, fill=white] 
    (\xone-0.494, 1.058) circle[radius=0.247];
    \node[anchor=center] at (\xone-0.494, 1.058) {\small $\Omega$};
    \draw[Virtual, mid-] (\xone+\sep+\rd, 1.552) -- (\xone+\sep+2*\rd, 1.552) arc[start angle=-90, end angle=90, x radius=0.494, y radius=-0.494] -- (\xone+\sep+\rd, 0.564);
    \whaMDouble{(\xone, 1.552)};
    \whaMDouble{(\xone+\sep, 1.552)};
    \whaNDouble{(\xone, 0.564)};
    \whaNDouble{(\xone+\sep, 0.564)};
\end{scope}
}

\newcommand{\whachannelonefour}[3]{
\begin{scope}[shift={(#1)}]
\def\sep{#2};
\def\seps{1.2};
\def\xone{0.};
\def\rd{0.247};
    \draw[Virtual] (\xone+0.1*\sep-\rd, 0.564) -- (\xone+\rd-0.5*\sep, 0.564);
    \draw[Virtual] (\xone+0.5*\sep-\rd, 0.564) -- (\xone+0.1*\sep+\rd, 0.564);
    \draw[Virtual, -mid] (\xone+1.5*\seps-\rd, 1.552) -- (\xone+\rd-1.5*\seps, 1.552);
    \filldraw[ultra thin, fill=white] 
    (\xone+0.1*\sep, 0.564) circle[radius=0.247];
    \node[anchor=center,scale=0.85] at (\xone+0.1*\sep, 0.564) {#3};
    \draw[Virtual] (\xone-1.5*\seps, 1.552) arc[start angle=90, end angle=180, radius=0.494];
    \draw[Virtual] (\xone-1.5*\seps, 0.564) arc[start angle=90, end angle=180, x radius=0.494, y radius=-0.494];
    \draw[Virtual] (\xone-1.5*\seps, 0.564) -- (\xone-0.5*\sep, 0.564) ;
    \filldraw[ultra thin, fill=white] 
    (\xone-0.494-1.5*\seps, 1.058) circle[radius=0.247];
    \node[anchor=center] at (\xone-0.494-1.5*\seps, 1.058) {\small $\Omega$};
    \draw[Virtual, mid-] (\xone+1.5*\seps+\rd, 1.552) -- (\xone+1.5*\seps+2*\rd, 1.552) arc[start angle=-90, end angle=90, x radius=0.494, y radius=-0.494] -- (\xone+1.5*\seps+\rd, 0.564)-- (\xone+0.5*\sep+\rd, 0.564);
    \whaMDouble{(\xone-0.5*\seps, 1.552)};
    \whaMDouble{(\xone+0.5*\seps, 1.552)};
    \whaMDouble{(\xone-1.5*\seps, 1.552)};
    \whaMDouble{(\xone+1.5*\seps, 1.552)};
    \whaNDouble{(\xone-0.5*\sep, 0.564)};
    \whaNDouble{(\xone+0.5*\sep, 0.564)};
\end{scope}
}

\newcommand{\wharhozeroloose}[3]{
\begin{scope}[shift={(#1)}]
\def\sep{#2};
\def\xone{0.};
\def\rd{0.247};
\draw[Virtual] (\xone, 0.564) arc[start angle=90, end angle=180, x radius=0.494, y radius=-0.494];
    \draw[Virtual] (\xone, 0.564) -- (\xone+\sep+\rd, 0.564);
    \draw[Virtual, -mid] (\xone+\sep-\rd, 1.552) -- (\xone+\rd, 1.552);
    \draw[Virtual] (\xone, 1.552) arc[start angle=90, end angle=180, radius=0.494];
    \filldraw[ultra thin, fill=white] 
    (\xone-0.494, 1.058) circle[radius=0.247];
    \node[anchor=center] at (\xone-0.494, 1.058) {\small $\Omega$};
    \draw[Virtual, mid-] (\xone+\sep+\rd, 1.552) -- (\xone+\sep+2*\rd, 1.552) arc[start angle=-90, end angle=90, x radius=0.494, y radius=-0.494] -- (\xone+\sep+\rd, 0.564);
    \ifnum#3=1
    \whaMsymb{(\xone, 1.552)};
    \whaMsymb{(\xone+\sep, 1.552)};
    \fi
    \ifnum#3=2
    \whaMDouble{(\xone, 1.552)};
    \whaMDouble{(\xone+\sep, 1.552)};
    \fi
\end{scope}
}

\newcommand{\wharholeft}[2]{
\begin{scope}[shift={(#1)}]
    \def\sep{1.2};
\def\xone{0.};
\def\rd{0.247};
\draw[Virtual] (\xone, 0.564) arc[start angle=90, end angle=180, x radius=0.494, y radius=-0.494];
    \draw[Virtual] (\xone, 0.564) -- (\xone+\sep+2*\rd, 0.564);
    \draw[Virtual, -mid] (\xone+\sep-\rd, 1.552) -- (\xone+\rd, 1.552);
    \draw[Virtual] (\xone, 1.552) arc[start angle=90, end angle=180, radius=0.494];
    \filldraw[ultra thin, fill=white] 
    (\xone-0.494, 1.058) circle[radius=0.247];
    \node[anchor=center] at (\xone-0.494, 1.058) {\small #2};
     \draw[Virtual] (\xone+\sep+\rd, 1.552) -- (\xone+\sep+2*\rd, 1.552); 
     \node[anchor=center, text=Virtual] at (\xone+\sep+4*\rd, 1.7) {$\overset{N-2}\cdots$};
     \node[anchor=center, text=Virtual] at (\xone+\sep+4*\rd, 0.564) {$\cdots$};
    \whaMsymb{(\xone, 1.552)};
    \whaMDouble{(\xone+\sep, 1.552)};
\end{scope}
}

\newcommand{\wharholefttwo}[1]{
\begin{scope}[shift={(#1)}]
    \def\sep{1.2};
    \def\seps{0.75};
\def\xone{0.};
\def\rd{0.247};
\draw[Virtual] (\xone, 0.564) arc[start angle=90, end angle=180, x radius=0.494, y radius=-0.494];
    \draw[Virtual] (\xone, 0.564) -- (\xone+\sep+3*\rd, 0.564);
    \draw[Virtual, -mid] (\xone+\sep-\rd, 1.552) -- (\xone+\rd, 1.552);
    \draw[Virtual] (\xone, 1.552) arc[start angle=90, end angle=180, radius=0.494];
     \draw[Virtual] (\xone+\sep+\rd, 1.552) -- (\xone+\sep+3*\rd, 1.552); 
     \node[anchor=center, text=Virtual] at (\xone+\sep+4*\rd, 1.7) {$\overset{N-2}\cdots$};
     \node[anchor=center, text=Virtual] at (\xone+\sep+4*\rd, 0.564) {$\cdots$};
    \whaMsymb{(\xone, 1.552)};
    \whaMDouble{(\xone+\sep, 1.552)};
     \filldraw[ultra thin, fill=white] 
    (\xone, 0.564) circle[radius=0.247];
    \node[anchor=center] at (\xone, 0.564) {\scriptsize $\Xi^{\frac{1}{2}}$};
    \filldraw[ultra thin, fill=white] 
    (\xone+\seps, 0.564) circle[radius=0.247];
    \node[anchor=center] at (\xone+\seps, 0.564) {$B'$};
    \filldraw[ultra thin, fill=white] 
    (\xone+2*\seps, 0.564) circle[radius=0.247];
    \node[anchor=center] at (\xone+2*\seps, 0.564) {\scriptsize $\Xi^{\frac{1}{2}}$};
\end{scope}
}

\newcommand{\wharhoright}[1]{
\begin{scope}[shift={(#1)}]
    \def\sep{1.2};
\def\xone{0.};
\def\rd{0.247};
    \draw[Virtual, -mid] (\xone-\rd, 1.552) -- (\xone-2*\rd, 1.552);
    \draw[Virtual, mid-] (\xone+\rd, 1.552) -- (\xone+2*\rd, 1.552) arc[start angle=-90, end angle=90, x radius=0.494, y radius=-0.494] -- (\xone+\rd, 0.564) -- (\xone-2*\rd, 0.564);
    \node[anchor=center, text=Virtual] at (\xone-4*\rd, 1.552) {$\cdots$};
     \node[anchor=center, text=Virtual] at (\xone-4*\rd, 0.564) {$\cdots$};
    \whaMsymb{(\xone, 1.552)};
\end{scope}
}

\newcommand{\whaMtrace}[2]{
\begin{scope}[shift={(#1)}]
\def\rd{0.247};
\def\rlen{0.318};
\def\hro{0.25};
\draw[Virtual] (\rd+\rlen+0.1, 0) -- (\rd, 0);
\draw[Virtual] (-\rd, 0) -- (-\rd-\rlen-0.1, 0);
\draw[bevel, -mid] (0.,\rd) -- (0., \rd+\rlen*\hro) arc[start angle=-180, end angle=0, x radius=\rd, y radius=-\rd] -- (2*\rd, 0.) -- (2*\rd, -\rd-\rlen*\hro) arc[start angle=0, end angle=180, x radius=\rd, y radius=-\rd] (0.,-\rd-\rlen*\hro) -- (0.,-\rd);
\filldraw[fill=whampdocolor] 
(0,0) circle[radius=0.247];
\node[anchor=center] at (0,0) {#2};
\end{scope}
}

\newcommand{\whaMtraceOp}[3]{
\begin{scope}[shift={(#1)}]
\def\rd{0.247};
\def\rlen{0.318};
\def\hro{0.25};
\draw[Virtual] (\rd+\rlen+0.1, 0) -- (\rd, 0);
\draw[Virtual] (-\rd, 0) -- (-\rd-\rlen-0.1, 0);
\draw[bevel, -mid] (0.,\rd) -- (0., 3*\rd+\rlen*\hro) arc[start angle=-180, end angle=0, x radius=\rd, y radius=-\rd] -- (2*\rd, 0.) -- (2*\rd, -\rd-\rlen*\hro) arc[start angle=0, end angle=180, x radius=\rd, y radius=-\rd] (0.,-\rd-\rlen*\hro) -- (0.,-\rd);
\filldraw[fill=whampdocolor] 
(0,0) circle[radius=\rd];
\node[anchor=center] at (0,0) {#2};
\filldraw[fill=whitetensorcolor] 
(0,2.5*\rd) circle[radius=\rd];
\node[anchor=center] at (0,2.5*\rd) {#3};
\end{scope}
}

\newcommand{\bendline}[1]{
\begin{scope}[shift={(#1)}]
\def\rd{0.247};
\def\rlen{0.318};
\def\hro{0.25};
\draw[bevel, -mid] (0.,3.2*\rd) -- (0., 3*\rd+\rlen*\hro) arc[start angle=-180, end angle=0, x radius=\rd, y radius=-\rd] -- (2*\rd, 0.) -- (2*\rd, -\rd-\rlen*\hro) arc[start angle=0, end angle=180, x radius=\rd, y radius=-\rd] (0.,-\rd-\rlen*\hro) -- (0.,-\rd);
\draw[]  (0.,\rd) -- (0.,1.8*\rd);
\end{scope}
}

\newcommand{\whaMtraceOpmul}[3]{
\begin{scope}[shift={(#1)}]
\def\rd{0.247};
\def\rlen{0.318};
\def\hro{0.25};
\draw[fill=whitetensorcolor, rounded corners=2pt] (-\rd,1.8*\rd) rectangle (2.47+\rd, 3.*\rd);
\draw (0,2.4*\rd) node {\scriptsize #3};
\draw[Virtual] (0.988+\rd+\rlen-0.1, 0) -- (-\rd-\rlen-0.1, 0);
\draw[Virtual] (2.47+\rd+\rlen+0.1, 0) -- (2.47-\rd-\rlen+0.1, 0);
\bendline{(0,0)};
\bendline{(0.988,0)};
\bendline{(2.47,0)};
\filldraw[fill=whampdocolor] 
(0,0) circle[radius=\rd];
\filldraw[fill=whampdocolor] 
(0.988,0) circle[radius=\rd];
\filldraw[fill=whampdocolor] 
(2.47,0) circle[radius=\rd];
\node[anchor=center] at (0,0) {#2};
\node[anchor=center] at (0.988,0) {#2};
\node[anchor=center] at (2.47,0) {#2};
\node[anchor=center] at (1.729, 0.1) {$\overset{l}\cdots$};
\end{scope}
}

\newcommand{\whaGcTensor}[6]{
    \begin{scope}[shift={(#1)}]
    \ifnum#5=-1
		\draw[Virtual] (0,0) -- (#2,0);
		\draw[-mid] (0,#3) -- (0,#2);
      \draw[-mid] (0,-#2) -- (0,-#3);
    \fi
    \ifnum#5=1
		\draw[Virtual] (-#2,0) -- (0,0);
		\draw[-mid] (0,#3) -- (0,#2);
      \draw[-mid] (0,-#2) -- (0,-#3);
    \fi
        \draw[fill=#6, rounded corners=2pt,\mthick] (-#3,-#3) rectangle (#3,#3);
		\draw (0,0) node {\scriptsize #4};
	\end{scope}
}

\newcommand{\whauTensor}[6]{
    \begin{scope}[shift={(#1)}]
        \draw[\mthick,-mid](-#4,-#5) -- (-#4,-#3);
        \draw[\mthick,-mid](-#4,#3) -- (-#4,#5);
        \draw[\mthick,-mid](#4,-#5) -- (#4,-#3);
        \draw[\mthick,-mid](#4,#3) -- (#4,#5);
        \draw[\mthick, fill=lcolor, rounded corners=2pt] (-#2,-#3) rectangle (#2, #3);
        \draw (0,0) node {\scriptsize #6};
    \end{scope}
}

\newcommand{\whaATensor}[3]{
\begin{scope}[shift={(#1)}]
\def\rd{0.247};
\def\rlen{0.318};

\ifnum#2=1
\draw[] (0, \rd) -- (0, \rd+\rlen);
\draw[Virtual] (\rd+\rlen, 0) -- (\rd, 0);
\draw[Virtual] (-\rd, 0) -- (-\rd-\rlen, 0);
\fi
\ifnum#2=3
\draw[] (0, -\rd) -- (0, -\rd-\rlen);
\draw[Virtual] (\rd+\rlen, 0) -- (\rd, 0);
\draw[Virtual] (-\rd, 0) -- (-\rd-\rlen, 0);
\fi
\draw[fill=tensorcolor, rounded corners=2pt,\mthick] (-\rd,-\rd) rectangle (\rd,\rd);
\node[anchor=center] at (0,0) {#3};

\end{scope}
}

\usetikzlibrary{decorations.pathmorphing, decorations.markings, arrows, arrows.meta, shapes, patterns}
\tikzset{baseline={([yshift=-.5ex]current bounding box.center)}}
\tikzset{every path/.style={ line width=0.5pt, line cap=round }}
\colorlet{Virtual}{RedOrange}
\tikzstyle{bevel} = [ preaction = { draw, white, line width=3pt,  line cap = round } ]
\tikzstyle{bevel wide} = [ preaction = { draw, white, line width=4pt,  line cap = round } ]
\tikzstyle{symb} = [ draw=black, fill=black, line width=0.4pt, inner sep=1.5pt ]
\tikzstyle{mysymb} = [ draw=black, fill=white, circle, line width=0.3pt, inner sep=1pt, font=\small ] 
\tikzstyle{symb large} = [ inner sep=2.1pt ]
\tikzstyle{symb small} = [ inner sep=1pt   ]
\tikzstyle{symb tiny} = [ inner sep=0.8pt ]
\tikzstyle{symb fdisk} = [ circle ]
\tikzstyle{symb disk} = [ circle ]
\tikzstyle{symb square} = [ rectangle ]
\tikzstyle{symb fsquare} = [ rectangle ]
\tikzstyle{Msymb}=[draw=black, fill=whampdocolor, circle, inner sep=1pt, font=\small]
\tikzstyle{Nsymb}=[draw=black, fill=whampdocolorw, circle, inner sep=1pt, font=\small]
\tikzstyle{-mid} = [ decoration={ markings, mark = at position 0.50*\pgfdecoratedpathlength+0.6*3pt with \arrow{>[width=2pt]} }, postaction={decorate} ]
\tikzstyle{mid-} = [ decoration={ markings, mark = at position 0.50*\pgfdecoratedpathlength+0.6*3pt with \arrow{<[width=2pt]} }, postaction={decorate} ]

\newcommand\subsetsim{\mathrel{%
  \ooalign{\raise0.2ex\hbox{$\subset$}\cr\hidewidth\raise-0.8ex\hbox{\scalebox{0.9}{$\sim$}}\hidewidth\cr}}}

\newcommand{\tb}{\hspace{0.2cm}}

\newcommand{\ra}{\rightarrow}

\newcommand{\mS}{\mathcal{S}}
\newcommand{\mT}{\mathcal{T}}
\newcommand{\mE}{\mathcal{E}}

\newcommand{\mL}{\mathcal{L}}
\newcommand{\tr}{\mathrm{Tr}}
\newcommand{\bo}{\mathbbm{1}}

\newcommand{\rd}{\mathrm{d}}

\newcommand{\mV}{\text{V}}
\newcommand{\mW}{\text{W}}
\newcommand{\dg}{\dagger}
\newcommand{\init}{\text{init}}
\newcommand{\Afib}{A_{\text{Fib}}}
\newcommand{\rof}{0.65}
\newcommand{\wt}{M'} 
\newcommand{\hbnt}{a} 
\newcommand{\Eone}{\text{I}}
\newcommand{\Etwo}{\text{II}}
\newcommand{\fQ}{\bm{Q}} 

\begin{document}
\title{Parent Lindbladians for Matrix Product Density Operators}

\author{Yuhan Liu}
\affiliation{Max Planck Institute of Quantum Optics, Hans-Kopfermann-Str. 1, Garching 85748, Germany}
\affiliation{Munich Center for Quantum Science and Technology (MCQST), Schellingstr. 4, 80799 M{\"{u}}nchen, Germany}

\author{Alberto Ruiz-de-Alarc\'{o}n}
\affiliation{Departamento de An\'alisis Matem\'atico y Matemática Aplicada, Universidad Complutense de Madrid, 28040 Madrid, Spain}
\affiliation{Department of Mathematics, CUNEF Universidad, 28040 Madrid, Spain}

\author{Georgios Styliaris}
\affiliation{Max Planck Institute of Quantum Optics, Hans-Kopfermann-Str. 1, Garching 85748, Germany}
\affiliation{Munich Center for Quantum Science and Technology (MCQST), Schellingstr. 4, 80799 M{\"{u}}nchen, Germany}

\author{Xiao-Qi Sun}
\affiliation{Max Planck Institute of Quantum Optics, Hans-Kopfermann-Str. 1, Garching 85748, Germany}
\affiliation{Munich Center for Quantum Science and Technology (MCQST), Schellingstr. 4, 80799 M{\"{u}}nchen, Germany}

\author{David P\'erez-Garc\'ia}
\affiliation{Departamento de An\'alisis Matem\'atico y Matemática Aplicada, Universidad Complutense de Madrid, 28040 Madrid, Spain}
\affiliation{Instituto de Ciencias Matem\'aticas (CSIC-UAM-UC3M-UCM), 28049 Madrid, Spain}

\author{J.~Ignacio Cirac}
\affiliation{Max Planck Institute of Quantum Optics, Hans-Kopfermann-Str. 1, Garching 85748, Germany}
\affiliation{Munich Center for Quantum Science and Technology (MCQST), Schellingstr. 4, 80799 M{\"{u}}nchen, Germany}

\date{\today}
\begin{abstract}
Understanding quantum phases of matter is a fundamental goal in physics. For pure states, the representatives of phases are the ground states of locally interacting Hamiltonians, which are also renormalization fixed points (RFPs). These RFP states are exactly described by tensor networks. Extending this framework to mixed states, matrix product density operators (MPDOs) which are RFPs are believed to encapsulate mixed-state phases of matter in one dimension, where non-trivial topological phases have already been shown to exist. However, to better motivate the physical relevance of those states, and in particular the physical relevance of the recently found non-trivial phases, it remains an open question whether such MPDO RFPs can be realized as steady states of local Lindbladians. In this work, we resolve this question by analytically constructing parent Lindbladians for MPDO RFPs. These Lindbladians are local, frustration-free, and exhibit minimal steady-state degeneracy. Interestingly, we find that parent Lindbladians possess a rich structure that distinguishes them from their Hamiltonian counterparts. In particular, we uncover an intriguing connection between the non-commutativity of the Lindbladian terms and the fact that the corresponding MPDO RFP belongs to a non-trivial phase.

\end{abstract}

\maketitle


\section{Introduction}

In quantum many-body systems, short-ranged interactions strongly constrain the entanglement of physically relevant states. This fact provides a natural foundation for theoretical tools that exploit the underlying entanglement structure.
Among these tools, tensor networks have been remarkably successful in revealing a variety of phenomena in strongly interacting many-body quantum systems~\cite{cirac2021matrix}. The tensor-network framework makes the spatial structure of the underlying correlations explicit while allowing for the efficient description of physically relevant many-body quantum states. For instance, matrix product states (MPS)~\cite{fannes1992finitely}, the prototypical class of pure tensor-network states in 1D, can approximate ground states of gapped local Hamiltonians using a number of parameters that scale only polynomially with system size~\cite{hastings2007area}. Analogously, matrix-product density operators (MPDO) are the mixed-state counterparts of MPS. They can provide an efficient description of
Gibbs states of local Hamiltonians~\cite{hastings2006solving,molnar2015approximating} and, via the bulk-boundary correspondence of tensor networks~\cite{cirac2011entanglement}, they appear as boundary states of 2D systems.

Beyond serving as a successful variational class~\cite{schollwock2011density}, tensor-network states also emerge as exact solutions to the ground state problem of local Hamiltonians. 
Specifically, every projected entangled-pair state~\cite{cirac2021matrix} (PEPS) -- the generalization of MPS to higher dimensions -- can be realized as the exact ground state of a local Hamiltonian. The resulting so-called parent Hamiltonians are
local, have the minimal possible ground-state degeneracy, and, in 1D, are also gapped~\cite{fannes1992finitely,cirac2021matrix}. The parent Hamiltonian construction thus endows pure tensor-network states -- MPS in particular -- with a direct physical interpretation.

Parent Hamiltonians also play a central role 
in the classification of quantum phases of matter~\cite{cirac2021matrix}.
PEPS, together with their parent Hamiltonians, provide representatives for all known (non-chiral)
gapped topological phases of matter in 2D, as they encompass string-net models~\cite{levin2005string,buerschaper2009explicit,levin2006detecting,kitaev2012models,lin2014generalizations,kim2024classifying}. In particular, this holds true even if PEPS are restricted to renormalization fixed point (RFP) states~\cite{cirac2021matrix}. This fact becomes both rigorous and explicit in 1D. There, RFP MPS emerge in the limit of a real-space renormalization flow.
Moreover, the complete phase classification can be conducted directly~\cite{chen2011classification,schuch2011classifying} and the associated MPS RFPs provide representatives for all phases.


Recently, there has been an increasing interest in classifying quantum phases of matter for mixed states~\cite{coser2019classification,de2022symmetry,lee2025symmetry,ma2023average,sang2024mixed,rakovszky2024defining,huang2024hydrodynamics,wang2023intrinsic,gu2024spontaneous,ma2024symmetry,you2024intrinsic,ellison2024towards,lessa2024strong,chen2024unconventional,ma2025topological}, where novel phases -- not present in pure states -- can emerge. Notably, for 1D MPDO, non-trivial phases exist even with short-range correlations~\cite{lessa2024mixed,sun2025}. This contrasts sharply with pure states, where non-trivial phases with short-range entanglement appear only after imposing symmetry protection~\cite{chen2011classification,schuch2011classifying}. Although a full classification of mixed-state phases is still unresolved, MPDO RFPs are expected to capture all possible phases, analogously to their pure-state counterparts. However, it has remained unclear whether MPDO RFPs admit a physical interpretation analogous to parent Hamiltonians for MPS. 





In this paper, we resolve this question by establishing the open system analog of the parent Hamiltonian construction for MPDO RFPs. Since this concerns mixed quantum states, in place of the Hamiltonian ground state, we consider steady states of a Markovian open-system evolution, described by a Lindblad generator~\cite{breuer2002theory}. \emph{When is a given MPDO the steady state of a local Lindbladian?} In analogy to the pure state case, we call this a parent Linbladian of the MPDO. We systematically answer this question, giving an explicit analytical parent Lindbladian construction for the case when the MPDO is an RFP (see \cref{tab:translation}). 
We show that the constructed parent Lindbladians satisfy many important properties: They are local, frustration-free, and with minimal steady-state degeneracy. In contrast to parent Hamiltonians, parent Lindbladians are not always commuting at the RFP and a nontrivial steady-state degeneracy does not necessarily imply long-range correlation (See \cref{tab:translation} for a summary and \cref{sec:sumres} for further discussion).

The problem of constructing a local Lindbladian with a unique (specified) tensor-network state as its fixed point has been considered before, most notably as a numerical algorithm under certain assumptions~\cite{bondarenko2021constructing}, in dissipative state preparation~\cite{diehl2008quantum,kraus2008preparation,verstraete2009quantum} and in the context of thermal states~\cite{kastoryano2016quantum,chen2023efficient,rouze2024efficient,guo2024designing,ding2024efficient}. In particular, Ref.~\cite{verstraete2009quantum} includes a simple construction to leverage the parent Hamiltonian of a given injective MPS to a parent Lindbladian with the specified MPS as its unique fixed point. Although this method can be used to construct a parent Lindbladian for the purification of a mixed state, it fails to give a parent Lindbladian for the mixed state itself, which is the main objective here. For thermal states, Ref.~\cite{guo2024designing} gives a construction for a parent Lindbladian for the case of stabilizer Hamiltonians.
Our construction contains thermal states of nearest-neighbor commuting Hamiltonians with zero correlation length but also includes states that are not necessarily full-rank, such as a generalization of (non-full-rank) thermal states and mixed states arising as boundaries of 2D topologically ordered systems.

\begin{table*}[ht]
    \centering
    \begin{tabular}{c|c|c}
    \hline\hline
       & Parent Hamiltonian $H$ & Parent Lindbladian  $\mathcal{L}$ (RFP)\\
       \hline
       Geometrically local & $H = \sum_i h_i$ & $\mathcal L = \sum_i \mathcal L_i$ \\
       \hline
       Local term construction & $h_i=\tau_i(\bo-VV^\dg)$ & $\mathcal{L}_i=\tau_i(\mT\circ\mS - \bo)$ \\
       \hline
       Frustration-free & $h_i |\psi_{\mathrm{MPS}}\rangle = 0 \; \forall i \quad (H \ge 0)$ & $\mathcal L_i (\rho_{\mathrm {MPDO}}) = 0 \; \forall i$\\
       \hline
       Degeneracy  & Minimal ground-state degeneracy & Minimal steady-state degeneracy  \\
       \hline
       Degeneracy implies long-range correlation & $\checkmark$ & Not always
       \\
       \hline
       Commuting for RFP? & $\checkmark$ & Not possible for a class of RFP 
       \\
    \hline \hline
    \end{tabular}
    \caption{Properties of parent Lindbladian for MPDO renormalization fixed-point, in comparison with parent Hamiltonian for MPS.} 
    \label{tab:translation}
\end{table*}

Our results endow MPDO, which are RFP, with a direct physical meaning -- they arise as the exact steady states of local dissipative dynamics. Moreover, we identify the class of MPDO whose parent Lindbladian is rapid-mixing, thereby allowing efficient mixed-state preparation. We envision our the parent Lindbladian construction can serve as a tool for classifying mixed-state quantum phases.

The remainder of the paper is organized as follows. In \cref{sec:sumres}, we summarize the results, outlining the core components for constructing parent Lindbladians and discussing their key properties. \Cref{sec:pre} provides an overview of the essential concepts of MPS and MPDO, offering the necessary background for the parent Lindbladian construction. The classification of MPDO RFPs is divided into two cases: simple and non-simple. In \cref{sec:general-comment}, we present a general framework for constructing parent Lindbladians for MPDO RFPs, with explicit treatments of the simple and non-simple cases detailed in \cref{sec:pL-simple} and \cref{sec:pL-non-simple}, respectively. Finally, we conclude in \cref{sec:outlook} by highlighting open questions and directions for future research.



\section{Summary of Results}
\label{sec:sumres}


We consider $N$ spins in 1D and a density operator, $\rho$, which is an RFP MPDO (as defined below). Our goal is to find a Lindbladian $\mathcal{L}^{(N)}$ such that: (i) $\mathcal{L}^{(N)}= \sum_{i=1}^N \mathcal{L}_i$, where $\mathcal{L}_i$'s are local Lindbladians themselves (i.e., act on the site $i$ and its $r$ neighbors, where $r$ is a fixed integer); (ii) $\mathcal{L}_i(\rho) = 0$ for all $i$; (iii) The convex set of states satisfying $\mathcal{L}^{(N)}(\rho)=0$ has a finite cardinality of extreme points independent of the system size $N$, and are RFP MPDOs. We will call such cardinality steady state degeneracy (SSD). We will call $\mathcal{L}$ the parent Lindbladian.


We recall that a Lindbladian is a superoperator that generates the time evolution of a density matrix capturing Markovian dynamics in open quantum systems~\cite{breuer2002theory}. An MPDO, specified by a tensor $M$, is a renormalization fixed-point (RFP) if there exist two local quantum channels (i.e., completely positive and trace-preserving maps) $\mT$ and $\mS$ that respectively convert a one-site tensor and a two-site tensor into each other~\cite{cirac2017matrix},
 \begin{align}\label{eq:MRFP-TS}
        \begin{array}{c}
        \begin{tikzpicture}[scale=1.,baseline={([yshift=-0.75ex] current bounding box.center)}
        ]
        \whaM{(0,0)}{2};
        \whaM{(3*\whadx,0)}{2};
        \whaM{(4*\whadx,0)}{2};
        \draw [->, thick](1.,0.4) to [out=45,in=135] (3*\whadx-1.,0.4);
        \draw (1.5*\whadx,0.85) node {$\mT$};
        \draw [->, thick](3*\whadx-1,-0.4) to [out=225,in=-45] (1.,-0.4);
        \draw (1.5*\whadx,-0.32) node {$\mS$};
        \end{tikzpicture}
        \end{array}
        \,.
    \end{align}
As expected physically, at the RFP, there is no (finite) characteristic length scale of the correlations in the system. This is made apparent by the fact that
the RFP of an MPDO has zero correlation length~\cite{cirac2017matrix}, in the sense that every two-point correlation function $\langle O_i O_j \rangle$ is independent of the distance $|i - j|$ between the two operators. 
Furthermore, the conditional mutual information $I(A:C|B)$ vanishes for any consecutive regions $A,B,C$ when the complement of $ABC$ is non-empty.

The two quantum channels of the RFP -- $\mT$ and $\mS$ --
allow for the straightforward construction of a local parent Lindbladian with desirable properties. For this, we first construct a local channel acting on two neighboring sites
\begin{equation}
    \mathcal{E}=\mT\circ\mS,
\end{equation}
and denote $\mathcal{E}_i=\tau_i(\mathcal{E})$ where $\tau_i$ translates the sites by amount $i$. The parent Lindbladian is then constructed as a summation of these terms,
\begin{align}
    \mathcal{L}^{(N)}&=\sum_{i=1}^{N}\mathcal{L}_i\\
    \mathcal{L}_i&=\mathcal{E}_i-\bo,
\end{align}
where $\mathcal{L}_i$ is understood as $\mathcal{L}_i\otimes\bo_{\text{rest}}$, acting trivially outside the support of $\mathcal{L}_i$. By construction, $\mathcal{L}_i$ is a valid Lindbladian operator, and the parent Lindbladian $\mathcal{L}$ is local. Conditions (i) and (ii) are obviously satisfied, and we will show in \cref{sec:pL-simple,sec:pL-non-simple} that (iii) is also fulfilled. 


We note that the above construction can be understood as a generalization of the MPS parent Hamiltonian. In that case, the role of $\mS, \mT$ is played an isometry $V$ and its hermitian conjugate $V^\dagger$. That is, $V^\dg$ is the isometry that brings two-site MPS tensor to one-site MPS tensor during renormalization, and $V$ brings one-site tensor to two-site tensor. The parent Hamiltonian $H = \sum_i h_i$ is then formed by choosing for each local term $h=\bo-VV^\dg$, which is completely analogous to our Lindbladian construction.

While the construction is intuitive, the major
contributions of the paper are to characterize several key theoretical properties of the parent Lindbladian:

\begin{itemize}
    \item \textbf{Local and frustration-free}: The parent Lindbladian is a summation of geometrically local terms, $\mathcal{L}^{(N)}=\sum_{i=1}^{N} \mathcal{L}_i$. A global steady state that satisfies $\mathcal{L}^{(N)}(\rho)=0$ is also in the space of steady states of each local term, $\mathcal{L}_i(\rho)=0$, making it frustration-free. 
 
     \item \textbf{Minimal steady-state degeneracy}: The Steady-State Degeneracy (SSD) of our constructed parent Lindbladian is minimal; no other local Lindbladian construction admits lower steady-state degeneracy.

     \item \textbf{Steady-state degeneracy does not necessarily imply long-range correlations}: Unlike MPS, we show the MPDO, whose parent Lindbladian hosts
     two or more steady states, does not necessarily have long-range correlations. For example, the parent Lindbladian of the state $\rho=\frac{1}{2^N}(\bo^{\otimes N}+\sigma_z^{\otimes N})$ has two degenerate steady states, but $\rho$ has only short-range correlations.  

     \item \textbf{Steady-state degeneracy does not classify mixed-state quantum phases}: Unlike MPS, we show the MPDOs, whose local parent Lindbladians host the same minimal steady-state degeneracy, may belong to different mixed-state quantum phases. For example, the parent Lindbladians for state $\rho=\frac{1}{2^N}(\bo^{\otimes N}+\sigma_z^{\otimes N})$ and $\rho=\frac{1}{2}(|0\rangle\langle0|^{\otimes N}+|1\rangle\langle1|^{\otimes N})$ both admit exactly two linearly independent steady states. Nevertheless, they belong to different phases. 
     
    \item \textbf{(Non)-commutativity of the parent Lindbladian}: The parent Lindbladian
    cannot always be chosen to be commuting ($[\mathcal{L}_i,\mathcal{L}_j]=0$) even at the RFP if we require it to be local and have minimal steady-state degeneracy.
    In contrast, the parent Hamiltonian at the MPS RFP can always be chosen as commuting, local, and with minimal ground-state degeneracy.

    \item \textbf{Rapid-mixing dynamics}: We identify a class of MPDO RFP whose parent Lindbladian is commuting, and thus rapid-mixing, in the sense that, starting from any initial state, the Lindbladian evolution brings it to the unique steady-state (or space of steady states) with vanishing error after time $t =O(\text{poly}\log N)$. For example, the parent Lindbladian for $\rho=\frac{1}{2^N}(\bo^{\otimes N}+\sigma_z^{\otimes N})$ is rapid-mixing and can bring product state $\rho=|0\rangle\langle0|^{\otimes N}$ to it after $t =O(\text{poly}\log N)$ with vanishing error. More interestingly, we will show in~\cref{sec:beyond} that the parent Lindbladian for the non-trivial density matrix~\cite{lessa2024mixed} $\rho_{CZX}^{(N)} = \frac{1}{2^N}(\bo_2^{\otimes N} + U^{(N)}_{CZX})$ is also rapid-mixing. 
\end{itemize}



The distinction of MPDO RFP between the so-called simple and non-simple~\cite{cirac2017matrix} will play an important role in our parent Lindbladian construction, which applies to both cases. The simple MPDO RFPs include, but are not limited to, Gibbs states of local commuting Hamiltonians ($e^{-\beta H}$ where $H=\sum_i h_i$ and $[h_i,h_j]=0$) with zero correlation length. However, the simple MPDOs under consideration can be non-full-rank mixed states, which can be understood as special limiting cases of Gibbs states. In that sense, our result goes beyond the usual Gibbs states of commuting Hamiltonians. More interestingly, the non-simple MPDO can correspond to boundaries of 2d topologically ordered systems. For this class, we will employ a systematic construction that exploits tools from $C^*$-weak Hopf algebras~\cite{molnar2022matrix,ruizdealarcon2024matrix,yoshiko2024haag,wang2024hopf}.


\section{Preliminaries}
\label{sec:pre}

In this section, we provide a minimal review of MPS and MPDO and set up the notation. For a comprehensive review, we refer the reader to Ref.~\cite{cirac2021matrix}. Regarding MPS, we introduce the concept of renormalization fixed points (RFPs) and the construction of parent Hamiltonians. Parent Hamiltonians enable the demonstration that MPS RFPs provide an exhaustive classification of MPS phases~\cite{schuch2011classifying}. Specifically, every MPS belongs to the same phase as its corresponding fixed point.

After establishing a unified structure for all MPS RFPs, we next introduce MPDOs and review the full characterization of their associated RFPs \cite{cirac2017matrix}, highlighting their significantly richer structures. In particular, we will treat simple and non-simple MPDO RFPs separately, as their distinct structures will lead to different parent Lindbladian constructions. This framework sets the stage for the parent Lindbladian construction, which we present in the following sections.


\subsection{Matrix Product State, Renormalization Fixed Points, and Parent Hamiltonian}
\subsubsection{Basics of MPS}
Consider a rank-3 tensor $A$ with graphical representation, 
\begin{align}
    \left(A^{i}\right)_{\alpha\beta} =
        \begin{array}{c}
        \begin{tikzpicture}[scale=1,baseline={([yshift=-0.65ex] current bounding box.center)}]
        \whaATensor{(0,0)}{1}{$A$};
		\draw (-0.75,0) node {$\alpha$};
		\draw (0.75,0) node {$\beta$};
		\draw (0,0.75) node {$i$};
        \end{tikzpicture}
        \end{array}
        \in \mathbb C
\end{align}
where $i=1,2,\cdots,d$ labels the Hilbert space and $\alpha,\beta=1,2,\cdots,D$ label the auxiliary (virtual) space. $d$ is the local Hilbert space dimension and $D$ is the bond dimension. $A^i$ can be viewed as a $D\times D$ matrix, $A^i\in\mathcal{M}_{D}$, where $\mathcal{M}_D$ denotes the matrix algebra of $D\times D$-dimensional complex matrices. 

Given the tensor $A$, one can concatenate $N$ copies to form a uniform MPS, 
\begin{equation}
    |\psi^{(N)}(A)\rangle=\sum_{\lbrace i\rbrace}\tr\left( A^{i_1} A^{i_2}\cdots A^{i_N}\right)|i_1 i_2 \cdots i_N\rangle.
\end{equation}
Graphically,
\begin{align}
|\psi^{(N)}(A)\rangle = 
    \begin{array}{c}
        \begin{tikzpicture}[scale=1.,baseline={([yshift=-0.75ex] current bounding box.center)}]
            \draw[Virtual] (1.235, 0.564) -- (0.741, 0.564);
    \draw[Virtual] (1.976, 0.564) -- (1.729, 0.564);
    \node[anchor=center] at (2.225, 0.661) {$\overset{N}\cdots$};
    \draw[Virtual] (2.717, 0.564) -- (2.47, 0.564);
    \draw[Virtual] (0.257, 0.564) arc[start angle=87.734, end angle=270, radius=0.247] -- (3.21, 0.071) arc[start angle=-90, end angle=90, radius=0.247];
    \draw[] (0.494, 0.811) -- (0.494, 1.129);
    \draw[] (1.482, 0.811) -- (1.482, 1.129);
    \draw[Virtual] (0.494, 0.564) -- (0.257, 0.564);
    \draw[] (2.964, 0.811) -- (2.964, 1.129);
    \draw[Virtual] (3.21, 0.564) -- (2.964, 0.564);
    \whaATensor{(0.494, 0.564)}{2}{$A$};
    \whaATensor{(1.482, 0.564)}{2}{$A$};
    \whaATensor{(2.964, 0.564)}{2}{$A$};
        \end{tikzpicture}
        \end{array}
        \,.
\end{align}
When restricted to a finite bond dimension $D$, MPS provides an efficient description of the underlying many-body states while incorporating by construction the entanglement area-law.

The MPS tensor has redundant degrees of freedom, in the sense that two different tensors may generate the same MPS $|\psi^{(N)}(A)\rangle$ for any $N$. One can use such redundant degrees of freedom, along with a finite number of blocking steps
(combining $n$ neighboring tensors to form a larger tensor with physical dimension $d^n$) 
to bring $A^i$ into a block-diagonal form where each block exhibits a desired structure.
Before showing this canonical form of the tensor, we first need to define the concept of injectivity \cite{cirac2017matrix}:

\begin{defn}[Normal and injective tensors]
An MPS tensor $A$ is called \emph{injective} if the matrices $A^i$ span the whole set of $\mathcal{M}_{D}$ matrices. An MPS tensor $A$ is a \emph{normal} tensor if it becomes injective after blocking. 
\end{defn}


\noindent In other words, $A$ is injective if any $D\times D$ matrix $X$ can be expanded as $X=\sum_{i=1}^d c_i A^i$, where $c_i\in\mathbb{C}$.

For any tensor $B$, it is always possible to obtain another tensor $A$ that generates the same MPS state for any $N$, and $A$ is in the following \textit{canonical form},
\begin{align}
    &A^i=\bigoplus_{\hbnt=1}^g (\mu_\hbnt\otimes A^i_\hbnt),\label{eqn:canonical}\\
    &\mathrm{where} \quad \mu_\hbnt=\mathrm{diag}\lbrace \mu_{\hbnt,1}, \mu_{\hbnt,2},\cdots,\mu_{\hbnt, n_\hbnt}\rbrace. \label{eqn:BNT}
\end{align}
Moreover, the tensors $\lbrace A_\hbnt\rbrace_{\hbnt=1,\cdots,g}$ can be taken to form a \textit{basis of normal tensors} (BNT): (i) Each $A_\hbnt$ is a normal tensor; (ii) There exists some $N_0$ such that for all $N>N_0$, $|\psi^{(N)}(A_\hbnt)\rangle$ are linearly independent. 
That is, a BNT explicitly separates the (normal) tensors generating proportional MPS.

The MPS state generated by $A$ can, therefore, be written as a summation,
\begin{align}
    |\psi^{(N)}(A)\rangle=\sum_{\hbnt=1}^g \left(\sum_{q=1}^{n_\hbnt} \mu_{\hbnt,q}^N\right) |\psi^{(N)}(A_\hbnt)\rangle.
\end{align}

\subsubsection{Renormalization Fixed Point}

In statistical mechanics, for phase classification, one defines the local coarse-graining transformation and the corresponding renormalization flow. The limit of renormalization flow, if it exists, is the renormalization fixed point. Different renormalization fixed points typically correspond to different phases.

In the context of MPS, coarse-graining is defined naturally by blocking neighboring MPS tensors and applying a local isometry to remove
short-range entanglement, while keeping the long-range entanglement intact~\cite{verstraete2009quantum}.
Such a renormalization flow is reversible, in the sense that there exists a local isometry to fine-grain the state. The renormalization fixed point is then defined as the limit of the renormalization flow \cite{cirac2017matrix}.
\begin{defn}[Renormalization fixed point]
A tensor $A$ is called a \emph{renormalization fixed point} (RFP) if it satisfies
\begin{align}
    A^{i_1} A^{i_2} = \sum_{j} V_{(i_1 i_2),j} A^{j},
\end{align}
where $V$ is an isometry $V^\dg V = \bo$. 
Graphically,
\begin{align} \label{eq:mps_rfp_def}
\begin{array}{c}
         \begin{tikzpicture}[scale=1,baseline={([yshift=-0.75ex] current bounding box.center)}]
        \whaATensor{(0,0)}{1}{$A$};
        \whaATensor{(\whadx,0)}{1}{$A$};
         \end{tikzpicture}
    \end{array}
    =
     \begin{array}{c}
    \begin{tikzpicture}[scale=.4,baseline={([yshift=-0.75ex] current bounding box.center)}]
        \vTensor{(\singledx*0.,1.5)}{1.2}{0.5}{\singledx*0.3}{1}{$V$};
        \foreach \x in {0,...,0}{
                \GTensor{(\singledx*0,0)}{1.}{.6}{\small $A$}{2}
        }
    \end{tikzpicture}
    \end{array}
    \,.
\end{align}
\label{def:RFP-pure}
\end{defn}

In the context of MPS, the phase classification is done conveniently using the RFP.
One can show that an MPS and its RFP are in the same phase using the parent Hamiltonian~\cite{schuch2011classifying}. The phase classification problem is then reduced to the phase classification for the RFP; MPS RFPs provide representatives for all 1D phases. 

At RFP, $|\psi^{(N)}(A)\rangle$ is a summation of linearly independent states $|\psi^{(N)}(A_\hbnt)\rangle$, 
each being a direct product of
states that are entangled only between neighboring sites. Specifically, by unitarily rotating the basis on each local site, $|\psi^{(N)}(A_\hbnt)\rangle$ takes the form of 
\begin{equation}
\begin{aligned}
    |\psi^{(N)}(A_\hbnt)\rangle&=|\omega(\Lambda_\hbnt)\rangle^{\otimes N}\\
    &\equiv\begin{array}{c}
            \begin{tikzpicture}[scale=.5,,baseline={([yshift=-0.75ex] current bounding box.center)}]
                \dmTensor{(0,0)}{1.}{.6}{\small $\Lambda_\hbnt$};
                \dmTensor{(2.3,0)}{1.}{.6}{\small $\Lambda_\hbnt$};
                \SingleDots{2.3*2,0}{\singledx/2};
                 \dmTensor{(2.3*3,0)}{1.}{.6}{\small $\Lambda_\hbnt$};
            \end{tikzpicture}
            \end{array}
            \;,
\end{aligned}
\end{equation}
where the matrix $\Lambda_\hbnt$ is diagonal, positive, and $\tr(\Lambda_\hbnt^2)=1$.
Thus
\begin{equation}
    |\psi^{(N)}(A)\rangle=\sum_{\hbnt=1}^g \left(\sum_{q=1}^{n_\hbnt} \mu_{\hbnt,q}^N\right) |\omega(\Lambda_\hbnt)\rangle^{\otimes N},
\end{equation}
where now $|\mu_{\hbnt,q}|=1$. 

\subsubsection{Parent Hamiltonian}
Given an MPS, one can always construct a local Hamiltonian such that this MPS is a ground state. A particularly useful construction is the parent Hamiltonian, which enjoys many nice properties. Given a uniform MPS generated by tensor $A$, one 
defines the vector space $G_L$ as the space spanned by all states $|\psi^{(L)}(A)\rangle_X$,
generated by $L$ consecutive 
tensors with arbitrary boundary condition $X$, 
\begin{align}
    G_L &= \lbrace  |\psi^{(L)}(A)\rangle_X|\forall X\in\mathcal{M}_{D}\rbrace\\
    |\psi^{(L)}(A)\rangle_X&=\sum_{\lbrace i\rbrace}\tr(X A^{i_1}\cdots A^{i_L})|i_1\cdots i_L\rangle.
\end{align}
$L$ is called the interaction length, and we choose large enough $L$ such that its orthogonal subspace $G_L^\perp$ in the $L$-site Hilbert space is not a null space, which is always possible. 
Consider any positive definite operator $h\geq 0$ such that $\mathrm{ker}(h)=G_L$ (one simple choice of $h$ is the projector to $G_L^\perp$). A parent Hamiltonian is constructed as \cite{perez2006matrix}
\begin{equation}
    H^{(N)}=\sum_{i=1}^{N} h_i,
\end{equation}
where $h_i=\tau_i(h)$ and $\tau_i$ translates the sites by
$i$. At the RFP, the choice of $h$ as the projector onto $G_L^\perp$ can be rewritten as 
\begin{equation}
    h=\bo-VV^\dg,
\end{equation}
where $V$ is the isometry from \cref{eq:mps_rfp_def}.

\begin{defn}[Frustration freeness]
    A Hamiltonian is \emph{frustration-free} if it can be written as a sum of terms $H^{(N)}=\sum_{i=1}^N h_i$, such that the lowest energy state of the full Hamiltonian is the lowest energy state of each individual term. Namely,
    \begin{equation}
        H^{(N)}|\psi\rangle=E_0|\psi\rangle\quad\Leftrightarrow\quad h_i|\psi\rangle=e^i_0|\psi\rangle, 
    \end{equation}
    where $E_0$ is the ground state energy of $H^{(N)}$ and $e_0^i$ is the ground state energy of $h_i$. 
\end{defn}
By construction, the parent Hamiltonian is translational invariant and frustration-free. The ground state of the full Hamiltonian satisfies $H^{(N)}|\psi\rangle=0$. Since each individual term $h_i$ is positive semidefinite, $\langle \psi|h_i|\psi\rangle\geq 0$. Therefore, $h_i|\psi\rangle=0$ holds for each individual state, verifying frustration freeness. 

Frustration-freeness plays an important role in proving that the ground-state degeneracy of a parent Hamiltonian is minimal when choosing sufficiently large $L$. In particular, the ground state degeneracy is $g$, the number of BNT elements in tensor $A$, and the ground state space is $\mathrm{Span}\lbrace|\psi^{(N)}(A_\hbnt)\rangle\rbrace$ \cite{perez2006matrix,cirac2021matrix}. Using the parent Hamiltonian of MPS at the RFP,
one can establish the phase classification of MPS: The number of elements in the BNT completely specifies the different phases~\cite{schuch2011classifying}. This number coincides with the ground-state degeneracy of the corresponding parent Hamiltonian.

Finally, the parent Hamiltonian is commuting $[h,\tau_j(h)]=0$ for all $j$ if the tensor is RFP. A nearest-neighbor parent Hamiltonian has $L=2$. One has the following theorem \cite{cirac2017matrix}.

\begin{thm}
    A tensor $A$ is an RFP iff for all $N>2$, $|\psi^{(N)}(A)\rangle$ is a ground state of a nearest-neighbor commuting parent Hamiltonian. 
\end{thm}

Let us summarize a few key properties of MPS that will be contrasted later in the case of MPDO with nontrivial distinctions: the overall structure of the non-injective MPS RFPs does not depend on the details of each BNT element. Moreover, the corresponding parent Hamiltonian is always commuting and straightforward to construct. However, as we will demonstrate in subsequent sections, the additional structure inherent to density matrices in the mixed-state case leads to a richer set of properties.


\subsection{Matrix Product Density Operator and Renormalization Fixed Points}

Analogously, one may ask whether mixed states on a one-dimensional chain also allow a tensor network description. Consider a rank-4 tensor $M$ with graphical representation, 
\begin{align}
    \left(M^{ij}\right)_{\alpha\beta} =
        \begin{array}{c}
        \begin{tikzpicture}[scale=1.,baseline={([yshift=-0.65ex] current bounding box.center)}]
		\draw (-0.75,0) node {$\alpha$};
		\draw (0.75,0) node {$\beta$};
		\draw (0,0.75) node {$i$};
        \draw (0,-0.75) node {$j$};
        \whaM{(0,0)}{2};
        \end{tikzpicture}
        \end{array}
        \in \mathbb C
\end{align}
where $i,j=1,2,\cdots,d$ are physical space indices and $\alpha,\beta=1,2,\cdots,D$ are auxiliary space indices.
A tensor $M$
generates a valid matrix-product density operator (MPDO)
\begin{equation}
\begin{aligned}
     \rho^{(N)}(M)=&\sum_{\lbrace i,j\rbrace}\tr\left( M^{i_1 j_1} M^{i_2 j_2}\cdots M^{i_N j_N}\right)\\
     &\quad  |i_1 i_2 \cdots i_N\rangle\langle j_1 j_2 \cdots j_N|,
\end{aligned}
\end{equation}
or graphically,
\begin{align}
\rho^{(N)}(M)=
    \begin{array}{c}
\begin{tikzpicture}[scale=1]
    \draw[Virtual] (1.235, 0.564) -- (0.741, 0.564);
    \draw[Virtual] (1.976, 0.564) -- (1.729, 0.564);
    \node[anchor=center] at (2.225, 0.661) {$\overset{N}\cdots$};
    \draw[Virtual] (2.717, 0.564) -- (2.47, 0.564);
    \draw[Virtual] (0.257, 0.564) arc[start angle=87.734, end angle=270, radius=0.247] -- (3.21, 0.071) arc[start angle=-90, end angle=90, radius=0.247];
    \draw[-mid] (0.494, 0.811) -- (0.494, 1.129);
    \draw[-mid] (1.482, 0.811) -- (1.482, 1.129);
    \draw[Virtual] (0.494, 0.564) -- (0.257, 0.564);
    \draw[-mid] (2.964, 0.811) -- (2.964, 1.129);
    \draw[Virtual] (3.21, 0.564) -- (2.964, 0.564);
    \draw[bevel, -mid] (0.494, 0) -- (0.494, 0.318);
    \draw[bevel, -mid] (1.482, 0) -- (1.482, 0.318);
    \draw[bevel, -mid] (2.964, 0) -- (2.964, 0.318);
    \whaMsymb{(0.494, 0.564)};
    \whaMsymb{(1.482, 0.564)};
    \whaMsymb{(2.964, 0.564)};
\end{tikzpicture}
        \end{array}
        \,,
\end{align}
if $\rho^{(N)}=(\rho^{(N)})^\dg\geq 0$.
More generally, one can include a boundary condition matrix $B$ in MPDO, 
\begin{equation}
\rho^{(N)}(M)_{B} :=
\begin{tikzpicture}[scale=1]
    \draw[Virtual] (2.223, 0.564) -- (1.729, 0.564);
    \draw[Virtual] (2.963, 0.564) -- (2.716, 0.564);
    \node[anchor=center] at (3.212, 0.661) {$\overset{N}\cdots$};
    \draw[Virtual] (3.704, 0.564) -- (3.457, 0.564);
    \draw[Virtual] (0.247, 0.564) arc[start angle=90, end angle=270, radius=0.247] -- (4.198, 0.071) arc[start angle=-90, end angle=90, radius=0.247];
    \draw[-mid] (1.482, 0.811) -- (1.482, 1.129);
    \draw[-mid] (2.469, 0.811) -- (2.469, 1.129);
    \draw[Virtual] (1.482, 0.564) -- (1.244, 0.564);
    \draw[-mid] (3.951, 0.811) -- (3.951, 1.129);
    \draw[Virtual] (4.198, 0.564) -- (3.951, 0.564);
    \draw[bevel, -mid] (1.482, 0) -- (1.482, 0.318);
    \draw[bevel, -mid] (2.469, 0) -- (2.469, 0.318);
    \draw[bevel, -mid] (3.951, 0) -- (3.951, 0.318);
     \whaMsymb{(1.482, 0.564)};
     \whaMsymb{(2.469, 0.564)};
     \whaMsymb{(3.951, 0.564)};
    \draw[Virtual] (1.235, 0.564) -- (0.741, 0.564);
    \filldraw[ultra thin, fill=white] 
    (0.494, 0.564) circle[radius=0.247];
    \node[anchor=center] at (0.494, 0.564) {$B$};
\end{tikzpicture}
~.
\end{equation}
Given $M$, deciding whether $\rho^{(N)}(M)$ is a positive operator for all system sizes is, in fact, an undecidable problem~\cite{kliesch2014matrix,de2016fundamental}. Nevertheless, as we will see, it is possible to bypass this difficulty by assuming an RFP condition.

What kind of mixed states admit an MPDO description? 
There are at least two physical scenarios where MPDOs appear: 
\begin{enumerate}
    \item Gibbs states of local Hamiltonians  \cite{beigi2011classification,molnar2015approximating,chen2020matrix};
    \item Boundaries of two-dimensional topologically ordered systems \cite{schuch2013topological,bultinck2017anyons}. 
\end{enumerate}
Although a complete classification of mixed-state quantum phases is still ongoing, it is believed that the renormalization fixed points of MPDOs capture mixed-state phases of matter. In the following, we introduce RFPs for MPDO, distinguishing between simple and non-simple RFP characterizations. 


\subsubsection{Renormalization fixed point}
One natural way to generalize the definition of renormalization to MPDO is
to define coarse-graining channels which map two consecutive tensors to one and also maintain long-range correlation. 
Unlike the renormalization flow for MPS, which is always reversible, the existence of another quantum channel that can reverse the coarse-graining is not guaranteed. We require the existence of such a channel to restore reversibility.
Unlike MPS, devising an explicit algorithm to find the coarse-graining quantum channel of a given MPDO tensor $M$ is still challenging (see a recent progress in \cite{kato2024exact}). Nevertheless, one can still define the renormalization fixed point for MPDO as the limit of renormalization flow \cite{cirac2017matrix}.
\begin{defn}[RFP for MPDO]\label{def:RFP_MPDO}
    A tensor $M$ generating MPDO is called a \emph{renormalization fixed point} (RFP) if there exist two quantum channels, $\mT$ and $\mS$, that act on the physical space and fulfill the following condition: 
    \begin{align}
        \begin{array}{c}
        \begin{tikzpicture}[scale=1.,baseline={([yshift=-0.75ex] current bounding box.center)}
        ]
       \whaM{(0,0)}{2};
        \whaM{(3*\whadx,0)}{2};
        \whaM{(4*\whadx,0)}{2};
        \draw [->, thick](1.,0.4) to [out=45,in=135] (3*\whadx-1.,0.4);
        \draw (1.5*\whadx,0.85) node {$\mT$};
        \draw [->, thick](3*\whadx-1,-0.4) to [out=225,in=-45] (1.,-0.4);
        \draw (1.5*\whadx,-0.32) node {$\mS$};
        \end{tikzpicture}
        \end{array}
        \,,
    \end{align}
    namely,
    \begin{equation}
        \begin{aligned}
            &\mT(M_{\alpha\beta})=\sum_\gamma M_{\alpha\gamma}\otimes M_{\gamma\beta}\\
            &\mS(\sum_\gamma M_{\alpha\gamma}\otimes M_{\gamma\beta})=M_{\alpha\beta}. 
        \end{aligned}
    \end{equation}
\end{defn}
One can show this definition agrees with \cref{def:RFP-pure} when the MPDO is pure. 

Next, we would like to characterize the RFP for MPDO. By treating indices $i,j$ as a compound index (namely, by vectorizing the MPDO), $M$ can be viewed as a rank-3 tensor, and one can apply the previous results for matrix product state to bring $M^{ij}$ into the canonical form with a BNT as in \cref{eqn:canonical,eqn:BNT}, namely,
\begin{align}
    &M^{ij}=\bigoplus_{\hbnt=1}^g (\mu_\hbnt\otimes M^{ij}_\hbnt),\\
    &\mathrm{where} \quad \mu_\hbnt=\mathrm{diag}\lbrace \mu_{\hbnt,1}, \mu_{\hbnt,2},\cdots,\mu_{\hbnt, n_\hbnt}\rbrace. 
\end{align}
We call this the \textit{horizontal canonical form} and denote the BNT as $\lbrace M_\hbnt\rbrace$. 

Consider evaluating a two-point correlation function, which
can be evaluated via the tensor-network transfer matrix. The transfer matrices $E$ for MPS and MPDO are defined as
\begin{align}
    \mathrm{MPS}:E=
    \begin{array}{c}
        \begin{tikzpicture}[scale=1.,baseline={([yshift=-0.75ex] current bounding box.center)}
        ]
        \whaATensor{(0,0)}{1}{$A$};
        \whaATensor{(0,\whadx*0.82)}{3}{$A^*$};
        \end{tikzpicture}
        \end{array},\quad
        \mathrm{MPDO}: E=
\begin{array}{c}
        \begin{tikzpicture}[scale=1.,baseline={([yshift=-0.75ex] current bounding box.center)}
        ]
         \whaMtrace{(0,0)}{$M$};
        \end{tikzpicture}
        \end{array}
        \,.
\end{align}
While all BNT element $A_\hbnt$ contribute to the transfer matrix for MPS, some BNT element $M_\hbnt$ for an MPDO may lead to nilpotent operators, meaning that $\tr_R(\rho^{(N)}(M_\hbnt))=0$ when we trace $R\leq D$ consecutive sites, and $N$ is sufficiently large. Therefore, such BNT elements do not contribute to the two-point correlation function when the separation of the two operators is larger than $R$. This motivates the following definition. 

\begin{defn}[Simple]
    A tensor generating an MPDO is \emph{simple} if none of the elements in its BNT is nilpotent. 
\end{defn}

Below, we give two examples of non-injective MPDO tensors; one is simple, and the other is non-simple. 

\begin{exmp}[Classical GHZ state] \label{ex:classical_ghz}
    Consider a tensor $M$ with $d=D=2$
    which has components $M_{00}^{00}=M_{11}^{11}=1$, and all other components
    vanishing. The corresponding MPDO is
\begin{equation}
    \rho^{(N)}(M)=|0\rangle\langle 0|^{\otimes N} + |1\rangle\langle 1|^{\otimes N}. 
\end{equation}
This is an example of a simple tensor. There are two elements in BNT, each with bond dimension 1: $M_{(\hbnt=0)}^{00}=1$ and $M_{(\hbnt=1)}^{11}=1$. None of them is nilpotent. Similar to the GHZ state, this MPDO has long-range correlations. Indeed, $\rho^{(N)}(M)$ can be obtained from the pure GHZ state by tracing out one site. 
\end{exmp}

\begin{exmp}[Boundary state of the Toric Code]\label{exmp:ToricCodeBd}
Consider a $d=D=2$ tensor $M$ with 
components $M_{00}^{00}=M^{11}_{00}=M^{00}_{11}=-M^{11}_{11}=1/2$ and all other components  vanishing.
The corresponding MPDO is
\begin{equation}
    \rho^{(N)}(M)=\frac{1}{2^N}(\bo^{\otimes N}+\sigma_z^{\otimes N}).
\end{equation}
This is an example of a non-simple tensor. There are two elements in BNT, each with bond dimension 1: $M_{(\hbnt=0)}^{00}=\frac{1}{2},M_{(\hbnt=0)}^{11}=\frac{1}{2}$ and $M_{(\hbnt=1)}^{00}=\frac{1}{2},M_{(\hbnt=1)}^{11}=-\frac{1}{2}$. $M_{(\hbnt=0)}$ generates $\bo^{\otimes N}$ and $M_{(\hbnt=1)}$, which is nilpotent, generates $\sigma_z^{\otimes N}$. After tracing any single site, the MPDO is proportional to a maximally mixed state and thus does not have long-range correlations. 
\end{exmp}

The
simple and non-simple MPDO RPF tensors have different structures, which we review below.

\subsubsection{Simple RFP characterization}
\begin{thm}
\label{thm:simple-RFP}
    If a tensor $M$ is in canonical form,  generates an MPDO, is an RFP, and is simple, then each BNT element $M_\hbnt$ takes the following form after a local
    unitary rotation: 
    \begin{align}
        \begin{array}{c}
        \begin{tikzpicture}[scale=1.,baseline={([yshift=-0.75ex] current bounding box.center)}
        ]
        \whaMorg{(0,0)}{2}{\small $M_a$}{whampdocolor};
        \end{tikzpicture}
        \end{array}=\bigoplus_{k\in S_\hbnt} 
        \begin{tikzpicture}[scale=.4,baseline={([yshift=-0.5ex] current bounding box.center)}]
        \whaGcTensor{(0,0)}{1.}{.5}{\small $l_k$}{1}{lcolor};
        \whaGcTensor{(\singledx,0)}{1.}{.5}{\small $r_k$}{-1}{lcolor};
        \end{tikzpicture}
        \,,
        \label{eqn:simple-RFP}
    \end{align}
    with
    \begin{align}
    \begin{array}{c}
        \begin{tikzpicture}[scale=.4,baseline={([yshift=-0.75ex] current bounding box.center)}
        ]
        \whauTensor{(0,0)}{1.2}{.5}{\singledx*0.5}{1.}{\small $\eta_{k,k'}$};
        \end{tikzpicture}
        \end{array}=
         \begin{tikzpicture}[scale=.4,baseline={([yshift=-0.75ex] current bounding box.center)}]
        \whaGcTensor{(0,0)}{1.}{.5}{\small $r_{k}$}{-1}{lcolor};
        \whaGcTensor{(\singledx,0)}{1.}{.5}{\small $l_{k'}$}{1}{lcolor};
        \end{tikzpicture}
        \geq 0,\quad \text{for}\quad k,k'\in S_\hbnt
        \,,
    \end{align}
    and the matrix $R$ defined below is a rank-one matrix
    \begin{equation}
        R_{k,k'}:=\tr(\eta_{k,k'})=a_{k} b_{k'}
    \end{equation}
    with normalization $\sum_{k\in S_\hbnt} a_k b_k = 1$.

    Different elements of a BNT are supported on orthogonal physical subspaces. They further satisfy $\eta_{k,k'}=0$ if $k\in S_\hbnt, k'\in S_{\hbnt'}$ and $\hbnt\neq \hbnt'$.  
\end{thm}

We emphasize that the direct sum in \cref{eqn:simple-RFP} is in the physical space (unlike \cref{eqn:canonical}, where the direct sum is in the auxiliary space). The local Hilbert space for each site is decomposed into a direct sum $\mathcal{H}=\oplus_\hbnt \mathcal{H}_\hbnt$.
Moreover, each $\mathcal{H}_\hbnt$ is further decomposed into another direct sum $\mathcal{H}_\hbnt=\oplus_{k\in S_\hbnt} \mathcal{H}^{(k)}$ with $\mathcal{H}^{(k)}=\mathcal{H}^{(k)}_l\otimes \mathcal{H}^{(k)}_r$. Denoting $d_l^{(k)},d_r^{(k)}$ the Hilbert space dimension of $\mathcal{H}^{(k)}_l,\mathcal{H}^{(k)}_r$ respectively, the physical Hilbert space dimension of a single site is $d=\sum_\hbnt\sum_{k\in S_\hbnt} d_l^{(k)} d_r^{(k)}$. 

By concatenating $M$, the MPDO generated by simple RFP tensor $M$ is thus a summation of mixed states, each of which is a direct product,
\begin{equation}
\begin{aligned}
    &\rho^{(N)}(M)=\\
    &\quad\bigoplus_{\hbnt=1}^g n_\hbnt \bigoplus_{k_1,\ldots,k_N\in S_\hbnt}\eta_{k_1,k_2}\otimes \eta_{k_2,k_3} \otimes\cdots\otimes \eta_{k_N,k_1}.
\end{aligned}
\label{eqn:simple-MPDO-RFP}
\end{equation}
We use the direct sum to emphasize different summands are supported on orthogonal subspaces. In fact, each summand labeled by $\hbnt$ in the above MPDO is the generalized Gibbs state of nearest-neighboring commuting Hamiltonian with zero correlation length,
\begin{equation}
    \rho^{(N)}\propto e^{-\sum_{i=1}^N h_i},
\end{equation}
where $h_i=\tau_i(h)$, $[h_i,h_j]=0$. By ``generalized'' here we mean that each $h$ includes the limiting case where $h=h_1+\beta h_2$ with $h_2\geq 0,[h_1,h_2]=0$ and $\beta\ra\infty$, so that $e^{-h}=P_2 e^{-h_1}$, with $P_2$ the projector onto the ground space of $h_2$. Therefore, the generalized Gibbs state can be rank-deficient.
Vice versa, the generalized Gibbs state of nearest-neighboring commuting Hamiltonian with zero correlation length can be rewritten in the form of \cref{eqn:simple-MPDO-RFP} \cite{cirac2017matrix,bravyi2005commutative,beigi2011classification}. 


The MPDO generated by the simple RFP tensor satisfies several properties physically expected for an RFP. First, it has zero correlation length $E^2=E$. Second, the conditional mutual information $I(A:C|B)$ vanishes for any consecutive regions $A,B,C$ when the complement of $ABC$ is non-empty, which directly leads to the above local Hilbert space decomposition \cite{hayden2004structure}. 

\begin{exmp}[Gibbs state]
    We provide two examples of simple MPDO RFPs in \cref{sec:example-simple}. The first one is a direct product of Werner states, therefore with a single $k$ in \cref{eqn:simple-RFP}: there is no further decomposition for the local Hilbert space at each site. The second one is constructed by two Werner states with different parameters, and has two $k$'s: the local Hilbert space is a direct sum $\mathcal{H} = \mathcal{H}^{(k=1)} \oplus \mathcal{H}^{(k=2)}$. For both cases, we show that they are indeed Gibbs states by constructing the Gibbs state Hamiltonians explicitly.
\end{exmp}

Finally, using the characterization of simple RFP, one can explicitly construct the renormalization channels $\mT$ and $\mS$. When the MPDO is injective, with only one horizontal BNT element, the channel $\mT$ that
transforms two sites to three sites and the channel $\mS$ that
transforms three sites to two sites read 
\begin{equation}
\begin{aligned}
    &\mT(X)=\\
    &\bigoplus_{k_1,k_2}\left[ \frac{1}{a_{k_1} b_{k_2}} \tr_{1r,2l}\left(Q_{k_1}\otimes Q_{k_2}(X)Q^\dg_{k_1}\otimes Q^\dg_{k_2}\right)\right.\\
    &\quad \left.\otimes\left(\bigoplus_{k'}(\eta_{k_1,k'}\otimes\eta_{k',k_2})\right)\right],
\end{aligned}
\label{eqn:inj-T}
\end{equation}
and
\begin{equation}
\begin{aligned}
    &\mS(X)=\bigoplus_{k_1,k_2,k_3}\left[\frac{1}{a_{k_1} b_{k_3}}\right.\\
    &\quad\tr_{1r,2,3l}\left(Q_{k_1}\otimes Q_{k_2}\otimes Q_{k_3}(X)Q^\dg_{k_1}\otimes Q^\dg_{k_2}\otimes Q^\dg_{k_3}\right)\\
    &\quad \otimes \eta_{k_1,k_3}\biggr].
\end{aligned}
\label{eqn:inj-S}
\end{equation}
Here $Q_k$ is an isometry $Q_k Q_k^\dg=\bo$ that projects onto the Hilbert subspace $\mathcal{H}^{(k)}$ and satisfies $\sum_k Q_k^\dg Q_k=\bo$. After applying the projector $Q_k$, the symbol $\tr_{1r}$ denotes taking partial trace on $\mathcal{H}^{(k)}_r$, and ``1'' labels the site of the action. Similarly for $\tr_{2l}$ the partial trace is taken on $\mathcal{H}_l^{(k)}$ on the second site. 
Intuitively, the quantum channels perform project, partial trace, and replace.
Graphically,
 \begin{align}
 \mT &=\bigoplus_{k_1,k_2,k'}\frac{1}{a_{k_1} b_{k_2}}\nonumber\\
 &\quad\quad
    \begin{array}{c}
        \begin{tikzpicture}[scale=.45,baseline={([yshift=-0.75ex] current bounding box.center)}
        ]
        \projector{(\singledx,0)}{1.3}{.5}{\singledx*0.5}{1.}{\small $\fQ_{k_1}$}{3};
        \projector{(\singledx*3,0)}{1.3}{.5}{\singledx*0.5}{1.}{\small $\fQ_{k_2}$}{4};
        \etatensor{(\singledx,2)}{1.3}{.5}{\singledx*0.5}{1.}{$\eta_{k_1,k'}$}
        \etatensor{(\singledx*3,2)}{1.3}{.5}{\singledx*0.5}{1.}{$\eta_{k',k_2}$}
        \end{tikzpicture}
        \end{array}
        \,,
    \end{align}
    and
    \begin{align}
 \mS&=\bigoplus_{k_1,k_2,k_3}\frac{1}{a_{k_1} b_{k_3}}\nonumber\\
 &\quad\quad
    \begin{array}{c}
        \begin{tikzpicture}[scale=.45,baseline={([yshift=-0.75ex] current bounding box.center)}
        ]
        \projector{(0,0)}{1.3}{.5}{\singledx*0.5}{1.}{\small $\fQ_{k_1}$}{5};
        \projector{(\singledx*4,0)}{1.3}{.5}{\singledx*0.5}{1.}{\small $\fQ_{k_3}$}{6};
        \projector{(\singledx*2,0)}{1.3}{.5}{\singledx*0.5}{1.}{\small $\fQ_{k_2}$}{2};
        \etatensor{(\singledx*2,2)}{1.3}{.5}{\singledx*0.5}{1.}{$\eta_{k_1,k_3}$}
        \end{tikzpicture}
        \end{array}
        \,,
    \end{align}
    where we employ the vectorized notation, using the upward arrows to denote the output Hilbert space degrees of freedom and the downward arrows to denote the input Hilbert space degrees of freedom. $\fQ_k$ in the graph thus denotes $Q_k\otimes Q_k^*$.

    We would like to comment that the channels $\mT$ and $\mS$ introduced above convert a two-site tensor and a three-site tensor into each other. By blocking the tensor once more, we can define new $\mT$ and $\mS$ that go between one-site tensor and two-site tensor; however, we will keep the form between two-site and three-site tensor due to its simplicity. 
    
    

\subsubsection{Non-simple RFP characterization}
In this section, we review the more general non-simple MPDO RFPs and their characterization \cite{cirac2017matrix}. The non-simple RFPs have several non-trivial features since they can emerge on the boundary of topologically ordered models in 2D. We will first briefly state the general characterization, then use the language of $C^*$-weak Hopf algebra to write down explicit RFP tensors $M$ and the corresponding MPDOs.  

For the general characterization, we need to use the canonical form with a BNT in the vertical, instead of a horizontal, direction. Starting from the rank-4 tensor $M$, we now treat indices $\alpha,\beta$ as a compound index and view $M$ as a rank-3 tensor. One can then apply the results for MPS (with the horizontal and vertical directions exchanged) to bring $M_{\alpha\beta}$ into canonical form with a corresponding BNT. We call this the \textit{vertical canonical form} and denote the BNT as $\lbrace M_\lambda \rbrace$. Moreover, it is proven that~\cite{cirac2017matrix} if an MPDO tensor is in the horizontal canonical form, it is automatically also in the vertical canonical form.

Consider the one-site tensor $M$ and the two-site tensor $K$ formed by concatenating two copies of $M$ horizontally,
\begin{align}
        \begin{array}{c}
        \begin{tikzpicture}[scale=1.,baseline={([yshift=-0.75ex] current bounding box.center)}
        ]
        \whaMorg{(0,0)}{2}{$K$}{whampdocolor};
        \end{tikzpicture}
        \end{array}
        :=
        \begin{array}{c}
        \begin{tikzpicture}[scale=1.,baseline={([yshift=-0.75ex] current bounding box.center)}
        ]
       \whaM{(0,0)}{2};
       \whaM{(\whadx,0)}{2};
        \end{tikzpicture}
        \end{array}
        \,.
    \end{align}
Up to isometries acting on the physical indices, they can be written in the vertical canonical form,
\begin{align}
    &M_{(\alpha\beta)}=\bigoplus_{\lambda} \mu_\lambda \otimes M_{(\alpha\beta),\lambda}\label{eqn:Mvertical}\\
    &K_{(\alpha\beta)}=\bigoplus_{\kappa} \nu_\kappa \otimes K_{(\alpha\beta),\kappa},\label{eqn:Kvertical}
\end{align}
where $\lbrace M_\lambda\rbrace$ and $\lbrace K_\kappa\rbrace$ are BNTs for $M$ and $K$, respectively, and $\mu_\lambda,\nu_\kappa$ are diagonal positive matrices. We define $m_\lambda=\tr(\mu_\lambda)$ and $n_\kappa=\tr(\nu_\kappa)$. Again, we emphasize the direct sum and direct product are in the physical space (vertical direction). 

\begin{thm}
A tensor $M$ generating an MPDO is an RFP iff the elements in BNTs $\lbrace M_\lambda\rbrace$ and $\lbrace K_\kappa\rbrace$ for one-site tensor $M$ and two-site tensor $K$ are related by a unitary $U_\lambda$ acting on the physical index
\begin{equation}
    K_{(\alpha\beta),\lambda }= U_\lambda M_{(\alpha\beta),\lambda} U_\lambda^\dg,
    \label{eqn:same-BNT}
\end{equation}
for any $\lambda$ and $m_\lambda=n_\lambda$~\cite{cirac2017matrix}.
\label{thm:non-simple}
\end{thm}

The multiplication in \cref{eqn:same-BNT} is
over the physical indices (along the vertical direction). In other words, at the RFP, the one-site tensor $M$ and two-site tensor $K$ share the same vertical BNT up to a unitary transformation. This feature can be used to construct the channels $\mT$ and $\mS$, and thereby the parent Lindbladian for the non-simple case in later sections. 

We comment that, in practice, it is more convenient to generate the non-simple MPDO RFP using $C^*$-weak Hopf algebras, as we introduce in \cref{subsubsec:nonSimpleRfpWha}; this construction also makes the MPDOs manifest as boundaries of 2D topologically ordered systems.
However, such constructions do not exhaust all non-simple RFPs~\cite{liu2025trading}. We will explicitly present a counterexample and construct its parent Lindbladian in~\cref{sec:beyond}. 

\subsection{Lindblad Master Equation}

Here we give a short introduction to the Lindbladian formalism.
The Linbladian operator $\mL$ is a linear map that generates a continuous dynamical semigroup of quantum channels $\mathcal{E}_{\text{evol},t}$ via $\mathcal{E}_{\text{evol},t}=e^{\mL t}$. Physically, the Lindblad master equation \cite{lindblad1976generators} governs the evolution of an open quantum system weakly coupled to the environment,
\begin{equation}
    \dot{\rho}:=\mathcal{L}(\rho)=-i[H,\rho]+\sum_k\left(L_k \rho L_k^\dg - \frac{1}{2}\lbrace L_k^\dg L_k,\rho\rbrace\right),
\end{equation}
where the Hamiltonian part $H$ governs the unitary part of the evolution, and $\lbrace L_k\rbrace$ are called jump operators, governing the dissipative part of the evolution. 
An initial density matrix $\rho_{\text{init}}$ evolves as $\rho(t)=e^{\mathcal{L}t}(\rho_{\text{init}})$. The evolution operator $e^{\mathcal{L}t}$ is completely positive and trace-preserving for all $t\ge0$. 

One can prove that the eigenvalues of $\mathcal{L}$ have a non-positive real part with at least one zero eigenvalue. In particular, the parent Lindbladian we construct takes the form of $
\mathcal{L}^{(N)}=\sum_{i=1}^N (\mathcal{E}_i-\bo)$, subsequently has $H=0$ and only jump operators, and all the non-zero eigenvalues have a strictly negative real part~\footnote{This is because we can rewrite $
\mathcal{L}^{(N)}=N(\mathcal{E}^{(N)}-\bo)$ where $\mathcal{E}^{(N)}=\frac{1}{N}\sum_{i=1}^N \mathcal{E}_i$, and the eigenvalues of the quantum channel $\mathcal{E}^{(N)}$ lie within the unit circle}. 

Given a Lindbladian operator, a steady state is a density matrix $\rho_\infty$ corresponding to the zero eigenvalue, satisfying $\mathcal{L}(\rho_\infty)=0$. Namely, it is a fixed point of the evolution operator $e^{\mathcal{L}t}(\rho_\infty)=\rho_\infty$. The set of steady states forms a convex set, and all steady states can be written as a convex combination of extreme points. In case there is a finite number of such points, we will call it steady-state degeneracy (SSD). If $\mathcal{L}$ has a unique steady state, any initial state $\rho_{\text{init}}$ will evolve towards the steady state after infinite time evolution, $\lim_{t\ra\infty} e^{\mathcal{L}t}(\rho_{\text{init}})=\rho_\infty$. If the steady state is not unique, the initial state will evolve towards a state in the convex set of steady states. 

\subsection{Mixed-state Quantum Phases}

For completeness, here we provide a definition of the mixed-state quantum phases via Lindbladian evolution~\cite{coser2019classification}. 

\begin{defn}
    A state $\rho_0$ can be driven fast to another state $\rho_1$, and we write $\rho_0\ra\rho_1$, if there exists a local and time-dependent Lindbladian $\mathcal{L}$ acting on the original system and a locality preserving ancillary system, such that for any $t\geq \text{poly} \log N$,
    \begin{equation}
        \|e^{t\mathcal{L}}(\rho_0\otimes\sigma_0)-\rho_1\otimes\sigma_1 \|_1 \leq \epsilon_N,
    \end{equation}
    with $\sigma_0$ and $\sigma_1$ respectively the initial and final states of the ancillas, 
    and $\epsilon_N\ra 0$ in the thermodynamic limit $N\ra \infty$. 

    Two states belong to the same phase if there exist two local Lindbladians described above such that $\rho_0\ra\rho_1$ and $\rho_1\ra \rho_0$.  
\end{defn}

Including the ancillary system in the definition is necessary for the phase equivalence to be transitive. 

In particular, if an MPDO $\rho^{(N)}$ admits a rapid-mixing parent Lindbladian that brings a product state $\rho_0$ to $\rho^{(N)}$, then it is in the trivial phase. Showing the reverse direction $\rho^{(N)}\ra \rho_0$ is easy by designing a parent Lindbladian for the product state. 

\section{Parent Lindbladian constructions}
\label{sec:general-comment}
After reviewing the characterization of the MPDO renormalization fixed point tensor, we are now ready to construct its
local parent Lindbladian. We use the channels $\mathcal{T}$ and $\mathcal{S}$ of the RFP to form a local quantum channel 
\begin{equation}
    \mathcal{E}=\mathcal{T}\circ \mathcal{S}
\end{equation}
and denote $\mathcal{E}_i=\tau_i(\mathcal{E})$, where $\tau_i$ translates the sites by amount $i$. We then construct the parent Lindbladian as,
\begin{align}
    \mathcal{L}_i&=\mathcal{E}_i-\bo,\\
    \mathcal{L}^{(N)}&=\sum_{i=1}^{N} \mathcal{L}_i. \label{eq:parent_lind_def}
\end{align}
In this construction, the jump operators of the local Lindbladians $\mathcal{L}_i$ are identified with the Kraus operators of $\mathcal{E}_i$ (note that $\sum_k L_k^\dg L_k=\bo$ due to the trace-preserving property of $\mathcal{E}_i$), and the Hamiltonian part of $\mathcal{L}_i$ is zero. The MPDO generated by $M$ is a steady state of $\mathcal{L}^{(N)}$ by construction, $\mathcal{L}^{(N)}(\rho^{(N)}(M))=0$, and the Lindbladian is geometrically local.

One may wonder whether we can find an easier solution by first constructing the parent Lindbladian for the purification of the mixed state and then taking a partial trace over the purifying system. There are two problems with this approach. First, not every MPDO allows for a local purification, as pointed out in \cite{de2016fundamental}. The second is that the quantum channel for partial trace cannot be generated by a Lindbladian. 
Therefore, this procedure fails to give a parent Lindbladian for the mixed state itself.

Next, we define several key properties for parent Lindbladians.

\begin{defn}[Frustration freeness]
    A Lindbladian is \emph{frustration-free} if it can be written as a sum of terms $\mathcal{L}^{(N)}=\sum_{i=1}^{N} \mathcal{L}_i$, such that every steady state of the full Lindbladian is the steady state of each individual term. Namely,
    \begin{equation}
        \mathcal{L}^{(N)}(\rho) = 0  \quad \Leftrightarrow \quad \mathcal{L}_i(\rho)=0.
    \end{equation}
\end{defn}

\begin{defn}[Minimal steady-state degeneracy]
    A parent Lindbladian $\mathcal{L}^{(N)}$ of an MPDO $\rho^{(N)}(M)$ reaches \emph{minimal steady-state degeneracy} (SSD) if there does not exist any
    other local Lindbladian
    with a lower-dimensional steady-state subspace which includes $\rho^{(N)}(M)$. 
\end{defn}

While the parent Lindbladian construction is straightforward, 
proving its properties, such as minimal steady-state degeneracy and frustration-freeness, is challenging. The rest of the text is devoted to proving the properties summarized in \cref{sec:sumres} for the parent Lindbladian defined in \cref{eq:parent_lind_def}.
There are at least two main differences compared to the parent Hamiltonian.

First, unlike the parent Hamiltonian, parent Lindbladian is not automatically frustration-free. Although the eigenvalues of $\mathcal{L}_i$ always have a non-positive real part, it is possible that a global steady-state is not a steady-state of each local term, as shown in the following counter-example,
\begin{exmp}
    Consider the following Lindbladian acting on a single qubit,
    \begin{equation}
        \begin{aligned}
            \mathcal{L}&=\mathcal{L}_1+\mathcal{L}_2\\
            \mathcal{L}_1(X)&=\tr(X)|0\rangle\langle 0|-X\\
            \mathcal{L}_2(X)&=\tr(X)|1\rangle\langle 1|-X.
        \end{aligned}
    \end{equation}
    The steady state of $\mathcal{L}$ is $\rho=\frac{1}{2}(|0\rangle\langle 0|+|1\rangle\langle 1|)$, and it is not a steady state of $\mathcal{L}_1$ or $\mathcal{L}_2$. 
\end{exmp}
However, if the parent Lindbladian is commuting, then frustration-freeness is guaranteed. 

\begin{lemma}
    A Lindbladian $\mathcal{L}=\sum_i \mathcal{L}_i$ with $[\mathcal{L}_i,\mathcal{L}_j]=0$ and $\mathcal{L}_i=\mathcal{E}_i-\bo$ is frustration-free. That is, if $\mathcal{L}(\rho)=0$ then $\mathcal{L}_i(\rho)=0$ for all $i$. 
\end{lemma}

\begin{proof}
    Using the local Lindbladian term, we define the projectors,
    \begin{equation}
    \label{eqn:projector-def}
    \begin{aligned}
         P_i&=\lim_{t\ra\infty} e^{t\mathcal{L}_i}\\
         P&=\lim_{t\ra\infty} e^{t\sum_i \mathcal{L}_i}=\prod_i P_i,
    \end{aligned}
    \end{equation}
    where we used $[\mathcal{L}_i,\mathcal{L}_j]=0$.
    The limit is well-defined because $\mathcal{L}_i$ does not have pure-imaginary eigenvalues by construction. $P_i$ is a projector $P_i^2=P_i$ but not necessarily Hermitian. Using $P_i^2=P_i$, one can show that $P_i$ as a linear map is non-defective: for each eigenvalue, the geometric multiplicity equals the algebraic multiplicity. Therefore, $P_i$ is diagonalizable. 

    The commutation of the Lindbladian $[\mathcal{L}_i,\mathcal{L}_j]=0$ leads to the commutation of the projectors $[P_i,P_j]=0$, and subsequently, they are simultaneously diagonalizable. Furthermore, the eigenvalues of $P_i$ can only be 0 or 1. For a steady state, $\mathcal{L}(\rho_\infty)=0$ leads to $P(\rho_\infty)=1$, and thus $P_i(\rho_\infty)=1$ for all $i$, equivalently,  $\mathcal{L}_i(\rho_\infty)=0$. This finishes the proof.
\end{proof}

Second, let the linear space of kernels of $\mathcal{L}_i$, which is equivalent to the linear space of fixed points of $\mathcal{E}$, be denoted as $\mathcal{F}_{\mathcal{E}}$. 
Define the vector space $G_L(M)$ as the space spanned by all states $\rho^{(L)}(M)_X$, which are states of $L$ consecutive sites generated by $M$ with arbitrary boundary condition $X$,
\begin{equation}
\label{eqn:GLM}
    \begin{aligned}
        &G_L(M) = \lbrace  \rho^{(L)}(M)_X|\forall X\in\mathcal{M}_{D}\rbrace
    \end{aligned}
\end{equation}
where
\begin{equation}
    \begin{aligned}
    &\rho^{(L)}(M)_X=\\
    &\sum_{\lbrace ij\rbrace}\tr(X M^{i_1 j_1}\cdots M^{i_L j_L})|i_1\cdots i_L\rangle\langle j_1\cdots j_L|,
    \end{aligned}
\end{equation}
and $L$ is the
interaction length of $\mathcal{E}$. Unlike parent Hamiltonians where $G_L=\mathrm{ker}(h)$, for a parent Lindbladian, we generally have 
\begin{equation}
    G_L(M)\subseteq \mathcal{F}_{\mathcal{E}}. 
\end{equation}
This is because the fixed-point space of a quantum channel
has a specific structure 
(See Theorem 6.14 of \cite{wolf2012quantum}) that might not be satisfied by $G_L(M)$ (see \cref{exm:two-werner} in \cref{sec:example-simple}).

Because of the above two differences, the techniques for proving parent Hamiltonian ground-state degeneracy
cannot be directly applied to determine the steady-state degeneracy of parent Lindbladians. In the following, we give explicit constructions of the channels $\mathcal{E}$ and show the above-constructed parent Lindbladian is still frustration-free and reaches the minimal steady-state degeneracy, along with other properties such as commuting, rapid-mixing, and
the nature of its correlations. We discuss the simple MPDO and non-simple MPDO in the next two sections, respectively. 

\section{Parent Lindbladian constructions: Simple MPDO}
\label{sec:pL-simple}

In this section, we explicitly construct the parent Lindbladians for simple MPDOs and prove their properties. We begin by considering the injective MPDO case, where the MPDO has only one BNT element in horizontal canonical form, and then extend to the non-injective MPDO case. The parent Lindbladian for injective MPDOs is commuting, has a unique steady state, and exhibits rapid mixing. For non-injective MPDOs, we prove a no-go theorem stating that the parent Lindbladian cannot be commuting. We also show that the steady-state degeneracy is minimal and equals $g$, the number of BNT elements in the horizontal canonical form of the tensor $M$.

\subsection{Injective MPDO}
Using the channels $\mT$ and $\mS$ in \cref{eqn:inj-T,eqn:inj-S}, the local channel $\mathcal{E}_i=\tau_i(\mT\circ \mS)$ acting on three consecutive sites $i,i+1,i+2$ reads
\begin{equation}
\begin{aligned}
    \mathcal{E}_i(X)&=\bigoplus_{k_1,k_3}\left[\frac{1}{a_{k_1} b_{k_3}}\right. \\
    &\quad \tr_{1r,2,3l}\left(Q_{k_1}\otimes \bo \otimes Q_{k_3}(X)Q^\dg_{k_1}\otimes \bo \otimes Q^\dg_{k_3}\right)\\
    &\quad \left.\otimes\left(\bigoplus_{k_2'}(\eta_{k_1,k_2'}\otimes\eta_{k_2',k_3})\right)\right].
\end{aligned}
\label{eqn:simplechannel}
\end{equation}
Graphically,
 \begin{align}
 \mathcal{E}_i&=\bigoplus_{k_1,k_3,k_2'}\frac{1}{a_{k_1} b_{k_3}}\nonumber\\
 &\quad\quad
    \begin{array}{c}
        \begin{tikzpicture}[scale=.45,baseline={([yshift=-0.75ex] current bounding box.center)}
        ]
        \projector{(0,0)}{1.3}{.5}{\singledx*0.5}{1.}{\small $\fQ_{k_1}$}{0};
        \projector{(\singledx*4,0)}{1.3}{.5}{\singledx*0.5}{1.}{\small $\fQ_{k_3}$}{1};
        \tracetensor{(\singledx*2,0)}{\singledx*0.5}{1.}
        \etatensor{(\singledx,2)}{1.3}{.5}{\singledx*0.5}{1.}{$\eta_{k_1,k_2'}$}
        \etatensor{(\singledx*3,2)}{1.3}{.5}{\singledx*0.5}{1.}{$\eta_{k_2',k_3}$}
        \end{tikzpicture}
        \end{array}
        \,,
    \end{align}
    where we employ the vectorized notation, using the upward arrows to denote the output Hilbert space degrees of freedom and the downward arrows to denote the input Hilbert space degrees of freedom.

    Using the graphical notation, it is easy to see the channels $\mathcal{E}_i$ commute. The product of neighboring channels is 
     \begin{align}
 \mathcal{E}_i\circ\mathcal{E}_{i+1}&=\bigoplus_{k_1,k_4,k_2',k_3'}\frac{1}{a_{k_1} b_{k_4}}\nonumber\\
 &
    \begin{array}{c}
        \begin{tikzpicture}[scale=.45,baseline={([yshift=-0.75ex] current bounding box.center)}
        ]
        \projector{(0,0)}{1.3}{.5}{\singledx*0.5}{1.}{\small $\fQ_{k_1}$}{0};
        \tracetensor{(\singledx*2,0)}{\singledx*0.5}{1.};
        \tracetensor{(\singledx*4,0)}{\singledx*0.5}{1.};
        \projector{(\singledx*6,0)}{1.3}{.5}{\singledx*0.5}{1.}{\small $\fQ_{k_4}$}{1};
        \etatensor{(\singledx,2)}{1.3}{.5}{\singledx*0.5}{1.}{$\eta_{k_1,k_2'}$}
        \etatensor{(\singledx*3,2)}{1.3}{.5}{\singledx*0.5}{1.}{$\eta_{k_2',k_3'}$}
        \etatensor{(\singledx*5,2)}{1.3}{.5}{\singledx*0.5}{1.}{$\eta_{k_3',k_4}$}
        \end{tikzpicture}
        \end{array}
        \,,
    \end{align}
    and similarly for $\mathcal{E}_{i+1}\circ\mathcal{E}_{i}$, 
    leading to
    \begin{align}
        \mathcal{E}_i\circ\mathcal{E}_{i+1}=\mathcal{E}_{i+1}\circ\mathcal{E}_{i},
    \end{align}
    and subsequently $[\mathcal{L}_i,\mathcal{L}_{i+1}]=0$. The non-overlapping Lindbladians trivially commute. Therefore, the parent Lindbladian is commuting for injective MPDO, which leads to frustration freeness. 

    Frustration freeness allows for a simple proof of the steady state uniqueness. Using frustration freeness, a global steady state $\rho$ with $\mathcal{L}(\rho)=0$ also satisfies $\mathcal{L}_i(\rho)=0$, or equivalently, $\mathcal{E}_i(\rho)=\rho$ for all $i$. Therefore, it must satisfy
    \begin{equation}
        \mathcal{E}_1\circ \mathcal{E}_2\circ\cdots\circ\mathcal{E}_N(\rho)=\rho. 
    \end{equation}
    The channel $\mathcal{E}_1\circ \mathcal{E}_2\circ\cdots\circ\mathcal{E}_N$ is nothing but the replacement channel, which can be shown graphically, for example, when $N=4$, 
    \begin{align}
 &\mathcal{E}_1\circ \mathcal{E}_2\circ\mathcal{E}_3\circ\mathcal{E}_4=\bigoplus_{k_1',k_2',k_3',k_4'}\nonumber\\
 &
    \begin{array}{c}
        \begin{tikzpicture}[scale=.45,baseline={([yshift=-0.75ex] current bounding box.center)}
        ]
        \tracetensor{(\singledx*0,0)}{\singledx*0.5}{1.};
        \tracetensor{(\singledx*2,0)}{\singledx*0.5}{1.};
        \tracetensor{(\singledx*4,0)}{\singledx*0.5}{1.};
        \tracetensor{(\singledx*6,0)}{\singledx*0.5}{1.};
        \etaright{(-\singledx*1,2)}{1.3}{.5}{\singledx*0.5}{1.}{}
        \etatensor{(\singledx,2)}{1.3}{.5}{\singledx*0.5}{1.}{$\eta_{k_1',k_2'}$}
        \etatensor{(\singledx*3,2)}{1.3}{.5}{\singledx*0.5}{1.}{$\eta_{k_2',k_3'}$}
        \etatensor{(\singledx*5,2)}{1.3}{.5}{\singledx*0.5}{1.}{$\eta_{k_3',k_4'}$}
        \etaleft{(\singledx*7,2)}{1.3}{.5}{\singledx*0.5}{1.}{$\eta_{k_4',k_1'}$}
        \end{tikzpicture}
        \end{array}
        \,.
    \end{align}
    This quantum channel takes the trace of the input state and replaces it with $\rho^{(N)}(M)$. The unique steady state of the parent Lindbladian is then $\rho_\infty=\rho^{(N)}(M)$.

\subsection{Rapid mixing dynamics}
\label{sec:rapid-mixing}

The parent Lindbladian of a simple MPDO has one additional nice property: it is a rapid-mixing. That is, starting from any initial $\rho_{\text{init}}$ and evolving for time $t$, the distance between the evolved state $e^{\mathcal{L}t}(\rho_{\text{init}})$ and the unique steady-state $\rho_\infty$ vanishes when $t =O(\text{poly}\log N)$ in the thermodynamic limit $N\ra\infty$~\cite{coser2019classification}. The rapid mixing dynamics is a direct consequence of the fact that the parent Lindbladian is commuting.

\begin{lemma}
    A local Lindbladian $\mathcal{L}^{(N)}=\sum_{i=1}^N \mathcal{L}_i$, where $[\mathcal{L}_i,\mathcal{L}_j]=0$ and $\mathcal{L}_i$ does not have pure-imaginary eigenvalues, has rapid mixing dynamics. 
\end{lemma}

Note that in our formalism, $\mathcal{L}_i$ does not have pure-imaginary eigenvalues by construction. 

\begin{proof}
    Since the Lindbladian is commuting, we define the projectors $P_i, P$ as in \cref{eqn:projector-def}. Consider the completely
bounded trace norm of the superoperator $e^{\mathcal{L}t}-P$,
    \begin{equation}
        \begin{aligned}
            &\|e^{\mathcal{L}t}-P \|_\lozenge \\& = \|\prod_i e^{\mathcal{L}_i t} -\prod_i P_i \|_\lozenge\\
            &\leq \sum_{i=1}^N \|e^{\mathcal{L}_1 t}e^{\mathcal{L}_2 t}\cdots e^{\mathcal{L}_{i-1} t}(e^{\mathcal{L}_i t}-P_i) P_{i+1}\cdots P_N \|_\lozenge\\
            &\leq \sum_{i=1}^N \|e^{\mathcal{L}_i t}-P_i\|_\lozenge,
        \end{aligned}
    \end{equation}
    where we use the triangle inequality, submultiplicativity of the completely bounded trace norm, and $\|\mathcal{E}\|_\lozenge=1$ for a quantum channel $\mathcal{E}$. 

    For each local term, one can use the Jordan decomposition of the superoperator to show that the norm is bounded by
    \begin{equation}
        \|e^{\mathcal{L}_i t}-P_i\|_\lozenge \leq  e^{-\gamma_i t}
    \end{equation}
    for some constant $\gamma_i$ since $\mathcal{L}_i$ acts locally on a Hilbert space of fixed dimensions. Subsequently, 
    \begin{equation}
        \|e^{\mathcal{L}t}-P \|_\lozenge\leq N e^{-\gamma t}
    \end{equation}
    where $\gamma=\min(\gamma_1,\gamma_2,\cdots, \gamma_N)$. 

    The error is then bounded by
    \begin{equation}
    \begin{aligned}
    \epsilon:&=\|e^{\mathcal{L}t}(\rho_{\text{init}})-\rho_\infty\|_1\\
        &=\|e^{\mathcal{L}t}(\rho_\init)-P(\rho_\init)\|_1\leq \| e^{\mathcal{L}t}-P\|_1 \|\rho_\init\|_1\\
        &=\| e^{\mathcal{L}t}-P\|_1\leq \| e^{\mathcal{L}t}-P\|_\lozenge\leq  N e^{-\gamma t}.
    \end{aligned}
    \end{equation}
    The norm in $\| e^{\mathcal{L}t}-P\|_1$ is the induced trace norm of the superoperator, which lower bounds the completely bounded trace norm.
    One can take $t=(\log N)^{1+\delta}$ for any positive $\delta$ and show $\epsilon\ra 0$ when $N\ra\infty$. This justifies that $\mathcal{L}$ is a rapid-mixing Lindbladian. 
\end{proof}

\subsection{Non-injective simple MPDO}
A natural guess for generalizing the local quantum channel for a non-injective MPDO is given by
\begin{equation}
\begin{aligned}
    &\mathcal{E}^{(\Eone)}_i(X)=\bigoplus_\hbnt \bigoplus_{k_1,k_3\in S_\hbnt}\frac{1}{a_{k_1} b_{k_3}} \\
    &\quad \tr_{1r,2,3l}\left(Q_{k_1}\otimes P_\hbnt\otimes Q_{k_3}(X)Q^\dg_{k_1} \otimes \bo \otimes Q^\dg_{k_3}\right)\\
    &\quad \otimes\left(\bigoplus_{k_2'\in S_\hbnt}(\eta_{k_1,k_2'}\otimes\eta_{k_2',k_3})\right),
\end{aligned}
\end{equation}
where $P_\hbnt=\sum_{k\in S_\hbnt}Q_k^\dg Q_k$ is the projector to Hilbert subspace corresponding to BNT element $\hbnt$ and satisfy $\sum_\hbnt P_\hbnt = \bo$.
By direct inspection, one can see $\mathcal{E}_i^{(\Eone)}$ as a linear map is completely positive (CP) and $\rho^{(N)}(M)$ in \cref{eqn:simple-MPDO-RFP} is a fixed point of this CP map. 
However, a valid quantum channel must also be trace-preserving (TP), while
\begin{equation}
\begin{aligned}
    \tr\left(\mathcal{E}_i^{(\Eone)} (X)\right)&=\tr\left(\left(\sum_\hbnt P_\hbnt\otimes P_\hbnt\otimes P_\hbnt\right) X\right)\\
    &\neq \tr(X).
\end{aligned}
\end{equation}
To form a CPTP map, denoting $P=\sum_\hbnt P_\hbnt\otimes P_\hbnt\otimes P_\hbnt$ and $P^\perp=\bo-P$, we construct the local quantum channel $\mathcal{E}_i$ as
\begin{align}
    &\mathcal{E}_i=\mathcal{E}_i^{(\Eone)}+\mathcal{E}_i^{(\Etwo)},\\
    &\mathrm{where}\quad \mathcal{E}_i^{(\Etwo)}(X)=\tr(P^\perp X) \rho_0,
\end{align}
and $\rho_0$ is a density matrix satisfying 
\begin{equation}
    \sum_\hbnt (P_\hbnt\otimes P_\hbnt\otimes P_\hbnt) \rho_0 (P_\hbnt\otimes P_\hbnt\otimes P_\hbnt)=\rho_0. 
\end{equation}
Note that this condition is stronger than $P\rho_0 P = \rho_0$. One can check that the second piece $\mathcal{E}_i^{(\Etwo)}$ is still a CP map, and the combined map $\mathcal{E}_i$ is indeed trace-preserving. 

By construction, the second piece $\mathcal{E}_i^{(\Etwo)}$ satisfies $\mathcal{E}_i^{(\Etwo)}\circ\mathcal{E}_i=0$, which leads to $\mathcal{E}_i\circ\mathcal{E}_i=\mathcal{E}^{(\Eone)}_i\circ\mathcal{E}_i$. Using the graphical representation, one can show that $\mathcal{E}_i^{(\Eone)}$ is idempotent, $\mathcal{E}_i^{(\Eone)}\circ \mathcal{E}_i^{(\Eone)}=\mathcal{E}_i^{(\Eone)}$. Combing the above properties, one can show the equivalence
\begin{equation}
    \mathcal{E}_i(X)=X\quad \Leftrightarrow\quad \mathcal{E}_i^{(\Eone)}(X)=X.
    \label{eqn:fixed-point-on-one}
\end{equation}
The fixed point space of local channel $\mathcal{E}_i$ is the same as the fixed point space of the CP map $\mathcal{E}_i^{(\Eone)}$. 

The proof of frustration-freeness is more involved, since $\mathcal E_i$ is no longer commuting. As we will show later, this is not merely an artifact of our choice of $\mathcal E_i^{(\Etwo)}$; in fact, non-injective simple MPDO RFPs never admit a commuting local parent Lindbladian (see \cref{sec:simple-no-go}).
To prove the parent Lindbladian is frustration-free, we make a convenient choice of $\rho_0$ as
\begin{equation}
    \rho_0 = \frac{1}{g}\bigoplus_{\hbnt=1}^g \bigoplus_{k_1 k_2 k_3\in S_\hbnt}\eta_{k_1,k_2}\otimes \eta_{k_2,k_3} \otimes \eta_{k_3,k_1}. 
    \label{eqn:simplerho0}
\end{equation}
Under this condition, $\mathcal{E}_i^{(\Eone)}(\rho_0)=\rho_0$ and
\begin{equation}
    \mathcal{E}_i^{(\Eone)}\circ \mathcal{E}_i^{(\Etwo)}=\mathcal{E}_i^{(\Etwo)}.
\end{equation}
Combined with the properties
\begin{align}
    \mathcal{E}_i^{(\Eone)}\circ \mathcal{E}_i^{(\Eone)}&=\mathcal{E}_i^{(\Eone)}\\
    \mathcal{E}_i^{(\Eone)}\circ \mathcal{E}_{i+1}^{(\Eone)}&=\mathcal{E}_{i+1}^{(\Eone)}\circ\mathcal{E}_i^{(\Eone)},
\end{align}
(while keeping in mind that $\mathcal{E}_i\circ \mathcal{E}_{i+1}\neq\mathcal{E}_{i+1}\circ\mathcal{E}_i$) one can prove frustration freeness using a ``flow diagram'' method. We leave the proof for \cref{sec:proof-ff}.

With frustration-freeness, steady-state degeneracy can be shown using a method similar to simple MPDOs. A global steady state satisfies $\mathcal{E}_i(\rho)=\rho$, or equivalently, $\mathcal{E}^{(\Eone)}_i(\rho)=\rho$ for all $i$. Consequently, it must satisfy
    \begin{equation}
        \mathcal{E}^{(\Eone)}_1\circ \mathcal{E}^{(\Eone)}_2\circ\cdots\circ\mathcal{E}^{(\Eone)}_N(\rho)=\rho. 
    \end{equation}
The CP map $ \mathcal{E}^{(\Eone)}_1\circ \mathcal{E}^{(\Eone)}_2\circ\cdots\circ\mathcal{E}^{(\Eone)}_N$ can similarly be found using graphical representation, where now
\begin{align}
 \mathcal{E}^{(\Eone)}_i&=\bigoplus_\hbnt \bigoplus_{k_1,k_3,k_2'\in S_\hbnt}\frac{1}{a_{k_1} b_{k_3}}\nonumber\\
 &\quad\quad
    \begin{array}{c}
        \begin{tikzpicture}[scale=.45,baseline={([yshift=-0.75ex] current bounding box.center)}
        ]
        \projector{(0,0)}{1.3}{.5}{\singledx*0.5}{1.}{\small $\fQ_{k_1}$}{0};
        \projector{(\singledx*4,0)}{1.3}{.5}{\singledx*0.5}{1.}{\small $\fQ_{k_3}$}{1};
        \projector{(\singledx*2,0)}{1.3}{.5}{\singledx*0.5}{1.}{\small $P_{\hbnt}$}{2};
        \etatensor{(\singledx,2)}{1.3}{.5}{\singledx*0.5}{1.}{$\eta_{k_1,k_2'}$}
        \etatensor{(\singledx*3,2)}{1.3}{.5}{\singledx*0.5}{1.}{$\eta_{k_2',k_3}$}
        \end{tikzpicture}
        \end{array}
        \,.
    \end{align}
Let's take $N=4$ as an example, 
 \begin{align}
 &\mathcal{E}^{(\Eone)}_1\circ \mathcal{E}^{(\Eone)}_2\circ\mathcal{E}^{(\Eone)}_3\circ\mathcal{E}^{(\Eone)}_4=\bigoplus_\hbnt \bigoplus_{k_1',k_2',k_3',k_4'\in S_\hbnt}\nonumber\\
 &
    \begin{array}{c}
        \begin{tikzpicture}[scale=.45,baseline={([yshift=-0.75ex] current bounding box.center)}
        ]
        \projector{(\singledx*0,0)}{1.3}{.5}{\singledx*0.5}{1.}{\small $P_{\hbnt}$}{2};
        \projector{(\singledx*2,0)}{1.3}{.5}{\singledx*0.5}{1.}{\small $P_{\hbnt}$}{2};
        \projector{(\singledx*4,0)}{1.3}{.5}{\singledx*0.5}{1.}{\small $P_{\hbnt}$}{2};
        \projector{(\singledx*6,0)}{1.3}{.5}{\singledx*0.5}{1.}{\small $P_{\hbnt}$}{2};
        \etaright{(-\singledx*1,2)}{1.3}{.5}{\singledx*0.5}{1.}{}
        \etatensor{(\singledx,2)}{1.3}{.5}{\singledx*0.5}{1.}{$\eta_{k_1',k_2'}$}
        \etatensor{(\singledx*3,2)}{1.3}{.5}{\singledx*0.5}{1.}{$\eta_{k_2',k_3'}$}
        \etatensor{(\singledx*5,2)}{1.3}{.5}{\singledx*0.5}{1.}{$\eta_{k_3',k_4'}$}
        \etaleft{(\singledx*7,2)}{1.3}{.5}{\singledx*0.5}{1.}{$\eta_{k_4',k_1'}$}
        \end{tikzpicture}
        \end{array}
        \,,
    \end{align}
and the generalization to arbitrary $N$ is straightforward. Therefore, the global steady-states take the form of 
\begin{equation}
\begin{aligned}
    &\rho_\infty^{(N)}=\\
    &\quad\bigoplus_{\hbnt=1}^g c_\hbnt \bigoplus_{k_1,\ldots,k_N\in S_p}\eta_{k_1,k_2}\otimes \eta_{k_2,k_3} \otimes\cdots\otimes \eta_{k_N,k_1},
\end{aligned}
\end{equation}
with $c_\hbnt\in\mathbb{C}$. A direct inspection shows all the states in the above form are steady states of $\mathcal{L}$. This concludes that the steady-state degeneracy (SSD) is $g$, the number of horizontal BNT elements, i.e.,
\begin{equation}
    (\text{SSD})=g.
\end{equation}
Moreover, the minimally possible steady-state degeneracy is $g$, as proven below.

\begin{thm}
\label{thm:minimal-SSD}
    Consider a simple MPDO $\rho^{(N)}(M)$ generated by the RFP tensor $M$ with $g$ BNT elements in the horizontal canonical form. A local Lindbladian $\mL^{(N)}$ that has $\rho^{(N)}(M)$ as its steady state must have at least $g$ independent steady states. In particular, the MPDO $\rho^{(N)}(M_\hbnt)$  generated by each BNT element is a steady state. 
\end{thm}

\begin{proof}
    Consider the vectoried state $|\rho(M)\rangle\!\rangle$ of $\rho(M)$ and define an operator in the vectorized space $H_\mL=\hat{\mL}^\dg \hat{\mL}$, where $\hat{\mL}$ is the representation of $\mL$ in the vectorized space. By construction, $H_\mL$ is positive semi-definite, satisfies $\langle\!\langle\rho(M)|H_\mL|\rho(M)\rangle\!\rangle=0$, and it is $k$-local (but not geometrically local) because each term in the expansion $H_\mL=\sum_{ij}\hat{\mL}^\dg_i \hat{\mL}_j$ is only supported on a finite number of sites.

    Now, if there are $g$ BNT elements in the horizontal canonical form of $M$, one can decompose $|\rho(M)\rangle\!\rangle=\sum_{\hbnt=1}^g |\rho(M_\hbnt)\rangle\!\rangle$, where $\langle\!\langle\rho(M_\hbnt)|\rho(M_{\hbnt'})\rangle\!\rangle=0$ for $\hbnt\neq \hbnt'$. We construct another vectorized state \[
    \begin{aligned}
&|\rho'(M)\rangle\!\rangle=2|\rho(M_{\hbnt=1})\rangle\!\rangle-|\rho(M)\rangle\!\rangle\\
&\quad=|\rho(M_{\hbnt=1})\rangle\!\rangle-|\rho(M_{\hbnt=2})\rangle\!\rangle-\cdots-|\rho(M_{\hbnt=g})\rangle\!\rangle.
    \end{aligned}
    \]
   Consider any region $A$ and its complement
   $\bar{A}$. By construction, 
   \[
   \tr_{\bar{A}}(|\rho(M)\rangle\!\rangle\langle\!\langle\rho(M))|=\tr_{\bar{A}}(|\rho'(M)\rangle\!\rangle\langle\!\langle\rho'(M)|),
   \]
   where we use the fact that different BNT elements are supported on orthogonal physical subspaces locally (\cref{thm:simple-RFP}). 
   This shows that each term $\hat{\mL}_i^\dg \hat{\mL}_j$ in $H_\mL$, being $k$-local, cannot distinguish $|\rho(M)\rangle\!\rangle$ and $|\rho'(M)\rangle\!\rangle$, and therefore $|\rho'(M)\rangle\!\rangle$ also fulfills $\langle\!\langle \rho'(M)|H_\mL|\rho'(M)\rangle\!\rangle=0$. Equivalently, $\hat{\mathcal{L}}|\rho'(M)\rangle\!\rangle=0$ and back to the unvectorized notation, $\mathcal{L}(\rho'(M))=0$. 

Since both $\rho(M)$ and $\rho'(M)$ are steady states of $\mL$, their summation, $\rho(M_{\hbnt=1})$, is also a steady state of $\mL$. This can be applied to any BNT element $\hbnt$ and finishes the proof that each $\rho(M_\hbnt)$ is a steady state of $\mL$. The steady-state degeneracy is at least $g$.
\end{proof}


This shows that our parent Lindbladian construction reaches the minimal steady-state degeneracy. Below we show an example of parent Lindbladian for the non-injective simple MPDO, the classical GHZ state.

\begin{exmp}[Classical GHZ state]
\label{ex:pLGHZ}
Let us consider again
\begin{equation}
    \rho^{(N)}(M)=|0\rangle\langle 0|^{\otimes N} + |1\rangle\langle 1|^{\otimes N}. 
\end{equation}
(see also \cref{ex:classical_ghz}).
In this example, there are two elements in the BNT, each with a bond dimension of 1. Therefore, it is sufficient to consider a two-site local channel rather than a three-site one. The two projectors for the two BNT elements are 
\[
P_{0}=\begin{pmatrix}
    1 & 0\\
    0 & 0
\end{pmatrix},\quad P_{1}=\begin{pmatrix}
    0 & 0\\
    0 & 1
\end{pmatrix}.
\]
Denoting $P^\perp=\bo-P_0\otimes P_0-P_1\otimes P_1$, the local channel takes the
form
\begin{equation}
\begin{aligned}
    \mathcal{E}_i(X)&=\tr\left((P_0\otimes P_0) X\right) P_0\otimes P_0 \\
    & + \tr\left((P_1\otimes P_1) X\right) P_1\otimes P_1 + \tr\left((P^\perp )X\right) \rho_0, 
\end{aligned}
\end{equation}
and $\rho_0$ can be chosen as $\rho_0=\frac{1}{2}(P_0\otimes P_0+P_1\otimes P_1)$. Using this local channel,
the parent Lindbladian has a two-dimensional space of steady states: 
\[
\rho_\infty^{(N)}=c_1 |0\rangle\langle 0|^{\otimes N} + c_2 |1\rangle\langle 1|^{\otimes N},
\]
with $c_1,c_2\geq 0$. 
\end{exmp}
    
\subsection{No-go theorem for commuting local Lindbladian}
\label{sec:simple-no-go}
In this section, we prove a no-go theorem regarding the existence of commuting local Lindbladians with minimal steady-state degeneracy. We first prove a lemma on long-range correlations of non-injective simple MPDO.
\begin{lemma}
\label{lemma:LRC}
    Consider a non-injective simple MPDO RFP $\rho^{(N)}(M)$ generated by the tensor $M$ with horizontal BNT $\lbrace M_\hbnt\rbrace$. The state $\rho^{(N)}=c_a \rho^{(N)}(M_\hbnt)+c_{\hbnt'}\rho^{(N)}(M_{\hbnt'})$ with $\hbnt\neq\hbnt', c_\hbnt>0, c_{\hbnt'}>0$ has long-range correlations, in the sense that there exist local operators $O^A,O^B$ such that the connected correlation function $\langle O^A O^B\rangle_c$ has a positive lower bound for arbitrarily large separation between region $A$ and region $B$.
\end{lemma}

\begin{proof}
We will prove the lemma by construction. 
    For a non-injective simple MPDO RFP $\rho^{(N)}(M)$ generated by tensor $M$ with horizontal BNT $\lbrace M_\hbnt\rbrace$, denote the transfer matrix for a BNT element $M_\hbnt$ as $E_\hbnt$, and the transfer matrix for $M_\hbnt$ with operator $O^A_{\hbnt}$ inserted as $E_\hbnt^{A}$. Diagrammatically,
    \begin{equation}
        E_\hbnt=\begin{array}{c}
        \begin{tikzpicture}[scale=1.,baseline={([yshift=-0.75ex] current bounding box.center)}
        ]
          \whaMtrace{(0,0)}{\small $M_\hbnt$};
        \end{tikzpicture}
        \end{array},\quad E_\hbnt^{A}=\begin{array}{c}
        \begin{tikzpicture}[scale=1.,baseline={([yshift=-0.75ex] current bounding box.center)}
        ]
        \whaMtraceOp{(0,0)}{\small $M_\hbnt$}{\small $O^A_{\hbnt}$};
        \end{tikzpicture}
        \end{array}
    \end{equation}
    The operator $O^A_{\hbnt}$ is an onsite operator supported on $\mathcal{H}_\hbnt$. 

    Next, we construct $O^A$ as $O^A=O^A_\hbnt-O^A_{\hbnt'}$ where $\hbnt\neq \hbnt'$, namely, $O^A_\hbnt$ and $O^A_{\hbnt'}$ are supported on orthogonal physical spaces.  
    The one-point function $\langle O^A\rangle$ can be computed as
    \begin{equation}
        \begin{aligned}
            &\langle O^A\rangle=\tr(O^A \rho^{(N)})\\
            &=c_\hbnt\tr(O^A_\hbnt\rho^{(N)}(M_\hbnt))-c_{\hbnt'}\tr(O^A_{\hbnt'}\rho^{(N)}(M_{\hbnt'}))\\
            &=c_\hbnt \tr(E_\hbnt E_\hbnt^A)-c_{\hbnt'}\tr(E_{\hbnt'}E_{\hbnt'}^A),
        \end{aligned}
    \end{equation}
    where the fixed point condition $E_\hbnt^2=E_\hbnt$ is applied. Similarly,
    \begin{equation}
        \langle O^B\rangle=c_\hbnt \tr(E_\hbnt E_\hbnt^B)-c_{\hbnt'}\tr(E_{\hbnt'}E_{\hbnt'}^B).
    \end{equation}
    One can always choose $O^A_\hbnt,O^A_{\hbnt'},O^B_{\hbnt},O^B_{\hbnt'}$ such that the one-point functions vanish, $\langle O^A\rangle=0$ and $\langle O^B\rangle=0$.

    For $O^A$ and $O^B$ supported on different sites, one can compute the two-point function as
    \begin{equation}
        \begin{aligned}
            &\langle O^A O^B\rangle=\tr(O^A O^B\rho^{(N)})\\
            &=c_\hbnt\tr(O^A_\hbnt O^B_\hbnt \rho^{(N)}(M_\hbnt))+c_{\hbnt'}\tr(O^A_{\hbnt'} O^B_{\hbnt'}\rho^{(N)}(M_{\hbnt'}))\\
            &=c_\hbnt \tr(E_\hbnt E_\hbnt^A E_\hbnt E_\hbnt^B)+c_{\hbnt'}\tr(E_{\hbnt'}E_{\hbnt'}^A E_{\hbnt'}E_{\hbnt'}^B),
        \end{aligned}
    \end{equation}
    where in the second line, there are no cross terms like $O_\hbnt^A O_{\hbnt'}^B$ because different BNT elements are supported on orthogonal physical subspaces locally (\cref{thm:simple-RFP}). Since $E_\hbnt$ is a rank-one matrix, one can further decompose
    \begin{equation}
        \begin{aligned}
&\langle O^A O^B\rangle\\
&=c_\hbnt \tr(E_\hbnt E_\hbnt^A)\tr(E_\hbnt E_\hbnt^B)+c_{\hbnt'}\tr(E_{\hbnt'}E_{\hbnt'}^A)\tr(E_{\hbnt'}E_{\hbnt'}^B)\\
&=c_\hbnt (1+\frac{c_\hbnt}{c_{\hbnt'}})\tr(E_\hbnt E_\hbnt^A)\tr(E_\hbnt E_\hbnt^B),
    \end{aligned}
    \end{equation}
    where in the last line, we use $\langle O^A\rangle=\langle O^B\rangle=0$. Since the one-point functions vanish, the connected two-point correlation function $\langle O^A O^B\rangle_c=\langle O^A O^B\rangle$, which is a positive finite value, independent of the separation between $A$ and $B$. 
\end{proof}

The above proof essentially shows that the non-injective MPDO RFPs have long-range correlations. The fact that different horizontal BNT elements are supported on orthogonal physical subspaces plays an important role in the above proof. This feature is true only for simple MPDOs. Non-simple MPDOs do not have this property and a state $\rho^{(N)}=c_a \rho^{(N)}(M_\hbnt)+c_{\hbnt'}\rho^{(N)}(M_{\hbnt'})$ may only have short-range correlations.

We are now ready to prove the following no-go theorem for commuting local Lindbladian. 

\begin{thm}
 A non-injective simple MPDO generated by an RFP tensor $M$ does not admit a local Lindbladian that has $\rho^{(N)}(M)$ as its steady state, is commuting, has minimal steady-state degeneracy, and has no purely imaginary eigenvalues.
\end{thm}

\begin{proof}
    We will use proof by contradiction. For this, we assume that there exists a local and commuting parent Lindbladian that has $\rho^{(N)}(M)$ as its steady state, and has $\mathrm{(SSD)}=g$ (recall that by \cref{thm:minimal-SSD}, the minimal steady-state degeneracy is $g$ and the space of steady states is spanned by $\lbrace \rho(M_p)\rbrace$). More precisely, consider a commuting local Lindbladian $\mathcal{L}^{(N)}=\sum_{i=1}^{N}\mathcal{L}_i$, where each summand $\mathcal{L}_i$ is local, acting on a region of length $L$, and $L$ is a finite number that does not grow with $N$. Since $[\mathcal{L}_i,\mathcal{L}_j]=0$, the time evolution operator $\mathcal{E}^{(N)}_{\mathrm{evol},t}=e^{\mathcal{L}^{(N)}t}$ can be written as the product 
    \begin{equation}
        \mathcal{E}^{(N)}_{\mathrm{evol},t}=\prod_{i=0}^{N-1}e^{\mathcal{L}_i t}=\prod_{l=0}^{L-1}\left(\prod_{j=0}^{N/L-1}e^{\mathcal{L}_{jL+l}t}\right).
    \end{equation}
    In other words, $\mathcal{E}^{(N)}_{\mathrm{evol},t}=e^{\mathcal{L}^{(N)}t}$ is a finite-depth local channel with depth $L$. 

    As a finite-depth local channel, $\mathcal{E}^{(N)}_{\mathrm{evol},t}$ has a strict light cone, and no correlation can be created outside the light cone. This also applies when $t\ra\infty$; recall that the limit exists and is equal to the steady-state projector $\mathcal{E}_{\mathrm{evol},\infty}$ since $\mL$ has no purely imaginary eigenvalues. Specifically, starting from an initial state $\rho^{(N)}_{\text{init}}$ which is a product state, the connected correlation function with $O(t)=\tr(\mathcal{E}_{\mathrm{evol},t}(\rho^{(N)}_{\text{init}})O)$ satisfies
    \begin{equation}
    (O_A O_B)_c(t)\equiv (O_A O_B)(t)-O_A(t) O_B(t)=0
    \label{eqn:Lieb-Robinson}
    \end{equation}
    when the separation between region $A$ and $B$ is larger than $2L$ for any operators $O_A,O_B$. Consequently, the steady state after infinite time evolution $\rho_\infty^{(N)}=\mathcal{E}_{\mathrm{evol},\infty}(\rho^{(N)}_{\text{init}})$ is still short-range correlated. 

    Now, consider initial states as product states,
    \begin{equation}
        \rho_{\init}^{(N)}=\rho_1\otimes \rho_2\otimes\cdots\otimes\rho_N,
    \end{equation}
    where $\rho_i$ is an arbitrary density matrix at site $i$. By assumption, the steady state is in the space spanned by $\lbrace \rho(M_\hbnt)\rbrace$. Suppose there exists a product state whose steady state is a superposition $c_1 \rho(M_\hbnt)+c_2 \rho(M_{\hbnt'})$ with $\hbnt\neq \hbnt'$. By \cref{lemma:LRC}, this state has long-range correlations and is in contradiction to \cref{eqn:Lieb-Robinson}. Therefore, the steady state of a product initial state does not contain superposition. 
    
    Since all possible product states $\rho_{\init}^{(N)}$ form a complete basis for the operators, all the states in the space of steady states shall be reached, and there must exist two initial states,
$\rho^{(N)}_{\text{init},\hbnt}$ and $\rho^{(N)}_{\text{init},\hbnt'}$, such that their corresponding steady states are the MPDOs generated by different BNT elements,
    \begin{equation}
        \begin{aligned}
            &\mathcal{E}_{\mathrm{evol},\infty}(\rho^{(N)}_{\text{init},\hbnt})=\rho(M_\hbnt)\\
            &\mathcal{E}_{\mathrm{evol},\infty}(\rho^{(N)}_{\text{init},\hbnt'})=\rho(M_{\hbnt'}),
        \end{aligned}
    \end{equation}
 which do not have long-range correlations. 
    Denote $\rho_{\init,\hbnt}^{(N)}=\rho_{1,\hbnt}\otimes\cdots\otimes\rho_{N,\hbnt}$ and similarly for $\rho_{\init,\hbnt'}^{(N)}$, 
    we next interpolate between $\rho_{i,\hbnt}$ and $\rho_{i,\hbnt'}$ by $\rho_i(s)$ where $0\leq s\leq 1$, and $\rho_i(s)$ is continuous in $s$, $\rho_i(0)=\rho_{i,\hbnt}$ and $\rho_i(1)=\rho_{i,\hbnt'}$. For example, one can choose
    \begin{equation}
        \rho_i(s)=(1-s)\rho_{i,\hbnt}+s\rho_{i,\hbnt'},
    \end{equation}
    and build an interpolated product initial state
    \begin{equation}
        \rho^{(N)}_{\text{init}}(s)=\rho_1(s)\otimes\rho_2(s)\otimes \cdots\otimes \rho_N(s). 
    \end{equation}
    The evolution operator $\mathcal{E}_{\mathrm{evol},\infty}$ as a linear map is continuous, so there must exist a $s'\in(0,1)$ such that $\mathcal{E}_{\mathrm{evol},\infty}(\rho^{(N)}(s'))= \frac{1}{2}(\rho^{(N)}(M_\hbnt)+  \rho^{(N)}(M_{\hbnt'}))$. This is since $\mathcal{E}_{\mathrm{evol},\infty}$ preserves positivity and $\rho^{(N)}(M_\hbnt)$, $\rho^{(N)}(M_\hbnt')$ are supported on orthogonal spaces, thus the convex-combination coefficients are indeed non-negative for all $s$. Again, by \cref{lemma:LRC} this state has long-range correlations and contradicts \cref{eqn:Lieb-Robinson}.
\end{proof}

The long-range correlations in $\rho^{(N)}(M)$ not only lead to a no-go theorem for a commuting local parent Lindbladian but also forbid rapid-mixing dynamics~\cite{bravyi2006lieb,poulin2010lieb}.

\section{Parent Lindbladian constructions: Non-simple MPDO}
\label{sec:pL-non-simple}



In this section, we explicitly construct the parent Lindbladian of non-simple MPDOs and prove their properties. In principle, $\mT$ and $\mS$ can be constructed using the vertical canonical forms of one-site tensor $M$ and two-site tensor $K$, subject to the condition in \cref{thm:non-simple}, which we demonstrate explicitly for two examples in \cref{sec:beyond} and \cref{sec:Fib}. For the purpose of proving the properties of the parent Lindbladian, we adopt a more convenient algebraic approach based on the $C^*$-weak Hopf algebra, which implicitly uses the horizontal canonical form.

We begin with a lightweight review of the characterization of non-simple MPDO RFPs using $C^*$-weak Hopf algebra, employing the graphical notation. Given a $C^*$-weak Hopf algebra $A$, one can generate an MPDO RFP, which corresponds to the boundary of a 2D topologically ordered system, and construct the renormalization channels $\mT$ and $\mS$. We leave the rigorous algebraic treatments behind the graphical notations in \cref{appendix:WHA}. Next, we construct the parent Lindbladian using $\mT$ and $\mS$. In the special case of $C^*$-Hopf algebras, the parent Lindbladian is commuting and exhibits rapid mixing, though the steady state is not unique. For a general $C^*$-weak Hopf algebra $A$, the steady-state degeneracy (SSD) depends on the boundary condition: with open boundary condition SSD equals $\text{dim}(A)$, while with periodic boundary condition SSD equals $g$, the number of BNT elements in the horizontal canonical form.

Finally, we present a case study of a non-simple MPDO RFP that lies beyond the framework of $C^*$-weak Hopf algebra, and construct its parent Lindbladian via the vertical canonical form. The conclusion regarding SSD remains unchanged in this case.

\subsection{Non-simple RFP characterization: \texorpdfstring{$C^*$}{C*}-weak Hopf algebras}
\label{subsubsec:nonSimpleRfpWha}

This subsection is devoted to recalling 
some elementary facts
for the construction of RFP MPDOs from the algebraic point of view in Refs.~\cite{molnar2022matrix,ruizdealarcon2024matrix,yoshiko2024haag}; see also \cref{appendix:WHA} for rigorous definitions and statements. 
The central algebraic structures in this context, introduced by Böhm and Szlach\'{a}nyi in \cite{bohm_coassociativec_1996}, are known as $C^*$-weak Hopf algebras, and are, in essence, complex finite-dimensional vector spaces $A$ such that both $A$ and their dual vector spaces $A^*$ are $C^*$-algebras, and which are furthermore compatible in the sense that their categories of $*$-algebra representations encompass unitary (multi)fusion categories \cite{etingof_fusion_2005, etingof_tensor_2015}. For example, any group algebra $A=\mathbb{C}G$ of a finite group $G$ can be promoted to a $C^*$-Hopf algebra  (See details in \cref{ex:group-algebra}).



In the context of tensor networks and for our purposes here, let us recall that, given any $C^*$-weak Hopf algebra $A$, one can systematically construct a rank-four tensor
\begin{equation}
    \begin{tikzpicture}[scale=1]
    \draw[-mid] (0.6, 0.811) -- (0.6, 1.129);
    \draw[-mid] (0.6, 0) -- (0.6, 0.318);
    \draw[Virtual, -mid] (1.199, 0.564) -- (0.847, 0.564);
    \whaMsymb{(0.6, 0.564)};
    \draw[Virtual, -mid] (0.353, 0.564) -- (0, 0.564);
\end{tikzpicture}
\end{equation}
generating RFP MPDOs for all system sizes $N\in\mathbb{N}$:
\begin{equation}
\rho^{(N)}(M) :=
\begin{tikzpicture}[scale=1]
    \draw[Virtual, -mid] (1.235, 0.564) -- (0.741, 0.564);
    \draw[Virtual] (1.976, 0.564) -- (1.729, 0.564);
    \node[anchor=center] at (2.225, 0.661) {$\overset{N}\cdots$};
    \draw[Virtual] (2.717, 0.564) -- (2.47, 0.564);
    \draw[Virtual] (0.257, 0.564) arc[start angle=87.734, end angle=270, radius=0.247] -- (3.21, 0.071) arc[start angle=-90, end angle=90, radius=0.247];
    \draw[-mid] (0.494, 0.811) -- (0.494, 1.129);
    \draw[-mid] (1.482, 0.811) -- (1.482, 1.129);
    \draw[Virtual] (0.494, 0.564) -- (0.257, 0.564);
    \draw[-mid] (2.964, 0.811) -- (2.964, 1.129);
    \draw[Virtual, -mid] (3.21, 0.564) -- (2.964, 0.564);
    \draw[bevel, -mid] (0.494, 0) -- (0.494, 0.318);
    \draw[bevel, -mid] (1.482, 0) -- (1.482, 0.318);
    \draw[bevel, -mid] (2.964, 0) -- (2.964, 0.318);
    \whaMsymb{(0.494, 0.564)};
    \whaMsymb{(1.482, 0.564)};
    \whaMsymb{(2.964, 0.564)};
\end{tikzpicture}
~,
\end{equation}
in the sense of \cref{def:RFP_MPDO} or, more generally,
\begin{equation}
\rho^{(N)}(M)_{B} :=
\begin{tikzpicture}[scale=1]
    \draw[Virtual, -mid] (2.223, 0.564) -- (1.729, 0.564);
    \draw[Virtual] (2.963, 0.564) -- (2.716, 0.564);
    \node[anchor=center] at (3.212, 0.661) {$\overset{N}\cdots$};
    \draw[Virtual] (3.704, 0.564) -- (3.457, 0.564);
    \draw[Virtual] (0.247, 0.564) arc[start angle=90, end angle=270, radius=0.247] -- (4.198, 0.071) arc[start angle=-90, end angle=90, radius=0.247];
    \draw[-mid] (1.482, 0.811) -- (1.482, 1.129);
    \draw[-mid] (2.469, 0.811) -- (2.469, 1.129);
    \draw[Virtual] (1.482, 0.564) -- (1.244, 0.564);
    \draw[-mid] (3.951, 0.811) -- (3.951, 1.129);
    \draw[Virtual, -mid] (4.198, 0.564) -- (3.951, 0.564);
    \draw[bevel, -mid] (1.482, 0) -- (1.482, 0.318);
    \draw[bevel, -mid] (2.469, 0) -- (2.469, 0.318);
    \draw[bevel, -mid] (3.951, 0) -- (3.951, 0.318);
     \whaMsymb{(1.482, 0.564)};
     \whaMsymb{(2.469, 0.564)};
     \whaMsymb{(3.951, 0.564)};
    \draw[Virtual, -mid] (1.235, 0.564) -- (0.741, 0.564);
    \filldraw[ultra thin, fill=white] 
    (0.494, 0.564) circle[radius=0.247];
    \node[anchor=center] at (0.494, 0.564) {$B$};
\end{tikzpicture}
~,
\end{equation}
for certain admissible boundary conditions $B$; see \cite{ruizdealarcon2024matrix}.
Furthermore, the tensor $M$ is in the horizontal 
canonical form by construction, for which we label the BNT as $a = 1,\ldots, g$ 
and let
\begin{equation}
\begin{tikzpicture}[scale=1]
    \draw[-mid] (0.6, 0.811) -- (0.6, 1.129);
    \draw[-mid] (0.6, 0) -- (0.6, 0.318);
    \draw[Virtual, -mid] (1.199, 0.564) -- (0.847, 0.564);
    \whaMsymb{(0.6, 0.564)};
    \draw[Virtual, -mid] (0.353, 0.564) -- (0, 0.564);
\end{tikzpicture}
=\sum_{a=1}^g
\begin{tikzpicture}[scale=1]
    \draw[-mid] (0.6, 0.811) -- (0.6, 1.129);
    \draw[-mid] (0.6, 0) -- (0.6, 0.318);
    \draw[Virtual, -mid] (1.199, 0.564) -- (0.847, 0.564);
    \whaMsymb{(0.6, 0.564)};
    \draw[Virtual, -mid] (0.353, 0.564) -- (0, 0.564);
    \node[anchor=center, font=\footnotesize, text=Virtual] at (1.025, 0.701) {$a$};
    \node[anchor=center, font=\footnotesize, text=Virtual] at (0.178, 0.701) {$a$};
\end{tikzpicture}
\end{equation}
stand for its corresponding block decomposition; here, the labels $a$ correspond to 
projecting onto the $a$-th block at the virtual level.
Furthermore, 
there exists a rank-four tensor $\wt$ with the same block decomposition, and two rank-two tensors $\Omega$ and $\Xi$ satisfying
\begin{equation}
\label{eq:BlackVsWhite}
\begin{tikzpicture}[scale=1]
    \draw[Virtual] (1.729, 0.423) -- (1.482, 0.423);
    \draw[Virtual, -mid] (2.469, 0.423) -- (1.729, 0.423);
    \draw[Virtual] (1.729, 1.411) -- (1.482, 1.411);
    \draw[-mid] (1.235, 0.67) -- (1.235, 1.164);
    \draw[bevel, -mid] (1.245, 1.658) arc[start angle=-163.45, end angle=0, x radius=0.247, y radius=-0.247] -- (1.729, 0.988) -- (1.729, 0.247) arc[start angle=0, end angle=163.45, x radius=0.247, y radius=-0.247];
     \whaMsymb{(1.235, 0.423)};
     \whaNsymb{(1.235, 1.411)};
    \draw[Virtual, mid-] (2.469, 1.411) -- (1.729, 1.411);
    \draw[Virtual, -mid] (0.988, 0.423) -- (0, 0.423);
    \draw[Virtual] (0.988, 1.411) -- (0, 1.411);
    \filldraw[ultra thin, fill=white] 
    (0.494, 1.411) circle[radius=0.247];
    \node[anchor=center, scale=0.85] at (0.494, 1.411) {$\Omega$};
\end{tikzpicture}
= \sum_a 
\begin{tikzpicture}[scale=1]
    \draw[Virtual, -mid] (0, 0.988) arc[start angle=-90, end angle=90, x radius=0.494, y radius=-0.494];
    \draw[Virtual, -mid] (1.235, 0) arc[start angle=90, end angle=270, x radius=0.494, y radius=-0.494];
    \node[anchor=center, font=\footnotesize, text=Virtual] at (0.318, 0.494) {$a$};
    \node[anchor=center, font=\footnotesize, text=Virtual] at (0.917, 0.494) {$a$};
\end{tikzpicture}
\end{equation}
and, on the other hand, 
\begin{equation}\label{eq:XiOmega}
    \sum_a
\begin{tikzpicture}[scale=1]
    \draw[Virtual, bevel, -mid] (0, 0.988) arc[start angle=-90, end angle=90, x radius=0.494, y radius=-0.494];
    \draw[Virtual, bevel, -mid] (1.727, 0.459) arc[start angle=4.001, end angle=354.819, x radius=0.494, y radius=-0.494];
    \node[anchor=center, font=\footnotesize, text=Virtual] at (0.318, 0.494) {$a$};
    \node[anchor=center, font=\footnotesize, text=Virtual] at (0.917, 0.494) {$a$};
    \filldraw[ultra thin, fill=white] 
    (1.693, 0.494) circle[radius=0.247];
    \node[anchor=center] at (1.693, 0.494) {$\Xi$};
\end{tikzpicture}
=
\begin{tikzpicture}[scale=1]
    \draw[Virtual, bevel, -mid] (0, 0.988) arc[start angle=-90, end angle=90, x radius=0.494, y radius=-0.494];
    \node[anchor=center,mysymb] at (0.494, 0.494) {$\Omega$};
\end{tikzpicture}
\end{equation}
See \cref{appendix:WHA} for explicit constructions of the tensors $M,M',\Omega$ and $\Xi$. Essentially, \cref{eq:BlackVsWhite} is a consequence of the horizontal canonical form.
Moreover, $\Omega$ 
is invertible and is proportional to the identity matrix in each block $a$, and therefore
it commutes at the virtual level with $M$ (and $M'$) by construction:
\begin{equation}
\label{eq:OmegaCommutes}
    \begin{tikzpicture}[scale=1]
    \draw[Virtual] (0.988, 0.564) -- (0, 0.564);
    \whaMsymb{(1.235, 0.564)};
    \filldraw[ultra thin, fill=white] 
    (0.494, 0.564) circle[radius=0.247];
    \node[anchor=center,scale=0.85] at (0.494, 0.564) {$\Omega$};
    \draw[-mid] (1.235, 0.811) -- (1.235, 1.129);
    \draw[-mid] (1.235, 0) -- (1.235, 0.318);
    \draw[Virtual, -mid] (1.834, 0.564) -- (1.482, 0.564);
\end{tikzpicture}
=
\begin{tikzpicture}[scale=1]
    \whaMsymb{(0.6, 0.564)};
    \draw[-mid] (0.6, 0.811) -- (0.6, 1.129);
    \draw[-mid] (0.6, 0) -- (0.6, 0.318);
    \draw[Virtual, -mid] (0.353, 0.564) -- (0, 0.564);
    \draw[Virtual] (1.834, 0.564) -- (0.847, 0.564);
    \filldraw[ultra thin, fill=white] 
    (1.341, 0.564) circle[radius=0.247];
    \node[anchor=center,scale=0.85] at (1.341, 0.564) {$\Omega$};
\end{tikzpicture};
\end{equation}
in other words, we say $\Omega$ is central with respect to $M$. 
By using \cref{eq:BlackVsWhite,eq:XiOmega} and the property of $\Omega$, one can easily derive that
\begin{equation}
\label{eq:MvsNXi}
\begin{tikzpicture}[scale=1]
    \draw[Virtual] (1.729, 0.423) -- (1.482, 0.423);
    \draw[Virtual, -mid] (3.951, 0.423) -- (3.21, 0.423);
    \draw[Virtual] (3.21, 0.423) -- (2.963, 0.423);
    \draw[Virtual, -mid] (2.469, 0.423) -- (1.729, 0.423);
    \draw[Virtual] (2.469, 1.411) -- (1.482, 1.411);
    \draw[Virtual] (3.951, 1.411) -- (3.21, 1.411);
    \filldraw[ultra thin, fill=white] 
    (2.117, 1.411) circle[radius=0.247];
    \node[anchor=center,scale=0.85] at (2.117, 1.411) {$\Xi$};
    \draw[-mid] (1.235, 0.67) -- (1.235, 1.164);
    \draw[Virtual, -mid] (0.988, 0.423) -- (0, 0.423);
    \draw[Virtual] (0.988, 1.411) -- (0, 1.411);
    \draw[-mid] (2.716, 0.67) -- (2.716, 1.164);
    \draw[bevel, -mid] (2.727, 1.658) arc[start angle=-163.45, end angle=0, x radius=0.247, y radius=-0.247] -- (3.21, 0.988) -- (3.21, 0.247) arc[start angle=0, end angle=163.45, x radius=0.247, y radius=-0.247];
    \draw[bevel, -mid] (1.245, 1.658) arc[start angle=-163.45, end angle=0, x radius=0.247, y radius=-0.247] -- (1.729, 0.988) -- (1.729, 0.247) arc[start angle=0, end angle=163.45, x radius=0.247, y radius=-0.247];
    \whaMsymb{(1.235, 0.423)};
    \whaMsymb{(2.716, 0.423)};
    \whaNsymb{(1.235, 1.411)};
    \whaNsymb{(2.716, 1.411)};
    \draw[Virtual] (3.21, 1.411) -- (2.963, 1.411);
    \filldraw[ultra thin, fill=white] 
    (0.494, 1.411) circle[radius=0.247];
    \node[anchor=center,scale=0.85] at (0.494, 1.411) {$\Omega$};
\end{tikzpicture}
=
\sum_a \begin{tikzpicture}[scale=1]
    \draw[Virtual, -mid] (0, 0.988) arc[start angle=-90, end angle=90, x radius=0.494, y radius=-0.494];
    \draw[Virtual, -mid] (1.235, 0) arc[start angle=90, end angle=270, x radius=0.494, y radius=-0.494];
    \node[anchor=center, font=\footnotesize, text=Virtual] at (0.318, 0.494) {$a$};
    \node[anchor=center, font=\footnotesize, text=Virtual] at (0.917, 0.494) {$a$};
\end{tikzpicture}
\end{equation}


Below we show the example of MPDO tensor when $A=\mathbb{C}\mathbb{Z}_2$. See \cref{sec:Fib} for a more advanced example where $A$ is the Fibonacci $C^*$-weak Hopf algebra.

\begin{exmp}[Boundary state of the Toric Code]\label{ex:CZ2_M_N}
Let us reconsider the boundary state of the toric code introduced in \cref{exmp:ToricCodeBd}. This MPDO tensor is in fact generated using $\mathbb{C}\mathbb{Z}_2$ group $C^*$-algebra, where the tensors $M$ and $\wt$ are given by
\[ 
\begin{tikzpicture}[scale=1]
    \whaM{(0,0)}{1};
\end{tikzpicture}
=
\begin{tikzpicture}[scale=1]
    \draw[-mid] (0.6, 0.811) -- (0.6, 1.129);
    \draw[-mid] (0.6, 0) -- (0.6, 0.318);
    \draw[Virtual, -mid] (1.199, 0.564) -- (0.847, 0.564);
    \whaNsymb{(0.6, 0.564)};
    \draw[Virtual, -mid] (0.353, 0.564) -- (0, 0.564);
\end{tikzpicture}
=  \frac{\bo}{2} \otimes |e\rangle\langle e| +  \frac{\sigma_z}{2}  \otimes |g\rangle\langle g| \]
where the first tensor factor corresponds to the endomorphism in the physical space and the second tensor factor to one in the virtual space, 
and the rank-two tensors acting in the virtual space are 
\[ 
\begin{tikzpicture}[scale=1]
    \node[anchor=center,mysymb] (no0) at (0.564, 0.121) {$\Omega$};
    \draw[Virtual, -mid] (no0) -- (0, 0.121);
    \draw[Virtual, -mid] (1.129, 0.121) -- (no0);
\end{tikzpicture}
=
\begin{tikzpicture}[scale=1]
    \node[anchor=center,mysymb] (no0) at (0.564, 0.121) {$\Xi$};
    \draw[Virtual, -mid] (no0) -- (0, 0.121);
    \draw[Virtual, -mid] (1.129, 0.121) -- (no0);
\end{tikzpicture}
= \frac{\bo}{2}.
\]
We show the explicit derivation using $C^*$-Hopf algebra in \cref{sec:CZ2-construction}.
\end{exmp}

Any tensor $M$ satisfying \cref{eq:BlackVsWhite,eq:XiOmega,eq:OmegaCommutes} defines RFP MPDOs in the sense of \cref{def:RFP_MPDO}. Indeed, on the one hand, the fine-graining quantum channel $\mathcal{T}$ is given by the following diagrammatic formula using the vectorized notation:  
\begin{equation}
\mathcal{T} =
\begin{tikzpicture}[scale=1]
    \draw[Virtual] (0.502, 1.42) arc[start angle=118.955, end angle=151.045, radius=0.494];
    \draw[Virtual, -mid] (0.309, 0.749) arc[start angle=-151.045, end angle=-90, radius=0.494] -- (1.235, 0.494);
    \draw[Virtual, -mid] (1.976, 1.482) -- (0.988, 1.482);
    \draw[Virtual, -mid] (1.729, 0.494) -- (2.716, 0.494) arc[start angle=-90, end angle=90, radius=0.494] -- (2.469, 1.482);
    \node[anchor=center] at (0.247, 0.988) {$\Omega$};
    \filldraw[ultra thin, fill=white] 
    (0.247, 0.988) circle[radius=0.247];
    \node[anchor=center] at (0.247, 0.988) {$\Omega$};
    \draw[-mid] (0.741, 1.729) -- (0.741, 1.976);
    \draw[bevel, mid-] (0.741, 1.235) arc[start angle=180, end angle=360, radius=0.247] -- (1.235, 1.976);
    \draw[-mid] (2.223, 1.729) -- (2.223, 1.976);
    \draw[bevel, mid-] (2.223, 1.235) arc[start angle=180, end angle=360, radius=0.247] -- (2.716, 1.976);
    \whaMsymb{(0.741, 1.482)};
    \whaMsymb{(2.223, 1.482)};
    \draw[-mid] (1.482, 0) -- (1.482, 0.247);
    \draw[bevel, -mid] (1.482, 0.741) arc[start angle=-180, end angle=0, x radius=0.247, y radius=-0.247] -- (1.976, 0);
    \whaNsymb{(1.482, 0.494)};
\end{tikzpicture}
~.
\end{equation}
On the other hand, the coarse-graining quantum channel $\mathcal{S}$ is given by a sum of the form
\begin{equation}
    \mathcal{S} = \mathcal{S}^{(\Eone)} + \mathcal{S}^{(\Etwo)},
\end{equation}
where the first summand is defined by the expression
\begin{equation}
\mathcal{S}^{(\Eone)} = 
\begin{tikzpicture}[scale=1]
    \draw[Virtual] (0.741, 0.494) arc[start angle=90, end angle=180, x radius=0.494, y radius=-0.494];
    \draw[Virtual, -mid] (1.235, 1.482) -- (0.741, 1.482) arc[start angle=90, end angle=151.045, radius=0.494];
    \draw[Virtual] (1.623, 0.494) -- (0.741, 0.494);
    \draw[Virtual] (2.716, 0.494) -- (1.623, 0.494);
    \draw[Virtual] (2.716, 0.494) arc[start angle=-90, end angle=90, radius=0.494] -- (1.729, 1.482);
    \draw[-mid] (0.741, 0) -- (0.741, 0.247);
    \draw[-mid] (2.223, 0) -- (2.223, 0.247);
    \draw[-mid] (1.482, 1.729) -- (1.482, 1.976);
    \draw[bevel, mid-] (1.482, 1.235) -- (1.482, 1.235) arc[start angle=180, end angle=360, radius=0.247] -- (1.976, 1.976);
    \draw[bevel, -mid] (2.223, 0.741) -- (2.223, 0.741) arc[start angle=-180, end angle=0, x radius=0.247, y radius=-0.247] -- (2.716, 0);
    \draw[bevel, -mid] (0.741, 0.741) -- (0.741, 0.741) arc[start angle=-180, end angle=0, x radius=0.247, y radius=-0.247] -- (1.235, 0);
    \whaMsymb{(1.482, 1.482)};
    \filldraw[ultra thin, fill=white] 
    (0.247, 0.988) circle[radius=0.247];
    \node[anchor=center] at (0.247, 0.988) {$\Omega$};
    \whaNsymb{(0.741, 0.494)};
    \filldraw[ultra thin, fill=white] 
    (1.623, 0.494) circle[radius=0.247];
    \node[anchor=center] at (1.623, 0.494) {$\Xi$};
    \whaNsymb{(2.223, 0.494)};
\end{tikzpicture}
~,
\end{equation}
which is completely positive and trace-non-increa\-sing, but not necessarily trace-preserving, and the second summand is such that
\begin{equation}\label{eq:S2annh}
\begin{tikzpicture}[scale=1]
    \draw[-mid] (0.741, 0.67) -- (0.741, 0.917);
    \draw[Virtual, -mid] (0.494, 0.423) -- (0, 0.423);
    \draw[-mid] (2.223, 0.67) -- (2.223, 0.917);
    \draw[-mid] (1.411, 1.411) -- (1.411, 1.905);
    \draw[mid-] (1.905, 1.411) -- (1.905, 1.905);
    \draw[Virtual] (1.976, 0.423) -- (0.988, 0.423);
    \draw[Virtual, -mid] (3.457, 0.423) -- (2.963, 0.423);
    \draw[Virtual] (2.963, 0.423) -- (2.469, 0.423);
    \draw[bevel, mid-] (0.751, 0.177) arc[start angle=-163.45, end angle=0, radius=0.247] -- (1.235, 0.917);
    \draw[bevel, mid-] (2.233, 0.177) arc[start angle=-163.45, end angle=0, radius=0.247] -- (2.716, 0.917);
    \filldraw[ultra thin, fill=white] (0.741, 0.917) arc[start angle=90, end angle=270, x radius=0.247, y radius=-0.247] -- (2.716, 1.411) arc[start angle=-90, end angle=90, x radius=0.247, y radius=-0.247] -- cycle;
    \node[anchor=center] at (1.658, 1.164) {$\mathcal{S}^{(\Etwo)}$};
    \whaMsymb{(0.741, 0.423)};
    \whaMsymb{(2.223, 0.423)};
\end{tikzpicture}
\equiv 0,
\end{equation}
i.e.~acts non-trivially only on the orthogonal of 
the space spanned by $\rho^{(2)}(M)_X$ where $X$ is an arbitrary boundary condition matrix, and makes $\mathcal{S}$ trace-preserving. We will make a particular choice of $\mS^{(\Etwo)}$ in the next subsection. 

It can be deduced that the effect of the fine-grained map followed by the coarse-grained map is trivial, i.e.~\( \mathcal{S}\circ \mathcal{T} \) leaves the tensor $M$ invariant; 
and the opposite combination, the quantum channel
\( \mathcal{T} \circ \mathcal{S} \)
is an idempotent map that 
leaves the two-site tensor invariant and stabilizes the RFP MPDOs. 
This makes \( \mathcal{T} \circ \mathcal{S} \) the main candidate for the local stabilizers on the Lindbladian construction in the next subsections.

\subsection{Parent Lindbladian}
With the notations of \cref{subsubsec:nonSimpleRfpWha}, let us consider the quantum channel
\(
\mathcal{E} := \mathcal{T}\circ \mathcal{S}
\),
which can be decomposed in the form
\begin{align}
    &\mathcal{E} = \mathcal{E}^{(\Eone)} + \mathcal{E}^{(\Etwo)},\\
    &\text{ where }\mathcal{E}^{(k)} := \mathcal{T}\circ \mathcal{S}^{(k)},\; k=\Eone,\Etwo.
\end{align}
In diagrammatic notation, it is easy to check that the first term takes the form
\begin{equation}
\mathcal{E}^{(\Eone)} =
\begin{tikzpicture}[scale=1]
    \draw[Virtual] (0.502, 0.626) arc[start angle=118.955, end angle=151.045, x radius=0.494, y radius=-0.494];
    \draw[Virtual] (1.376, 0.564) -- (0.988, 0.564);
    \draw[Virtual] (1.976, 0.564) -- (1.87, 0.564);
    \draw[-mid] (0.741, 0) -- (0.741, 0.318);
    \draw[-mid] (0.741, 1.799) -- (0.741, 2.117);
    \draw[Virtual, -mid] (1.976, 1.552) -- (1.235, 1.552);
    \draw[-mid] (2.223, 1.799) -- (2.223, 2.117);
    \filldraw[ultra thin, fill=white] 
    (1.623, 0.564) circle[radius=0.247];
    \draw[bevel, -mid] (0.751, 0.811) arc[start angle=-163.45, end angle=0, x radius=0.247, y radius=-0.247] -- (1.235, 0);
    \draw[bevel, mid-] (2.233, 1.305) arc[start angle=-163.45, end angle=0, radius=0.247] -- (2.716, 2.117);
    \node[anchor=center,scale=0.85] at (1.623, 0.564) {$\Xi$};
    \draw[Virtual] (1.235, 1.552) -- (0.988, 1.552);
    \draw[bevel, mid-] (0.751, 1.305) arc[start angle=-163.45, end angle=0, radius=0.247] -- (1.235, 2.117);
    \whaMsymb{(0.741, 1.552)};
    \draw[Virtual] (0.502, 1.49) arc[start angle=118.955, end angle=151.045, radius=0.494];
    \filldraw[ultra thin, fill=white] 
    (0.247, 1.058) circle[radius=0.247];
    \node[anchor=center,scale=0.85] at (0.247, 1.058) {$\Omega$};
    \draw[Virtual, mid-] (2.469, 1.552) -- (2.716, 1.552) arc[start angle=-90, end angle=90, x radius=0.494, y radius=-0.494] -- (2.469, 0.564);
    \draw[-mid] (2.223, 0) -- (2.223, 0.318);
    \draw[bevel, -mid] (2.233, 0.811) arc[start angle=-163.45, end angle=0, x radius=0.247, y radius=-0.247] -- (2.716, 0);
    \whaMsymb{(2.223, 1.552)};
    \whaNsymb{(0.741, 0.564)};
    \whaNsymb{(2.223, 0.564)};
\end{tikzpicture}
~,
\end{equation}
and, moreover, that the RFP MPDOs with arbitrary boundary conditions are invariant under the action of $\mathcal{E}^{(\Eone)}$ (and hence of  $\mathcal{E}$ due to \cref{eq:S2annh}), since
\begin{equation}\label{eq:E1Stab}
\begin{tikzpicture}[scale=1]
    \draw[Virtual] (0.502, 1.319) arc[start angle=118.955, end angle=151.045, x radius=0.494, y radius=-0.494];
    \draw[-mid] (0.741, 2.492) -- (0.741, 2.739);
    \draw[Virtual, -mid] (1.976, 2.245) -- (1.235, 2.245);
    \draw[-mid] (2.223, 2.492) -- (2.223, 2.739);
    \draw[bevel, mid-] (2.233, 1.999) arc[start angle=-163.45, end angle=0, radius=0.247] -- (2.716, 2.739);
    \draw[Virtual] (1.235, 2.245) -- (0.988, 2.245);
    \draw[bevel, mid-] (0.751, 1.999) arc[start angle=-163.45, end angle=0, radius=0.247] -- (1.235, 2.739);
    \whaMsymb{(0.741, 2.245)};
    \draw[Virtual] (0.502, 2.184) arc[start angle=118.955, end angle=151.045, radius=0.494];
    \whaMsymb{(2.223, 2.245)};
    \filldraw[ultra thin, fill=white] 
    (0.247, 1.752) circle[radius=0.247];
    \node[anchor=center] at (0.247, 1.752) {$\Omega$};
    \draw[Virtual, mid-] (2.469, 2.245) -- (2.716, 2.245) arc[start angle=-90, end angle=90, x radius=0.494, y radius=-0.494] -- (2.469, 1.258);
    \draw[Virtual] (2.469, 1.258) -- (0.988, 1.258);
    \filldraw[ultra thin, fill=white] 
    (1.623, 1.258) circle[radius=0.247];
    \node[anchor=center] at (1.623, 1.258) {$\Xi$};
    \draw[-mid] (0.741, 0.764) -- (0.741, 1.011);
    \draw[bevel, -mid] (0.751, 1.504) arc[start angle=-163.45, end angle=0, x radius=0.247, y radius=-0.247] -- (1.235, 0.764) -- (1.235, 0.27) arc[start angle=-5.094, end angle=185.094, x radius=0.248, y radius=-0.248];
    \draw[-mid] (2.223, 0.764) -- (2.223, 1.011);
    \draw[bevel, -mid] (2.233, 1.504) arc[start angle=-163.45, end angle=0, x radius=0.247, y radius=-0.247] -- (2.716, 0.764) -- (2.716, 0.27) arc[start angle=0, end angle=180, x radius=0.247, y radius=-0.247];
    \whaNsymb{(0.741, 1.258)};
    \whaNsymb{(2.223, 1.258)};
    \draw[Virtual] (1.976, 0.517) -- (0.988, 0.517);
    \whaMsymb{(0.741, 0.517)};
    \draw[Virtual, -mid] (0.494, 0.517) -- (0, 0.517);
    \draw[Virtual, -mid] (3.21, 0.517) -- (2.469, 0.517);
    \whaMsymb{(2.223, 0.517)};
\end{tikzpicture}
=
\begin{tikzpicture}[scale=1]
    \draw[-mid] (0.741, 0.67) -- (0.741, 0.917);
    \draw[Virtual] (1.976, 0.423) -- (0.988, 0.423);
    \draw[Virtual, -mid] (0.494, 0.423) -- (0, 0.423);
    \draw[Virtual, -mid] (3.21, 0.423) -- (2.469, 0.423);
    \draw[-mid] (2.223, 0.67) -- (2.223, 0.917);
    \draw[bevel, mid-] (0.751, 0.177) arc[start angle=-163.45, end angle=0, radius=0.247] -- (1.235, 0.917);
    \draw[bevel, mid-] (2.233, 0.177) arc[start angle=-163.45, end angle=0, radius=0.247] -- (2.716, 0.917);
    \whaMsymb{(2.223, 0.423)};
    \whaMsymb{(0.741, 0.423)};
\end{tikzpicture}
\end{equation}
as an immediate consequence of \cref{eq:MvsNXi}.


Let us denote $\mathcal{E}_i=\tau_i(\mathcal{E})$ where $\tau_i$ translates the operator by $i$ sites and
define the linear maps
\begin{align}
    \mathcal{E}^{(k)}_{i,...,i+\ell}=\mathcal{E}_i^{(k)}\circ\mathcal{E}_{i+1}^{(k)}\circ\cdots\circ\mathcal{E}_{i+\ell}^{(k)},
\end{align}
for $k=\Eone,\Etwo$,
 where 
 the operators are tensorized with the identity map elsewhere.
For example, for $\ell = 2$ it is easy to check that it takes the form
\begin{equation}\label{eq:defE_3sites}
\mathcal{E}^{(\Eone)}_{i,i+1,i+2} = 
\begin{tikzpicture}[scale=1]
    \draw[Virtual] (0.502, 0.626) arc[start angle=118.955, end angle=151.045, x radius=0.494, y radius=-0.494];
    \draw[Virtual] (0.502, 1.49) arc[start angle=118.955, end angle=151.045, radius=0.494];
    \draw[Virtual] (1.376, 0.564) -- (0.988, 0.564);
    \draw[Virtual] (4.198, 0.564) -- (1.87, 0.564);
    \draw[bevel, -mid] (0.751, 0.811) arc[start angle=-163.45, end angle=0, x radius=0.247, y radius=-0.247] -- (1.235, 0);
    \draw[-mid] (0.741, 1.799) -- (0.741, 2.117);
    \draw[Virtual, -mid] (1.976, 1.552) -- (0.988, 1.552);
    \draw[-mid] (2.223, 1.799) -- (2.223, 2.117);
    \draw[Virtual, -mid] (4.198, 0.564) arc[start angle=-90, end angle=90, radius=0.494] -- (3.951, 1.552);
    \draw[-mid] (3.704, 0) -- (3.704, 0.318);
    \draw[-mid] (3.704, 1.799) -- (3.704, 2.117);
    \draw[bevel, mid-] (3.714, 1.305) arc[start angle=-163.45, end angle=0, radius=0.247] -- (4.198, 2.117);
    \filldraw[ultra thin, fill=white] 
    (0.247, 1.058) circle[radius=0.247];
    \filldraw[ultra thin, fill=white] 
    (3.104, 0.564) circle[radius=0.247];
    \filldraw[ultra thin, fill=white] 
    (1.623, 0.564) circle[radius=0.247];
    \node[anchor=center] at (1.623, 0.564) {$\Xi$};
    \node[anchor=center] at (0.247, 1.058) {$\Omega$};
    \node[anchor=center] at (3.104, 0.564) {$\Xi$};
    \draw[bevel, -mid] (2.233, 0.811) arc[start angle=-163.45, end angle=0, x radius=0.247, y radius=-0.247] -- (2.716, 0);
    \draw[bevel, -mid] (3.714, 0.811) arc[start angle=-163.45, end angle=0, x radius=0.247, y radius=-0.247] -- (4.198, 0);
    \draw[-mid] (0.741, 0) -- (0.741, 0.318);
    \draw[-mid] (2.223, 0) -- (2.223, 0.318);
    \draw[Virtual, -mid] (3.457, 1.552) -- (2.469, 1.552);
    \draw[bevel, mid-] (2.233, 1.305) arc[start angle=-163.45, end angle=0, radius=0.247] -- (2.716, 2.117);
    \draw[bevel, mid-] (0.751, 1.305) arc[start angle=-163.45, end angle=0, radius=0.247] -- (1.235, 2.117);
    \whaMsymb{(0.741, 1.552)};
    \whaMsymb{(2.223, 1.552)};
    \whaMsymb{(3.704, 1.552)};
    \whaNsymb{(3.704, 0.564)};
    \whaNsymb{(2.223, 0.564)};
    \whaNsymb{(0.741, 0.564)};
\end{tikzpicture}~.
\end{equation}
This is because $\mathcal{E}_{i}^{(\Eone)}\circ\mathcal{E}_{i+1}^{(\Eone)}$ is given by the expression 
\[ 
\begin{tikzpicture}[scale=1]
    \draw[Virtual] (0.502, 0.556) arc[start angle=118.955, end angle=151.044, x radius=0.494, y radius=-0.494];
    \draw[Virtual, -mid] (0.502, 1.42) arc[start angle=118.955, end angle=151.044, radius=0.494];
    \draw[Virtual] (1.623, 0.494) -- (0.741, 0.494);
    \draw[Virtual] (2.716, 0.494) -- (1.623, 0.494);
    \draw[Virtual, -mid] (2.716, 0.494) arc[start angle=-90, end angle=90, radius=0.494] -- (2.469, 1.482);
    \draw[-mid] (0.741, 0) -- (0.741, 0.247);
    \draw[-mid] (2.223, 0) -- (2.223, 0.247);
    \draw[bevel, -mid] (2.233, 0.741) arc[start angle=-163.45, end angle=0, x radius=0.247, y radius=-0.247] -- (2.716, 0);
    \draw[bevel, -mid] (0.751, 0.741) arc[start angle=-163.45, end angle=0, x radius=0.247, y radius=-0.247] -- (1.235, 0);
    \draw[-mid] (0.741, 1.729) -- (0.741, 3.951);
    \draw[Virtual, -mid] (1.976, 1.482) -- (0.988, 1.482);
    \draw[-mid] (2.223, 1.729) -- (2.223, 2.223);
    \draw[Virtual] (2.223, 2.469) arc[start angle=90, end angle=180, x radius=0.494, y radius=-0.494];
    \draw[Virtual, -mid] (2.963, 3.457) -- (2.223, 3.457) arc[start angle=90, end angle=180, radius=0.494];
    \draw[Virtual] (3.104, 2.469) -- (2.223, 2.469);
    \draw[Virtual] (4.198, 2.469) -- (3.104, 2.469);
    \draw[Virtual, -mid] (4.198, 2.469) arc[start angle=-90, end angle=90, radius=0.494] -- (2.963, 3.457);
    \draw[-mid] (3.704, 0) -- (3.704, 2.223);
    \draw[bevel] (3.714, 2.716) arc[start angle=-163.45, end angle=0, x radius=0.247, y radius=-0.247] -- (4.198, 2.223);
    \draw[-mid] (2.223, 3.704) -- (2.223, 3.951);
    \draw[mid-] (2.233, 3.21) arc[start angle=-163.45, end angle=0, radius=0.247] -- (2.716, 3.951);
    \draw[Virtual, -mid] (3.704, 3.457) -- (2.223, 3.457);
    \draw[-mid] (3.704, 3.704) -- (3.704, 3.951);
    \draw[mid-] (3.714, 3.21) arc[start angle=-163.45, end angle=0, radius=0.247] -- (4.198, 3.951);
    \draw[bevel, mid-] (2.233, 1.235) arc[start angle=-163.45, end angle=0, radius=0.247] -- (2.716, 2.223) -- (2.716, 2.646) arc[start angle=0, end angle=163.45, radius=0.247];
    \filldraw[ultra thin, fill=white] 
    (3.104, 2.469) circle[radius=0.247];
    \filldraw[ultra thin, fill=white] 
    (1.729, 2.963) circle[radius=0.247];
    \filldraw[ultra thin, fill=white] 
    (1.623, 0.494) circle[radius=0.247];
    \filldraw[ultra thin, fill=white] 
    (0.247, 0.988) circle[radius=0.247];
    \node[anchor=center] at (1.623, 0.494) {$\Xi$};
    \node[anchor=center] at (0.247, 0.988) {$\Omega$};
    \node[anchor=center] at (3.104, 2.469) {$\Xi$};
    \node[anchor=center] at (1.729, 2.963) {$\Omega$};
    \whaMsymb{(2.223, 3.457)};
    \whaMsymb{(3.704, 3.457)};
    \whaNsymb{(2.223, 2.469)};
    \whaMsymb{(2.223, 1.482)};
    \whaNsymb{(3.704, 2.469)};
    \draw[mid-] (4.198, 0) -- (4.198, 2.223);
    \draw[bevel] (0.751, 1.235) arc[start angle=-163.45, end angle=0, radius=0.247] -- (1.235, 1.729);
    \whaMsymb{(0.741, 1.482)};
    \whaNsymb{(2.223, 0.494)};
    \whaNsymb{(0.741, 0.494)}
    \draw[mid-] (1.235, 1.729) -- (1.235, 3.951);
\end{tikzpicture}~,
\]
and by virtue of \cref{eq:BlackVsWhite}, this is equal to
\[
\begin{tikzpicture}[scale=1]
    \draw[Virtual] (0.741, 0.494) arc[start angle=90, end angle=180, x radius=0.494, y radius=-0.494];
    \draw[Virtual, -mid] (1.482, 1.482) -- (0.741, 1.482) arc[start angle=90, end angle=180, radius=0.494];
    \draw[Virtual] (1.376, 0.494) -- (0.988, 0.494);
    \draw[Virtual] (3.069, 0.494) -- (2.469, 0.494);
    \draw[Virtual, -mid] (3.069, 0.494) arc[start angle=-89.206, end angle=89.206, radius=0.494] -- (2.946, 1.482);
    \draw[-mid] (0.741, 0) -- (0.741, 0.494);
    \draw[-mid] (2.223, 0) -- (2.223, 0.494);
    \draw[bevel, -mid] (2.233, 0.741) arc[start angle=-163.45, end angle=0, x radius=0.247, y radius=-0.247] -- (2.716, 0);
    \draw[bevel, -mid] (0.751, 0.741) arc[start angle=-163.45, end angle=0, x radius=0.247, y radius=-0.247] -- (1.235, 0);
    \draw[-mid] (0.741, 1.729) -- (0.741, 3.951);
    \draw[bevel] (0.751, 1.235) arc[start angle=-163.45, end angle=0, radius=0.247] -- (1.235, 1.729);
    \draw[Virtual, -mid] (1.905, 1.482) -- (0.741, 1.482);
    \draw[Virtual, -mid] (2.858, 3.457) -- (1.87, 3.457) arc[start angle=90, end angle=270, radius=0.494];
    \draw[Virtual] (4.198, 2.469) -- (3.104, 2.469);
    \draw[Virtual, -mid] (4.198, 2.469) arc[start angle=-90, end angle=90, radius=0.494] -- (2.963, 3.457);
    \draw[-mid] (3.704, 0) -- (3.704, 2.223);
    \draw[bevel] (3.714, 2.716) arc[start angle=-163.45, end angle=0, x radius=0.247, y radius=-0.247] -- (4.198, 2.223);
    \draw[-mid] (2.223, 3.457) -- (2.223, 3.951);
    \draw[bevel, mid-] (2.233, 3.21) arc[start angle=-163.45, end angle=0, radius=0.247] -- (2.716, 3.951);
    \draw[Virtual, -mid] (3.704, 3.457) -- (2.223, 3.457);
    \draw[-mid] (3.704, 3.457) -- (3.704, 3.951);
    \draw[bevel, mid-] (3.714, 3.21) arc[start angle=-163.45, end angle=0, radius=0.247] -- (4.198, 3.951);
    \draw[Virtual] (2.963, 1.482) arc[start angle=90, end angle=257.658, x radius=0.494, y radius=-0.494];
    \draw[Virtual] (1.87, 2.469) arc[start angle=-90, end angle=90, x radius=0.494, y radius=-0.494];
    \filldraw[ultra thin, fill=white] 
    (0.247, 0.988) circle[radius=0.247];
    \filldraw[ultra thin, fill=white] 
    (1.623, 0.494) circle[radius=0.247];
    \filldraw[ultra thin, fill=white] 
    (3.104, 2.469) circle[radius=0.247];
    \draw[Virtual] (1.976, 0.494) -- (1.87, 0.494);
    \draw[mid-] (4.198, 0) -- (4.198, 2.223);
    \node[anchor=center] at (1.623, 0.494) {$\Xi$};
    \node[anchor=center] at (0.247, 0.988) {$\Omega$};
    \node[anchor=center] at (3.104, 2.469) {$\Xi$};
    \draw[mid-] (1.235, 1.729) -- (1.235, 3.951);
    \whaNsymb{(0.741, 0.494)};
    \whaMsymb{(0.741, 1.482)};
    \whaNsymb{(2.223, 0.494)};
    \whaMsymb{(2.223, 3.457)};
    \whaMsymb{(3.704, 3.457)};
    \whaNsymb{(3.704, 2.469)};
\end{tikzpicture}~,
\]
which corresponds exactly to the expression introduced in \cref{eq:defE_3sites}. Similarly, one can also show $\mathcal{E}_{i+1}^{(\Eone)}\circ\mathcal{E}_i^{(\Eone)}$ is given by the same expression in \cref{eq:defE_3sites}, and therefore the two linear maps $\mathcal{E}_{i+1}^{(\Eone)}$ and $\mathcal{E}_i^{(\Eone)}$ commute,
\begin{equation}
    \mathcal{E}_i^{(\Eone)}\circ\mathcal{E}_{i+1}^{(\Eone)}=\mathcal{E}_{i+1}^{(\Eone)}\circ\mathcal{E}_{i}^{(\Eone)}.
\end{equation}
Due to this commutation, the definition of $\mathcal{E}^{(\Eone)}_{i,\ldots,i+\ell}$ is independent of the order of the products.

Furthermore, in the particular case of $C^*$-Hopf algebras it holds that $\mathcal{E} = \mathcal{E}^{(\Eone)}$ is already trace-preserving and hence $\mathcal{E}^{(\Etwo)} \equiv 0$; this includes the case of group $C^*$-algebras. 
Hence, the parent Lindbladian is commuting for the MPDO RFP generated by $C^*$-Hopf algebra, and subsequently, the parent Lindbladian is frustration-free and is rapid-mixing (the same proof in \cref{sec:rapid-mixing} applies). Define the quantum channels for open boundary condition (OBC) 
\begin{equation}
\begin{aligned}    \mathcal{E}_{1,2,\ldots,N}^{(\Eone)\text{OBC}}:=\mathcal{E}^{(\Eone)}_{1,2,\ldots,N}
\end{aligned}
\end{equation}
and periodic boundary condition (PBC),
\begin{equation}
\mathcal{E}_{1,2,\ldots,N}^{(\Eone)\text{PBC}}:=\mathcal{E}^{(\Eone)}_{N}\circ\mathcal{E}^{(\Eone)}_{1,2,\ldots,N},
\end{equation}
where $\mathcal{E}_N^\Eone$ acts on sites $N$ and 1, closing the boundary. 
The global steady state of the parent Lindbladian $\mL$ is, therefore, the fixed point of the channel
\begin{equation}
    \mathcal{E}_{1,2,\ldots,N}^{(\Eone)\text{O/PBC}}(\rho)=\rho,
\end{equation}
depending on the boundary condition. 

However, one cannot expect that $\mathcal{E}=\mathcal{E}^{(\Eone)}$ for $C^*$-weak Hopf algebras in general; see e.g.~\cref{sec:Fib}. Similar to the non-injective simple case, there are degrees of freedom in the choice of $\mE^{(\Etwo)}$. Here, we make a convenient choice of,
\begin{equation}
    \mE^{(\Etwo)}\propto\begin{array}{c}
\begin{tikzpicture}[scale=1]
    \draw[Virtual, -mid] (2.223, 0.564) -- (1.729, 0.564);
    \draw[Virtual] (2.963, 0.564) -- (2.716, 0.564);
    \draw[Virtual] (0.247, 0.564) arc[start angle=90, end angle=270, radius=0.247] -- (2.963, 0.071) arc[start angle=-90, end angle=90, radius=0.247];
    \draw[-mid] (1.482, 0.811) -- (1.482, 1.129);
    \draw[-mid] (2.469, 0.811) -- (2.469, 1.129);
    \draw[Virtual] (1.482, 0.564) -- (1.244, 0.564);
    \draw[bevel, mid-] (1.482, 0.318) arc[start angle=-163.45, end angle=0, radius=0.247] -- (1.976, 1.129);
    \draw[bevel, mid-] (2.469, 0.318) arc[start angle=-163.45, end angle=0, radius=0.247] -- (2.963, 1.129);
    \whaMsymb{(2.469, 0.564)};
    \whaMsymb{(1.482, 0.564)};
    \draw[Virtual, -mid] (1.235, 0.564) -- (0.741, 0.564);
    \filldraw[ultra thin, fill=white] 
    (0.494, 0.564) circle[radius=0.247];
    \node[anchor=center] at (0.494, 0.564) {$\Omega$};
    \haprojector{(2, -0.5)}{$P^\perp$}{1};
    \end{tikzpicture}
    \end{array}
    \label{eqn:choiceE2WHA}
\end{equation}
where the normalization constant can be determined by $\Omega$, 
and $P^\perp$ is a projector onto the orthogonal space of the space generated by two-site MPDO,
\begin{equation}
\begin{array}{c}
\begin{tikzpicture}[scale=1]
\draw[Virtual] (-0.75, -0.818) -- (0.25, -0.818);
\draw[Virtual,mid-] (-1.75, -0.818) -- (-1.25, -0.818);
\draw[Virtual,mid-] (0.75, -0.818) -- (1.5, -0.818);
\haMtensor{(-1,-0.818)}{$M$};
\haMtensor{(0.5,-0.818)}{$M$};
\haprojector{(0,0)}{$P^\perp$}{0};
\end{tikzpicture}
\end{array}=0.
\end{equation}
By construction, $\mathcal{E}^{(\Etwo)}$ satisfies the properties,
\begin{align}
\label{eq:propsEE1E2wha}
    &\mathcal{E}^{(\Etwo)}\circ\mathcal{E}=0\\
    &\mathcal{E}^{(\Eone)}\circ\mathcal{E}^{(\Etwo)}=\mathcal{E}^{(\Etwo)},
\end{align}
and along with the fact that $\mathcal{E}^{(\Eone)}\circ\mathcal{E}^{(\Eone)}=\mathcal{E}^{(\Eone)}$ one can show $\mathcal{E}(X)=X$ if and only if $\mathcal{E}^{(\Eone)}(X)=X$, same as \cref{eqn:fixed-point-on-one}. These properties also allow us to prove the frustration-freeness, where we show the proof in \cref{sec:proof-ff}. Being frustration-free, the global steady state of the parent Lindbladian is, again, the fixed point of the quantum channel, 
\begin{equation}
    \mathcal{E}^{(\Eone)\text{O/PBC}}_{1,2,\ldots,N}(\rho)=\rho.
    \label{eqn:OPBC}
\end{equation}

Let us now characterize the
parent Lindbladian fixed points, equivalently, the solution of \cref{eqn:OPBC}, 
in two different cases, depending on whether we consider open or periodic boundary conditions. As we will see, the degeneracy of the steady state is different for these two boundary conditions.

\subsection{Open boundary conditions} 
For an open finite quantum system of $N$ sites in a one-dimensional interval, labeled $1,\ldots, N$, we consider the local Lindbladian defined by
\begin{equation}
 \mathcal{L}^{(N)} = \sum_{i=1}^{N-1} \mathcal{L}_{i}, \quad \mathcal{L}_{i} := \mathcal{E}_{i}-\bo.
 \end{equation}
 As an immediate consequence of \cref{eq:E1Stab}, we have
\( \mathcal{L}_{i}( \rho^{(N)}(M)_B) = 0 \) for any boundary condition $B$, and consequently, \( \mathcal{L}^{(N)}( \rho^{(N)}(M)_B) = 0\). 
In fact, these are all the solutions of the fixed points of the linear map $\mathcal{E}_{1,2,\ldots,N}^{(\Eone),\text{OBC}}$, and equivalently the steady states of the global Lindbladian. Indeed, by expressing $\mathcal{E}_{1,2,\ldots,N}^{(\Eone),\text{OBC}}$ graphically as a generalization of \cref{eq:defE_3sites}, any steady state $\rho^{(N)}_\infty$ satisfies $\rho^{(N)}_\infty = \mathcal{E}^{(\Eone)}_{1,\ldots, N}(\rho^{(N)}_\infty) = \rho^{(N)}(M)_X$ for some boundary condition $X$.
The number of independent admissible boundary condition $B$ equals the dimension of the $C^*$-weak Hopf algebra $A$, and thus
\begin{equation}
    (\text{SSD})=\text{dim}(A). 
\end{equation}



\subsection{Periodic boundary conditions}
For a finite quantum chain system of $N$ sites, i.e.~with periodic boundary conditions, let
\begin{equation} \mathcal{L}^{(N)} = \sum_{i=1}^{N} \mathcal{L}_{i}, \quad \mathcal{L}_{i} := \mathcal{E}_i-\bo;
\end{equation}
where the sites $1$ and $N+1$ are identified.
For the sake of simplicity, let us note that it is sufficient to evaluate
$\mathcal{E}_{1,2,\ldots,N}^{(\Eone)\text{PBC}}$ with system size $N=2$
\begin{equation}
    \mathcal{E}_{1,2,\ldots,N}^{(\Eone)\text{PBC}}=\mathcal{E}^{(\Eone)}_{2} \circ \mathcal{E}^{(\Eone)}_{1},
\end{equation}
which 
is given by the expression
\begin{equation}
\begin{tikzpicture}[scale=1]
    \draw[Virtual] (0.986, 0.494) arc[start angle=90, end angle=180, x radius=0.494, y radius=-0.494];
    \draw[Virtual, -mid] (1.48, 1.482) -- (0.986, 1.482) arc[start angle=90, end angle=180, radius=0.494];
    \draw[Virtual] (1.868, 0.494) -- (0.986, 0.494);
    \draw[Virtual, -mid] (2.467, 1.482) -- (1.48, 1.482);
    \draw[Virtual, -mid] (1.233, 2.469) arc[start angle=-90, end angle=90, radius=0.494];
    \draw[Virtual] (0.739, 2.469) -- (0.492, 2.469);
    \node[anchor=center, text=Virtual] at (4.382, 2.457) {$\cdots$};
    \node[anchor=center, text=Virtual] at (4.382, 3.449) {$\cdots$};
    \draw[Virtual] (3.208, 0.494) -- (2.115, 0.494);
    \draw[Virtual, -mid] (3.208, 0.494) arc[start angle=-90, end angle=90, radius=0.494] -- (2.961, 1.482);
    \draw[Virtual] (2.475, 2.531) arc[start angle=118.956, end angle=151.044, x radius=0.494, y radius=-0.494];
    \draw[Virtual] (2.475, 3.395) arc[start angle=118.956, end angle=151.044, radius=0.494];
    \filldraw[ultra thin, fill=white] 
    (2.22, 2.963) circle[radius=0.247];
    \filldraw[ultra thin, fill=white] 
    (1.973, 0.494) circle[radius=0.247];
    \filldraw[ultra thin, fill=white] 
    (0.492, 0.988) circle[radius=0.247];
    \filldraw[ultra thin, fill=white] 
    (3.596, 2.469) circle[radius=0.247];
    \node[anchor=center, text=Virtual] at (0.234, 2.458) {$\cdots$};
    \node[anchor=center, text=Virtual] at (0.234, 3.446) {$\cdots$};
    \draw[Virtual] (0.739, 3.457) -- (0.492, 3.457);
    \node[anchor=center] at (0.492, 0.988) {$\Omega$};
    \node[anchor=center] at (3.596, 2.469) {$\Xi$};
    \node[anchor=center] at (2.22, 2.963) {$\Omega$};
    \node[anchor=center] at (1.973, 0.494) {$\Xi$};
    \draw[Virtual] (3.349, 2.469) -- (2.961, 2.469);
    \draw[Virtual] (4.09, 2.469) -- (3.843, 2.469);
    \draw[Virtual] (4.09, 3.457) -- (2.961, 3.457);
    \draw[-mid] (2.714, 1.729) -- (2.714, 2.223);
    \draw[bevel, -mid] (2.724, 1.235) arc[start angle=-163.45, end angle=0, radius=0.247] -- (3.208, 2.223) -- (3.208, 2.646) arc[start angle=0, end angle=163.45, radius=0.247];
    \whaMsymb{(2.714, 1.482)};
    \whaNsymb{(2.714, 2.469)};
    \draw[-mid] (0.986, 1.729) -- (0.986, 2.223);
    \draw[bevel, mid-] (0.996, 1.235) arc[start angle=-163.45, end angle=0, radius=0.247] -- (1.48, 2.223) -- (1.48, 2.646) arc[start angle=0, end angle=163.45, radius=0.247];
    \whaMsymb{(0.986, 1.482)};
    \whaNsymb{(0.986, 2.469)};
    \draw[-mid] (0.986, 3.704) -- (0.986, 3.951);
    \draw[bevel, mid-] (0.996, 3.21) arc[start angle=-163.45, end angle=0, radius=0.247] -- (1.48, 3.951);
    \whaMsymb{(0.986, 3.457)};
    \draw[-mid] (2.714, 3.704) -- (2.714, 3.951);
    \draw[bevel, mid-] (2.724, 3.21) arc[start angle=-163.45, end angle=0, radius=0.247] -- (3.208, 3.951);
    \whaMsymb{(2.714, 3.457)};
    \draw[-mid] (0.986, 0) -- (0.986, 0.247);
    \draw[bevel, -mid] (0.996, 0.741) arc[start angle=-163.45, end angle=0, x radius=0.247, y radius=-0.247] -- (1.48, 0);
    \whaNsymb{(0.986, 0.494)};
    \draw[-mid] (2.714, 0) -- (2.714, 0.247);
    \draw[bevel, -mid] (2.724, 0.741) arc[start angle=-163.45, end angle=0, x radius=0.247, y radius=-0.247] -- (3.208, 0);
    \whaNsymb{(2.714, 0.494)};
\end{tikzpicture}
\end{equation}
where the open left and right indices are identified, respectively, and again as a direct consequence of \cref{eq:BlackVsWhite}, this is equal to
\begin{equation}
\sum_a
\begin{tikzpicture}[scale=1]
    \draw[Virtual] (3.702, 0.494) -- (1.726, 0.494);
    \filldraw[ultra thin, fill=white] 
    (2.361, 0.494) circle[radius=0.247];
    \filldraw[ultra thin, fill=white] 
    (3.949, 0.494) circle[radius=0.247];
    \node[anchor=center] at (2.361, 0.494) {$\Xi$};
    \node[anchor=center] at (3.949, 0.494) {$\Xi$};
    \draw[Virtual, -mid] (4.443, 1.482) -- (3.455, 1.482);
    \draw[Virtual] (3.455, 1.482) -- (1.726, 1.482);
    \draw[Virtual] (4.443, 0.494) -- (4.196, 0.494);
    \node[anchor=center, text=Virtual] at (4.734, 0.482) {$\cdots$};
    \node[anchor=center, text=Virtual] at (4.735, 1.473) {$\cdots$};
    \draw[-mid] (1.48, 1.482) -- (1.48, 1.976);
    \draw[bevel, mid-] (1.49, 1.235) arc[start angle=-163.45, end angle=0, radius=0.247] -- (1.973, 1.976);
    \draw[-mid] (1.48, 0) -- (1.48, 0.247);
    \draw[bevel, -mid] (1.49, 0.741) arc[start angle=-163.45, end angle=0, x radius=0.247, y radius=-0.247] -- (1.973, 0);
    \whaMsymb{(1.48, 1.482)};
    \whaNsymb{(1.48, 0.494)};
    \draw[-mid] (2.961, 1.482) -- (2.961, 1.976);
    \draw[bevel, mid-] (2.971, 1.235) arc[start angle=-163.45, end angle=0, radius=0.247] -- (3.455, 1.976);
    \draw[-mid] (2.961, 0) -- (2.961, 0.247);
    \draw[bevel, -mid] (2.971, 0.741) arc[start angle=-163.45, end angle=0, x radius=0.247, y radius=-0.247] -- (3.455, 0);
    \whaMsymb{(2.961,1.482)};
    \whaNsymb{(2.961, 0.494)};
    \draw[Virtual, mid-] (1.233, 0.494) -- (0.492, 0.494);
    \draw[Virtual, -mid] (1.233, 1.482) -- (0.492, 1.482);
    \node[anchor=center, text=Virtual] at (0.234, 0.483) {$\cdots$};
    \node[anchor=center, text=Virtual] at (0.234, 1.471) {$\cdots$};
    \node[anchor=center, font=\footnotesize, text=Virtual] at (0.856, 1.621) {$a$};
    \node[anchor=center, font=\footnotesize, text=Virtual] at (0.857, 0.633) {$a$};
\end{tikzpicture}~.
\label{eqn:E1PBC}
\end{equation}

Let $B$ be central with respect to $M$ (namely, $B$ is proportional to the identity matrix in each sector $a$) and let us prove that the previous quantum channel $\mathcal{E}_{1,2}^{(\Eone)\text{PBC}}$ leaves $\rho^{(2)}(M)_B$ invariant:
\begin{equation}
\begin{aligned}
&\mathcal{E}_{1,2}^{(\Eone)\text{PBC}}(\rho^{(2)}(M)_B)=\\
&\sum_a
\begin{tikzpicture}[scale=1]
    \draw[Virtual] (3.702, 1.729) -- (1.726, 1.729);
    \filldraw[ultra thin, fill=white] 
    (2.361, 1.729) circle[radius=0.247];
    \filldraw[ultra thin, fill=white] 
    (3.949, 1.729) circle[radius=0.247];
    \node[anchor=center] at (2.361, 1.729) {$\Xi$};
    \node[anchor=center] at (3.949, 1.729) {$\Xi$};
    \draw[Virtual, -mid] (4.443, 2.716) -- (3.455, 2.716);
    \draw[Virtual] (3.455, 2.716) -- (1.726, 2.716);
    \draw[Virtual] (4.443, 1.729) -- (4.196, 1.729);
    \node[anchor=center, text=Virtual] at (4.734, 1.716) {$\cdots$};
    \node[anchor=center, text=Virtual] at (4.735, 2.708) {$\cdots$};
    \draw[-mid] (1.48, 2.716) -- (1.48, 3.21);
    \draw[bevel, mid-] (1.49, 2.47) arc[start angle=-163.45, end angle=0, radius=0.247] -- (1.973, 3.21);
    \draw[-mid] (1.48, 0.988) -- (1.48, 1.482);
    \draw[-mid] (2.961, 2.716) -- (2.961, 3.21);
    \draw[bevel, mid-] (2.971, 2.47) arc[start angle=-163.45, end angle=0, radius=0.247] -- (3.455, 3.21);
    \draw[-mid] (2.961, 0.988) -- (2.961, 1.482);
    \whaMsymb{(2.961, 2.716)};
    \draw[Virtual, mid-] (1.233, 1.729) -- (0.492, 1.729);
    \draw[Virtual, -mid] (1.233, 2.716) -- (0.492, 2.716);
    \node[anchor=center, text=Virtual] at (0.234, 1.718) {$\cdots$};
    \node[anchor=center, text=Virtual] at (0.234, 2.706) {$\cdots$};
    \node[anchor=center, font=\footnotesize, text=Virtual] at (0.856, 2.856) {$a$};
    \node[anchor=center, font=\footnotesize, text=Virtual] at (0.857, 1.868) {$a$};
    \draw[Virtual] (4.196, 0.741) -- (0.986, 0.741);
    \draw[bevel, -mid] (1.49, 1.975) arc[start angle=-163.45, end angle=0, x radius=0.247, y radius=-0.247] -- (1.973, 1.235) -- (1.973, 0.564) arc[start angle=0, end angle=163.45, x radius=0.247, y radius=-0.247];
    \draw[bevel, -mid] (2.971, 1.975) arc[start angle=-163.45, end angle=0, x radius=0.247, y radius=-0.247] -- (3.455, 1.235) -- (3.455, 0.564) arc[start angle=0, end angle=163.45, x radius=0.247, y radius=-0.247];
    \filldraw[ultra thin, fill=white] 
    (0.862, 0.741) circle[radius=0.247];
    \node[anchor=center] at (0.862, 0.741) {$B$};
    \draw[Virtual] (0.615, 0.741) arc[start angle=90, end angle=270, radius=0.37];
    \draw[Virtual] (4.196, 0) -- (0.615, 0);
    \draw[Virtual] (4.196, 0) arc[start angle=90, end angle=270, radius=-0.37];
    \whaNsymb{(1.48, 1.729)};
    \whaMsymb{(1.48, 2.716)};
    \whaMsymb{(2.961, 0.741)};
    \whaMsymb{(1.48, 0.741)};
    \whaNsymb{(2.961, 1.729)};
\end{tikzpicture}
\end{aligned}
\end{equation}
which, by virtue of \cref{eq:MvsNXi}, equals
\begin{equation}
\label{eqn:pbc-action}
\begin{aligned}
\\=&
\sum_a
\begin{tikzpicture}[scale=1]
    \draw[Virtual] (4.443, 1.235) -- (3.455, 1.235);
    \draw[Virtual] (3.455, 1.235) -- (1.726, 1.235);
    \draw[Virtual] (4.443, 0.247) -- (4.196, 0.247);
    \node[anchor=center, text=Virtual] at (4.734, 0.235) {$\cdots$};
    \node[anchor=center, text=Virtual] at (4.735, 1.226) {$\cdots$};
    \draw[-mid] (1.48, 1.482) -- (1.48, 1.729);
    \draw[bevel, mid-] (1.49, 0.988) arc[start angle=-163.45, end angle=0, radius=0.247] -- (1.973, 1.729);
    \whaMsymb{(1.48, 1.235)};
    \draw[-mid] (2.961, 1.482) -- (2.961, 1.729);
    \draw[bevel, mid-] (2.971, 0.988) arc[start angle=-163.45, end angle=0, radius=0.247] -- (3.455, 1.729);
    \whaMsymb{(2.961, 1.235)};
    \draw[Virtual] (2.22, 0.247) -- (3.702, 0.247);
    \draw[Virtual, -mid] (1.233, 1.235) -- (0.492, 1.235);
    \node[anchor=center, text=Virtual] at (0.234, 1.224) {$\cdots$};
    \node[anchor=center, font=\footnotesize, text=Virtual] at (0.856, 1.374) {$a$};
    \draw[Virtual, mid-] (1.233, 0.247) -- (0.492, 0.247);
    \node[anchor=center, font=\footnotesize, text=Virtual] at (0.856, 0.387) {$a$};
    \node[anchor=center, text=Virtual] at (0.234, 0.236) {$\cdots$};
    \filldraw[ultra thin, fill=white] (1.48, 0) arc[start angle=90, end angle=270, x radius=0.247, y radius=-0.247] -- (1.973, 0.494) arc[start angle=-90, end angle=90, x radius=0.247, y radius=-0.247] -- cycle;
    \node[anchor=center] at (1.726, 0.247) {$\Omega^{-1}$};
    \filldraw[ultra thin, fill=white] 
    (2.591, 0.247) circle[radius=0.247];
    \node[anchor=center] at (2.591, 0.247) {$B$};
    \filldraw[ultra thin, fill=white] 
    (3.949, 0.247) circle[radius=0.247];
    \node[anchor=center] at (3.949, 0.247) {$\Xi$};
\end{tikzpicture}
\end{aligned}
\end{equation}
and since both $B$ and $\Omega^{-1}$ are central with respect to $M$,
i.e.~proportional to the identity in each sector,
\begin{equation}
\begin{aligned}
\\=&
\sum_a
\begin{tikzpicture}[scale=1]
    \filldraw[ultra thin, fill=white] 
    (5.184, 0.247) circle[radius=0.247];
    \node[anchor=center] at (5.184, 0.247) {$\Xi$};
    \draw[Virtual] (5.678, 1.235) -- (4.69, 1.235);
    \draw[Virtual] (5.678, 0.247) -- (5.431, 0.247);
    \node[anchor=center, text=Virtual] at (5.969, 0.235) {$\cdots$};
    \node[anchor=center, text=Virtual] at (5.97, 1.226) {$\cdots$};
    \draw[Virtual] (1.233, 0.247) -- (4.937, 0.247);
    \filldraw[ultra thin, fill=white] 
    (1.48, 1.235) circle[radius=0.247];
    \node[anchor=center] at (1.48, 1.235) {$B$};
    \draw[Virtual] (2.467, 1.235) -- (1.726, 1.235);
    \draw[Virtual, -mid] (1.233, 1.235) -- (0.492, 1.235);
    \node[anchor=center, text=Virtual] at (0.234, 1.224) {$\cdots$};
    \node[anchor=center, font=\footnotesize, text=Virtual] at (0.856, 1.374) {$a$};
    \node[anchor=center, text=Virtual] at (0.234, 0.236) {$\cdots$};
    \node[anchor=center, font=\footnotesize, text=Virtual] at (0.856, 0.387) {$a$};
    \draw[Virtual, mid-] (1.233, 0.247) -- (0.492, 0.247);
    \draw[Virtual] (5.184, 1.235) -- (2.467, 1.235);
    \draw[-mid] (3.208, 1.235) -- (3.208, 1.729);
    \draw[bevel, mid-] (3.218, 0.988) arc[start angle=-163.45, end angle=0, radius=0.247] -- (3.702, 1.729);
    \whaMsymb{(3.208, 1.235)};
    \draw[-mid] (4.69, 1.235) -- (4.69, 1.729);
    \draw[bevel, mid-] (4.7, 0.988) arc[start angle=-163.45, end angle=0, radius=0.247] -- (5.184, 1.729);
    \whaMsymb{(4.69, 1.235)};
    \filldraw[ultra thin, fill=white] (2.079, 0.988) arc[start angle=90, end angle=270, x radius=0.247, y radius=-0.247] -- (2.573, 1.482) arc[start angle=-90, end angle=90, x radius=0.247, y radius=-0.247] -- cycle;
    \node[anchor=center] at (2.326, 1.235) {$\Omega^{-1}$};
\end{tikzpicture}
\end{aligned}
\end{equation}
finally, as a consequence of \cref{eq:XiOmega},
\begin{equation}
\begin{aligned}
\\=&
\begin{tikzpicture}[scale=1]
    \draw[Virtual] (4.937, 0.423) -- (3.949, 0.423);
    \draw[Virtual] (3.949, 0.423) -- (2.22, 0.423);
    \node[anchor=center, text=Virtual] at (5.229, 0.415) {$\cdots$};
    \draw[-mid] (1.973, 0.423) -- (1.973, 0.917);
    \draw[bevel, mid-] (1.984, 0.177) arc[start angle=-163.45, end angle=0, radius=0.247] -- (2.467, 0.917);
    \whaMsymb{(1.973, 0.423)};
    \draw[-mid] (3.455, 0.423) -- (3.455, 0.917);
    \draw[bevel, mid-] (3.465, 0.177) arc[start angle=-163.45, end angle=0, radius=0.247] -- (3.949, 0.917);
    \whaMsymb{(3.455, 0.423)};
    \draw[Virtual, -mid] (0.986, 0.423) -- (0.492, 0.423);
    \node[anchor=center, text=Virtual] at (0.234, 0.412) {$\cdots$};
    \filldraw[ultra thin, fill=white] 
    (1.232, 0.423) circle[radius=0.247];
    \node[anchor=center] at (1.232, 0.423) {$B$};
    \draw[Virtual] (1.726, 0.423) -- (1.479, 0.423);
\end{tikzpicture}=\rho^{(2)}(M)_B
\end{aligned}
\end{equation}
as we wanted to prove. 
Therefore $\mathcal{L}^{(2)}(\rho^{(2)}(M)_B) = 0$ for any boundary weight $B$ which is central with respect to $M$. 

This can be easily generalized to arbitrary $N\geq 2$ to conclude that $\rho^{(N)}(M)_B$ with a central $B$ is a steady state of $\mL^{(N)}$; and there are $g$ independent choices of $B$ that are central, where $g$ is the number of BNT. In fact, these are all the steady states of $\mL$ by noting that $\mathcal{E}_{1,2,\ldots,N}^{(\Eone)\text{PBC}}$ can be expressed as a generalization of \cref{eqn:E1PBC}, which is a summation over BNT elements $a$, and each summand is tracing and replacing. To summarize, when taking the periodic boundary condition, the steady-state degeneracy of $\mL$ is
\begin{equation}
    (\text{SSD})=g.
\end{equation}

Moreover, it can be proven that the minimal steady-state degeneracy is $g$ when the $C^*$-weak Hopf algebra is biconnected, meaning there is a unique vacuum sector $e$. In this case, the transfer matrix for a horizontal BNT element $E_a=0$ unless $a=e$. In other words, all BNT elements are simple except the one associated with the vacuum sector. Biconnected $C^*$-weak Hopf algebras include unitary fusion categories, which are the cases of primary interest (see precise definition in \cref{appendix:WHA}).

\begin{thm}
\label{thm:minimal-SSD-WHA}
    Consider a non-simple MPDO $\rho^{(N)}(M)$ generated by RFP tensor $M$ with $g$ BNT elements in the horizontal canonical form, constructed from a biconnected $C^*$-weak Hopf algebra.  A local Lindbladian $\mL^{(N)}$ that has $\rho^{(N)}(M)$ as its steady state must have at least $g$ independent steady states. 
\end{thm}

\begin{proof}
    Using the horizontal canonical form, one can decompose the state into $\rho^{(N)}(M)=\sum_{\hbnt=1}^g \rho^{(N)}(M_\hbnt)$. Let's construct an operator 
    \begin{equation}
        \begin{aligned}
            &\rho^{'(N)}(M)=\rho^{(N)}(M)-2\rho^{(N)}(M_1)\\
            &\quad =-\rho^{(N)}(M_1)+\rho^{(N)}(M_2)+\cdots+\rho^{(N)}(M_g), 
        \end{aligned}
    \end{equation}
    where $\hbnt=1$ is not the vacuum sector $e$. Since $E_\hbnt=0$ for all non-vacuum sectors, the reduced density matrix for $\rho^{(N)}(M_\hbnt)$ vanishes when $\hbnt\neq e$. 
    Therefore, when tracing out one site, the reduced density matrix for $\rho^{(N)}(M)$ and $\rho^{'(N)}(M)$ are the same and equal to the reduced density matrix for $\rho^{(N)}(M_e)$. Any term $\mL_i$ in a local Lindbladian $\mL$ cannot distinguish $\rho^{(N)}(M)$ and $\rho^{'(N)}(M)$; if $\rho^{(N)}(M)$ is a steady state of $\mL$, then $\rho^{'(N)}(M)$ must also be a steady state. Finally, there are $(g-1)$ independent ways to construct $\rho^{(N)}(M)-2\rho^{(N)}(M_\hbnt)$, $\hbnt\neq e$. We conclude that there are at least $g$ steady states of a local Lindbladian that has $\rho^{(N)}(M)$ as its steady state.  
\end{proof}



Therefore, under the periodic boundary condition for $\mL$, our constructed parent Lindbladian reaches the minimal steady-state degeneracy. 

Finally, we would like to point out the implications of the above results for the particular case of $C^*$-Hopf algebras.

\begin{thm}
    An MPDO RFP generated by $C^*$-Hopf algebras exhibits only short-range correlations and admits a commuting local rapid-mixing parent Lindbladian with minimal degenerate steady states under periodic boundary conditions. The minimal degeneracy equals $g$, the number of BNT elements in the horizontal canonical form of the MPDO tensor, which is larger than 1. 
\end{thm}

The fact that this class of RFPs only exhibits short-range correlations is a direct consequence of the existence of rapid-mixing local parent Lindbladian~\cite{bravyi2006lieb,poulin2010lieb}. It can also be understood by noting that the transfer matrix $E=\sum_\hbnt E_\hbnt = E_e$ is a rank-one matrix. This demonstrates that steady-state degeneracy does not necessarily imply the existence of long-range correlations for MPDOs. Moreover, it shows that steady-state degeneracy alone cannot classify mixed-state quantum phases. For example, MPDO RFPs generated by $C^*$-Hopf algebras and non-injective simple RFPs may share the same steady-state degeneracy, yet they belong to different phases. This is demonstrated in the following example.

\begin{exmp}[Boundary state of the Toric Code]
\label{exm:CZ2-parent-Lindbladian}
As a continuation of \cref{ex:CZ2_M_N}, let us note that the local quantum channel acting on two sites is given by the explicit formula
    \begin{equation}
        \mathcal{E}(X)=
        \frac{1}{4}\left(\tr(X)\bo_2^{\otimes 2}+\tr(\sigma_z^{\otimes 2} X) \sigma_z^{\otimes 2}\right).
    \end{equation}
    The corresponding local Lindbladian is rapid-mi\-xing. It is easy to verify that it brings the product state $|0\rangle\langle 0|^{\otimes N}$ to $(\bo^{\otimes N}+\sigma_z^{\otimes N})/2^N$ and the product states $|0\rangle\langle 0|^{\otimes (N-1)}\otimes |1\rangle\langle 1|$ to $(\bo^{\otimes N}-\sigma_z^{\otimes N})/2^N$. The steady states of the global Lindbladian take the form of
    \begin{equation}
        \rho^{(N)}_\infty=\frac{1}{2^N} \bo_2^{\otimes N} + c \sigma_z^{\otimes N}
    \end{equation}
    with $c\in(-2^{-N},2^{-N})$. 
    There are two independent steady states, independent of the boundary condition. Indeed, $g=\text{dim}(A)=2$ for $A=\mathbb{C}\mathbb{Z}_2$. However, unlike \cref{ex:pLGHZ}, the coefficient of $\bo_2^{\otimes N}$ is fixed to $1/2^N$ for $\rho^{(N)}_\infty$ to be a valid density matrix. This is a general feature for non-simple MPDO. 
    This state is in the trivial phase because the parent Lindbladian drives a product state fast to it, and therefore, it is in a different phase from the classical GHZ state. 
\end{exmp}

We provide a more complicated example for the parent Lindbladian construction in \cref{sec:Fib} for the Fibonacci $C^*$-weak Hopf algebra $A=\text{Fib}$. In this example, $g=2$ while $\text{dim}(A)=13$. 

\subsection{Beyond $C^*$-weak Hopf algera: a case study}
\label{sec:beyond}

Having reached general conclusions for non-simple RFPs constructed from $C^*$-weak Hopf algebras, in this final section, we present a case study of a non-simple RFP beyond the framework of $C^*$-weak Hopf algebras~\cite{lessa2024mixed,liu2025trading}. The density matrix under concern is the CZX density matrix associated with a nontrivial 3-cocycle. 

\begin{exmp}[CZX density matrix]
Consider the density matrix of even system size $N$~\cite{lessa2024mixed}
    \begin{equation}
    \rho_{CZX}^{(N)} = \frac{1}{2^N}(\bo_2^{\otimes N} + U^{(N)}_{CZX}),
\end{equation}
where $U^{(N)}_{CZX}$ is a nontrivial representation of $\mathbb{Z}_2$ group element associated with a nontrivial 3-cocycle,
\begin{equation}
    U^{(N)}_{CZX} = \prod_{i=1}^N CZ_{i,i+1} \prod_{i=1}^N (\sigma_x)_i.
\end{equation}
This density matrix can be represented by an MPDO tensor $\tilde{M}$ with $d=2,D=3$, 
\begin{equation}
\begin{aligned}
    \tilde{M}^{00}&=\frac{1}{2}\begin{pmatrix}
        1 & 0 & 0\\
        0 & 0 & 0\\
        0 & 0 & 0
    \end{pmatrix}, \quad
    \tilde{M}^{11}=\frac{1}{2}\begin{pmatrix}
        1 & 0 & 0\\
        0 & 0 & 0\\
        0 & 0 & 0
    \end{pmatrix}\\
    \tilde{M}^{01}&=\frac{1}{2}\begin{pmatrix}
        0 & 0 & 0\\
        0 & 1 & 1\\
        0 & 0 & 0
    \end{pmatrix},\quad 
     \tilde{M}^{10}=\frac{1}{2}\begin{pmatrix}
        0 & 0 & 0\\
        0 & 0 & 0\\
        0 & 1 & -1
    \end{pmatrix}. 
\end{aligned}
\end{equation}
This is an example of a non-simple RFP with two elements in the horizontal BNT, the first with bond dimension $D=1$ generating $\bo_2^{\otimes N}$ and the second with $D=2$ generating $U_{CZX}^{(N)}$. After tracing any single site, the MPDO is proportional to a maximally mixed state and does not have long-range correlations. Nevertheless, $\rho_{CZX}^{(N)}$ is in the non-trivial phase due to multipartite entanglement~\cite{lessa2024mixed}.
\end{exmp}

One can prove that this MPDO is an RFP after blocking two sites by showing the vertical canonical form verifies~\cref{thm:non-simple} (see details in~\cref{app:CZX}). However, this RFP lacks the $C^*$-weak Hopf algebra structure, as carefully analyzed in \cite{liu2025trading}. In particular, the counit is absent in $A$ and $A$ is not cosemisimple, therefore the previous constructions of renormalization channels $\mT$ and $\mS$ from $C^*$-weak Hopf algebra do not apply. Nevertheless, one can still construct the renormalization channels from the vertical canonical form, as described below.  

\begin{exmp}[Parent Lindbladian for the CZX density matrix]
\label{ex:CZX-ex}
    Using the vertical canonical form of $M$ and $K$, one can construct the channel $\mT$ and $\mS$ as follows. 
    Denote $M$ as the two-site tensor (blocking two sites of $\tilde{M}$) and $K$ as the four-site tensor (blocking four sites of $\tilde{M}$) with vertical canonical form as in~\cref{eqn:Mvertical,eqn:Kvertical}, 
    \begin{align}
    &U M_{(\alpha\beta)} U^\dg=\bigoplus_{\lambda=1}^2 \mu_\lambda \otimes M_{(\alpha\beta),\lambda}\\
    &V K_{(\alpha\beta)} V^\dg=\bigoplus_{\lambda=1}^2 \nu_\lambda \otimes M_{(\alpha\beta),\lambda},
\end{align}
    that verifies the conditions in~\cref{thm:non-simple}, $m_\lambda:=\tr[\mu_\lambda]=n_\lambda:=\tr[\nu_\lambda]$. There are two elements in the vertical BNT, each with physical dimension $d_\lambda=2$, and both $\mu_1,\mu_2$ are one-dimensional, both $\nu_1$ and $\nu_2$ are four-dimensional. Using the vertical canonical form, the channel $\mT$ bringing $M$ to $K$ is
\begin{equation}
\begin{aligned}
    &\mT(X)=V^\dg\left[\bigoplus_{\lambda=1}^2(\frac{\nu_\lambda}{m_\lambda}\otimes P_\lambda UXU^\dg P_\lambda^\dg)\right]V
\end{aligned}
\end{equation}
where 
\begin{equation}
    P_1 = \left(\begin{array}{c|c}
        \bo_2 & 0_{2\times 2} 
    \end{array}\right),\quad 
    P_2 = \left(\begin{array}{c|c}
        0_{2\times 2}  & \bo_2
    \end{array}\right)
\end{equation}
are the isometries projecting to the first and second BNT, respectively.

Similarly, the four-site MPDO $K$ to two-site MPDO $M$ channel $\mS$ is
\begin{equation}
    \begin{aligned}
        &\mS(X)=U^\dg \left[\bigoplus_{\lambda=1}^2\tr_1(P_\lambda' V X V^\dg (P_\lambda')^\dg) \right]U
    \end{aligned}
\end{equation}
where 
\begin{equation}
    P_1' = \left(\begin{array}{c|c}
        \bo_8 & 0_{8\times 8} 
    \end{array}\right),\quad 
    P_2' = \left(\begin{array}{c|c}
        0_{8\times 8}  & \bo_8
    \end{array}\right)
\end{equation}
and $\tr_1$ denotes tracing over the multiplicity degrees of freedom. The explicit forms of $M_\lambda,\mu_\lambda,\nu_\lambda,U,V$ are given in~\cref{app:CZX}. Knowing the renormalization channels $\mT,\mS$, the parent Lindbladian is constructed by the standard algorithm as in~\cref{sec:general-comment}. 

This parent Lindbladian is commuting and therefore rapid-mixing. The steady states of the global Lindbladian take the form of
\begin{equation}
    \rho_\infty^{(N)}=\frac{1}{2^N}(\bo_2^{\otimes N}+ c U_{CZX}^{(N)}),\quad c\in[-1,1],
\end{equation}
hence the steady-state degeneracy equals 2. 
We leave the detailed derivation in~\cref{app:CZX}. 
\end{exmp}

\section{Outlook}
\label{sec:outlook}
In this paper, we systematically construct the parent Lindbladians for the renormalization fixed-points of MPDO. Our constructed parent Lindbladian satisfies the important properties of local, frustration-free, and reaches minimal steady-state degeneracy. In comparison to pure states, the parent Lindbladian has unique properties that are distinct from the parent Hamiltonian.

There are at least three interesting questions that deserve further study. 
The first open question concerns the parent Lindbladians with symmetries. Given the construction of our parent Lindbladians, they naturally inherit the symmetries present in the state.  As such, this construction could be potentially useful in exploring symmetry-protected topological phases for mixed states~\cite{coser2019classification,de2022symmetry,ma2023average,sang2024mixed,rakovszky2024defining,huang2024hydrodynamics,wang2023intrinsic,gu2024spontaneous,ma2024symmetry,you2024intrinsic,ellison2024towards,lessa2024strong,chen2024unconventional}. However, in this work, we have focused exclusively on topological phases without symmetry protection. Investigating parent Lindbladians with symmetries is left for future study. 

Second, can we construct the parent Lindbladian for a general MPDO that is not at RFP? A natural generalization is still taking $\mathcal{E}=\mT\circ\mS$ and $\mathcal{L}=\mathcal{E}-\bo$, but with $\mT,\mS$ being the renormalization quantum channels for a general MPDO,
 \begin{align}
        \begin{array}{c}
        \begin{tikzpicture}[scale=1.,baseline={([yshift=-0.75ex] current bounding box.center)}
        ]
        \whaMorg{(0,0)}{2}{$\tilde{M}$}{mpdotcolor};
        \whaM{(3*\whadx,0)}{2};
        \whaM{(4*\whadx,0)}{2};
        \draw [->, thick](1.,0.4) to [out=45,in=135] (3*\whadx-1.,0.4);
        \draw (1.5*\whadx,0.85) node {$\mT$};
        \draw [->, thick](3*\whadx-1,-0.4) to [out=225,in=-45] (1.,-0.4);
        \draw (1.5*\whadx,-0.32) node {$\mS$};
        \end{tikzpicture}
        \end{array}
        \,,
    \end{align}
where the tensor $\tilde{M}$ after coarse-graining is not the same as $M$. 
We note that the parent Hamiltonian construction $h=\bo-VV^\dg$ applies to general MPS that is not at RFP, and it is natural to expect the above generalization to also apply to general MPDOs. One difficulty, however, is that there is no universal algorithm to find the renormalization quantum channels $\mT$ and $\mS$ when away from MPDO RFPs; furthermore, such channels do not exist for certain MPDOs. We note a recent progress in \cite{kato2024exact}, and it is reasonable to expect that our construction can be applied to the MPDO where renormalization channels exist and are known. 

The third question comes from an interesting numerical observation that the parent Lindbladians constructed from the vertical canonical form and from the $C^*$-weak Hopf algebra (which is associated with the horizontal canonical form) appears identical for the examples studied (see \cref{sec:Fib}). This finding hints at a potential duality between the horizontal and vertical canonical forms at MPDO RFPs, requiring further study.

\section*{Acknowledgement}
We are very grateful to Zhiyuan Wang, Ángela Capel, and Kohtaro Kato for discussions.
This work was in part supported by Munich Quantum Valley (MQV), which is supported by the Bavarian state government with funds from the Hightech Agenda Bayern Plus; and by Munich Center for Quantum Science and Technology (MCQST), funded by the Deutsche Forschungsgemeinschaft (DFG) under Germany’s Excellence Strategy (EXC2111-390814868); and by funding from the German Federal Ministry of Education and Research (BMBF) through FermiQP.
This work was in part supported by the Alexander von Humboldt Foundation.
This work was in part supported by the Spanish Ministry of Science and Innovation MCIN/AEI/10.13039/501100011033 (CEX2023-001347-S, PID2020-113523GB-I00, PID2023-146758NB-I00), Universidad Complutense de Madrid (FEI-EU-22-06), and the Ministry for Digital Transformation and of Civil Service of the Spanish Government through the QUANTUM ENIA project call – Quantum Spain project, and by the European Union through the Recovery, Transformation and Resilience Plan – NextGenerationEU within the framework of the Digital Spain 2026 Agenda.
This research was supported in part by Perimeter Institute for Theoretical Physics. Research at Perimeter Institute is supported by the Government of Canada through the Department of Innovation, Science, and Economic Development, and by the Province of Ontario through the Ministry of Colleges and Universities.

\bibliography{main}

\begin{thebibliography}{70}%
\makeatletter
\providecommand \@ifxundefined [1]{%
 \@ifx{#1\undefined}
}%
\providecommand \@ifnum [1]{%
 \ifnum #1\expandafter \@firstoftwo
 \else \expandafter \@secondoftwo
 \fi
}%
\providecommand \@ifx [1]{%
 \ifx #1\expandafter \@firstoftwo
 \else \expandafter \@secondoftwo
 \fi
}%
\providecommand \natexlab [1]{#1}%
\providecommand \enquote  [1]{``#1''}%
\providecommand \bibnamefont  [1]{#1}%
\providecommand \bibfnamefont [1]{#1}%
\providecommand \citenamefont [1]{#1}%
\providecommand \href@noop [0]{\@secondoftwo}%
\providecommand \href [0]{\begingroup \@sanitize@url \@href}%
\providecommand \@href[1]{\@@startlink{#1}\@@href}%
\providecommand \@@href[1]{\endgroup#1\@@endlink}%
\providecommand \@sanitize@url [0]{\catcode `\\12\catcode `\$12\catcode
  `\&12\catcode `\#12\catcode `\^12\catcode `\_12\catcode `\%12\relax}%
\providecommand \@@startlink[1]{}%
\providecommand \@@endlink[0]{}%
\providecommand \url  [0]{\begingroup\@sanitize@url \@url }%
\providecommand \@url [1]{\endgroup\@href {#1}{\urlprefix }}%
\providecommand \urlprefix  [0]{URL }%
\providecommand \Eprint [0]{\href }%
\providecommand \doibase [0]{https://doi.org/}%
\providecommand \selectlanguage [0]{\@gobble}%
\providecommand \bibinfo  [0]{\@secondoftwo}%
\providecommand \bibfield  [0]{\@secondoftwo}%
\providecommand \translation [1]{[#1]}%
\providecommand \BibitemOpen [0]{}%
\providecommand \bibitemStop [0]{}%
\providecommand \bibitemNoStop [0]{.\EOS\space}%
\providecommand \EOS [0]{\spacefactor3000\relax}%
\providecommand \BibitemShut  [1]{\csname bibitem#1\endcsname}%
\let\auto@bib@innerbib\@empty
\bibitem [{\citenamefont {{Cirac}}\ \emph {et~al.}(2021)\citenamefont
  {{Cirac}}, \citenamefont {{P{\'e}rez-Garc{\'\i}a}}, \citenamefont
  {{Schuch}},\ and\ \citenamefont {{Verstraete}}}]{cirac2021matrix}%
  \BibitemOpen
  \bibfield  {author} {\bibinfo {author} {\bibfnamefont {J.~I.}\ \bibnamefont
  {{Cirac}}}, \bibinfo {author} {\bibfnamefont {D.}~\bibnamefont
  {{P{\'e}rez-Garc{\'\i}a}}}, \bibinfo {author} {\bibfnamefont
  {N.}~\bibnamefont {{Schuch}}},\ and\ \bibinfo {author} {\bibfnamefont
  {F.}~\bibnamefont {{Verstraete}}},\ }\bibfield  {title} {\bibinfo {title}
  {{Matrix product states and projected entangled pair states: Concepts,
  symmetries, theorems}},\ }\href
  {https://doi.org/10.1103/RevModPhys.93.045003} {\bibfield  {journal}
  {\bibinfo  {journal} {Reviews of Modern Physics}\ }\textbf {\bibinfo {volume}
  {93}},\ \bibinfo {eid} {045003} (\bibinfo {year} {2021})}\BibitemShut
  {NoStop}%
\bibitem [{\citenamefont {Fannes}\ \emph {et~al.}(1992)\citenamefont {Fannes},
  \citenamefont {Nachtergaele},\ and\ \citenamefont
  {Werner}}]{fannes1992finitely}%
  \BibitemOpen
  \bibfield  {author} {\bibinfo {author} {\bibfnamefont {M.}~\bibnamefont
  {Fannes}}, \bibinfo {author} {\bibfnamefont {B.}~\bibnamefont
  {Nachtergaele}},\ and\ \bibinfo {author} {\bibfnamefont {R.~F.}\ \bibnamefont
  {Werner}},\ }\bibfield  {title} {\bibinfo {title} {Finitely correlated states
  on quantum spin chains},\ }\href {https://doi.org/10.1007/BF02099178}
  {\bibfield  {journal} {\bibinfo  {journal} {Communications in mathematical
  physics}\ }\textbf {\bibinfo {volume} {144}},\ \bibinfo {pages} {443}
  (\bibinfo {year} {1992})}\BibitemShut {NoStop}%
\bibitem [{\citenamefont {Hastings}(2007)}]{hastings2007area}%
  \BibitemOpen
  \bibfield  {author} {\bibinfo {author} {\bibfnamefont {M.~B.}\ \bibnamefont
  {Hastings}},\ }\bibfield  {title} {\bibinfo {title} {An area law for
  one-dimensional quantum systems},\ }\href
  {https://doi.org/10.1088/1742-5468/2007/08/P08024} {\bibfield  {journal}
  {\bibinfo  {journal} {Journal of statistical mechanics: {T}heory and
  experiment}\ }\textbf {\bibinfo {volume} {2007}},\ \bibinfo {pages} {P08024}
  (\bibinfo {year} {2007})}\BibitemShut {NoStop}%
\bibitem [{\citenamefont {Hastings}(2006)}]{hastings2006solving}%
  \BibitemOpen
  \bibfield  {author} {\bibinfo {author} {\bibfnamefont {M.~B.}\ \bibnamefont
  {Hastings}},\ }\bibfield  {title} {\bibinfo {title} {Solving gapped
  {H}amiltonians locally},\ }\href {https://doi.org/10.1103/PhysRevB.73.085115}
  {\bibfield  {journal} {\bibinfo  {journal} {Physical Review B}\ }\textbf
  {\bibinfo {volume} {73}},\ \bibinfo {pages} {085115} (\bibinfo {year}
  {2006})}\BibitemShut {NoStop}%
\bibitem [{\citenamefont {Molnar}\ \emph {et~al.}(2015)\citenamefont {Molnar},
  \citenamefont {Schuch}, \citenamefont {Verstraete},\ and\ \citenamefont
  {Cirac}}]{molnar2015approximating}%
  \BibitemOpen
  \bibfield  {author} {\bibinfo {author} {\bibfnamefont {A.}~\bibnamefont
  {Molnar}}, \bibinfo {author} {\bibfnamefont {N.}~\bibnamefont {Schuch}},
  \bibinfo {author} {\bibfnamefont {F.}~\bibnamefont {Verstraete}},\ and\
  \bibinfo {author} {\bibfnamefont {J.~I.}\ \bibnamefont {Cirac}},\ }\bibfield
  {title} {\bibinfo {title} {Approximating {G}ibbs states of local
  {H}amiltonians efficiently with projected entangled pair states},\ }\href
  {https://doi.org/10.1103/PhysRevB.91.045138} {\bibfield  {journal} {\bibinfo
  {journal} {Physical Review B}\ }\textbf {\bibinfo {volume} {91}},\ \bibinfo
  {pages} {045138} (\bibinfo {year} {2015})}\BibitemShut {NoStop}%
\bibitem [{\citenamefont {Cirac}\ \emph {et~al.}(2011)\citenamefont {Cirac},
  \citenamefont {Poilblanc}, \citenamefont {Schuch},\ and\ \citenamefont
  {Verstraete}}]{cirac2011entanglement}%
  \BibitemOpen
  \bibfield  {author} {\bibinfo {author} {\bibfnamefont {J.~I.}\ \bibnamefont
  {Cirac}}, \bibinfo {author} {\bibfnamefont {D.}~\bibnamefont {Poilblanc}},
  \bibinfo {author} {\bibfnamefont {N.}~\bibnamefont {Schuch}},\ and\ \bibinfo
  {author} {\bibfnamefont {F.}~\bibnamefont {Verstraete}},\ }\bibfield  {title}
  {\bibinfo {title} {Entanglement spectrum and boundary theories with projected
  entangled-pair states},\ }\href {https://doi.org/10.1103/PhysRevB.83.245134}
  {\bibfield  {journal} {\bibinfo  {journal} {Physical Review B}\ }\textbf
  {\bibinfo {volume} {83}},\ \bibinfo {pages} {245134} (\bibinfo {year}
  {2011})}\BibitemShut {NoStop}%
\bibitem [{\citenamefont {Schollw{\"o}ck}(2011)}]{schollwock2011density}%
  \BibitemOpen
  \bibfield  {author} {\bibinfo {author} {\bibfnamefont {U.}~\bibnamefont
  {Schollw{\"o}ck}},\ }\bibfield  {title} {\bibinfo {title} {The density-matrix
  renormalization group in the age of matrix product states},\ }\href
  {https://doi.org/https://doi.org/10.1016/j.aop.2010.09.012} {\bibfield
  {journal} {\bibinfo  {journal} {Annals of Physics}\ }\textbf {\bibinfo
  {volume} {326}},\ \bibinfo {pages} {96} (\bibinfo {year} {2011})}\BibitemShut
  {NoStop}%
\bibitem [{\citenamefont {Levin}\ and\ \citenamefont
  {Wen}(2005)}]{levin2005string}%
  \BibitemOpen
  \bibfield  {author} {\bibinfo {author} {\bibfnamefont {M.~A.}\ \bibnamefont
  {Levin}}\ and\ \bibinfo {author} {\bibfnamefont {X.-G.}\ \bibnamefont
  {Wen}},\ }\bibfield  {title} {\bibinfo {title} {String-net condensation: A
  physical mechanism for topological phases},\ }\href
  {https://doi.org/10.1103/PhysRevB.71.045110} {\bibfield  {journal} {\bibinfo
  {journal} {Physical Review B}\ }\textbf {\bibinfo {volume} {71}},\ \bibinfo
  {pages} {045110} (\bibinfo {year} {2005})}\BibitemShut {NoStop}%
\bibitem [{\citenamefont {Buerschaper}\ \emph {et~al.}(2009)\citenamefont
  {Buerschaper}, \citenamefont {Aguado},\ and\ \citenamefont
  {Vidal}}]{buerschaper2009explicit}%
  \BibitemOpen
  \bibfield  {author} {\bibinfo {author} {\bibfnamefont {O.}~\bibnamefont
  {Buerschaper}}, \bibinfo {author} {\bibfnamefont {M.}~\bibnamefont
  {Aguado}},\ and\ \bibinfo {author} {\bibfnamefont {G.}~\bibnamefont
  {Vidal}},\ }\bibfield  {title} {\bibinfo {title} {Explicit tensor network
  representation for the ground states of string-net models},\ }\href
  {https://doi.org/10.1103/PhysRevB.79.085119} {\bibfield  {journal} {\bibinfo
  {journal} {Physical Review B}\ }\textbf {\bibinfo {volume} {79}},\ \bibinfo
  {pages} {085119} (\bibinfo {year} {2009})}\BibitemShut {NoStop}%
\bibitem [{\citenamefont {Levin}\ and\ \citenamefont
  {Wen}(2006)}]{levin2006detecting}%
  \BibitemOpen
  \bibfield  {author} {\bibinfo {author} {\bibfnamefont {M.}~\bibnamefont
  {Levin}}\ and\ \bibinfo {author} {\bibfnamefont {X.-G.}\ \bibnamefont
  {Wen}},\ }\bibfield  {title} {\bibinfo {title} {Detecting topological order
  in a ground state wave function},\ }\href
  {https://doi.org/10.1103/PhysRevLett.96.110405} {\bibfield  {journal}
  {\bibinfo  {journal} {Physical review letters}\ }\textbf {\bibinfo {volume}
  {96}},\ \bibinfo {pages} {110405} (\bibinfo {year} {2006})}\BibitemShut
  {NoStop}%
\bibitem [{\citenamefont {Kitaev}\ and\ \citenamefont
  {Kong}(2012)}]{kitaev2012models}%
  \BibitemOpen
  \bibfield  {author} {\bibinfo {author} {\bibfnamefont {A.}~\bibnamefont
  {Kitaev}}\ and\ \bibinfo {author} {\bibfnamefont {L.}~\bibnamefont {Kong}},\
  }\bibfield  {title} {\bibinfo {title} {Models for gapped boundaries and
  domain walls},\ }\href {https://doi.org/10.1007/s00220-012-1500-5} {\bibfield
   {journal} {\bibinfo  {journal} {Communications in Mathematical Physics}\
  }\textbf {\bibinfo {volume} {313}},\ \bibinfo {pages} {351} (\bibinfo {year}
  {2012})}\BibitemShut {NoStop}%
\bibitem [{\citenamefont {Lin}\ and\ \citenamefont
  {Levin}(2014)}]{lin2014generalizations}%
  \BibitemOpen
  \bibfield  {author} {\bibinfo {author} {\bibfnamefont {C.-H.}\ \bibnamefont
  {Lin}}\ and\ \bibinfo {author} {\bibfnamefont {M.}~\bibnamefont {Levin}},\
  }\bibfield  {title} {\bibinfo {title} {Generalizations and limitations of
  string-net models},\ }\href {https://doi.org/10.1103/PhysRevB.89.195130}
  {\bibfield  {journal} {\bibinfo  {journal} {Physical Review B}\ }\textbf
  {\bibinfo {volume} {89}},\ \bibinfo {pages} {195130} (\bibinfo {year}
  {2014})}\BibitemShut {NoStop}%
\bibitem [{\citenamefont {Kim}\ and\ \citenamefont
  {Ranard}(2024)}]{kim2024classifying}%
  \BibitemOpen
  \bibfield  {author} {\bibinfo {author} {\bibfnamefont {I.~H.}\ \bibnamefont
  {Kim}}\ and\ \bibinfo {author} {\bibfnamefont {D.}~\bibnamefont {Ranard}},\
  }\bibfield  {title} {\bibinfo {title} {Classifying {2D} topological phases:
  mapping ground states to string-nets},\ }\href
  {https://arxiv.org/abs/2405.17379} {\bibfield  {journal} {\bibinfo  {journal}
  {arXiv:2405.17379}\ } (\bibinfo {year} {2024})}\BibitemShut {NoStop}%
\bibitem [{\citenamefont {Chen}\ \emph {et~al.}(2011)\citenamefont {Chen},
  \citenamefont {Gu},\ and\ \citenamefont {Wen}}]{chen2011classification}%
  \BibitemOpen
  \bibfield  {author} {\bibinfo {author} {\bibfnamefont {X.}~\bibnamefont
  {Chen}}, \bibinfo {author} {\bibfnamefont {Z.-C.}\ \bibnamefont {Gu}},\ and\
  \bibinfo {author} {\bibfnamefont {X.-G.}\ \bibnamefont {Wen}},\ }\bibfield
  {title} {\bibinfo {title} {Classification of gapped symmetric phases in
  one-dimensional spin systems},\ }\href
  {https://doi.org/10.1103/PhysRevB.83.035107} {\bibfield  {journal} {\bibinfo
  {journal} {Phys. Rev. B}\ }\textbf {\bibinfo {volume} {83}},\ \bibinfo
  {pages} {035107} (\bibinfo {year} {2011})}\BibitemShut {NoStop}%
\bibitem [{\citenamefont {Schuch}\ \emph {et~al.}(2011)\citenamefont {Schuch},
  \citenamefont {P\'erez-Garc\'{\i}a},\ and\ \citenamefont
  {Cirac}}]{schuch2011classifying}%
  \BibitemOpen
  \bibfield  {author} {\bibinfo {author} {\bibfnamefont {N.}~\bibnamefont
  {Schuch}}, \bibinfo {author} {\bibfnamefont {D.}~\bibnamefont
  {P\'erez-Garc\'{\i}a}},\ and\ \bibinfo {author} {\bibfnamefont
  {I.}~\bibnamefont {Cirac}},\ }\bibfield  {title} {\bibinfo {title}
  {Classifying quantum phases using matrix product states and projected
  entangled pair states},\ }\href {https://doi.org/10.1103/PhysRevB.84.165139}
  {\bibfield  {journal} {\bibinfo  {journal} {Phys. Rev. B}\ }\textbf {\bibinfo
  {volume} {84}},\ \bibinfo {pages} {165139} (\bibinfo {year}
  {2011})}\BibitemShut {NoStop}%
\bibitem [{\citenamefont {{Coser}}\ and\ \citenamefont
  {{P{\'e}rez-Garc{\'\i}a}}(2019)}]{coser2019classification}%
  \BibitemOpen
  \bibfield  {author} {\bibinfo {author} {\bibfnamefont {A.}~\bibnamefont
  {{Coser}}}\ and\ \bibinfo {author} {\bibfnamefont {D.}~\bibnamefont
  {{P{\'e}rez-Garc{\'\i}a}}},\ }\bibfield  {title} {\bibinfo {title}
  {{Classification of phases for mixed states via fast dissipative
  evolution}},\ }\href {https://doi.org/10.22331/q-2019-08-12-174} {\bibfield
  {journal} {\bibinfo  {journal} {Quantum}\ }\textbf {\bibinfo {volume} {3}},\
  \bibinfo {pages} {174} (\bibinfo {year} {2019})}\BibitemShut {NoStop}%
\bibitem [{\citenamefont {de~Groot}\ \emph {et~al.}(2022)\citenamefont
  {de~Groot}, \citenamefont {Turzillo},\ and\ \citenamefont
  {Schuch}}]{de2022symmetry}%
  \BibitemOpen
  \bibfield  {author} {\bibinfo {author} {\bibfnamefont {C.}~\bibnamefont
  {de~Groot}}, \bibinfo {author} {\bibfnamefont {A.}~\bibnamefont {Turzillo}},\
  and\ \bibinfo {author} {\bibfnamefont {N.}~\bibnamefont {Schuch}},\
  }\bibfield  {title} {\bibinfo {title} {Symmetry protected topological order
  in open quantum systems},\ }\href {https://doi.org/10.22331/q-2022-11-10-856}
  {\bibfield  {journal} {\bibinfo  {journal} {Quantum}\ }\textbf {\bibinfo
  {volume} {6}},\ \bibinfo {pages} {856} (\bibinfo {year} {2022})}\BibitemShut
  {NoStop}%
\bibitem [{\citenamefont {Lee}\ \emph {et~al.}(2025)\citenamefont {Lee},
  \citenamefont {You},\ and\ \citenamefont {Xu}}]{lee2025symmetry}%
  \BibitemOpen
  \bibfield  {author} {\bibinfo {author} {\bibfnamefont {J.~Y.}\ \bibnamefont
  {Lee}}, \bibinfo {author} {\bibfnamefont {Y.-Z.}\ \bibnamefont {You}},\ and\
  \bibinfo {author} {\bibfnamefont {C.}~\bibnamefont {Xu}},\ }\bibfield
  {title} {\bibinfo {title} {Symmetry protected topological phases under
  decoherence},\ }\href
  {https://doi.org/https://doi.org/10.22331/q-2025-01-23-1607} {\bibfield
  {journal} {\bibinfo  {journal} {Quantum}\ }\textbf {\bibinfo {volume} {9}},\
  \bibinfo {pages} {1607} (\bibinfo {year} {2025})}\BibitemShut {NoStop}%
\bibitem [{\citenamefont {{Ma}}\ and\ \citenamefont
  {{Wang}}(2023)}]{ma2023average}%
  \BibitemOpen
  \bibfield  {author} {\bibinfo {author} {\bibfnamefont {R.}~\bibnamefont
  {{Ma}}}\ and\ \bibinfo {author} {\bibfnamefont {C.}~\bibnamefont {{Wang}}},\
  }\bibfield  {title} {\bibinfo {title} {{Average Symmetry-Protected
  Topological Phases}},\ }\href {https://doi.org/10.1103/PhysRevX.13.031016}
  {\bibfield  {journal} {\bibinfo  {journal} {Physical Review X}\ }\textbf
  {\bibinfo {volume} {13}},\ \bibinfo {eid} {031016} (\bibinfo {year}
  {2023})}\BibitemShut {NoStop}%
\bibitem [{\citenamefont {Sang}\ \emph {et~al.}(2024)\citenamefont {Sang},
  \citenamefont {Zou},\ and\ \citenamefont {Hsieh}}]{sang2024mixed}%
  \BibitemOpen
  \bibfield  {author} {\bibinfo {author} {\bibfnamefont {S.}~\bibnamefont
  {Sang}}, \bibinfo {author} {\bibfnamefont {Y.}~\bibnamefont {Zou}},\ and\
  \bibinfo {author} {\bibfnamefont {T.~H.}\ \bibnamefont {Hsieh}},\ }\bibfield
  {title} {\bibinfo {title} {Mixed-state quantum phases: Renormalization and
  quantum error correction},\ }\href
  {https://doi.org/10.1103/PhysRevX.14.031044} {\bibfield  {journal} {\bibinfo
  {journal} {Physical Review X}\ }\textbf {\bibinfo {volume} {14}},\ \bibinfo
  {pages} {031044} (\bibinfo {year} {2024})}\BibitemShut {NoStop}%
\bibitem [{\citenamefont {Rakovszky}\ \emph {et~al.}(2024)\citenamefont
  {Rakovszky}, \citenamefont {Gopalakrishnan},\ and\ \citenamefont {von
  Keyserlingk}}]{rakovszky2024defining}%
  \BibitemOpen
  \bibfield  {author} {\bibinfo {author} {\bibfnamefont {T.}~\bibnamefont
  {Rakovszky}}, \bibinfo {author} {\bibfnamefont {S.}~\bibnamefont
  {Gopalakrishnan}},\ and\ \bibinfo {author} {\bibfnamefont {C.}~\bibnamefont
  {von Keyserlingk}},\ }\bibfield  {title} {\bibinfo {title} {Defining stable
  phases of open quantum systems},\ }\href
  {https://doi.org/10.1103/PhysRevX.14.041031} {\bibfield  {journal} {\bibinfo
  {journal} {Physical Review X}\ }\textbf {\bibinfo {volume} {14}},\ \bibinfo
  {pages} {041031} (\bibinfo {year} {2024})}\BibitemShut {NoStop}%
\bibitem [{\citenamefont {Huang}\ \emph {et~al.}(2024)\citenamefont {Huang},
  \citenamefont {Qi}, \citenamefont {Zhang},\ and\ \citenamefont
  {Lucas}}]{huang2024hydrodynamics}%
  \BibitemOpen
  \bibfield  {author} {\bibinfo {author} {\bibfnamefont {X.}~\bibnamefont
  {Huang}}, \bibinfo {author} {\bibfnamefont {M.}~\bibnamefont {Qi}}, \bibinfo
  {author} {\bibfnamefont {J.-H.}\ \bibnamefont {Zhang}},\ and\ \bibinfo
  {author} {\bibfnamefont {A.}~\bibnamefont {Lucas}},\ }\bibfield  {title}
  {\bibinfo {title} {Hydrodynamics as the effective field theory of
  strong-to-weak spontaneous symmetry breaking},\ }\href
  {https://arxiv.org/abs/2407.08760} {\bibfield  {journal} {\bibinfo  {journal}
  {arXiv:2407.08760}\ } (\bibinfo {year} {2024})}\BibitemShut {NoStop}%
\bibitem [{\citenamefont {Wang}\ \emph {et~al.}(2023)\citenamefont {Wang},
  \citenamefont {Wu},\ and\ \citenamefont {Wang}}]{wang2023intrinsic}%
  \BibitemOpen
  \bibfield  {author} {\bibinfo {author} {\bibfnamefont {Z.}~\bibnamefont
  {Wang}}, \bibinfo {author} {\bibfnamefont {Z.}~\bibnamefont {Wu}},\ and\
  \bibinfo {author} {\bibfnamefont {Z.}~\bibnamefont {Wang}},\ }\bibfield
  {title} {\bibinfo {title} {Intrinsic mixed-state topological order without
  quantum memory},\ }\href {https://arxiv.org/abs/2307.13758} {\bibfield
  {journal} {\bibinfo  {journal} {arXiv:2307.13758}\ } (\bibinfo {year}
  {2023})}\BibitemShut {NoStop}%
\bibitem [{\citenamefont {Gu}\ \emph {et~al.}(2024)\citenamefont {Gu},
  \citenamefont {Wang},\ and\ \citenamefont {Wang}}]{gu2024spontaneous}%
  \BibitemOpen
  \bibfield  {author} {\bibinfo {author} {\bibfnamefont {D.}~\bibnamefont
  {Gu}}, \bibinfo {author} {\bibfnamefont {Z.}~\bibnamefont {Wang}},\ and\
  \bibinfo {author} {\bibfnamefont {Z.}~\bibnamefont {Wang}},\ }\bibfield
  {title} {\bibinfo {title} {Spontaneous symmetry breaking in open quantum
  systems: strong, weak, and strong-to-weak},\ }\href
  {https://arxiv.org/abs/2406.19381} {\bibfield  {journal} {\bibinfo  {journal}
  {arXiv:2406.19381}\ } (\bibinfo {year} {2024})}\BibitemShut {NoStop}%
\bibitem [{\citenamefont {Ma}\ and\ \citenamefont
  {Turzillo}(2024)}]{ma2024symmetry}%
  \BibitemOpen
  \bibfield  {author} {\bibinfo {author} {\bibfnamefont {R.}~\bibnamefont
  {Ma}}\ and\ \bibinfo {author} {\bibfnamefont {A.}~\bibnamefont {Turzillo}},\
  }\bibfield  {title} {\bibinfo {title} {Symmetry protected topological phases
  of mixed states in the doubled space},\ }\href
  {https://arxiv.org/abs/2403.13280} {\bibfield  {journal} {\bibinfo  {journal}
  {arXiv:2403.13280}\ } (\bibinfo {year} {2024})}\BibitemShut {NoStop}%
\bibitem [{\citenamefont {You}\ and\ \citenamefont
  {Oshikawa}(2024)}]{you2024intrinsic}%
  \BibitemOpen
  \bibfield  {author} {\bibinfo {author} {\bibfnamefont {Y.}~\bibnamefont
  {You}}\ and\ \bibinfo {author} {\bibfnamefont {M.}~\bibnamefont {Oshikawa}},\
  }\bibfield  {title} {\bibinfo {title} {Intrinsic symmetry-protected
  topological mixed state from modulated symmetries and hierarchical structure
  of boundary anomaly},\ }\href {https://doi.org/10.1103/PhysRevB.110.165160}
  {\bibfield  {journal} {\bibinfo  {journal} {Physical Review B}\ }\textbf
  {\bibinfo {volume} {110}},\ \bibinfo {pages} {165160} (\bibinfo {year}
  {2024})}\BibitemShut {NoStop}%
\bibitem [{\citenamefont {Ellison}\ and\ \citenamefont
  {Cheng}(2024)}]{ellison2024towards}%
  \BibitemOpen
  \bibfield  {author} {\bibinfo {author} {\bibfnamefont {T.}~\bibnamefont
  {Ellison}}\ and\ \bibinfo {author} {\bibfnamefont {M.}~\bibnamefont
  {Cheng}},\ }\bibfield  {title} {\bibinfo {title} {Towards a classification of
  mixed-state topological orders in two dimensions},\ }\href
  {https://arxiv.org/abs/2405.02390} {\bibfield  {journal} {\bibinfo  {journal}
  {arXiv:2405.02390}\ } (\bibinfo {year} {2024})}\BibitemShut {NoStop}%
\bibitem [{\citenamefont {Lessa}\ \emph {et~al.}(2024)\citenamefont {Lessa},
  \citenamefont {Ma}, \citenamefont {Zhang}, \citenamefont {Bi}, \citenamefont
  {Cheng},\ and\ \citenamefont {Wang}}]{lessa2024strong}%
  \BibitemOpen
  \bibfield  {author} {\bibinfo {author} {\bibfnamefont {L.~A.}\ \bibnamefont
  {Lessa}}, \bibinfo {author} {\bibfnamefont {R.}~\bibnamefont {Ma}}, \bibinfo
  {author} {\bibfnamefont {J.-H.}\ \bibnamefont {Zhang}}, \bibinfo {author}
  {\bibfnamefont {Z.}~\bibnamefont {Bi}}, \bibinfo {author} {\bibfnamefont
  {M.}~\bibnamefont {Cheng}},\ and\ \bibinfo {author} {\bibfnamefont
  {C.}~\bibnamefont {Wang}},\ }\bibfield  {title} {\bibinfo {title}
  {Strong-to-weak spontaneous symmetry breaking in mixed quantum states},\
  }\href {https://arxiv.org/abs/2405.03639} {\bibfield  {journal} {\bibinfo
  {journal} {arXiv:2405.03639}\ } (\bibinfo {year} {2024})}\BibitemShut
  {NoStop}%
\bibitem [{\citenamefont {Chen}\ and\ \citenamefont
  {Grover}(2024)}]{chen2024unconventional}%
  \BibitemOpen
  \bibfield  {author} {\bibinfo {author} {\bibfnamefont {Y.-H.}\ \bibnamefont
  {Chen}}\ and\ \bibinfo {author} {\bibfnamefont {T.}~\bibnamefont {Grover}},\
  }\bibfield  {title} {\bibinfo {title} {Unconventional topological mixed-state
  transition and critical phase induced by self-dual coherent errors},\ }\href
  {https://doi.org/10.1103/PhysRevB.110.125152} {\bibfield  {journal} {\bibinfo
   {journal} {Physical Review B}\ }\textbf {\bibinfo {volume} {110}},\ \bibinfo
  {pages} {125152} (\bibinfo {year} {2024})}\BibitemShut {NoStop}%
\bibitem [{\citenamefont {Ma}\ \emph {et~al.}(2025)\citenamefont {Ma},
  \citenamefont {Zhang}, \citenamefont {Bi}, \citenamefont {Cheng},\ and\
  \citenamefont {Wang}}]{ma2025topological}%
  \BibitemOpen
  \bibfield  {author} {\bibinfo {author} {\bibfnamefont {R.}~\bibnamefont
  {Ma}}, \bibinfo {author} {\bibfnamefont {J.-H.}\ \bibnamefont {Zhang}},
  \bibinfo {author} {\bibfnamefont {Z.}~\bibnamefont {Bi}}, \bibinfo {author}
  {\bibfnamefont {M.}~\bibnamefont {Cheng}},\ and\ \bibinfo {author}
  {\bibfnamefont {C.}~\bibnamefont {Wang}},\ }\bibfield  {title} {\bibinfo
  {title} {Topological phases with average symmetries: The decohered, the
  disordered, and the intrinsic},\ }\href
  {https://doi.org/https://doi.org/10.1103/PhysRevX.15.021062} {\bibfield
  {journal} {\bibinfo  {journal} {Physical Review X}\ }\textbf {\bibinfo
  {volume} {15}},\ \bibinfo {pages} {021062} (\bibinfo {year}
  {2025})}\BibitemShut {NoStop}%
\bibitem [{\citenamefont {{Lessa}}\ \emph {et~al.}(2024)\citenamefont
  {{Lessa}}, \citenamefont {{Cheng}},\ and\ \citenamefont
  {{Wang}}}]{lessa2024mixed}%
  \BibitemOpen
  \bibfield  {author} {\bibinfo {author} {\bibfnamefont {L.~A.}\ \bibnamefont
  {{Lessa}}}, \bibinfo {author} {\bibfnamefont {M.}~\bibnamefont {{Cheng}}},\
  and\ \bibinfo {author} {\bibfnamefont {C.}~\bibnamefont {{Wang}}},\
  }\bibfield  {title} {\bibinfo {title} {{Mixed-state quantum anomaly and
  multipartite entanglement}},\ }\href {https://arxiv.org/abs/2401.17357}
  {\bibfield  {journal} {\bibinfo  {journal} {arXiv:2401.17357}\ } (\bibinfo
  {year} {2024})}\BibitemShut {NoStop}%
\bibitem [{\citenamefont {Sun}(2025)}]{sun2025}%
  \BibitemOpen
  \bibfield  {author} {\bibinfo {author} {\bibfnamefont {X.-Q.}\ \bibnamefont
  {Sun}},\ }\bibfield  {title} {\bibinfo {title} {Anomalous matrix product
  operator symmetries and 1{D} mixed-state phases},\ }\href
  {https://arxiv.org/abs/2504.16985} {\bibfield  {journal} {\bibinfo  {journal}
  {arXiv:2504.16985}\ } (\bibinfo {year} {2025})}\BibitemShut {NoStop}%
\bibitem [{\citenamefont {Breuer}\ and\ \citenamefont
  {Petruccione}(2002)}]{breuer2002theory}%
  \BibitemOpen
  \bibfield  {author} {\bibinfo {author} {\bibfnamefont {H.-P.}\ \bibnamefont
  {Breuer}}\ and\ \bibinfo {author} {\bibfnamefont {F.}~\bibnamefont
  {Petruccione}},\ }\href@noop {} {\emph {\bibinfo {title} {The theory of open
  quantum systems}}}\ (\bibinfo  {publisher} {Oxford University Press, USA},\
  \bibinfo {year} {2002})\BibitemShut {NoStop}%
\bibitem [{\citenamefont {Bondarenko}(2021)}]{bondarenko2021constructing}%
  \BibitemOpen
  \bibfield  {author} {\bibinfo {author} {\bibfnamefont {D.}~\bibnamefont
  {Bondarenko}},\ }\bibfield  {title} {\bibinfo {title} {Constructing k-local
  parent lindbladians for matrix product density operators},\ }\href
  {https://arxiv.org/abs/2110.13134} {\bibfield  {journal} {\bibinfo  {journal}
  {arXiv:2110.13134}\ } (\bibinfo {year} {2021})}\BibitemShut {NoStop}%
\bibitem [{\citenamefont {{Diehl}}\ \emph {et~al.}(2008)\citenamefont
  {{Diehl}}, \citenamefont {{Micheli}}, \citenamefont {{Kantian}},
  \citenamefont {{Kraus}}, \citenamefont {{B{\"u}chler}},\ and\ \citenamefont
  {{Zoller}}}]{diehl2008quantum}%
  \BibitemOpen
  \bibfield  {author} {\bibinfo {author} {\bibfnamefont {S.}~\bibnamefont
  {{Diehl}}}, \bibinfo {author} {\bibfnamefont {A.}~\bibnamefont {{Micheli}}},
  \bibinfo {author} {\bibfnamefont {A.}~\bibnamefont {{Kantian}}}, \bibinfo
  {author} {\bibfnamefont {B.}~\bibnamefont {{Kraus}}}, \bibinfo {author}
  {\bibfnamefont {H.~P.}\ \bibnamefont {{B{\"u}chler}}},\ and\ \bibinfo
  {author} {\bibfnamefont {P.}~\bibnamefont {{Zoller}}},\ }\bibfield  {title}
  {\bibinfo {title} {{Quantum states and phases in driven open quantum systems
  with cold atoms}},\ }\href {https://doi.org/10.1038/nphys1073} {\bibfield
  {journal} {\bibinfo  {journal} {Nature Physics}\ }\textbf {\bibinfo {volume}
  {4}},\ \bibinfo {pages} {878} (\bibinfo {year} {2008})}\BibitemShut {NoStop}%
\bibitem [{\citenamefont {{Kraus}}\ \emph {et~al.}(2008)\citenamefont
  {{Kraus}}, \citenamefont {{B{\"u}chler}}, \citenamefont {{Diehl}},
  \citenamefont {{Kantian}}, \citenamefont {{Micheli}},\ and\ \citenamefont
  {{Zoller}}}]{kraus2008preparation}%
  \BibitemOpen
  \bibfield  {author} {\bibinfo {author} {\bibfnamefont {B.}~\bibnamefont
  {{Kraus}}}, \bibinfo {author} {\bibfnamefont {H.~P.}\ \bibnamefont
  {{B{\"u}chler}}}, \bibinfo {author} {\bibfnamefont {S.}~\bibnamefont
  {{Diehl}}}, \bibinfo {author} {\bibfnamefont {A.}~\bibnamefont {{Kantian}}},
  \bibinfo {author} {\bibfnamefont {A.}~\bibnamefont {{Micheli}}},\ and\
  \bibinfo {author} {\bibfnamefont {P.}~\bibnamefont {{Zoller}}},\ }\bibfield
  {title} {\bibinfo {title} {{Preparation of entangled states by quantum Markov
  processes}},\ }\href {https://doi.org/10.1103/PhysRevA.78.042307} {\bibfield
  {journal} {\bibinfo  {journal} {Phys. Rev. A}\ }\textbf {\bibinfo {volume}
  {78}},\ \bibinfo {pages} {042307} (\bibinfo {year} {2008})}\BibitemShut
  {NoStop}%
\bibitem [{\citenamefont {Verstraete}\ \emph {et~al.}(2009)\citenamefont
  {Verstraete}, \citenamefont {Wolf},\ and\ \citenamefont
  {Ignacio~Cirac}}]{verstraete2009quantum}%
  \BibitemOpen
  \bibfield  {author} {\bibinfo {author} {\bibfnamefont {F.}~\bibnamefont
  {Verstraete}}, \bibinfo {author} {\bibfnamefont {M.~M.}\ \bibnamefont
  {Wolf}},\ and\ \bibinfo {author} {\bibfnamefont {J.}~\bibnamefont
  {Ignacio~Cirac}},\ }\bibfield  {title} {\bibinfo {title} {Quantum computation
  and quantum-state engineering driven by dissipation},\ }\href
  {https://doi.org/10.1038/nphys1342} {\bibfield  {journal} {\bibinfo
  {journal} {Nature physics}\ }\textbf {\bibinfo {volume} {5}},\ \bibinfo
  {pages} {633} (\bibinfo {year} {2009})}\BibitemShut {NoStop}%
\bibitem [{\citenamefont {{Kastoryano}}\ and\ \citenamefont
  {{Brand{\~a}o}}(2016)}]{kastoryano2016quantum}%
  \BibitemOpen
  \bibfield  {author} {\bibinfo {author} {\bibfnamefont {M.~J.}\ \bibnamefont
  {{Kastoryano}}}\ and\ \bibinfo {author} {\bibfnamefont {F.~G.~S.~L.}\
  \bibnamefont {{Brand{\~a}o}}},\ }\bibfield  {title} {\bibinfo {title}
  {{Quantum Gibbs Samplers: The Commuting Case}},\ }\href
  {https://doi.org/10.1007/s00220-016-2641-8} {\bibfield  {journal} {\bibinfo
  {journal} {Communications in Mathematical Physics}\ }\textbf {\bibinfo
  {volume} {344}},\ \bibinfo {pages} {915} (\bibinfo {year}
  {2016})}\BibitemShut {NoStop}%
\bibitem [{\citenamefont {Chen}\ \emph {et~al.}(2023)\citenamefont {Chen},
  \citenamefont {Kastoryano},\ and\ \citenamefont
  {Gily{\'e}n}}]{chen2023efficient}%
  \BibitemOpen
  \bibfield  {author} {\bibinfo {author} {\bibfnamefont {C.-F.}\ \bibnamefont
  {Chen}}, \bibinfo {author} {\bibfnamefont {M.~J.}\ \bibnamefont
  {Kastoryano}},\ and\ \bibinfo {author} {\bibfnamefont {A.}~\bibnamefont
  {Gily{\'e}n}},\ }\bibfield  {title} {\bibinfo {title} {An efficient and exact
  noncommutative quantum {G}ibbs sampler},\ }\href
  {https://arxiv.org/abs/2311.09207} {\bibfield  {journal} {\bibinfo  {journal}
  {arXiv:2311.09207}\ } (\bibinfo {year} {2023})}\BibitemShut {NoStop}%
\bibitem [{\citenamefont {Rouz{\'e}}\ \emph {et~al.}(2024)\citenamefont
  {Rouz{\'e}}, \citenamefont {Fran{\c{c}}a},\ and\ \citenamefont
  {Alhambra}}]{rouze2024efficient}%
  \BibitemOpen
  \bibfield  {author} {\bibinfo {author} {\bibfnamefont {C.}~\bibnamefont
  {Rouz{\'e}}}, \bibinfo {author} {\bibfnamefont {D.~S.}\ \bibnamefont
  {Fran{\c{c}}a}},\ and\ \bibinfo {author} {\bibfnamefont {{\'A}.~M.}\
  \bibnamefont {Alhambra}},\ }\bibfield  {title} {\bibinfo {title} {Efficient
  thermalization and universal quantum computing with quantum {G}ibbs
  samplers},\ }\href {https://arxiv.org/abs/2403.12691} {\bibfield  {journal}
  {\bibinfo  {journal} {arXiv:2403.12691}\ } (\bibinfo {year}
  {2024})}\BibitemShut {NoStop}%
\bibitem [{\citenamefont {Guo}\ \emph {et~al.}(2024)\citenamefont {Guo},
  \citenamefont {Hart}, \citenamefont {Chen}, \citenamefont {Friedman},\ and\
  \citenamefont {Lucas}}]{guo2024designing}%
  \BibitemOpen
  \bibfield  {author} {\bibinfo {author} {\bibfnamefont {J.}~\bibnamefont
  {Guo}}, \bibinfo {author} {\bibfnamefont {O.}~\bibnamefont {Hart}}, \bibinfo
  {author} {\bibfnamefont {C.-F.}\ \bibnamefont {Chen}}, \bibinfo {author}
  {\bibfnamefont {A.~J.}\ \bibnamefont {Friedman}},\ and\ \bibinfo {author}
  {\bibfnamefont {A.}~\bibnamefont {Lucas}},\ }\bibfield  {title} {\bibinfo
  {title} {Designing open quantum systems with known steady states: Davies
  generators and beyond},\ }\href {https://arxiv.org/abs/2404.14538} {\bibfield
   {journal} {\bibinfo  {journal} {arXiv:2404.14538}\ } (\bibinfo {year}
  {2024})}\BibitemShut {NoStop}%
\bibitem [{\citenamefont {{Ding}}\ \emph {et~al.}(2024)\citenamefont {{Ding}},
  \citenamefont {{Li}},\ and\ \citenamefont {{Lin}}}]{ding2024efficient}%
  \BibitemOpen
  \bibfield  {author} {\bibinfo {author} {\bibfnamefont {Z.}~\bibnamefont
  {{Ding}}}, \bibinfo {author} {\bibfnamefont {B.}~\bibnamefont {{Li}}},\ and\
  \bibinfo {author} {\bibfnamefont {L.}~\bibnamefont {{Lin}}},\ }\bibfield
  {title} {\bibinfo {title} {{Efficient quantum Gibbs samplers with
  Kubo--Martin--Schwinger detailed balance condition}},\ }\href
  {https://doi.org/10.48550/arXiv.2404.05998} {\bibfield  {journal} {\bibinfo
  {journal} {arXiv e-prints}\ ,\ \bibinfo {eid} {arXiv:2404.05998}} (\bibinfo
  {year} {2024})}\BibitemShut {NoStop}%
\bibitem [{\citenamefont {{Cirac}}\ \emph {et~al.}(2017)\citenamefont
  {{Cirac}}, \citenamefont {{P{\'e}rez-Garc{\'\i}a}}, \citenamefont
  {{Schuch}},\ and\ \citenamefont {{Verstraete}}}]{cirac2017matrix}%
  \BibitemOpen
  \bibfield  {author} {\bibinfo {author} {\bibfnamefont {J.~I.}\ \bibnamefont
  {{Cirac}}}, \bibinfo {author} {\bibfnamefont {D.}~\bibnamefont
  {{P{\'e}rez-Garc{\'\i}a}}}, \bibinfo {author} {\bibfnamefont
  {N.}~\bibnamefont {{Schuch}}},\ and\ \bibinfo {author} {\bibfnamefont
  {F.}~\bibnamefont {{Verstraete}}},\ }\bibfield  {title} {\bibinfo {title}
  {{Matrix product density operators: Renormalization fixed points and boundary
  theories}},\ }\href {https://doi.org/10.1016/j.aop.2016.12.030} {\bibfield
  {journal} {\bibinfo  {journal} {Annals of Physics}\ }\textbf {\bibinfo
  {volume} {378}},\ \bibinfo {pages} {100} (\bibinfo {year}
  {2017})}\BibitemShut {NoStop}%
\bibitem [{\citenamefont {Moln{\'a}r}\ \emph {et~al.}(2022)\citenamefont
  {Moln{\'a}r}, \citenamefont {{Ruiz-de-Alarc{\'o}n}}, \citenamefont
  {Garre-Rubio}, \citenamefont {Schuch}, \citenamefont {Cirac},\ and\
  \citenamefont {P{\'e}rez-Garc{\'i}a}}]{molnar2022matrix}%
  \BibitemOpen
  \bibfield  {author} {\bibinfo {author} {\bibfnamefont {A.}~\bibnamefont
  {Moln{\'a}r}}, \bibinfo {author} {\bibfnamefont {A.}~\bibnamefont
  {{Ruiz-de-Alarc{\'o}n}}}, \bibinfo {author} {\bibfnamefont {J.}~\bibnamefont
  {Garre-Rubio}}, \bibinfo {author} {\bibfnamefont {N.}~\bibnamefont {Schuch}},
  \bibinfo {author} {\bibfnamefont {J.~I.}\ \bibnamefont {Cirac}},\ and\
  \bibinfo {author} {\bibfnamefont {D.}~\bibnamefont {P{\'e}rez-Garc{\'i}a}},\
  }\href@noop {} {\bibinfo {title} {Matrix product operator algebras {I}:
  representations of weak {Hopf} algebras and projected entangled pair states}}
  (\bibinfo {year} {2022}),\ \Eprint {https://arxiv.org/abs/2204.05940}
  {arXiv:2204.05940} \BibitemShut {NoStop}%
\bibitem [{\citenamefont {{Ruiz-de-Alarc{\'o}n}}\ \emph
  {et~al.}(2024)\citenamefont {{Ruiz-de-Alarc{\'o}n}}, \citenamefont
  {Garre-Rubio}, \citenamefont {Moln{\'a}r},\ and\ \citenamefont
  {P{\'e}rez-Garc{\'i}a}}]{ruizdealarcon2024matrix}%
  \BibitemOpen
  \bibfield  {author} {\bibinfo {author} {\bibfnamefont {A.}~\bibnamefont
  {{Ruiz-de-Alarc{\'o}n}}}, \bibinfo {author} {\bibfnamefont {J.}~\bibnamefont
  {Garre-Rubio}}, \bibinfo {author} {\bibfnamefont {A.}~\bibnamefont
  {Moln{\'a}r}},\ and\ \bibinfo {author} {\bibfnamefont {D.}~\bibnamefont
  {P{\'e}rez-Garc{\'i}a}},\ }\bibfield  {title} {\bibinfo {title} {Matrix
  product operator algebras {II}: phases of matter for {1D} mixed states},\
  }\href {https://doi.org/10.1007/s11005-024-01778-z} {\bibfield  {journal}
  {\bibinfo  {journal} {Letters in Mathematical Physics}\ }\textbf {\bibinfo
  {volume} {114}},\ \bibinfo {pages} {43} (\bibinfo {year} {2024})}\BibitemShut
  {NoStop}%
\bibitem [{\citenamefont {Ogata}\ \emph {et~al.}(2025)\citenamefont {Ogata},
  \citenamefont {P{\'e}rez-Garc{\'i}a},\ and\ \citenamefont
  {{Ruiz-de-Alarc{\'o}n}}}]{yoshiko2024haag}%
  \BibitemOpen
  \bibfield  {author} {\bibinfo {author} {\bibfnamefont {Y.}~\bibnamefont
  {Ogata}}, \bibinfo {author} {\bibfnamefont {D.}~\bibnamefont
  {P{\'e}rez-Garc{\'i}a}},\ and\ \bibinfo {author} {\bibfnamefont
  {A.}~\bibnamefont {{Ruiz-de-Alarc{\'o}n}}},\ }\href@noop {} {\bibinfo {title}
  {Haag duality for 2-d topologically ordered tensor network states}} (\bibinfo
  {year} {2025}),\ \bibinfo {note} {in preparation.}\BibitemShut {Stop}%
\bibitem [{\citenamefont {Wang}(2024)}]{wang2024hopf}%
  \BibitemOpen
  \bibfield  {author} {\bibinfo {author} {\bibfnamefont {Z.}~\bibnamefont
  {Wang}},\ }\bibfield  {title} {\bibinfo {title} {Hopf algebras and solvable
  unitary circuits},\ }\href {https://arxiv.org/abs/2409.17215} {\bibfield
  {journal} {\bibinfo  {journal} {arXiv:2409.17215}\ } (\bibinfo {year}
  {2024})}\BibitemShut {NoStop}%
\bibitem [{\citenamefont {{Perez-Garcia}}\ \emph {et~al.}(2007)\citenamefont
  {{Perez-Garcia}}, \citenamefont {{Verstraete}}, \citenamefont {{Wolf}},\ and\
  \citenamefont {{Cirac}}}]{perez2006matrix}%
  \BibitemOpen
  \bibfield  {author} {\bibinfo {author} {\bibfnamefont {D.}~\bibnamefont
  {{Perez-Garcia}}}, \bibinfo {author} {\bibfnamefont {F.}~\bibnamefont
  {{Verstraete}}}, \bibinfo {author} {\bibfnamefont {M.~M.}\ \bibnamefont
  {{Wolf}}},\ and\ \bibinfo {author} {\bibfnamefont {J.~I.}\ \bibnamefont
  {{Cirac}}},\ }\bibfield  {title} {\bibinfo {title} {{Matrix Product State
  Representations}},\ }\href {https://arxiv.org/abs/quant-ph/0608197}
  {\bibfield  {journal} {\bibinfo  {journal} {Quantum Info. Comput.}\ }\textbf
  {\bibinfo {volume} {7}},\ \bibinfo {pages} {401–430} (\bibinfo {year}
  {2007})}\BibitemShut {NoStop}%
\bibitem [{\citenamefont {{Kliesch}}\ \emph {et~al.}(2014)\citenamefont
  {{Kliesch}}, \citenamefont {{Gross}},\ and\ \citenamefont
  {{Eisert}}}]{kliesch2014matrix}%
  \BibitemOpen
  \bibfield  {author} {\bibinfo {author} {\bibfnamefont {M.}~\bibnamefont
  {{Kliesch}}}, \bibinfo {author} {\bibfnamefont {D.}~\bibnamefont {{Gross}}},\
  and\ \bibinfo {author} {\bibfnamefont {J.}~\bibnamefont {{Eisert}}},\
  }\bibfield  {title} {\bibinfo {title} {{Matrix-Product Operators and States:
  NP-Hardness and Undecidability}},\ }\href
  {https://doi.org/10.1103/PhysRevLett.113.160503} {\bibfield  {journal}
  {\bibinfo  {journal} {\prl}\ }\textbf {\bibinfo {volume} {113}},\ \bibinfo
  {eid} {160503} (\bibinfo {year} {2014})}\BibitemShut {NoStop}%
\bibitem [{\citenamefont {{De las Cuevas}}\ \emph {et~al.}(2016)\citenamefont
  {{De las Cuevas}}, \citenamefont {{Cubitt}}, \citenamefont {{Cirac}},
  \citenamefont {{Wolf}},\ and\ \citenamefont
  {{P{\'e}rez-Garc{\'\i}a}}}]{de2016fundamental}%
  \BibitemOpen
  \bibfield  {author} {\bibinfo {author} {\bibfnamefont {G.}~\bibnamefont {{De
  las Cuevas}}}, \bibinfo {author} {\bibfnamefont {T.~S.}\ \bibnamefont
  {{Cubitt}}}, \bibinfo {author} {\bibfnamefont {J.~I.}\ \bibnamefont
  {{Cirac}}}, \bibinfo {author} {\bibfnamefont {M.~M.}\ \bibnamefont
  {{Wolf}}},\ and\ \bibinfo {author} {\bibfnamefont {D.}~\bibnamefont
  {{P{\'e}rez-Garc{\'\i}a}}},\ }\bibfield  {title} {\bibinfo {title}
  {{Fundamental limitations in the purifications of tensor networks}},\ }\href
  {https://doi.org/10.1063/1.4954983} {\bibfield  {journal} {\bibinfo
  {journal} {Journal of Mathematical Physics}\ }\textbf {\bibinfo {volume}
  {57}},\ \bibinfo {eid} {071902} (\bibinfo {year} {2016})}\BibitemShut
  {NoStop}%
\bibitem [{\citenamefont {{Beigi}}(2012)}]{beigi2011classification}%
  \BibitemOpen
  \bibfield  {author} {\bibinfo {author} {\bibfnamefont {S.}~\bibnamefont
  {{Beigi}}},\ }\bibfield  {title} {\bibinfo {title} {{Classification of the
  phases of 1D spin chains with commuting Hamiltonians}},\ }\href
  {https://doi.org/10.1088/1751-8113/45/2/025306} {\bibfield  {journal}
  {\bibinfo  {journal} {Journal of Physics A Mathematical General}\ }\textbf
  {\bibinfo {volume} {45}},\ \bibinfo {eid} {025306} (\bibinfo {year}
  {2012})}\BibitemShut {NoStop}%
\bibitem [{\citenamefont {Chen}\ \emph {et~al.}(2020)\citenamefont {Chen},
  \citenamefont {Kato},\ and\ \citenamefont {Brandao}}]{chen2020matrix}%
  \BibitemOpen
  \bibfield  {author} {\bibinfo {author} {\bibfnamefont {C.-F.}\ \bibnamefont
  {Chen}}, \bibinfo {author} {\bibfnamefont {K.}~\bibnamefont {Kato}},\ and\
  \bibinfo {author} {\bibfnamefont {F.~G.}\ \bibnamefont {Brandao}},\
  }\bibfield  {title} {\bibinfo {title} {Matrix product density operators: When
  do they have a local parent {H}amiltonian?},\ }\href
  {https://arxiv.org/abs/2010.14682} {\bibfield  {journal} {\bibinfo  {journal}
  {arXiv:2010.14682}\ } (\bibinfo {year} {2020})}\BibitemShut {NoStop}%
\bibitem [{\citenamefont {{Schuch}}\ \emph {et~al.}(2013)\citenamefont
  {{Schuch}}, \citenamefont {{Poilblanc}}, \citenamefont {{Cirac}},\ and\
  \citenamefont {{P{\'e}rez-Garc{\'\i}a}}}]{schuch2013topological}%
  \BibitemOpen
  \bibfield  {author} {\bibinfo {author} {\bibfnamefont {N.}~\bibnamefont
  {{Schuch}}}, \bibinfo {author} {\bibfnamefont {D.}~\bibnamefont
  {{Poilblanc}}}, \bibinfo {author} {\bibfnamefont {J.~I.}\ \bibnamefont
  {{Cirac}}},\ and\ \bibinfo {author} {\bibfnamefont {D.}~\bibnamefont
  {{P{\'e}rez-Garc{\'\i}a}}},\ }\bibfield  {title} {\bibinfo {title}
  {{Topological Order in the Projected Entangled-Pair States Formalism:
  Transfer Operator and Boundary Hamiltonians}},\ }\href
  {https://doi.org/10.1103/PhysRevLett.111.090501} {\bibfield  {journal}
  {\bibinfo  {journal} {\prl}\ }\textbf {\bibinfo {volume} {111}},\ \bibinfo
  {eid} {090501} (\bibinfo {year} {2013})}\BibitemShut {NoStop}%
\bibitem [{\citenamefont {{Bultinck}}\ \emph {et~al.}(2017)\citenamefont
  {{Bultinck}}, \citenamefont {{Mari{\"e}n}}, \citenamefont {{Williamson}},
  \citenamefont {{{\c{S}}ahino{\u{g}}lu}}, \citenamefont {{Haegeman}},\ and\
  \citenamefont {{Verstraete}}}]{bultinck2017anyons}%
  \BibitemOpen
  \bibfield  {author} {\bibinfo {author} {\bibfnamefont {N.}~\bibnamefont
  {{Bultinck}}}, \bibinfo {author} {\bibfnamefont {M.}~\bibnamefont
  {{Mari{\"e}n}}}, \bibinfo {author} {\bibfnamefont {D.~J.}\ \bibnamefont
  {{Williamson}}}, \bibinfo {author} {\bibfnamefont {M.~B.}\ \bibnamefont
  {{{\c{S}}ahino{\u{g}}lu}}}, \bibinfo {author} {\bibfnamefont
  {J.}~\bibnamefont {{Haegeman}}},\ and\ \bibinfo {author} {\bibfnamefont
  {F.}~\bibnamefont {{Verstraete}}},\ }\bibfield  {title} {\bibinfo {title}
  {{Anyons and matrix product operator algebras}},\ }\href
  {https://doi.org/10.1016/j.aop.2017.01.004} {\bibfield  {journal} {\bibinfo
  {journal} {Annals of Physics}\ }\textbf {\bibinfo {volume} {378}},\ \bibinfo
  {pages} {183} (\bibinfo {year} {2017})}\BibitemShut {NoStop}%
\bibitem [{\citenamefont {Kato}(2024)}]{kato2024exact}%
  \BibitemOpen
  \bibfield  {author} {\bibinfo {author} {\bibfnamefont {K.}~\bibnamefont
  {Kato}},\ }\bibfield  {title} {\bibinfo {title} {Exact renormalization group
  flow for matrix product density operators},\ }\href
  {https://arxiv.org/abs/2410.22696} {\bibfield  {journal} {\bibinfo  {journal}
  {arXiv:2410.22696}\ } (\bibinfo {year} {2024})}\BibitemShut {NoStop}%
\bibitem [{\citenamefont {Bravyi}\ and\ \citenamefont
  {Vyalyi}(2005)}]{bravyi2005commutative}%
  \BibitemOpen
  \bibfield  {author} {\bibinfo {author} {\bibfnamefont {S.}~\bibnamefont
  {Bravyi}}\ and\ \bibinfo {author} {\bibfnamefont {M.}~\bibnamefont
  {Vyalyi}},\ }\bibfield  {title} {\bibinfo {title} {Commutative version of the
  local hamiltonian problem and common eigenspace problem},\ }\href
  {https://arxiv.org/abs/quant-ph/0308021} {\bibfield  {journal} {\bibinfo
  {journal} {Quantum Information \& Computation}\ }\textbf {\bibinfo {volume}
  {5}},\ \bibinfo {pages} {187} (\bibinfo {year} {2005})}\BibitemShut {NoStop}%
\bibitem [{\citenamefont {{Hayden}}\ \emph {et~al.}(2004)\citenamefont
  {{Hayden}}, \citenamefont {{Jozsa}}, \citenamefont {{Petz}},\ and\
  \citenamefont {{Winter}}}]{hayden2004structure}%
  \BibitemOpen
  \bibfield  {author} {\bibinfo {author} {\bibfnamefont {P.}~\bibnamefont
  {{Hayden}}}, \bibinfo {author} {\bibfnamefont {R.}~\bibnamefont {{Jozsa}}},
  \bibinfo {author} {\bibfnamefont {D.}~\bibnamefont {{Petz}}},\ and\ \bibinfo
  {author} {\bibfnamefont {A.}~\bibnamefont {{Winter}}},\ }\bibfield  {title}
  {\bibinfo {title} {{Structure of States Which Satisfy Strong Subadditivity of
  Quantum Entropy with Equality}},\ }\href
  {https://doi.org/10.1007/s00220-004-1049-z} {\bibfield  {journal} {\bibinfo
  {journal} {Communications in Mathematical Physics}\ }\textbf {\bibinfo
  {volume} {246}},\ \bibinfo {pages} {359} (\bibinfo {year}
  {2004})}\BibitemShut {NoStop}%
\bibitem [{\citenamefont {Liu}\ \emph {et~al.}(2025)\citenamefont {Liu},
  \citenamefont {Molnar}, \citenamefont {Sun}, \citenamefont {Verstraete},
  \citenamefont {Kato},\ and\ \citenamefont {Lootens}}]{liu2025trading}%
  \BibitemOpen
  \bibfield  {author} {\bibinfo {author} {\bibfnamefont {Y.}~\bibnamefont
  {Liu}}, \bibinfo {author} {\bibfnamefont {A.}~\bibnamefont {Molnar}},
  \bibinfo {author} {\bibfnamefont {X.-Q.}\ \bibnamefont {Sun}}, \bibinfo
  {author} {\bibfnamefont {F.}~\bibnamefont {Verstraete}}, \bibinfo {author}
  {\bibfnamefont {K.}~\bibnamefont {Kato}},\ and\ \bibinfo {author}
  {\bibfnamefont {L.}~\bibnamefont {Lootens}},\ }\bibfield  {title} {\bibinfo
  {title} {Trading mathematical for physical simplicity: Bialgebraic structures
  in matrix product operator symmetries},\ }\href
  {https://arxiv.org/abs/2509.03600} {\bibfield  {journal} {\bibinfo  {journal}
  {arXiv:2509.03600}\ } (\bibinfo {year} {2025})}\BibitemShut {NoStop}%
\bibitem [{\citenamefont {{Lindblad}}(1976)}]{lindblad1976generators}%
  \BibitemOpen
  \bibfield  {author} {\bibinfo {author} {\bibfnamefont {G.}~\bibnamefont
  {{Lindblad}}},\ }\bibfield  {title} {\bibinfo {title} {{On the generators of
  quantum dynamical semigroups}},\ }\href {https://doi.org/10.1007/BF01608499}
  {\bibfield  {journal} {\bibinfo  {journal} {Communications in Mathematical
  Physics}\ }\textbf {\bibinfo {volume} {48}},\ \bibinfo {pages} {119}
  (\bibinfo {year} {1976})}\BibitemShut {NoStop}%
\bibitem [{Note1()}]{Note1}%
  \BibitemOpen
  \bibinfo {note} {This is because we can rewrite $ \protect \mathcal
  {L}^{(N)}=N(\protect \mathcal {E}^{(N)}-\protect \mathbbm {1})$ where
  $\protect \mathcal {E}^{(N)}=\protect \frac {1}{N}\DOTSB \sum@ \slimits@
  _{i=1}^N \protect \mathcal {E}_i$, and the eigenvalues of the quantum channel
  $\protect \mathcal {E}^{(N)}$ lie within the unit circle}\BibitemShut
  {NoStop}%
\bibitem [{\citenamefont {Wolf}(2012)}]{wolf2012quantum}%
  \BibitemOpen
  \bibfield  {author} {\bibinfo {author} {\bibfnamefont {M.~M.}\ \bibnamefont
  {Wolf}},\ }\bibfield  {title} {\bibinfo {title} {Quantum channels \&
  operations: Guided tour},\ }\href
  {https://citeseerx.ist.psu.edu/document?repid=rep1&type=pdf&doi=afa58291b0b8bd47504acb1ab8f553f0b37685cf}
  {\bibfield  {journal} {\bibinfo  {journal} {Lecture notes}\ } (\bibinfo
  {year} {2012})}\BibitemShut {NoStop}%
\bibitem [{\citenamefont {Bravyi}\ \emph {et~al.}(2006)\citenamefont {Bravyi},
  \citenamefont {Hastings},\ and\ \citenamefont {Verstraete}}]{bravyi2006lieb}%
  \BibitemOpen
  \bibfield  {author} {\bibinfo {author} {\bibfnamefont {S.}~\bibnamefont
  {Bravyi}}, \bibinfo {author} {\bibfnamefont {M.~B.}\ \bibnamefont
  {Hastings}},\ and\ \bibinfo {author} {\bibfnamefont {F.}~\bibnamefont
  {Verstraete}},\ }\bibfield  {title} {\bibinfo {title} {Lieb-{R}obinson bounds
  and the generation of correlations and topological quantum order},\ }\href
  {https://doi.org/10.1103/PhysRevLett.97.050401} {\bibfield  {journal}
  {\bibinfo  {journal} {Phys. Rev. Lett.}\ }\textbf {\bibinfo {volume} {97}},\
  \bibinfo {pages} {050401} (\bibinfo {year} {2006})}\BibitemShut {NoStop}%
\bibitem [{\citenamefont {{Poulin}}(2010)}]{poulin2010lieb}%
  \BibitemOpen
  \bibfield  {author} {\bibinfo {author} {\bibfnamefont {D.}~\bibnamefont
  {{Poulin}}},\ }\bibfield  {title} {\bibinfo {title} {{Lieb-Robinson Bound and
  Locality for General Markovian Quantum Dynamics}},\ }\href
  {https://doi.org/10.1103/PhysRevLett.104.190401} {\bibfield  {journal}
  {\bibinfo  {journal} {\prl}\ }\textbf {\bibinfo {volume} {104}},\ \bibinfo
  {eid} {190401} (\bibinfo {year} {2010})}\BibitemShut {NoStop}%
\bibitem [{\citenamefont {B\"{o}hm}\ and\ \citenamefont
  {Szlach\'{a}nyi}(1996)}]{bohm_coassociativec_1996}%
  \BibitemOpen
  \bibfield  {author} {\bibinfo {author} {\bibfnamefont {G.}~\bibnamefont
  {B\"{o}hm}}\ and\ \bibinfo {author} {\bibfnamefont {K.}~\bibnamefont
  {Szlach\'{a}nyi}},\ }\bibfield  {title} {\bibinfo {title} {A {Coassociative}
  {C}*-{Quantum} {Group} with {Non}-{Integral} {Dimensions}},\ }\href
  {https://doi.org/10.1007/BF01815526} {\bibfield  {journal} {\bibinfo
  {journal} {Letters in Mathematical Physics}\ }\textbf {\bibinfo {volume}
  {38}},\ \bibinfo {pages} {437} (\bibinfo {year} {1996})}\BibitemShut
  {NoStop}%
\bibitem [{\citenamefont {Etingof}\ \emph {et~al.}(2005)\citenamefont
  {Etingof}, \citenamefont {Nikshych},\ and\ \citenamefont
  {Ostrik}}]{etingof_fusion_2005}%
  \BibitemOpen
  \bibfield  {author} {\bibinfo {author} {\bibfnamefont {P.}~\bibnamefont
  {Etingof}}, \bibinfo {author} {\bibfnamefont {D.}~\bibnamefont {Nikshych}},\
  and\ \bibinfo {author} {\bibfnamefont {V.}~\bibnamefont {Ostrik}},\
  }\bibfield  {title} {\bibinfo {title} {On fusion categories},\ }\href
  {https://doi.org/10.4007/annals.2005.162.581} {\bibfield  {journal} {\bibinfo
   {journal} {Annals of Mathematics}\ }\textbf {\bibinfo {volume} {162}},\
  \bibinfo {pages} {581} (\bibinfo {year} {2005})}\BibitemShut {NoStop}%
\bibitem [{\citenamefont {Etingof}\ \emph {et~al.}(2015)\citenamefont
  {Etingof}, \citenamefont {Gelaki}, \citenamefont {Nikshych},\ and\
  \citenamefont {Ostrik}}]{etingof_tensor_2015}%
  \BibitemOpen
  \bibinfo {editor} {\bibfnamefont {P.~I.}\ \bibnamefont {Etingof}}, \bibinfo
  {editor} {\bibfnamefont {S.}~\bibnamefont {Gelaki}}, \bibinfo {editor}
  {\bibfnamefont {D.}~\bibnamefont {Nikshych}},\ and\ \bibinfo {editor}
  {\bibfnamefont {V.}~\bibnamefont {Ostrik}},\ eds.,\ \href@noop {} {\emph
  {\bibinfo {title} {Tensor categories}}},\ \bibinfo {series} {Mathematical
  surveys and monographs}\ No.\ \bibinfo {number} {volume 205}\ (\bibinfo
  {publisher} {American Mathematical Society},\ \bibinfo {address} {Providence,
  Rhode Island},\ \bibinfo {year} {2015})\BibitemShut {NoStop}%
\bibitem [{\citenamefont {B\"{o}hm}\ \emph {et~al.}(1999)\citenamefont
  {B\"{o}hm}, \citenamefont {Nill},\ and\ \citenamefont
  {Szlach\'{a}nyi}}]{bohm_weak_1999}%
  \BibitemOpen
  \bibfield  {author} {\bibinfo {author} {\bibfnamefont {G.}~\bibnamefont
  {B\"{o}hm}}, \bibinfo {author} {\bibfnamefont {F.}~\bibnamefont {Nill}},\
  and\ \bibinfo {author} {\bibfnamefont {K.}~\bibnamefont {Szlach\'{a}nyi}},\
  }\bibfield  {title} {\bibinfo {title} {Weak {Hopf} {Algebras} {I}: {Integral}
  {Theory} and {C}*-structure},\ }\href
  {https://doi.org/10.1006/jabr.1999.7984} {\bibfield  {journal} {\bibinfo
  {journal} {Journal of Algebra}\ }\textbf {\bibinfo {volume} {221}},\ \bibinfo
  {pages} {385} (\bibinfo {year} {1999})}\BibitemShut {NoStop}%
\bibitem [{\citenamefont {B\"{o}hm}\ and\ \citenamefont
  {Szlach\'{a}nyi}(2000)}]{bohm_weak_2000}%
  \BibitemOpen
  \bibfield  {author} {\bibinfo {author} {\bibfnamefont {G.}~\bibnamefont
  {B\"{o}hm}}\ and\ \bibinfo {author} {\bibfnamefont {K.}~\bibnamefont
  {Szlach\'{a}nyi}},\ }\bibfield  {title} {\bibinfo {title} {Weak {Hopf}
  {Algebras} {II}: {Representation} theory, dimensions and the {Markov}
  trace},\ }\href {https://doi.org/10.1006/jabr.2000.8379} {\bibfield
  {journal} {\bibinfo  {journal} {Journal of Algebra}\ }\textbf {\bibinfo
  {volume} {233}},\ \bibinfo {pages} {156} (\bibinfo {year}
  {2000})}\BibitemShut {NoStop}%
\bibitem [{Note2()}]{Note2}%
  \BibitemOpen
  \bibinfo {note} {The map $J:A\to A$ is given by \( J(x) := \xi ^{1/2}T(x) \xi
  ^{-1/2} \), where $T:A\to A$ is given in Ref.~\cite
  {molnar2022matrix,ruizdealarcon2024matrix}}\BibitemShut {NoStop}%
\bibitem [{Note3()}]{Note3}%
  \BibitemOpen
  \bibinfo {note} {We note two typos $\Delta (e_{2,13})=e_{1,12}\otimes
  e_{2,13}+e_{2,13}\otimes e_{1,22}+\zeta e_{2,12}\otimes e_{2,31}-\zeta ^2
  e_{2,13}\otimes e_{2,33}$ and $\Delta (e_{2,23})=e_{1,22}\otimes
  e_{2,23}+e_{2,23}\otimes e_{1,12}+\zeta e_{2,32}\otimes e_{2,21}-\zeta ^2
  e_{2,33}\otimes e_{2,23}$.}\BibitemShut {Stop}%
\end{thebibliography}%

\appendix

\section{Example of Simple MPDO RFP}
\label{sec:example-simple}
In this section, we give two examples of simple MPDO RFP, which has the structure as in \cref{thm:simple-RFP}. We construct their Gibbs state Hamiltonian and parent Lindbladian and discuss their properties.  
\begin{exmp}[Product of Werner state]
\label{exm:werner}
    Consider an MPDO RFP as a product of Werner state $W(b)$, where $W(b)$ is supported on $\mathbb{C}^2\otimes \mathbb{C}^2$ and $b\in[0,1]$,
    \begin{equation}
    \begin{aligned}
        \rho^{(N)} &= W(b)\otimes W(b)\otimes\cdots\otimes W(b)\\
        W(b) &= \frac{1}{4}(\bo\otimes\bo - b(\sigma_x\otimes\sigma_x + \sigma_y\otimes \sigma_y+ \sigma_z\otimes \sigma_z)). 
    \end{aligned}
    \end{equation}
    The MPDO tensor $M$ only has a single BNT element and a single $k=1$, and we identify $\eta_{1,1}=W(b)$. Therefore, the Hilbert space for each site is decomposed as $\mathcal{H}=\mathcal{H}_l\otimes\mathcal{H}_r$ with $d_l=2,d_r=2$ and $d=d_l d_r=4$. The explicit MPDO tensor representation is injective and has bond dimension $D=4$,
    \begin{equation}
    \begin{aligned}
        &\begin{tikzpicture}[scale=.4,baseline={([yshift=-0.5ex] current bounding box.center)}]
        \GcTensor{(0,0)}{1.}{.5}{\small $r_1$}{-1}{lcolor};
        \draw (1.4,0) node {\scriptsize 1};
        \end{tikzpicture}=\frac{1}{2}\bo,\quad\quad
        \begin{tikzpicture}[scale=.4,baseline={([yshift=-0.5ex] current bounding box.center)}]
        \GcTensor{(0,0)}{1.}{.5}{\small $l_1$}{1}{lcolor};
        \draw (-1.4,0) node {\scriptsize 1};
        \end{tikzpicture}=\frac{1}{2}\bo,\\
        &\begin{tikzpicture}[scale=.4,baseline={([yshift=-0.5ex] current bounding box.center)}]
        \GcTensor{(0,0)}{1.}{.5}{\small $r_1$}{-1}{lcolor};
        \draw (1.4,0) node {\scriptsize 2};
        \end{tikzpicture}=\frac{\sqrt{b}}{2}\sigma_x,\quad\quad
        \begin{tikzpicture}[scale=.4,baseline={([yshift=-0.5ex] current bounding box.center)}]
        \GcTensor{(0,0)}{1.}{.5}{\small $l_1$}{1}{lcolor};
        \draw (-1.4,0) node {\scriptsize 2};
        \end{tikzpicture}=-\frac{\sqrt{b}}{2}\sigma_x,\\
        &\begin{tikzpicture}[scale=.4,baseline={([yshift=-0.5ex] current bounding box.center)}]
        \GcTensor{(0,0)}{1.}{.5}{\small $r_1$}{-1}{lcolor};
        \draw (1.4,0) node {\scriptsize 3};
        \end{tikzpicture}=\frac{\sqrt{b}}{2}\sigma_y,\quad\quad
        \begin{tikzpicture}[scale=.4,baseline={([yshift=-0.5ex] current bounding box.center)}]
        \GcTensor{(0,0)}{1.}{.5}{\small $l_1$}{1}{lcolor};
        \draw (-1.4,0) node {\scriptsize 3};
        \end{tikzpicture}=-\frac{\sqrt{b}}{2}\sigma_y,\\
        &\begin{tikzpicture}[scale=.4,baseline={([yshift=-0.5ex] current bounding box.center)}]
        \GcTensor{(0,0)}{1.}{.5}{\small $r_1$}{-1}{lcolor};
        \draw (1.4,0) node {\scriptsize 4};
        \end{tikzpicture}=\frac{\sqrt{b}}{2}\sigma_z,\quad\quad
        \begin{tikzpicture}[scale=.4,baseline={([yshift=-0.5ex] current bounding box.center)}]
        \GcTensor{(0,0)}{1.}{.5}{\small $l_1$}{1}{lcolor};
        \draw (-1.4,0) node {\scriptsize 4};
        \end{tikzpicture}=-\frac{\sqrt{b}}{2}\sigma_z
        \,.
    \end{aligned}
    \end{equation}

    When $b\in[0,1)$, $\eta_{1,1}=W(b)$ is full-rank, and this MPDO is a Gibbs state of nearest-neighboring commuting Hamiltonian $\rho^{(N)}=e^{-H^{(N)}}$, with Hamiltonian
    \begin{equation}
        H^{(N)}=\sum_{i=1}^{N} h_i
    \end{equation}
    where $h_i=\tau_i(h)$ and $h$ is supported on two consecutive sites as 
    \begin{equation}
        h = \bo_{\mathcal{H}_l}\otimes (-\log \eta_{1,1})\otimes\bo_{\mathcal{H}_r}.
    \end{equation}
    The identity $e^{\bo\otimes A}=\bo\otimes e^A$ is used in the derivation. 

    The channel $\mathcal{E}$ in parent Lindbladian is simply a replacement channel,
    \begin{equation}
        \mathcal{E}(X)=\tr_{1r,2,3l}(X)\otimes (\eta_{11}\otimes\eta_{11}).
    \end{equation}
\end{exmp}

\begin{exmp}[Variation of Werner state]
\label{exm:two-werner}
    Consider an MPDO RFP generated by tensor $M$ that has a single BNT element and two $k$'s, $k=1$ and $k=2$, with components
    \begin{equation}
    \begin{aligned}
        &\begin{tikzpicture}[scale=.4,baseline={([yshift=-0.5ex] current bounding box.center)}]
        \GcTensor{(0,0)}{1.}{.5}{\small $r_k$}{-1}{lcolor};
        \draw (1.4,0) node {\scriptsize 1};
        \end{tikzpicture}=\frac{1}{2\sqrt{2}}\bo,\quad\quad
        \begin{tikzpicture}[scale=.4,baseline={([yshift=-0.5ex] current bounding box.center)}]
        \GcTensor{(0,0)}{1.}{.5}{\small $l_k$}{1}{lcolor};
        \draw (-1.4,0) node {\scriptsize 1};
        \end{tikzpicture}=\frac{1}{2\sqrt{2}}\bo,\\
        &\begin{tikzpicture}[scale=.4,baseline={([yshift=-0.5ex] current bounding box.center)}]
        \GcTensor{(0,0)}{1.}{.5}{\small $r_k$}{-1}{lcolor};
        \draw (1.4,0) node {\scriptsize 2};
        \end{tikzpicture}=\frac{\sqrt{b_k}}{2\sqrt{2}}\sigma_x,\quad\quad
        \begin{tikzpicture}[scale=.4,baseline={([yshift=-0.5ex] current bounding box.center)}]
        \GcTensor{(0,0)}{1.}{.5}{\small $l_k$}{1}{lcolor};
        \draw (-1.4,0) node {\scriptsize 2};
        \end{tikzpicture}=-\frac{\sqrt{b_k}}{2\sqrt{2}}\sigma_x,\\
        &\begin{tikzpicture}[scale=.4,baseline={([yshift=-0.5ex] current bounding box.center)}]
        \GcTensor{(0,0)}{1.}{.5}{\small $r_k$}{-1}{lcolor};
        \draw (1.4,0) node {\scriptsize 3};
        \end{tikzpicture}=\frac{\sqrt{b_k}}{2\sqrt{2}}\sigma_y,\quad\quad
        \begin{tikzpicture}[scale=.4,baseline={([yshift=-0.5ex] current bounding box.center)}]
        \GcTensor{(0,0)}{1.}{.5}{\small $l_k$}{1}{lcolor};
        \draw (-1.4,0) node {\scriptsize 3};
        \end{tikzpicture}=-\frac{\sqrt{b_k}}{2\sqrt{2}}\sigma_y,\\
        &\begin{tikzpicture}[scale=.4,baseline={([yshift=-0.5ex] current bounding box.center)}]
        \GcTensor{(0,0)}{1.}{.5}{\small $r_k$}{-1}{lcolor};
        \draw (1.4,0) node {\scriptsize 4};
        \end{tikzpicture}=\frac{\sqrt{b_k}}{2\sqrt{2}}\sigma_z,\quad\quad
        \begin{tikzpicture}[scale=.4,baseline={([yshift=-0.5ex] current bounding box.center)}]
        \GcTensor{(0,0)}{1.}{.5}{\small $l_k$}{1}{lcolor};
        \draw (-1.4,0) node {\scriptsize 4};
        \end{tikzpicture}=-\frac{\sqrt{b_k}}{2\sqrt{2}}\sigma_z
        \,,
    \end{aligned}
    \end{equation}
    where $b_1,b_2\in[0,1]$. The Hilbert space for each site is decomposed as $\mathcal{H}=(\mathcal{H}_l^{(1)}\otimes\mathcal{H}_r^{(1)})\oplus(\mathcal{H}_l^{(2)}\otimes\mathcal{H}_r^{(2)})$ with $d_l^{(1)}=d_r^{(1)}=d_l^{(2)}=d_r^{(2)}=2$, and $d=\sum_k d_l^{(k)} d_r^{(k)}=8$. The tensor is injective with bond dimension $D=4$. Explicitly,
    \begin{equation}
        \begin{array}{c}
        \begin{tikzpicture}[scale=1.,baseline={([yshift=-0.65ex] current bounding box.center)}]
		\draw (-0.75,0) node {$\alpha$};
		\draw (0.75,0) node {$\beta$};
        \whaM{(0,0)}{2};
        \end{tikzpicture}
        \end{array}=\begin{pmatrix}
            (l_1)_\alpha\otimes (r_1)_\beta & \\
            & (l_2)_\alpha\otimes (r_2)_\beta
        \end{pmatrix}.
    \end{equation}
    By construction,
    \begin{equation}
    \begin{aligned}
        \eta_{1,1}&=\frac{1}{2}W(b_1),\quad \eta_{2,2}=\frac{1}{2}W(b_2)\\
        \eta_{1,2}&=\eta_{2,1}=\frac{1}{2}W(\sqrt{b_1 b_2}),
    \end{aligned}
    \end{equation}
    and $\tr(\eta_{k,k'})=1/2$ for $k,k'=1,2$, satisfying $\sum_k \tr(\eta_{k,k})=1$. 

    When $b_1,b_2\in[0,1)$, $\eta_{k,k'}$ are full-rank for $k,k'=1,2$, and this MPDO is a Gibbs state of nearest-neighboring commuting Hamiltonian $\rho^{(N)}=e^{-H^{(N)}}$, with Hamiltonian
    \begin{equation}
        H^{(N)}=\sum_{i=1}^{N} h_i
    \end{equation}
    where $h_i=\tau_i(h)$ and $h$ is supported on two consecutive sites as 
    \begin{equation}
        h = \bigoplus_{k_1,k_2}\bo_{\mathcal{H}_l^{(k_1)}}\otimes (-\log \eta_{k_1,k_2})\otimes\bo_{\mathcal{H}_r^{(k_2)}}.
    \end{equation}
    The identities $e^{\bo\otimes A}=\bo\otimes e^A$ and $e^{A\oplus B}=e^A\oplus e^B$ are used to derive
    \begin{equation}
    \begin{aligned}
        &B_{i,i+1}:= e^{-\tau_i(h)\otimes \bo_{\mathrm{rest}}}\\
        &\quad= \left[\bigoplus_{k_i,k_{i+1}}\bo_{\mathcal{H}_l^{(k_i)}}\otimes  \eta_{k_i,k_{i+1}}\otimes\bo_{\mathcal{H}_r^{(k_{i+1})}}\right]\otimes \bo_{\mathrm{rest}},
    \end{aligned}
    \end{equation}
    leading to
    \begin{equation}
        \begin{aligned}
             \rho^{(N)}(M)&=\bigoplus_{k_1,\cdots,k_N}\eta_{k_1,k_2}\otimes \eta_{k_2,k_3} \otimes\cdots\otimes \eta_{k_N,k_1}\\
             &=\prod_{i=1}^{N}B_{i,i+1}. 
        \end{aligned}
    \end{equation}
   Indeed, the MPDO is a Gibbs state of commuting local Hamiltonian. 

    Finally, the channel $\mathcal{E}$ in parent Lindbladian is constructed as in \cref{eqn:simplechannel}, with isometries
    \[
    \begin{aligned}
    &Q_1=\begin{pmatrix}
        1 & 0 & 0 & 0 & 0 & 0 & 0 & 0\\
        0 & 1 & 0 & 0 & 0 & 0 & 0 & 0\\
        0 & 0 & 1 & 0 & 0 & 0 & 0 & 0\\
        0 & 0 & 0 & 1 & 0 & 0 & 0 & 0
    \end{pmatrix}\\
    &Q_2=\begin{pmatrix}
        0 & 0 & 0 & 0 & 1 & 0 & 0 & 0\\
        0 & 0 & 0 & 0 & 0 & 1 & 0 & 0\\
        0 & 0 & 0 & 0 & 0 & 0 & 1 & 0\\
        0 & 0 & 0 & 0 & 0 & 0 & 0 & 1
    \end{pmatrix},
    \end{aligned}
    \]
     and $\eta$'s constructed above. This is an example for $G_L\subseteq \mathcal{F}_{\mathcal{E}}$ commented in \cref{sec:general-comment}. In this example, $L=3$, and $G_L$ has dimension $\mathrm{dim}(G_{L=3})=D^2=16$ since the MPDO is injective. On the other hand, the fixed point space of $\mathcal{E}$ is
    \begin{equation}
        \begin{aligned}
            &\mathcal{F}_{\mathcal{E}}=\bigoplus_{k_1,k_3}\\
            &\mathcal{M}_{d_l^{(k_1)}}\otimes\left(\bigoplus_{k_2}\eta_{k_1,k_2}\otimes \eta_{k_2,k_3}\right)\otimes\mathcal{M}_{d_r^{(k_3)}},
        \end{aligned}
    \end{equation}
    where $\mathcal{M}_{d}$ is the full matrix algebra of $d\times d$ matrix. Therefore, the fixed point space has dimension $\mathrm{dim}(\mathcal{F}_{\mathcal{E}_i})=\sum_{k_1,k_3}(d_l^{(k_1)})^2 (d_r^{(k_3)})^2=64$, and $G_L\subseteq \mathcal{F}_{\mathcal{E}}$. Due to this mismatch, the techniques for proving ground state degeneracy for parent Hamiltonian do not directly apply to parent Lindbladian. For example, we cannot use the intersection property in \cite{perez2006matrix}. 
\end{exmp}

\section{Construction of RFP MPDOs using \texorpdfstring{\emph{C*}}{C*}-weak Hopf algebras}\label{appendix:WHA}

This appendix is devoted to introducing and recalling the rigorous definitions and statements for the construction of RFP MPDOs based on $C^*$-weak Hopf algebras. We also refer the reader to \cite{molnar2022matrix,ruizdealarcon2024matrix,yoshiko2024haag} for detailed proofs.

\subsection{Algebraic framework}

Let us first recall that a coalgebra is a complex vector space $C$ endowed with a linear map $\Delta:C\to C\otimes C$, called comultiplication, which is comultiplicative, i.e.
\begin{equation}
( \Delta \otimes \mathrm{Id} ) \circ \Delta = (\mathrm{Id}\otimes \Delta) \circ \Delta
\end{equation}
and a linear functional $\epsilon:C\to\mathbb{C}$, called counit, which is compatible in the sense that
\begin{equation}
( \epsilon \otimes \mathrm{Id} ) \circ \Delta = \mathrm{Id} = (\mathrm{Id}\otimes \epsilon) \circ \Delta.
\end{equation}
For successive applications of the comultiplication, 
\begin{equation}
\Delta^{(1)} := \Delta,\quad \Delta^{(n+1)} :=( \Delta\otimes\mathrm{Id}^{\otimes n})\circ \Delta^{(n)}
\end{equation}
for all $n\in\mathbb{N}$. In general, it is more appropriate to employ Sweedler's notation, in which we denote
\begin{equation}
x_{(1)}\otimes x_{(2)}\otimes \cdots \otimes x_{(n+1)} := \Delta^{(n)}(x),
\end{equation}
a shorthand notation for sums $\sum_{i} x_{1,i} \otimes \cdots \otimes x_{n+1,i}$; note that the terms $x_{(i)}$ above are not individual tensor factors.
It turns out that representations of coalgebras (more concretely, comodules) are in one-to-one correspondence with tensors generating MPS \cite{molnar2022matrix}, and the action of the comultiplication is understood as a notion of system size extensivity for these kinds of states.


\begin{defn}\label{def:WHA}
    A \emph{$C^\ast$-weak Hopf algebra} $A$ is a complex finite-dimensional $C^*$-algebra endowed with the structure of a coalgebra, such that the comultiplication is a $*$-algebra map, i.e.~for all $x\in A$:
    \begin{equation} \Delta(xy) = \Delta(x)\Delta(y), \quad \Delta(x^*) = \Delta(x)^*
    \end{equation}
        where the product and the $*$-operation is defined on $A\otimes A$ component-wise,
       the counit is \emph{weakly multiplicative}, i.e.
        for all $x,y,z\in A$:
        \begin{equation}
            \varepsilon(xyz) = \varepsilon(x y_{(1)})\varepsilon(y_{(2)}z)
            = \varepsilon(x y_{(2)})\varepsilon(y_{(1)}z),
        \end{equation}
       the unit is \emph{weakly comultiplicative}, i.e.
        \begin{equation}
            1_{(1)} \otimes 1_{(2)} \otimes 1_{(3)} = 1_{(1)} \otimes 1_{(2)} 1_{(1')}\otimes 1_{(2')}
        \end{equation}
        where the prime symbol is intended to distinguish between different coproduct instances of $1\in A$,
    and there exists a linear map, $S:A\to A$, called antipode, which is antimultiplicative, anticomultiplicative, and satisfies for all $x\in A$:
        \begin{align}
            S(x_{(1)}) x_{(2)} &= \varepsilon(1_{(1)} x) 1_{(2)}
        \\
            x_{(1)} S(x_{(2)}) &= \varepsilon(x 1_{(2)}) 1_{(1)}.
        \end{align}
        A \emph{$C^*$-Hopf algebra} is a $C^*$-weak Hopf algebra with
        \[ \Delta(1) = 1\otimes 1. \]
\end{defn}

The axioms of $C^*$-weak Hopf algebras are such that the notion is self-dual, i.e.~the dual vector space $A^* = \operatorname{Hom}(A,\mathbb{C})$ is endowed with the structure of a $C^\ast$-weak Hopf algebra \cite{bohm_weak_1999}, and such that category of $*$-algebra representations are unitary (multi)fusion category \cite{etingof_fusion_2005, etingof_tensor_2015}. First, the monoidal product is given by
$\Phi_1\boxtimes \Phi_2 := (\Phi_1\otimes \Phi_2)\circ \Delta$,
and there exists a \emph{trivial} $*$-representation in the sense of relaxed monoidal categories, denoted $(\mathcal{H}_\varepsilon,\Phi_\varepsilon)$ \cite{bohm_weak_2000}.
Second, for any finite $*$-re\-pre\-sen\-ta\-tion $(\mathcal{H},\Phi)$ there exists a $*$-re\-pre\-sen\-ta\-tion of $A$, given by the expression
\[
    \bar\Phi(x) = (\Phi\circ J)(x)^{\mathrm{t}},
\]
on the dual Hilbert space, for certain linear map $J:A\to A$ \cite{bohm_weak_2000}\footnote{The map $J:A\to A$ is given by
\( J(x) := \xi^{1/2}T(x) \xi^{-1/2} \), where $T:A\to A$ is given in Ref.~\cite{molnar2022matrix,ruizdealarcon2024matrix}}; this generalizes the notion of dual group-representations. We note that, in general, $J\neq S$.
Furthermore, it is a semisimple category and the set $\operatorname{Irr}A$ of equivalence classes of irreducible representations is finite, also known as the set of \emph{sectors} of $A$. 
The unusual feature of the trivial $*$-re\-pre\-sen\-ta\-tion in this setting is that it can be decomposable \cite{bohm_weak_1999,bohm_weak_2000}; if this is not the case, $A$ is said to be \emph{connected} or \emph{pure}, and if both trivial $*$-representations of $A$ and $A^*$ are indecomposable, $A$ is said to be \emph{biconnected}. We note that categories of $*$-representation of biconnected C*-weak Hopf algebras encompass unitary fusion categories \cite{etingof_fusion_2005}.

\begin{exmp}
\label{ex:group-algebra}
    Let $G$ be a finite group. Then, the group algebra $A = \mathbb{C}G$, endowed with the linear extensions of the maps defined by
    \[ \Delta(e_i) = e_i\otimes e_i, \; \varepsilon: e_i \mapsto 1, \; S(e_i) = e_i^* = e_i^{-1},  \]
    for all $e_i\in G$, defines a $C^*$-Hopf algebra.
\end{exmp}

\subsection{Tensor network constructions}

In terms of tensor networks, given a $C^*$-weak Hopf algebra $A$ and two faithful $*$-representations $\Phi$ and $\Psi$ of $A$ and $A^*$, respectively, one can define the rank-four tensor
\begin{equation}
\label{eqn:black-tensor-constr}
\begin{tikzpicture}[scale=1]
    \draw[-mid] (0.6, 0.688) -- (0.6, 1.129);
    \draw[-mid] (0.6, 0) -- (0.6, 0.441);
    \draw[Virtual, -mid] (1.199, 0.564) -- (0.723, 0.564);
    \draw[Virtual, -mid] (0.476, 0.564) -- (0, 0.564);
    \filldraw[black, ultra thin] 
    (0.6, 0.564) circle[radius=0.123];
\end{tikzpicture}
=
\sum_{i=1}^{\operatorname{dim}A} \Phi(e_i)\otimes \Psi(e^i),
\end{equation}
where $e_1,\ldots, e_d\in A$ stands for a basis of $A$ and $e^1,\ldots, e^d\in A^*$ stands for the dual basis. 
Here, the first tensor factor corresponds to an endomorphism at the physical level and the second one corresponds to an endomorphism to the virtual level.

For any $x\in A$, there is $b(x)\in\Psi(A^*)\subset\operatorname{End}(\mathcal{K})$ satisfying the following:
\begin{equation}
\begin{tikzpicture}[scale=0.75]
    \draw[-mid] (0.847, 0.688) -- (0.847, 1.129);
    \draw[Virtual] (0.723, 0.564) -- (0.247, 0.564);
    \draw[-mid] (1.835, 0.688) -- (1.835, 1.129);
    \draw[Virtual, -mid] (2.434, 0.564) -- (1.958, 0.564);
    \draw[Virtual, -mid] (1.835, 0.564) -- (0.847, 0.564);
    \draw[-mid] (3.81, 0.688) -- (3.81, 1.129);
    \draw[Virtual, -mid] (3.687, 0.564) -- (3.21, 0.564);
    \node[anchor=center, text=Virtual] at (2.819, 0.7) {$\overset{N}\cdots$};
    \draw[Virtual] (0.247, 0.564) arc[start angle=90, end angle=270, radius=0.247] -- (3.924, 0.071) arc[start angle=-90, end angle=90, radius=0.247];
    \filldraw[ultra thin, fill=Virtual] 
    (0.359, 0.564) circle[radius=0.123];
    \draw[bevel, -mid] (0.847, 0) -- (0.847, 0.441);
    \draw[bevel, -mid] (1.835, 0) -- (1.835, 0.441);
    \draw[bevel, -mid] (3.81, 0) -- (3.81, 0.441);
    \node[anchor=center, font=\footnotesize, text=Virtual] at (0.35, 0.904) {$b(x)$};
    \filldraw[black, ultra thin] 
    (3.81, 0.564) circle[radius=0.123];
    \filldraw[black, ultra thin] 
    (0.847, 0.564) circle[radius=0.123];
    \filldraw[black, ultra thin] 
    (1.835, 0.564) circle[radius=0.123];
\end{tikzpicture}
= (\Phi^{\otimes N}\circ \Delta^{(N-1)})(x).
\end{equation}
Moreover, provided a decomposition of the $*$-representation $\Psi = \Psi_1^{\otimes n_1}\oplus\cdots\oplus\Psi_g^{\otimes n_g}$ in terms of irreducible $*$-re\-pre\-sen\-tations $\Psi_1,\ldots,\Psi_g$, the rank-four tensor introduced above can be decomposed as a sum of the rank-four tensors denoted
\begin{equation}
\begin{tikzpicture}[scale=1]
    \draw[-mid] (0.6, 0.688) -- (0.6, 1.129);
    \draw[-mid] (0.6, 0) -- (0.6, 0.441);
    \draw[Virtual, -mid] (1.199, 0.564) -- (0.723, 0.564);
    \draw[Virtual, -mid] (0.476, 0.564) -- (0, 0.564);
    \node[anchor=center, font=\footnotesize, text=Virtual] at (0.24, 0.701) {$a$};
    \node[anchor=center, font=\footnotesize, text=Virtual] at (0.937, 0.701) {$a$};
    \filldraw[black, ultra thin] 
    (0.6, 0.564) circle[radius=0.123];
\end{tikzpicture}
= \sum_{i=1}^{\operatorname{dim}A} \Phi(e_{i})\otimes \Psi_a^{\otimes n_a}(e^i);
\end{equation}
or, from another point of view, these tensors induce a BNT. In particular, the tensor is in horizontal Canonical Form.
The duality manifests as the existence of another tensor,
\begin{equation}
\label{eqn:App-Cstar-white}
\begin{tikzpicture}[scale=1]
    \draw[-mid] (0.6, 0.688) -- (0.6, 1.129);
    \draw[-mid] (0.6, 0) -- (0.6, 0.441);
    \draw[Virtual, -mid] (1.199, 0.564) -- (0.723, 0.564);
    \draw[Virtual, -mid] (0.476, 0.564) -- (0, 0.564);
    \filldraw[ultra thin, fill=white] 
    (0.6, 0.564) circle[radius=0.123];
\end{tikzpicture}
= \sum_{i=1}^{\operatorname{dim}A} (\Phi\circ J)(e_i)\otimes \Psi(e^i).
\end{equation}

\subsection{Pulling-through identities}

In this context, it is fundamental to recall the existence of a positive idempotent $\Omega\in A$ satisfying a pulling-through identity, which can be regarded as a microscopical description of the topological order phenomena. The following result summarizes the fundamental properties needed in the framework.

\begin{thm}
    Let $A$ be a $C^*$-weak Hopf algebra. Then, there exists $\Omega\in A$ such that it is a positive idempotent which is invariant under both $S$ and $J$:
    \begin{equation} \Omega^2 = \Omega = \Omega^* = S(\Omega) = J(\Omega) \end{equation}
    and it is a cocentral element, i.e.~it satisfies
    \begin{equation}  \Omega_{(2)}\otimes\Omega_{(1)} = \Omega_{(1)}\otimes\Omega_{(2)} \end{equation} 
    or, equivalently, for any $n\in\mathbb{N}$ and any shift permutation $\sigma$ of $\{1,\ldots, n\}$, it holds that
    \begin{equation} 
        \Omega_{(\sigma(1))}\otimes\cdots\otimes \Omega_{(\sigma(n))} = \Omega_{1}\otimes\cdots\otimes \Omega_{(n)};
    \end{equation} 
    and moreover, there exists an invertible positive element $\xi\in A$ such that $J(\xi) = \xi$ and, for all $x,y\in A$,
    \begin{align}  x\xi^\frac{1}{2} J(\Omega_{(1)}) \otimes \Omega_{(2)} = \xi^\frac{1}{2} J(\Omega_{(1)})  \otimes  \Omega_{(2)} x,
    \\ J(\Omega_{(1)}) \xi^\frac{1}{2} y \otimes \Omega_{(2)} = J(\Omega_{(1)}) \xi^\frac{1}{2} \otimes y \Omega_{(2)},\label{eqn:pull-through}\end{align} 
    known as pulling-through identities. Let $\omega\in A^*$ stand for the dual analogue of $\Omega\in A$, then
    \begin{equation}
    \omega(\xi^{\frac{1}{2}} \Omega_{(1)} \xi^{\frac{1}{2}} J(x)) \Omega_{(2)} = x
    \label{eqn:omega-identity}
    \end{equation}
    for all $x\in A$.
\end{thm}

See \cite{ruizdealarcon2024matrix,yoshiko2024haag} for a proof.

\begin{remark}
\begin{enumerate}
\item If the C*-weak Hopf algebra $A$ is biconnected,
\[
    \Omega = \mathcal{D}^{-2}(d_1 \mathrm{x}_1 + \cdots + d_g \mathrm{x}_g),
\]
where $\mathrm{x}_1,\ldots,\mathrm{x}_g\in A\cong A^{**}$ stand for the characters of the irreducible $*$-re\-pre\-sen\-ta\-tions of $A^*$, $d_1,\ldots,d_g > 0$ are the associated quantum dimensions or Frobenius-Perron dimensions, and $\mathcal{D}^2 := d_1^2 + \cdots + d_g^2$ is the total quantum dimension.
\item In the particular case of $C^*$-Hopf algebras, see \cref{def:WHA}, the element introduced above $\Omega\in A$ coincides with the well-known Haar integral \cite{ruizdealarcon2024matrix}, $J = S$, and $\xi = \mathcal{D} 1$.
\item More concretely, for group algebras $A = \mathbb{C}G$, $$\Omega = |G|^{-1}\sum_{g\in G} g, \quad \omega = \delta_e, $$
     where $\delta_e$ is the dual basis element of $e\in G$.
\end{enumerate}
\end{remark}

\subsection{Renormalization fixed point MPDOs}

\begin{prop}
    The tensor
    \begin{equation}
    \label{eqn:M-tensor-constr}
    \begin{tikzpicture}[scale=1]
    \draw[-mid] (0.6, 0.811) -- (0.6, 1.129);
    \draw[-mid] (0.6, 0) -- (0.6, 0.318);
    \draw[Virtual, -mid] (1.199, 0.564) -- (0.847, 0.564);
    \whaMsymb{(0.6, 0.564)};
    \draw[Virtual, -mid] (0.353, 0.564) -- (0, 0.564);
\end{tikzpicture}
:= 
\begin{tikzpicture}[scale=1]
    \draw (0.6, 0.688) -- (0.6, 1.129);
    \draw[-mid] (0.6, 0) -- (0.6, 0.441);
    \draw[Virtual, -mid] (1.199, 0.564) -- (0.723, 0.564);
    \draw[Virtual, -mid] (0.476, 0.564) -- (0, 0.564);
    \filldraw[black, ultra thin] 
    (0.6, 0.564) circle[radius=0.123];
    \filldraw[ultra thin, fill=Virtual] 
    (0.6, 0.908) circle[radius=0.123];
    \node[anchor=center, font=\footnotesize] at (1.094, 0.908) {$b(\omega)$};
\end{tikzpicture}
\end{equation}
defines RFP MPDOs in the sense of \cref{def:RFP_MPDO}.
\end{prop}

Moreover, let us define the rank-four tensor  introduced in \cref{subsubsec:nonSimpleRfpWha} by the expression
\begin{align}
\label{eq:ExplicitDefMInv}
    \begin{tikzpicture}[scale=1]
    \draw[-mid] (0.6, 0.811) -- (0.6, 1.129);
    \draw[-mid] (0.6, 0) -- (0.6, 0.318);
    \draw[Virtual, -mid] (1.199, 0.564) -- (0.847, 0.564);
    \whaNsymb{(0.6, 0.564)};
    \draw[Virtual, -mid] (0.353, 0.564) -- (0, 0.564);
\end{tikzpicture}
&:= 
\begin{tikzpicture}[scale=1]
    \draw (0.6, 0.688) -- (0.6, 1.129);
    \draw (0.6, 0) -- (0.6, 0.441);
    \draw[Virtual, -mid] (1.199, 0.564) -- (0.723, 0.564);
    \draw[Virtual, -mid] (0.476, 0.564) -- (0, 0.564);
    \filldraw[ultra thin, fill=white] 
    (0.6, 0.564) circle[radius=0.123];
    \filldraw[ultra thin, fill=GreenYellow!50] 
    (0.6, 0.22) circle[radius=0.123];
    \node[anchor=center, font=\footnotesize] at (1.305, 0.22) {$\Phi(\xi)^{1/2}$};
    \filldraw[ultra thin, fill=GreenYellow!50] 
    (0.6, 0.908) circle[radius=0.123];
    \node[anchor=center, font=\footnotesize] at (1.305, 0.908) {$\Phi(\xi)^{1/2}$};
\end{tikzpicture}~,\\
\intertext{the rank-two tensor $\Omega$ as the one given by}
    \begin{tikzpicture}[scale=1]
    \draw[Virtual, -mid] (1.199, 0.564) -- (0.847, 0.564);
    \filldraw[ultra thin, fill=white] 
    (0.6, 0.564) circle[radius=0.247];
    \node[anchor=center] at (0.6, 0.564) {$\Omega$};
    \draw[Virtual, -mid] (0.353, 0.564) -- (0, 0.564);
\end{tikzpicture}
&:= b(\Omega),\\
\intertext{and the rank-two tensor $\Xi$ is defined by}
\label{eq:ExplicitDefXi}
    \begin{tikzpicture}[scale=1]
    \draw[Virtual, -mid] (1.199, 0.564) -- (0.847, 0.564);
    \filldraw[ultra thin, fill=white] 
    (0.6, 0.564) circle[radius=0.247];
    \node[anchor=center] at (0.6, 0.564) {$\Xi$};
    \draw[Virtual, -mid] (0.353, 0.564) -- (0, 0.564);
\end{tikzpicture}
&:= \Psi(\hat{\xi})^{-1}
\end{align}
where $\hat{\xi}\in A^*$ stands for the dual analogue of $\xi$. By the above definitions, \cref{eq:BlackVsWhite} is simply a consequence of the pulling-through identity \cref{eqn:omega-identity}, and \cref{eq:XiOmega} is translated to 
\[
\tr(\Psi_a(\hat{\xi}^{-1}))=d_a/\mathcal{D}^2. 
\]

On the one hand, the explicit expression of the fine-graining quantum channel $\mathcal{T}$ is the following:
\begin{align}
\label{eqn:wha-channelT}
\mathcal{T}(X) := \mathrm{Tr}((\Phi\circ J)&(\xi^\frac{1}{2}\Omega_{(2)} \xi^\frac{1}{2})X) \\ & b(\omega)^{\otimes 2}(\Phi^{\otimes 2}\circ \Delta)(\Omega_{(1)});\notag
\end{align}
for any $X\in\operatorname{End}(\mathcal{H})$. 
%
%
We refer the reader to \cite{ruizdealarcon2024matrix} for a proof of the facts that it is completely positive, trace-preserving and, in addition, it satisfies the fine-graining condition
\[ \mathcal{T}(b(\omega)\Phi(x)) = b(\omega)^{\otimes 2} (\Phi^{\otimes 2}\circ\Delta)(x). \]
On the other hand, the coarse-graining quantum channel $\mathcal{S}:\operatorname{End}(\mathcal{H} \otimes \mathcal{H})\to \operatorname{End}(\mathcal{H})$ is given by
\[
    \mathcal{S} := \mathcal{S}^{(\Eone)} + \mathcal{S}^{(\Etwo)},
\]
where the first summand is defined by the expression
\begin{align}
\label{eqn:wha-channelS1}
\mathcal{S}^{(\Eone)}(Y) := \mathrm{Tr}((\Phi^{\otimes 2}\circ\Delta \circ J)(\xi^\frac{1}{2} \Omega_{(2)}&\xi^\frac{1}{2}) Y)
\\ &b(\omega) \Phi(\Omega_{(1)})\notag
\end{align}
and is supported on the symmetric subspace, satisfying the coarse-graining condition
\[
\mathcal{S}^{(\Eone)}(b(\omega)^{\otimes 2}(\Phi^{\otimes 2}\circ \Delta)(x)) = b(\omega)\Phi(x)
\]
for all $x\in A$, and the second part is given by the expression on the complementary subspace,
\begin{equation}
\label{eqn:wha-channelS2}
    \mathcal{S}^{(\Etwo)}(Y) := \mathrm{Tr}(Y-(\Phi^{\otimes 2}\circ \Delta)(1)Y) \sigma_0
\end{equation}
where $\sigma_0\in\operatorname{End}(\mathcal{H}\otimes\mathcal{H})$ is any mixed state; as commented above, this part is completely positive by construction and completes the first part making $\mathcal{S}$ a trace-preserving linear map. We typically consider $\sigma_0 := \omega(\Omega)^{-1} b(\omega)\Phi(\Omega)$ for simplicity. Note that
\[
\mathcal{S}^{(\Etwo)}(b(\omega)^{\otimes 2}(\Phi^{\otimes 2}\circ \Delta)(x)) = 0.
\]



\begin{widetext}
\section{Proof of Frustration Freeness}
\label{sec:proof-ff}
In this section, we prove the frustration freeness of a general parent Lindbladian (not necessarily commuting) using a ``flow diagram'' method. Define 
\begin{equation}
\mathcal{E}^{(N)}=\frac{1}{N}\sum_{i=1}^N \mathcal{E}_i,    
\end{equation}
which is related to parent Lindbladian by $\mathcal{L}^{(N)}=N(\mathcal{E}^{(N)}-\bo)$, and the steady state of $\mathcal{L}^{(N)}$ corresponds to the fixed point of $\mathcal{E}^{(N)}$. The proof is inspired by the observation that if $\rho$ is a fixed point of $\mathcal{E}^{(N)}$, then it is also a fixed point of $(\mathcal{E}^{(N)})^{m}$ for any positive integer $m$. One can rewrite the fixed-point $\rho$ into a summation
\begin{equation}
\begin{aligned}
    \rho&=\frac{1}{N}\mathcal{E}^{(\Eone)}_1(\rho)+\frac{1}{N}\mathcal{E}^{(\Etwo)}_1(\rho) + \frac{1}{N}\mathcal{E}^{(\Eone)}_2(\rho)+\frac{1}{N}\mathcal{E}^{(\Etwo)}_2(\rho)+ \cdots + \frac{1}{N}\mathcal{E}^{(\Eone)}_N(\rho)+\frac{1}{N}\mathcal{E}^{(\Etwo)}_N(\rho),
\end{aligned}
\label{eqn:rho-decomp}
\end{equation}
and ask how each of these terms gets modified by successive applications of $\mathcal{E}^{(N)}$ and draw the corresponding flow diagram.
We note that all $\mathcal{E}_i^{(\Eone)}$ and $\mathcal{E}_i^{(\Etwo)}$ are CP maps and each term in the decomposition is positive semi-definite. 

Since the simple injective MPDOs and the non-simple MPDOs generated by $C^*$-Hopf algebra are guaranteed to be frustration-free due to the commutation of their parent Lindbladians, we only need to consider the remaining two scenarios: (1) simple non-injective simple MPDOs, and (2) non-simple MPDOs generated by a $C^*$-weak Hopf algebra which is not a $C^*$-Hopf algebra. 
For non-injective simple MPDOs, by choosing $\rho_0$ in $\mathcal{E}_i^{(\Etwo)}$ as \cref{eqn:simplerho0}, the pieces in the channels satisfy
    \begin{align}
        \mathcal{E}_i^{(\Eone)}(\rho_0)&=\rho_0\label{eqn:rho0eig}\\
        \mathcal{E}_i^{(\Eone)}\circ \mathcal{E}_i^{(\Etwo)}&=\mathcal{E}_i^{(\Etwo)}\label{eqn:e12}\\
        \mathcal{E}_i^{(\Etwo)}\circ\mathcal{E}_i^{(\Eone)}&=0\label{eqn:e21vanish}\\
        \mathcal{E}_i^{(\Eone)}\circ \mathcal{E}_i^{(\Eone)}&=\mathcal{E}_i^{(\Eone)}\label{eqn:idempotent}\\
    \mathcal{E}_i^{(\Eone)}\circ \mathcal{E}_{i+1}^{(\Eone)}&=\mathcal{E}_{i+1}^{(\Eone)}\circ\mathcal{E}_i^{(\Eone)}\label{eqn:ecommute}.
    \end{align}
For MPDO generated by $C^*$-weak Hopf algebra, by choosing $\mathcal{E}^{(\Etwo)}$ as \cref{eqn:choiceE2WHA}, the pieces in the channels also satisfy \cref{eqn:rho0eig,eqn:e12,eqn:e21vanish,eqn:idempotent,eqn:ecommute}. 
We will use these properties to draw the flow diagram. 
We organize the terms with different structures into separate boxes. For every action of $\mathcal{E}^{(N)}$, $\rho$ is rewritten, flowing between boxes.
Since $\mathcal{E}^{(N)}$ is a quantum channel, the summation of the trace of terms in different boxes is equal to one.
By iterating this process for large enough $m$, we prove below that all the terms of a fixed-point $\rho$ in \cref{eqn:rho-decomp} will flow to the final box, in the form of 
\begin{equation}
    \rho=\mathcal{E}^{(\Eone)}_1\circ \mathcal{E}^{(\Eone)}_2\circ\cdots\circ\mathcal{E}^{(\Eone)}_N(\rho).
    \label{eqn:ssform}
\end{equation}
Using \cref{eqn:e21vanish,eqn:idempotent,eqn:ecommute}, one can then show such a fixed-point state is frustration-free: $\mathcal{E}_i(\rho)=\mathcal{E}^{(\Eone)}_i(\rho)=\rho$, leading to $\mathcal{L}_i(\rho)=0$. In the following, we separate the treatments of the two cases.

\subsection{Non-injective simple MPDO}
Below, we prove the frustration-freeness of parent Lindbladian for non-injective simple MPDO RFP with system size $N=5$ as an explicit example. Generalization to larger $N$ is straightforward.

\begin{prop}
    The parent Lindbladian for non-injective simple MPDO RFP with system size $N=5$ is frustration-free. 
\end{prop}

\begin{proof}
We first decompose fixed-point $\rho$ as in \cref{eqn:rho-decomp}. 
Under an action of $\mathcal{E}^{(N)}$, the decomposition terms of $\rho$ will flow to each other, which includes the following arrows,
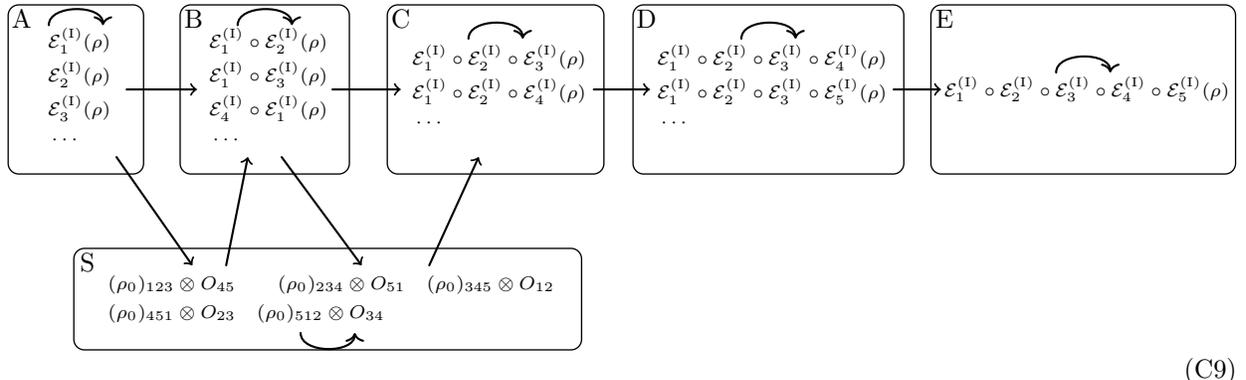
\begin{equation}
    \begin{array}{c}
        \begin{tikzpicture}[scale=.45,baseline={([yshift=-0.75ex] current bounding box.center)}
        ]
        \def\boxh{6.2};
        \def\boxv{6.2};
        \def\eqnh{1};
        \def\eqnht{0.8};
        \def\sf{0.};
        \def\posiB{0.9*\boxh};
        \def\posiC{2.*\boxh};
        \def\posiD{3.3*\boxh};
        \def\posiE{4.8*\boxh};
        \def\posiS{1.2*\boxh};
        \EquationBox{(0,0)}{2}{2.5}{boxc1}{
        $\begin{aligned}
            &\mathcal{E}_1^{(\Eone)}(\rho)\\
            &\mathcal{E}_2^{(\Eone)}(\rho)\\
            &\mathcal{E}_3^{(\Eone)}(\rho)\\
            &\cdots
        \end{aligned}$
        }{A}
        \draw [->, thick](-0.8,\eqnh*1.9) to [out=90,in=90] (0.8,\eqnh*1.9) ;
        \EquationBox{(\posiB,0)}{2.5}{2.5}{boxc2}{
        $\begin{aligned}
            &\mathcal{E}_1^{(\Eone)}\circ \mathcal{E}_2^{(\Eone)}(\rho)\\
            &\mathcal{E}_1^{(\Eone)}\circ \mathcal{E}_3^{(\Eone)}(\rho)\\
            &\mathcal{E}_4^{(\Eone)}\circ \mathcal{E}_1^{(\Eone)}(\rho)\\
            &\cdots
        \end{aligned}$
        }{B}
        \draw [->, thick](-0.8+\posiB,\eqnh*1.9) to [out=90,in=90] (0.8+\posiB,\eqnh*1.9);
        \EquationBox{(\posiC,0.)}{3.2}{2.5}{boxc3}{
        $\begin{aligned}
            &\mathcal{E}_1^{(\Eone)}\circ \mathcal{E}_2^{(\Eone)}\circ\mathcal{E}_3^{(\Eone)}(\rho)\\
            &\mathcal{E}_1^{(\Eone)}\circ \mathcal{E}_2^{(\Eone)}\circ\mathcal{E}_4^{(\Eone)}(\rho)\\
            &\cdots
        \end{aligned}$
        }{C}
        \draw [->, thick](-0.8+\posiC,\eqnh*1.5) to [out=90,in=90] (0.8+\posiC,\eqnh*1.5) ;
        \EquationBox{(\posiD,0.)}{4}{2.5}{boxc4}{
        $\begin{aligned}
            &\mathcal{E}_1^{(\Eone)}\circ \mathcal{E}_2^{(\Eone)}\circ\mathcal{E}_3^{(\Eone)}\circ\mathcal{E}_4^{(\Eone)}(\rho)\\
            &\mathcal{E}_1^{(\Eone)}\circ \mathcal{E}_2^{(\Eone)}\circ\mathcal{E}_3^{(\Eone)}\circ\mathcal{E}_5^{(\Eone)}(\rho)\\
            &\cdots
        \end{aligned}$
        }{D}
        \draw [->, thick](-0.8+\posiD,\eqnh*1.5) to [out=90,in=90] (0.8+\posiD,\eqnh*1.5) ;
        \EquationBox{(\posiE,0.)}{4.5}{2.5}{boxc5}{
        $\begin{aligned}
            \mathcal{E}_1^{(\Eone)}\circ \mathcal{E}_2^{(\Eone)}\circ\mathcal{E}_3^{(\Eone)}\circ\mathcal{E}_4^{(\Eone)}\circ\mathcal{E}_5^{(\Eone)}(\rho)
        \end{aligned}$
        }{E}
        \draw [->, thick](-0.8+\posiE,\eqnh*0.5) to [out=90,in=90] (0.8+\posiE,\eqnh*0.5) ;
        \EquationBox{(\posiS,-\boxv)}{7.5}{1.5}{lcolor}{
        $\begin{aligned}
            &(\rho_0)_{123}\otimes O_{45}\quad \quad (\rho_0)_{234}\otimes O_{51}\quad(\rho_0)_{345}\otimes O_{12}\\
            & (\rho_0)_{451}\otimes O_{23} \quad (\rho_0)_{512}\otimes O_{34}
        \end{aligned}$
        }{S}
        \draw [->, thick](-0.8+\posiS,-\eqnh-\boxv) to [out=-90,in=-90] (0.8+\posiS,-\eqnh-\boxv);
        \draw [-to,  thick](1.5,0) -- (\posiB-2,0);
        \draw [-to,  thick](2+\posiB,0) -- (\posiC-2.8,0);
        \draw [-to,  thick](2.9+\posiC,0) -- (\posiD-3.5,0);
        \draw [-to,  thick](3.7+\posiD,0) -- (\posiE-4.2,0);
        \draw [-to,  thick](1.2,-2) -- (\posiS-4,1-\boxv);
        \draw [-to, thick](\posiS-3,1-\boxv) -- (\posiB-0.5,-2);
        \draw [-to,  thick](\posiB+0.5,-2) -- (\posiS+1,1-\boxv);
        \draw [-to, thick](\posiS+3,1-\boxv) -- (\posiC-0.5,-2);
        \end{tikzpicture}
        \end{array}
\end{equation}
We group all possible terms in boxes labeled by ``A'', ``B'', ``C'', ``D'', ``E'' and ``S''. The terms in box B come from two neighboring actions of the form $\mathcal{E}^{(\Eone)}_i\circ\mathcal{E}^{(\Eone)}_j$, where ``neighboring'' means the support of $\mathcal{E}^{(\Eone)}_i$ and $\mathcal{E}^{(\Eone)}_j$ overlap; namely, $j$ can be $i-2,i-1,i+1,i+2$. 
The terms in box C come from three neighboring actions of the form $\mathcal{E}^{(\Eone)}_i\circ\mathcal{E}^{(\Eone)}_j\circ\mathcal{E}_k^{(\Eone)}$, where ``neighboring'' means the support of $\mathcal{E}^{(\Eone)}_i$ and $\mathcal{E}^{(\Eone)}_j$ overlap, and the support of $\mathcal{E}^{(\Eone)}_j$ and $\mathcal{E}^{(\Eone)}_k$ overlap. For example, when considering a sufficiently large $N$, say $N > 8$, the term $\mathcal{E}_1^{(\Eone)}\circ\mathcal{E}_2^{(\Eone)}\circ \mathcal{E}_3^{(\Eone)}(\rho)$ belongs to box C, while $\mathcal{E}_1^{(\Eone)}\circ\mathcal{E}_2^{(\Eone)}\circ \mathcal{E}_6^{(\Eone)}(\rho)$ belongs to box B.
Similarly for boxes ``D'' and ``E''. Finally, the terms in box ``S'' come from an action of $\mathcal{E}^{(\Etwo)}$, taking the form of $\mathcal{E}_i^{(\Etwo)}(\rho)=(\rho_0)_{i,i+1,i+2} \otimes O_{i+3,i+4}$.  $\mathcal{E}^{(\Etwo)}$ is a completely-positive linear map and $\rho$ is a density matrix, leading to $O_{i+3,i+4}$ being positive semi-definite. 

We justify the following facts:
\begin{enumerate}
\item \textbf{Either arrow $\text{A}\Rightarrow\text{B}$ or arrow $\text{A}\Rightarrow\text{S}$ must exist.} Without loss of generality, let's start from term $\mathcal{E}_1^{(\Eone)}(\rho)$. When acting $\mathcal{E}_2$, $\mathcal{E}_2^{(\Eone)}$ will bring it to term $\mathcal{E}_1^{(\Eone)}\circ\mathcal{E}_2^{(\Eone)}$ in box B, and $\mathcal{E}_2^{(\Etwo)}$ will bring it to term $(\rho_0)_{234} \otimes O_{51}$ in box S. Since $\mathcal{E}_2$ is a CPTP map, at least one of these two terms is non-zero.  Therefore, an arrow must exist for either box A to box B or box A to box S. 
    \item \textbf{The arrow $\text{S}\Rightarrow\text{B}$ must exist. } Without loss of generality, let's start from term $((\rho_0)_{123}\otimes O_{45})$, and decompose $O_{45}=\sum_{\hbnt \hbnt'} (P_\hbnt\otimes\bo) O_{45} (P_{\hbnt'}\otimes\bo)$. Since $O_{45}\gneq 0$, for each $\hbnt$, $(P_\hbnt\otimes\bo) O_{45} (P_\hbnt\otimes\bo)\geq 0$ and $\tr(O_{45})=\sum_\hbnt \tr((P_\hbnt\otimes\bo) O_{45} (P_\hbnt\otimes\bo))>0$. Therefore, there exists at least one $\hbnt$ such that $(P_\hbnt\otimes\bo) O_{45} (P_\hbnt\otimes\bo)\gneq 0$, and for this $a$ we have a lower bound for the trace $\tr((P_\hbnt\otimes\bo) O_{45} (P_\hbnt\otimes\bo))\geq\frac{1}{g}\tr(O_{45})$, where $g$ is the number of elements in the BNT. For this $\hbnt$ and with explicit choice of $\rho_0$ in \cref{eqn:simplerho0}, $(P_\hbnt\otimes P_\hbnt\otimes P_\hbnt\otimes P_\hbnt\otimes \bo)((\rho_0)_{123}\otimes O_{45})(P_\hbnt\otimes P_\hbnt\otimes P_\hbnt\otimes P_\hbnt\otimes \bo)\gneq 0$, leading to $(\bo\otimes P_\hbnt\otimes P_\hbnt\otimes P_\hbnt\otimes \bo)((\rho_0)_{123}\otimes O_{45})(\bo\otimes P_\hbnt\otimes P_\hbnt\otimes P_\hbnt\otimes \bo)\gneq 0$, and subsequently,
\[
\mathcal{E}^{(\Eone)}_2((\rho_0)_{123}\otimes O_{45}) \gneq 0.
\]
The action $\frac{1}{N}\mathcal{E}_2^{(\Eone)}$ will bring the term $(\rho_0)_{123}\otimes O_{45}$ from box S to box B. 
Therefore, the arrows from box S to box B must exist. 

Furthermore, by noting that $\tr(P_a\otimes P_a\otimes P_a (\rho_0) P_a\otimes P_a\otimes P_a )=\frac{1}{g}\tr(\rho_0)$ by construction, we can put a lower bound for the trace,
\[
\tr(\mathcal{E}^{(\Eone)}_2((\rho_0)_{123}\otimes O_{45}))\geq \frac{1}{g^2} \tr((\rho_0)_{123}\otimes O_{45})),
\]
where the lower bound $1/g^2$ is a constant.

\item \textbf{Either arrow $\text{B}\Rightarrow\text{C}$ or arrows $\text{B}\Rightarrow\text{S}\Rightarrow\text{C}$ must exist.} Without loss of generality, let's start from term $\mathcal{E}_1^{(\Eone)}\circ\mathcal{E}_2^{(\Eone)}(\rho)$ in box C and consider the action of $\frac{1}{N}\mathcal{E}_3$. If $\mathcal{E}_3^{(\Eone)}\circ\mathcal{E}_1^{(\Eone)}\circ\mathcal{E}_2^{(\Eone)}(\rho)\gneq 0$, then $\frac{1}{N}\mathcal{E}_3$ will bring $\mathcal{E}_1^{(\Eone)}\circ\mathcal{E}_2^{(\Eone)}(\rho)$ from box B to box C. 

If, however, that $\mathcal{E}_3^{(\Eone)}\circ\mathcal{E}_1^{(\Eone)}\circ\mathcal{E}_2^{(\Eone)}(\rho)=0$, we then consider $\mathcal{E}_3^{(\Etwo)}\circ(\mathcal{E}_1^{(\Eone)}\circ\mathcal{E}_2^{(\Eone)}(\rho))$, which is term $(\rho_0)_{345}\otimes O_{12}$ in box B. Due to the structure of $\mathcal{E}_1^{(\Eone)}\circ\mathcal{E}_2^{(\Eone)}$, $\mathcal{E}_1^{(\Eone)}\circ\mathcal{E}_2^{(\Eone)}(\rho)$ satisfies
\[
\sum_\hbnt (P_\hbnt\otimes P_\hbnt\otimes \bo\otimes\bo\otimes \bo)(\mathcal{E}_1^{(\Eone)}\circ\mathcal{E}_2^{(\Eone)}(\rho)) (P_\hbnt\otimes P_\hbnt\otimes \bo\otimes\bo\otimes \bo)=\mathcal{E}_1^{(\Eone)}\circ\mathcal{E}_2^{(\Eone)}(\rho),
\]
and therefore by acting  $\mathcal{E}_3^{(\Etwo)}$ which does not act on site 1 and 2, $\mathcal{E}_3^{(\Etwo)}\circ(\mathcal{E}_1^{(\Eone)}\circ\mathcal{E}_2^{(\Eone)}(\rho))$ can be written as a sum of direct product,
\[
\mathcal{E}_3^{(\Etwo)}\circ(\mathcal{E}_1^{(\Eone)}\circ\mathcal{E}_2^{(\Eone)}(\rho))=\sum_\hbnt  (O_\hbnt)_{12}\otimes (\rho_0)_{345}, 
\]
where $O_\hbnt\geq 0$ and $(P_\hbnt\otimes P_\hbnt \otimes \bo\otimes \bo\otimes\bo)(O_\hbnt)_{12}(P_\hbnt\otimes P_\hbnt \otimes \bo\otimes \bo\otimes\bo)=(O_\hbnt)_{12}$. There exists at least one $\hbnt$ such that $O_\hbnt\gneq 0$, and for this $\hbnt$ we have a lower bound for the trace $\tr((P_\hbnt\otimes P_\hbnt \otimes \bo\otimes \bo\otimes\bo)(O_\hbnt\otimes\rho_0))\geq \frac{1}{g}\tr(\mathcal{E}_3^{(\Etwo)}\circ(\mathcal{E}_1^{(\Eone)}\circ\mathcal{E}_2^{(\Eone)}(\rho)))$. For this $\hbnt$ and with explicit choice of $\rho_0$ in \cref{eqn:simplerho0}, $(P_\hbnt\otimes P_\hbnt\otimes P_\hbnt\otimes P_\hbnt\otimes P_\hbnt)(\mathcal{E}_3^{(\Etwo)}\circ(\mathcal{E}_1^{(\Eone)}\circ\mathcal{E}_2^{(\Eone)}(\rho)))(P_\hbnt\otimes P_\hbnt\otimes P_\hbnt\otimes P_\hbnt\otimes P_\hbnt)\gneq 0$, therefore by acting on another $\mathcal{E}_1^{(\Eone)}\circ\mathcal{E}_2^{(\Eone)}$,
\[
\mathcal{E}_1^{(\Eone)}\circ\mathcal{E}_2^{(\Eone)}\circ \mathcal{E}_3^{(\Etwo)}\circ(\mathcal{E}_1^{(\Eone)}\circ\mathcal{E}_2^{(\Eone)}(\rho))\gneq 0.
\]
Using \cref{eqn:e12}, we see this term is in box C. Therefore, either arrow from box B to box C or from box B to box C via box S must exist. Furthermore, we have a lower bound for the trace
\[
\tr(\mathcal{E}_1^{(\Eone)}\circ\mathcal{E}_2^{(\Eone)}\circ \mathcal{E}_3^{(\Etwo)}\circ(\mathcal{E}_1^{(\Eone)}\circ\mathcal{E}_2^{(\Eone)}(\rho)))\geq \frac{1}{g^2} \tr(\mathcal{E}_3^{(\Etwo)}\circ(\mathcal{E}_1^{(\Eone)}\circ\mathcal{E}_2^{(\Eone)}(\rho))),
\]
where the lower bound $1/g^2$ is a constant.

\item \textbf{There are no arrows from boxes C, D, E to box S. } At $N=5$, the terms in box C are already in the desired space. Specifically, taking $\mathcal{E}_1^{(\Eone)}\circ\mathcal{E}_2^{(\Eone)}\circ\mathcal{E}_3^{(\Eone)}(\rho)$ as an example, it already satisfies 
\[
    \sum_\hbnt (P_\hbnt\otimes P_\hbnt\otimes P_\hbnt\otimes P_\hbnt\otimes P_\hbnt)(\mathcal{E}_1^{(\Eone)}\circ\mathcal{E}_2^{(\Eone)}\circ\mathcal{E}_3^{(\Eone)}(\rho))(P_\hbnt\otimes P_\hbnt\otimes P_\hbnt\otimes P_\hbnt\otimes P_\hbnt) = \mathcal{E}_1^{(\Eone)}\circ\mathcal{E}_2^{(\Eone)}\circ\mathcal{E}_3^{(\Eone)}(\rho)
    \]
    because the support of $\mathcal{E}_1^{(\Eone)}\circ\mathcal{E}_2^{(\Eone)}\circ\mathcal{E}_3^{(\Eone)}$ covers all the five sites. When acting $\mathcal{E}_4$ only $\mathcal{E}_4^{(\Eone)}$ leads to non-zero result and brings it to term $\mathcal{E}_1^{(\Eone)}\circ\mathcal{E}_2^{(\Eone)}\circ\mathcal{E}_3^{(\Eone)}\circ\mathcal{E}_4^{(\Eone)}(\rho)$ in box D, while acting $\mathcal{E}_4^{(\Etwo)}$ gives zero. Similarly, the terms in box D are already in the desired space, and further actions of $\mathcal{E}_i$ can only bring the terms to themselves or the term in box E. Finally, acting $\mathcal{E}_i$ to terms in box E bring them to themselves. Therefore, no arrows are from boxes C, D, and E to box S.
\end{enumerate}
With the above arrows, it is clear that the trace of terms in box E will keep increasing as long as there exist terms in previous boxes, and subsequently, a steady state can only be supported in box E, taking the form of \cref{eqn:ssform}. Therefore, the parent Lindbladian is frustration-free.

\end{proof}

The above proof can be generalized to arbitrary $N$ by repeating the argument in point 3. For example, starting from $\mathcal{E}_1^{(\Eone)}\circ \mathcal{E}_2^{(\Eone)}\circ \mathcal{E}_3^{(\Eone)}(\rho)$, if $\mathcal{E}^{(\Eone)}_4\circ(\mathcal{E}_1^{(\Eone)}\circ \mathcal{E}_2^{(\Eone)}\circ \mathcal{E}_3^{(\Eone)}(\rho))=0$, then one can similarly show
\[
\mathcal{E}_1^{(\Eone)}\circ \mathcal{E}_2^{(\Eone)}\circ \mathcal{E}_3^{(\Eone)}\circ \mathcal{E}_4^{(\Etwo)}\circ (\mathcal{E}_1^{(\Eone)}\circ \mathcal{E}_2^{(\Eone)}\circ \mathcal{E}_3^{(\Eone)}(\rho))\gneq 0,
\]
justifying either arrow $\text{C}\Rightarrow\text{D}$ or arrows $\text{C}\Rightarrow\text{S}\Rightarrow\text{D}$ must exist.

\subsection{MPDO from \texorpdfstring{$\boldsymbol{C^*}$}{C*}-weak Hopf algebras}
In this section, we prove the frustration-freeness of parent Lindbladian for non-simple MPDO RFPs generated by $C^*$-weak Hopf algebras. We make a particular choice for $\mathcal{E}^{(\Etwo)}$ as in \cref{eqn:choiceE2WHA}, which takes the explicit form of
\begin{equation}
    \mathcal{E}^{(\Etwo)}(X)
    :=\tr(P^\perp X)\rho_0, \quad
    \rho_0 
    := \frac{1}{\omega(\Omega)}b(\omega)^{\otimes 2}(\Phi^{\otimes 2}\circ\Delta)(\Omega),
\end{equation}
where $P^\perp:=\bo-(\Phi^{\otimes 2}\circ\Delta)(1)$ and $\rho_0$ satisfies \cref{eqn:rho0eig}. 

\subsubsection{Open boundary condition}

Let us first consider system size $N=4$ and open boundary conditions as an explicit example, where under open boundary condition $\mathcal{E}^{(N)}$ is defined as $\mathcal{E}^{(N)}=\frac{1}{N-1}\sum_{i=1}^{N-1}\mathcal{E}_i$.

\begin{prop}
    The parent Lindbladian for non-simple MPDO RFPs generated by $ C*$-Weak Hopf algebras under the open boundary condition with a system size of $N=4$ is frustration-free. 
\end{prop}

\begin{proof}
    We first decompose the fixed-point $\rho$ as in \cref{eqn:rho-decomp}. Under an action of $\mathcal{E}^{(N)}$, the decomposition terms of $\rho$ flow to each other as
    \begin{equation}
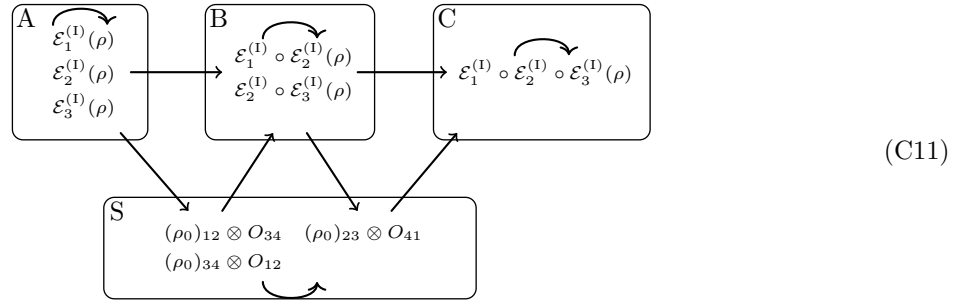
    \begin{array}{c}
    \begin{tikzpicture}[scale=.45,baseline={([yshift=-0.75ex] current bounding box.center)}
            ]
        \def\boxh{6.2};
        \def\boxv{5.2};
        \def\eqnh{1};
        \def\eqnht{0.8};
        \def\sf{0.};
        \def\posiB{\boxh};
        \def\posiC{2.2*\boxh};
        \def\posiS{\boxh};
        \EquationBox{(0,0)}{2}{2.}{boxc1}{
        $\begin{aligned}
            &\mathcal{E}_1^{(\Eone)}(\rho)\\
            &\mathcal{E}_2^{(\Eone)}(\rho)\\
            &\mathcal{E}_3^{(\Eone)}(\rho)
        \end{aligned}$
        }{A}
        \draw [->,  thick](-0.8,\eqnh*1.4) to [out=90,in=90] (0.8,\eqnh*1.4) ;
        \EquationBox{(\posiS,-\boxv)}{5.5}{1.5}{lcolor}{
        $\begin{aligned}
            &(\rho_0)_{12}\otimes O_{34}\quad (\rho_0)_{23}\otimes O_{41}\\
            &(\rho_0)_{34}\otimes O_{12}
        \end{aligned}$
        }{S}
        \draw [->, thick](-0.8+\posiS,-\eqnh-\boxv) to [out=-90,in=-90] (0.8+\posiS,-\eqnh-\boxv);
        \EquationBox{(\posiB,0)}{2.5}{2.}{boxc2}{
        $\begin{aligned}
            &\mathcal{E}_1^{(\Eone)}\circ \mathcal{E}_2^{(\Eone)}(\rho)\\
            &\mathcal{E}_2^{(\Eone)}\circ \mathcal{E}_3^{(\Eone)}(\rho)
        \end{aligned}$
        }{B}
        \draw [->, thick](-0.8+\posiB,\eqnh*1.) to [out=90,in=90] (0.8+\posiB,\eqnh*1.) ;
        \EquationBox{(2.2*\boxh,0.)}{3.2}{2.}{boxc3}{
        $\begin{aligned}
            \mathcal{E}_1^{(\Eone)}\circ \mathcal{E}_2^{(\Eone)}\circ\mathcal{E}_3^{(\Eone)}(\rho)
        \end{aligned}$
        }{C}
        \draw [->,  thick](-0.8+\posiC,\eqnh*0.5) to [out=90,in=90] (0.8+\posiC,\eqnh*0.5);
        \draw [-to, thick](1.5,0) -- (\posiB-2,0);
        \draw [-to,  thick](2+\posiB,0) -- (\posiC-2.8,0);
        \draw [-to,  thick](1.2,-1.8) -- (\posiS-3,1.1-\boxv);
        \draw [-to, thick](\posiS-2,1.1-\boxv) -- (\posiB-0.5,-1.8);
        \draw [-to, thick](\posiB+0.5,-1.8) -- (\posiS+2,1.1-\boxv);
        \draw [-to, thick](\posiS+3,1.1-\boxv) -- (\posiC-2.5,-1.8);
    \end{tikzpicture}
    \end{array}
    \end{equation}
    We group all possible terms in five boxes labeled by ``A'', ``B'', ``C'', ``S''. The terms in box B come from two neighboring actions of the form $\mathcal{E}^{(\Eone)}_i\circ\mathcal{E}^{(\Eone)}_j$, where ``neighboring'' means the support of $\mathcal{E}^{(\Eone)}_i$ and $\mathcal{E}^{(\Eone)}_j$ overlap; namely, $j$ can be $i-1,i+1$. The terms in box C come from three neighboring actions of $\mathcal{E}^{(\Eone)}$. Finally, the terms in box ``S'' come from an action of $\mathcal{E}^{(\Etwo)}$ on a density matrix and therefore $O_{i,i+1}\geq 0$. 

    We justify the following facts.
    \begin{enumerate}
        \item \textbf{Either arrow $\text{A}\Rightarrow\text{B}$ or arrow $\text{A}\Rightarrow\text{S}$ must exist.} The argument is the same as for non-injective simple MPDOs.
        \item \textbf{The arrow $\text{S}\Rightarrow\text{B}$ must exist. } Without loss of generality, let's start from the term $(\rho_0)_{12}\otimes O_{34}$ and act $\mathcal{E}_2^{(\Eone)}$, and show
        \begin{equation}
            \label{eq:SubtletyWHA}
            \mathcal{E}_2^{(\Eone)}\circ((\rho_0)_{12}\otimes O_{34})\gneq 0.
        \end{equation}
This expression can be simplified to 
    \[
\mathcal{E}_2^{(\Eone)}\circ((\rho_0)_{12}\otimes O_{34}) \equiv
\begin{tikzpicture}[scale=1]
    \draw[Virtual] (2.223, 1.976) arc[start angle=90, end angle=180, x radius=0.494, y radius=-0.494];
    \draw[Virtual] (2.469, 2.963) -- (2.223, 2.963) arc[start angle=90, end angle=151.044, radius=0.494];
    \draw[Virtual] (2.858, 1.976) -- (2.469, 1.976);
    \draw[Virtual] (4.198, 1.976) -- (3.351, 1.976);
    \draw[Virtual] (4.198, 1.976) arc[start angle=-90, end angle=90, radius=0.494] -- (3.951, 2.963);
    \draw[-mid] (3.704, 1.235) -- (3.704, 1.729);
    \draw[bevel, -mid] (3.714, 2.222) arc[start angle=-163.45, end angle=0, x radius=0.247, y radius=-0.247] -- (4.198, 1.235);
    \filldraw[ultra thin, fill=white] 
    (1.729, 2.469) circle[radius=0.247];
    \node[anchor=center] at (1.729, 2.469) {$\Omega$};
    \filldraw[ultra thin, fill=white] 
    (3.104, 1.976) circle[radius=0.247];
    \node[anchor=center] at (3.104, 1.976) {$\Xi$};
    \draw[Virtual] (0.502, 0.926) arc[start angle=118.956, end angle=151.044, radius=0.494];
    \draw[Virtual, -mid] (1.976, 0.988) -- (0.988, 0.988);
    \draw[-mid] (0.741, 1.235) -- (0.741, 1.482);
    \draw[bevel, mid-] (0.741, 0.741) arc[start angle=180, end angle=360, radius=0.247] -- (1.235, 1.482);
    \draw[Virtual] (0.741, 0) arc[start angle=90, end angle=180, x radius=0.494, y radius=-0.494];
    \draw[Virtual] (2.716, 0) -- (0.741, 0);
    \draw[Virtual] (2.716, 0) arc[start angle=-90, end angle=90, radius=0.494] -- (1.729, 0.988);
    \filldraw[ultra thin, fill=white] 
    (0.247, 0.494) circle[radius=0.247];
    \node[anchor=center] at (0.247, 0.494) {$\Omega$};
    \draw[-mid] (5.186, 1.235) -- (5.186, 1.729);
    \draw[mid-] (5.68, 1.235) -- (5.68, 1.729);
    \filldraw[ultra thin, fill=black!2] (3.704, 1.235) arc[start angle=90, end angle=270, radius=0.247] -- (5.68, 0.741) arc[start angle=-90, end angle=90, radius=0.247] -- cycle;
    \node[anchor=center] at (4.692, 0.988) {$O_{34}$};
    \node[anchor=center] at (0.741, 0.988) {$M$};
    \filldraw[ultra thin, fill=Goldenrod] 
    (0.741, 0.988) circle[radius=0.247];
    \node[anchor=center] at (0.741, 0.988) {$M$};
    \draw[bevel, -mid] (2.233, 2.222) arc[start angle=-163.45, end angle=0, x radius=0.247, y radius=-0.247] -- (2.716, 1.411) -- (2.716, 0.741) arc[start angle=-0.885, end angle=180.885, x radius=0.247, y radius=-0.247];
    \draw[-mid] (2.223, 1.235) -- (2.223, 1.729);
    \node[anchor=center] at (2.223, 0.988) {$M$};
    \filldraw[ultra thin, fill=Goldenrod] 
    (2.223, 0.988) circle[radius=0.247];
    \node[anchor=center] at (2.223, 0.988) {$M$};
    \draw[Virtual, -mid] (3.457, 2.963) -- (2.469, 2.963);
    \draw[-mid] (3.704, 3.21) -- (3.704, 3.457);
    \draw[bevel, mid-] (3.714, 2.717) arc[start angle=-163.45, end angle=0, radius=0.247] -- (4.198, 3.457);
    \filldraw[ultra thin, fill=Goldenrod] 
    (3.704, 2.963) circle[radius=0.247];
    \node[anchor=center] at (3.704, 2.963) {$M$};
    \draw[-mid] (2.223, 3.21) -- (2.223, 3.457);
    \draw[bevel, mid-] (2.233, 2.717) arc[start angle=-163.45, end angle=0, radius=0.247] -- (2.716, 3.457);
    \filldraw[ultra thin, fill=Goldenrod] 
    (2.223, 2.963) circle[radius=0.247];
    \node[anchor=center] at (2.223, 2.963) {$M$};
     \whaNsymb{(3.704, 1.976)};
    \whaNsymb{(2.223, 1.976)};
\end{tikzpicture}
=
\begin{tikzpicture}[scale=1]
    \draw[Virtual] (2.858, 1.976) -- (2.611, 1.976);
    \draw[Virtual] (4.198, 1.976) -- (3.351, 1.976);
    \draw[Virtual] (4.198, 1.976) arc[start angle=-90, end angle=90, radius=0.494] -- (3.21, 2.963);
    \draw[-mid] (3.704, 1.235) -- (3.704, 1.729);
    \filldraw[ultra thin, fill=white] 
    (3.104, 1.976) circle[radius=0.247];
    \node[anchor=center] at (3.104, 1.976) {$\Xi$};
    \draw[Virtual] (0.502, 0.926) arc[start angle=118.956, end angle=151.044, radius=0.494];
    \draw[Virtual] (0.741, 0) arc[start angle=90, end angle=180, x radius=0.494, y radius=-0.494];
    \draw[Virtual] (2.716, 0) -- (0.741, 0);
    \draw[Virtual] (2.716, 0) arc[start angle=-90, end angle=90, radius=0.494] -- (2.716, 0.988);
    \filldraw[ultra thin, fill=white] 
    (0.247, 0.494) circle[radius=0.247];
    \node[anchor=center] at (0.247, 0.494) {$\Omega$};
    \draw[-mid] (5.186, 1.235) -- (5.186, 1.729);
    \draw[mid-] (5.68, 1.235) -- (5.68, 1.729);
    \draw[Virtual, -mid] (2.611, 0.988) arc[start angle=90, end angle=270, x radius=0.494, y radius=-0.494];
    \draw[Virtual] (2.716, 0.988) -- (2.611, 0.988);
    \draw[Virtual, -mid] (1.976, 2.963) .. controls (1.129, 2.716) and (1.87, 1.235) .. (0.988, 0.988);
    \draw[-mid] (0.741, 1.235) -- (0.741, 1.482);
    \draw[bevel, mid-] (0.741, 0.741) arc[start angle=180, end angle=360, radius=0.247] -- (1.235, 1.482);
    \node[anchor=center] at (0.741, 0.988) {$M$};
    \filldraw[ultra thin, fill=Goldenrod] 
    (0.741, 0.988) circle[radius=0.247];
    \node[anchor=center] at (0.741, 0.988) {$M$};
    \draw[bevel, -mid] (3.714, 2.222) arc[start angle=-163.45, end angle=0, x radius=0.247, y radius=-0.247] -- (4.198, 1.235);
    \filldraw[ultra thin, fill=black!2] (3.704, 1.235) arc[start angle=90, end angle=270, radius=0.247] -- (5.68, 0.741) arc[start angle=-90, end angle=90, radius=0.247] -- cycle;
    \node[anchor=center] at (4.692, 0.988) {$O_{34}$};
    \draw[Virtual, -mid] (3.457, 2.963) -- (2.469, 2.963);
    \draw[-mid] (3.704, 3.21) -- (3.704, 3.457);
    \draw[bevel, mid-] (3.714, 2.717) arc[start angle=-163.45, end angle=0, radius=0.247] -- (4.198, 3.457);
    \filldraw[ultra thin, fill=Goldenrod] 
    (3.704, 2.963) circle[radius=0.247];
    \node[anchor=center] at (3.704, 2.963) {$M$};
    \draw[-mid] (2.223, 3.21) -- (2.223, 3.457);
    \draw[bevel, mid-] (2.233, 2.717) arc[start angle=-163.45, end angle=0, radius=0.247] -- (2.716, 3.457);
    \filldraw[ultra thin, fill=Goldenrod] 
    (2.223, 2.963) circle[radius=0.247];
    \node[anchor=center] at (2.223, 2.963) {$M$};
    \whaNsymb{(3.704, 1.976)};
\end{tikzpicture}
    \]
by virtue of \cref{eq:BlackVsWhite}. 

Now, let us recall the definition of the tensors $M'$ and $\Xi$ in \cref{eq:ExplicitDefMInv,eq:ExplicitDefXi}. One can also prove that the element $\xi\in A$ can be factorized in the form $\xi = \xi_L\xi_R$ for some positive invertible elements $\xi_L,\xi_R\in A$, which in particular have the properties of $\xi_L x=\hat{\xi}_L(x_{(1)}) x_{(2)}, x \xi_R=\hat{\xi}_R (x_{(1)}) x_{(2)}$; we refer the reader to Lemma 3.5 in \cite{ruizdealarcon2024matrix} for a proof. Using these properties, 
\begin{equation}
\label{eqn:action-right}
\begin{tikzpicture}[scale=1]
    \draw[Virtual,mid-] (0.988, 0.988) -- (0, 0.988);
    \draw[Virtual,-mid] (0, 0.988) -- (0.247, 0.988);
    \draw[Virtual] (1.482, 0.988) -- (1.729, 0.988) -- (2.117, 0.988);
    \draw[bevel, -mid] (1.235, 1.235) -- (1.235, 1.235) arc[start angle=180, end angle=360, x radius=0.247, y radius=-0.247] -- (1.729, 0);
    \draw[-mid] (1.235, 0) -- (1.235, 0.741);
    \whaNsymb{(1.235, 0.988)};
    \filldraw[ultra thin, fill=white] 
    (0.494, 0.988) circle[radius=0.247];
    \node[anchor=center] at (0.494, 0.988) {$\Xi$};
\end{tikzpicture}
=
\begin{tikzpicture}[scale=1]
    \draw[Virtual] (0.512, 0.988) -- (0, 0.988);
    \draw[-mid] (1.235, 0) -- (1.235, 0.864);
    \draw[Virtual] (0.758, 0.988) -- (1.729, 0.988) -- (2.117, 0.988);
    \filldraw[ultra thin, fill=LimeGreen!50] 
    (0.635, 0.988) circle[radius=0.123];
    \filldraw[ultra thin, fill=GreenYellow!50] 
    (1.235, 0.432) circle[radius=0.123];
    \node[anchor=center, font=\footnotesize] at (2.469, 0.432) {$\Phi(\xi)^{1/2}$};
    \node[anchor=center, font=\footnotesize] at (0.635, 1.376) {$\Psi(\hat\xi)^{-1}$};
    \node[anchor=center, font=\footnotesize] at (0.529, 0.432) {$\Phi(\xi)^{1/2}$};
    \draw[bevel, -mid] (1.235, 1.235) -- (1.235, 1.235) arc[start angle=-180, end angle=0, x radius=0.247, y radius=-0.247] -- (1.729, 0);
    \draw (1.235, 1.111) -- (1.235, 1.235);
    \filldraw[ultra thin, fill=white] 
    (1.235, 0.988) circle[radius=0.123];
    \filldraw[ultra thin, fill=GreenYellow!50] 
    (1.729, 0.432) circle[radius=0.123];
\end{tikzpicture}
=
\begin{tikzpicture}[scale=1]
    \draw[-mid] (1.447, 0) -- (1.447, 0.864);
    \filldraw[ultra thin, fill=GreenYellow!50] 
    (1.447, 0.432) circle[radius=0.123];
    \node[anchor=center, font=\footnotesize] at (2.823, 0.432) {$\Phi(g)^{-\frac{1}{2}}$};
    \node[anchor=center, font=\footnotesize] at (0.6, 0.432) {$\Phi(g)^{-\frac{1}{2}}$};
    \draw[Virtual] (0.97, 0.988) -- (1.941, 0.988) -- (2.329, 0.988);
    \draw[bevel, -mid] (1.447, 1.235) -- (1.447, 1.235) arc[start angle=-180, end angle=0, x radius=0.247, y radius=-0.247] -- (1.941, 0);
    \draw (1.447, 1.111) -- (1.447, 1.235);
    \filldraw[ultra thin, fill=GreenYellow!50] 
    (1.941, 0.432) circle[radius=0.123];
    \filldraw[ultra thin, fill=white] 
    (1.447, 0.988) circle[radius=0.123];
\end{tikzpicture}
=
\begin{tikzpicture}[scale=1]
    \draw[Virtual] (1.694, 0.988) -- (1.941, 0.988) -- (2.329, 0.988);
    \draw (1.447, 0) -- (1.447, 0.741);
    \draw[bevel, -mid] (1.447, 1.235) -- (1.447, 1.235) arc[start angle=-180, end angle=0, x radius=0.247, y radius=-0.247] -- (1.941, 0);
    \whaNsymb{(1.447, 0.988)};
    \draw[Virtual, -mid] (0.565, 0.988) -- (0.812, 0.988) -- (1.2, 0.988);
    \filldraw[ultra thin, fill=GreenYellow!50] 
    (1.447, 0.37) circle[radius=0.123];
    \filldraw[ultra thin, fill=GreenYellow!50] 
    (1.941, 0.37) circle[radius=0.123];
    \node[anchor=center, font=\footnotesize] at (2.823, 0.37) {$\Phi(\xi_L)^{-1}$};
    \node[anchor=center, font=\footnotesize] at (0.6, 0.37) {$\Phi(\xi_L)^{-1}$};
\end{tikzpicture}
\end{equation}
where $g:=\xi_L \xi_R^{-1}$, 
i.e.~the previous expression takes the form
\[
\mathcal{E}_2^{(\Eone)}\circ((\rho_0)_{12}\otimes O_{34})
=
\begin{tikzpicture}[scale=1]
    \draw[-mid] (3.21, 0.494) -- (3.21, 0.988);
    \draw[mid-] (3.704, 0.494) -- (3.704, 0.988);
    \draw[Virtual] (0.502, 2.161) arc[start angle=118.956, end angle=151.044, radius=0.494];
    \draw[Virtual, -mid] (1.976, 2.223) -- (0.988, 2.223);
    \draw[Virtual] (0.741, 1.235) arc[start angle=90, end angle=180, x radius=0.494, y radius=-0.494];
    \draw[Virtual] (2.716, 1.235) -- (0.741, 1.235);
    \draw[Virtual] (2.716, 1.235) arc[start angle=-90, end angle=90, radius=0.494] -- (1.729, 2.223);
    \filldraw[ultra thin, fill=white] 
    (0.247, 1.729) circle[radius=0.247];
    \node[anchor=center] at (0.247, 1.729) {$\Omega$};
    \draw[bevel, -mid] (1.739, 1.481) arc[start angle=-163.45, end angle=0, x radius=0.247, y radius=-0.247] -- (2.223, 0.494);
    \draw[-mid] (0.741, 2.469) -- (0.741, 2.716);
    \draw[bevel, mid-] (0.751, 1.976) arc[start angle=-163.45, end angle=0, radius=0.247] -- (1.235, 2.716);
    \filldraw[ultra thin, fill=Goldenrod] 
    (0.741, 2.223) circle[radius=0.247];
    \node[anchor=center] at (0.741, 2.223) {$M$};
    \draw[-mid] (1.606, 2.469) -- (1.606, 2.716);
    \draw[bevel, mid-] (1.616, 1.976) arc[start angle=-163.45, end angle=0, radius=0.247] -- (2.1, 2.716);
    \filldraw[ultra thin, fill=Goldenrod] 
    (1.606, 2.223) circle[radius=0.247];
    \node[anchor=center] at (1.606, 2.223) {$M$};
    \draw[-mid] (2.469, 2.469) -- (2.469, 2.716);
    \draw[bevel, mid-] (2.48, 1.976) arc[start angle=-163.45, end angle=0, radius=0.247] -- (2.963, 2.716);
    \filldraw[ultra thin, fill=Goldenrod] 
    (2.469, 2.223) circle[radius=0.247];
    \node[anchor=center] at (2.469, 2.223) {$M$};
    \draw (1.729, 0.494) -- (1.729, 0.988);
    \filldraw[ultra thin, fill=GreenYellow!50] 
    (1.729, 0.741) circle[radius=0.123];
    \filldraw[ultra thin, fill=GreenYellow!50] 
    (2.223, 0.741) circle[radius=0.123];
    \filldraw[ultra thin, fill=black!2] (1.729, 0.494) arc[start angle=90, end angle=270, radius=0.247] -- (3.704, 0) arc[start angle=-90, end angle=90, radius=0.247] -- cycle;
    \node[anchor=center] at (2.716, 0.247) {$O_{34}$};
    \whaNsymb{(1.729, 1.235)};
\end{tikzpicture}
=((\mathcal{T}\otimes\mathrm{Id})\circ \mathcal{T})( (\Phi(\xi_L^{-1})\otimes\mathrm{Id}) O_{34} (\Phi(\xi_L^{-1})^\dagger\otimes\mathrm{Id})),
\]
and since $\mathcal{T}$ is completely positive, we conclude that the term in \cref{eq:SubtletyWHA} is positive. Moreover, it is non-zero since its trace is of the form $\operatorname{Tr}(\Phi(\xi_L)^{-2}\otimes\mathrm{Id})O_{34})$ by the previous equation and the fact that $\mathcal{T}$ is trace-preserving, and $\xi_L$ is positive definite, leading to that $\operatorname{Tr}(\Phi(\xi_L)^{-2}\otimes\mathrm{Id})O_{34})>0$.  Therefore, the arrow from box S to box B must exist.

Furthermore, using the von Neumann trace inequality, the trace is lower bounded by $\tr(\Phi(\xi_L)^{-2}\otimes\mathrm{Id})O_{34})\geq r\tr(O_{34})$ where $r=\lambda_{\min} (\Phi(\xi_L)^{-2})$ is the minimal eigenvalue of $\Phi(\xi_L)^{-2}$, which is a constant given a $C^*$-weak Hopf algebra. Since $\mT$ is trace-preserving, we can put a lower bound for the trace
\[
\tr(\mathcal{E}_2^{(\Eone)}\circ((\rho_0)_{12}\otimes O_{34}))\geq r \tr((\rho_0)_{12}\otimes O_{34}),
\]
where $r$ is a constant. 

\item \textbf{Either arrow $\text{B}\Rightarrow\text{C}$ or arrows $\text{B}\Rightarrow\text{S}\Rightarrow\text{C}$ must exist.} Without loss of generality, let's start with the term $\mathcal{E}_2^{(\Eone)}\circ\mathcal{E}_3^{(\Eone)}(\rho)$ in box C and consider the action of $\mathcal{E}_1$. If $\mathcal{E}_1^{(\Eone)}\circ\mathcal{E}_2^{(\Eone)}\circ\mathcal{E}_3^{(\Eone)}(\rho)\gneq 0$, then $\frac{1}{N}\mathcal{E}_1$ will bring $\mathcal{E}_2^{(\Eone)}\circ\mathcal{E}_3^{(\Eone)}(\rho)$ from box B to box C. 

If, however, that $\mathcal{E}^{(\Eone)}_1\circ(\mathcal{E}_2^{(\Eone)}\circ\mathcal{E}_3^{(\Eone)}(\rho))=0$, we then need to consider $\mathcal{E}^{(\Etwo)}_1\circ(\mathcal{E}_2^{(\Eone)}\circ\mathcal{E}_3^{(\Eone)}(\rho))$, which is term $(\rho_0)_{12}\otimes O_{34}$ in box B. We need to show by acting another $\mathcal{E}_2^{(\Eone)}\circ\mathcal{E}_3^{(\Eone)}$, 
\[
\mathcal{E}_2^{(\Eone)}\circ\mathcal{E}_3^{(\Eone)}\circ (\mathcal{E}^{(\Etwo)}_1\circ\mathcal{E}_2^{(\Eone)}\circ\mathcal{E}_3^{(\Eone)}(\rho))\gneq 0,
\]
which, by noting that $[\mathcal{E}_3^{(\Eone)},\mathcal{E}_1^{(\Etwo)}]=0$ because these two linear maps do not overlap and using $[\mathcal{E}_i^{(\Eone)},\mathcal{E}_j^{(\Eone)}]=0$, $\mathcal{E}_i^{(\Eone)}\circ\mathcal{E}_i^{(\Eone)}=\mathcal{E}_i^{(\Eone)}$, amounts to
\begin{equation}
\mathcal{E}_2^{(\Eone)}\circ (\mathcal{E}^{(\Etwo)}_1\circ\mathcal{E}_2^{(\Eone)}\circ\mathcal{E}_3^{(\Eone)}(\rho))=\mathcal{E}_2^{(\Eone)}\circ((\rho_0)_{12}\otimes O_{34})\gneq 0.
\end{equation}
This is true based on the proof in point 3; using \cref{eqn:e12}, we see this term is in box C. 

When starting from $\mathcal{E}_1^{(\Eone)}\circ\mathcal{E}_2^{(\Eone)}(\rho)$, we would need to consider the action of $\mathcal{E}_3$ and finish the proof using a similar equation as \cref{eqn:action-right} with $\Xi$ on the right of $M'$. Therefore, either arrow from box B to box C or from box B to box C via box S must exist. 

\item \textbf{There are no arrows from box C to box S.} This is a direct result of \cref{eqn:e21vanish}. 
\end{enumerate}

With the above arrows, it is clear that the trace of terms in box C will keep increasing as long as there exist terms in previous boxes, and subsequently, a steady state can only be supported in box C, taking the form of $\rho=\mathcal{E}^{(\Eone)}_1\circ \mathcal{E}^{(\Eone)}_2\circ\cdots\circ\mathcal{E}^{(\Eone)}_{N-1}(\rho)$. Therefore, the parent Lindbladian is frustration-free.

\end{proof}

The above proof can be generalized to arbitrary $N$ with the open boundary condition by repeating the argument in point 3. 

\subsubsection{Periodic boundary condition}
Finally, for the periodic boundary condition, one more step is needed to close the boundary. Let's start from a term $\mathcal{E}^{(\Eone)}_1\circ \mathcal{E}^{(\Eone)}_2\circ\cdots\circ\mathcal{E}^{(\Eone)}_{N-1}(\rho)\gneq 0$ in the second to last box and act $\mathcal{E}_N$ on it. If $\mathcal{E}^{(\Eone)}_N\circ(\mathcal{E}^{(\Eone)}_1\circ \mathcal{E}^{(\Eone)}_2\circ\cdots\circ\mathcal{E}^{(\Eone)}_{N-1}(\rho))\gneq 0$, the resulting term is in the final box. If, however, that $\mathcal{E}^{(\Eone)}_N\circ(\mathcal{E}^{(\Eone)}_1\circ \mathcal{E}^{(\Eone)}_2\circ\cdots\circ\mathcal{E}^{(\Eone)}_{N-1}(\rho))=0$, we need to consider the term $\mathcal{E}^{(\Etwo)}_N\circ(\mathcal{E}^{(\Eone)}_1\circ \mathcal{E}^{(\Eone)}_2\circ\cdots\circ\mathcal{E}^{(\Eone)}_{N-1}(\rho))$. By acting $\mathcal{E}_{N-1}^{(\Eone)}\circ \mathcal{E}_{1}^{(\Eone)}$, one can show
\begin{equation}
    \mathcal{E}_{N-1}^{(\Eone)}\circ \mathcal{E}_{1}^{(\Eone)}\circ (\mathcal{E}^{(\Etwo)}_N\circ(\mathcal{E}^{(\Eone)}_1\circ \mathcal{E}^{(\Eone)}_2\circ\cdots\circ\mathcal{E}^{(\Eone)}_{N-1}(\rho)))\gneq 0,
\end{equation}
which is a term in the final box. In the following, we provide the proof. 

By assumption, the $C^*$-weak Hopf algebra is biconnected with a unique vacuum sector $e$. 
Denote $\tilde{\rho}=\mathcal{E}^{(\Eone)}_1\circ \mathcal{E}^{(\Eone)}_2\circ\cdots\circ\mathcal{E}^{(\Eone)}_{N-1}(\rho)$ and the corresponding boundary condition matrix as $B=b(x)$ where $x\in A$ is positive semi-definite. W.l.o,g., we choose $B=b(x)$ to be block diagonal in the same basis as $\Psi$, i.e., $B=\bigoplus_a B_a$. Recalling \cref{eqn:pbc-action} and note that $\Omega_a$ is proportional to identity, if $\mathcal{E}^{(\Eone)}_N (\tilde{\rho})=0$, then $\tr(B_a \Xi_a)=\tr(\Xi_a^{1/2} B_a \Xi_a^{1/2})=0$ for any $a$ (for example, in Fibonacci $C^*$-WHA, choosing $x=e_{2,11}$ would lead to $\tr(B_a \Xi_a)=0$ for all $a$). In this case we consider $\mathcal{E}_N^{(\Etwo)}(\tilde{\rho})$, which amounts to partially tracing $\tilde{\rho}$ and replacing the traced degrees of freedom by $\rho_0$, diagramatically. 
\[
\mathcal{E}_N^{(\Etwo)}(\tilde{\rho})=\frac{1}{\omega(\Omega)}
\begin{array}{c}
\begin{tikzpicture}[scale=1]
\wharhozeroloose{(-3,0)}{0.987}{2};
    \draw[Virtual, -mid] (2.223, 1.552) -- (1.729, 1.552);
    \draw[Virtual] (2.963, 1.552) -- (2.716, 1.552);
    \draw[Virtual, mid-] (4.198, 1.552) -- (4.691, 1.552);
    \node[anchor=center] at (3.212, 1.68) {$\overset{N}\cdots$};
    \draw[Virtual] (3.704, 1.552) -- (3.457, 1.552);
    \draw[Virtual] (0.247, 1.552) arc[start angle=90, end angle=270, radius=0.494] -- (5.432, 0.564) arc[start angle=-90, end angle=90, radius=0.494];
    \draw[Virtual] (1.482, 1.552) -- (1.244, 1.552);
    \draw[Virtual, -mid] (4.198, 1.552) -- (3.951, 1.552);
    \draw[Virtual] (5.185, 1.552) -- (5.432, 1.552);
     \whaMtracesgl{(1.482, 1.552)}{\small $M$};
     \whaMDouble{(2.469, 1.552)};
     \whaMDouble{(3.951, 1.552)};
     \whaMtracesgl{(4.938, 1.552)}{\small $M$};
    \draw[Virtual, -mid] (1.235, 1.552) -- (0.741, 1.552);
    \filldraw[ultra thin, fill=white] 
    (0.494, 1.552) circle[radius=0.247];
    \node[anchor=center] at (0.494, 1.552) {$B$};
\end{tikzpicture}
\,.
\end{array}
\]
The partial trace amounts to replacing the boundary condition matrix
$B\ra B'=E B E$
where $E$ is the transfer matrix. By construction, the transfer matrix $E=\Psi(\omega)$ is rank-one and is positive semi-definite. Therefore, $B'\propto E\gneq 0$ and in particular, $B'_e\propto E_e \gneq 0$.


We next act $\mathcal{E}_{N-1}^{(\Eone)}\circ\mathcal{E}_1^{(\Eone)}$. Diagrammatically, 
\[
\begin{aligned}
\mathcal{E}_{N-1}^{(\Eone)}\circ\mathcal{E}_1^{(\Eone)}\circ(\mathcal{E}_N^{(\Etwo)}(\tilde{\rho}))&=
\begin{array}{c}
\begin{tikzpicture}[scale=1]
\whachannelone{(0,0)}{2.0};
\whachannelone{(4.0,0)}{2.0};
\wharhozeroloose{(2.0,-1.976)}{2.0}{1};
\wharholeft{(6.0,-1.976)}{$B'$};
\wharhoright{(0,-1.976)};
\end{tikzpicture}
\end{array}\\
&=
\begin{array}{c}
\begin{tikzpicture}[scale=1]
    \whachannelonefour{(0,0)}{2}{$\Xi^2$};
    \wharholeft{(1.0,-1.976)}{$B'$};
\wharhoright{(-1.0,-1.976)};
\end{tikzpicture}
\end{array}=
\begin{array}{c}
\begin{tikzpicture}[scale=1]
    \whachannelonefour{(0,0)}{2}{$\Xi$};
    \wharholefttwo{(1.0,-1.976)};
\wharhoright{(-1.0,-1.976)};
\end{tikzpicture}
\end{array}
\end{aligned}
\]
where we use Lemma 3.5 in \cite{ruizdealarcon2024matrix} to move the position of $\Xi$ and \cref{eq:BlackVsWhite}. 
The boundary condition matrix is replaced by 
 $B'\ra B''=\Xi^{1/2} B' \Xi^{1/2}$ 
where $\Xi^{1/2}$ is positive definite (and each block $\Xi_a$ is positive definite). In particular, $B''_e=\Xi_e^{1/2} B'_e \Xi_e^{1/2}\gneq 0$. At least for the vacuum sector $\tr(\Xi_e^{1/2} B''_e \Xi_e^{1/2})>0$, and therefore $\mathcal{E}^{(\Eone)}$ acting on this state must be non-zero. 
With this, we conclude
\[
\mathcal{E}_{N-1}^{(\Eone)}\circ \mathcal{E}_{1}^{(\Eone)}\circ (\mathcal{E}^{(\Etwo)}_N\circ(\mathcal{E}^{(\Eone)}_1\circ \mathcal{E}^{(\Eone)}_2\circ\cdots\circ\mathcal{E}^{(\Eone)}_{N-1}(\rho)))\gneq 0.
\]
This finishes the proof. 
\end{widetext}

\section{Example of \texorpdfstring{$\mathbb{C}\mathbb{Z}_2$}{CZ2}}
\label{sec:CZ2-construction}
In this section, we show explicitly the construction in \cref{ex:CZ2_M_N} of the MPDO tensor for the boundary state of the toric code, using $\mathbb{C}\mathbb{Z}_2$ group $C^*$-algebra. The MPDO tensor is constructed by \cref{eqn:black-tensor-constr} and \cref{eqn:M-tensor-constr}. Hence, we need to identify the basis $\lbrace e_i\rbrace$ of $A$ and the dual basis $\lbrace e^i\rbrace$ of $A^*$, their faithful $*$-representations $\Phi,\Psi$, and $b(\omega)$. 

Since $A=\mathbb{C}\mathbb{Z}_2$ is a group algebra, its basis elements are identified with the group elements, and we denote the two basis elements as $e_1, e_z$; namely, the dimension of $A$ is 2. They obey the multiplication rule of $\mathbb{Z}_2$ group, and their faithful representation is the direct sum of the two irreps of $\mathbb{Z}_2$ group,
\begin{equation}
    \Phi(e_1)=\bo_2,\quad \Phi(e_z)=\sigma_z. 
\end{equation}
For the dual algebra $A^*$, its multiplication rule is induced by the comultiplication rule in $A$, $\Delta(e_g)=e_g\otimes e_g$. The multiplication of $A^*$ is thus $e^i\cdot e^j=\sum_k \delta_{ijk} e^k$, and the faithful representation of $A^*$ is
\begin{equation}
    \Psi(e^1)=\begin{pmatrix}
        1 & 0\\
        0 & 0
    \end{pmatrix},\quad \Psi(e^z)=\begin{pmatrix}
        0 & 0 \\
        0 & 1
    \end{pmatrix}.
\end{equation}
Finally, the boundary condition matrix $b(\omega)$ is found using $\omega=e^1$ and by definition $\tr(b(\omega)\Phi(x))=\omega(x)$ for any $x\in A$. The result is $b(\omega)=\bo_2/2$. 

Assemble everything, we obtain the MPDO tensor
\begin{equation}
\begin{aligned}
    \begin{array}{c}
        \begin{tikzpicture}[scale=1.,baseline={([yshift=-0.65ex] current bounding box.center)}]
        \whaM{(0,0)}{2};
        \end{tikzpicture}
        \end{array}&=\frac{\bo_2}{2}\otimes\begin{pmatrix}
            1 & 0\\
            0 & 0
        \end{pmatrix}+\frac{\sigma_z}{2}\otimes\begin{pmatrix}
            0 & 0\\
            0 & 1
        \end{pmatrix},
\end{aligned}
\end{equation}
where the first tensor factor corresponds to the endomorphism in the physical space and the second tensor factor to one in the virtual space, as given in \cref{ex:CZ2_M_N}. The parent Lindbladian of this MPDO is given in \cref{exm:CZ2-parent-Lindbladian}. 

\section{Example of Fibonacci}\label{sec:Fib}
In this section, we construct explicitly the MPDO RFP for the Fibonacci weak Hopf algebra $\Afib$ and its parent Lindbladian as a more complicated example. $\Afib$ is not a Hopf algebra because $\Delta(1)\neq 1\otimes 1$. This necessitates the second piece $\mathcal{E}^{(\Etwo)}$ in the parent Lindbladian construction. We will show the parent Lindbladian construction from two different ways (horizontal canonical form and $C^*$-weak Hopf algebra) and see that they give the same Lindbladian.

\subsection{Fibonacci MPDO construction}
To construct the MPDO RFP using \cref{eqn:black-tensor-constr} and \cref{eqn:M-tensor-constr}, we need to identify the basis $\lbrace e_i\rbrace$ of $A$ and its dual basis $\lbrace e^i\rbrace$ of $A^*$, their faithful $*$-representations $\Phi$ and $\Psi$, and $b(\omega)$. The Fibonacci $C^*$-weak Hopf algebra is defined as the direct sum $\mathcal{M}_{2}(\mathbb{C})\oplus \mathcal{M}_{3}(\mathbb{C})$ of full matrix $C^*$-algebras with complex coefficients. The basis elements of $\Afib$ is identified with the matrix units $e_{1,ij},i,j=1,2$ in $\mathcal{M}_{2}(\mathbb{C})$ and  $e_{2,ij},i,j=1,2,3$ in $\mathcal{M}_{3}(\mathbb{C})$. Namely, the dimension of $A$ is 13. The multiplication rule comes from the matrix multiplication, and the faithful representation $\Phi$ is simply the $5\times 5$ matrix as a direct sum of the two irreducible $*$-representations labeled by $a=1$ and $a=\tau$. For example, 
\begin{equation}
    \Phi(e_{1,12})=\begin{pmatrix}
        0 & 1 & & & \\
        0 & 0 & & & \\
        & & 0 & 0 & 0\\
        & & 0 & 0 & 0\\
        & & 0 & 0 & 0
    \end{pmatrix}.
\end{equation}

To obtain the MPDO tensor, we still need the faithful representation $\Psi$ of $A^*$. Let's use $I$ to denote the combined index $(\mathrm{a},ij)$, where $\mathrm{a}=1,2$. $\Psi$ can be found by noticing that $\Afib$ is self-dual, and therefore there exists a ``natural basis'' $\lbrace \tilde{e}^I \rbrace$, which is related to the dual basis $\lbrace e^I\rbrace$ by a transformation $\tilde{e}^I=\sum_J R^{IJ} e^J$, such that the faithful representation in the natural basis reads
\begin{equation}
    \Psi(\tilde{e}^I)=\Phi(e_I).
\end{equation}
The transformation matrix $R$, in general, is not unique. It is the solution of a consistency relation between multiplication and comultiplication. Specifically, the multiplication and comultiplication on $A$ can be specified using the basis elements
\[\begin{aligned}
     e_I\cdot e_J&=\sum_K \lambda_{IJ}^K e_K\\
     \Delta(e_K)&=\sum_{IJ}\Lambda^{IJ}_K e_I\otimes e_J.
\end{aligned}\]
By definition of dual algebra, they induce the comultiplication and multiplication in $A^*$,
\[
\begin{aligned}
    e^I\cdot e^J&=\sum_K \Lambda_K^{IJ} e^K\\
    \Delta(e^K)&=\sum_{IJ} \lambda_{IJ}^K e^I\otimes e^J.
\end{aligned}
\]
As a natural basis, $\tilde{e}^I\cdot\tilde{e}^J=\sum_K \lambda_{IJ}^K \tilde{e}^K$ and therefore the consistency relation requires
\begin{equation}
\label{eqn:consistencyR}
    \sum_{I' J'} R^{II'} R^{JJ'} \Lambda_S^{I' J'}=\sum_K \lambda_{IJ}^K R^{KS}.
\end{equation}
For $\Afib$, the multiplication rule $\lambda_{IJ}^K$ comes from matrix multiplication and the comultiplication rule $\Lambda_K^{IJ}$ is given explicitly in Example 2.5 of \cite{ruizdealarcon2024matrix}\footnote{We note two typos $\Delta(e_{2,13})=e_{1,12}\otimes e_{2,13}+e_{2,13}\otimes e_{1,22}+\zeta e_{2,12}\otimes e_{2,31}-\zeta^2 e_{2,13}\otimes e_{2,33}$ and $\Delta(e_{2,23})=e_{1,22}\otimes e_{2,23}+e_{2,23}\otimes e_{1,12}+\zeta e_{2,32}\otimes e_{2,21}-\zeta^2 e_{2,33}\otimes e_{2,23}$.}, and one possible solution for transformation $R$ of \cref{eqn:consistencyR} is~\cite{wang2024hopf}
\begin{equation}
\begin{aligned}
    &\tilde{e}^{1,11}= e^{1,11},\quad \tilde{e}^{1,12}=e^{2,11},\quad \tilde{e}^{1,21}=e^{2,22},\\
    &\tilde{e}^{1,22}=\zeta^4 e^{1,22}+\zeta^2 e^{2,33}\\
    &\tilde{e}^{2,11}=e^{1,12},\quad 
    \tilde{e}^{2,12}=e^{2,12},\quad\tilde{e}^{2,13}=e^{2,13}\\
    &\tilde{e}^{2,21}=\zeta^{-2} e^{2,21},\quad \tilde{e}^{2,22}=e^{1,21},\quad\tilde{e}^{2,23}=\zeta^{-1} e^{2,31}\\
    &\tilde{e}^{2,31}=\zeta^{-1} e^{2,23},\quad \tilde{e}^{2,32}=e^{2,32},\\
    & \tilde{e}^{2,33}=\zeta^2 e^{1,22}-\zeta^2 e^{2,33},
\end{aligned}
\label{eqn:eqnR}
\end{equation}
where $\zeta^2=2/(1+\sqrt{5})$. 

Finally, the boundary condition matrix $b(\omega)$ can be found by using character $\chi_\hbnt$ of Irrep class $\hbnt\in\text{Irr}(A)$ (namely, $\chi_\hbnt=\tr\circ\Phi_\hbnt$, which is the dual analog of $\mathrm{x}_a$), 
\[
\begin{aligned}
\chi_1&=e^{1,11}+e^{1,22}\\
\chi_\tau&=e^{2,11}+e^{2,22}+e^{2,33}\\
\omega&=\frac{1}{\mathcal{D}^2}(\delta_1 \chi_1+\delta_\tau \chi_\tau),
\end{aligned}
\]
where $\delta_1=1, \delta_\tau=\zeta^{-2}$ are the quantum dimensions of each Irrep, $\mathcal{D}^2=\sum_a \delta_a^2=1+\zeta^{-4}$ is the total quantum dimenstion, and by definition $\tr(b(\omega)\Phi(x))=\omega(x)$ for any $x\in A$. The result is 
\[
b(\omega)=\frac{\delta_1}{\mathcal{D}^2}\bo_2 \oplus \frac{\delta_\tau}{\mathcal{D}^2}\bo_3.
\]
In fact, one can show that in general, $b(\omega)=\oplus_\hbnt (\delta_\hbnt/\mathcal{D}^2)\bo_{\hbnt}$. 

Assemble everything, we obtain the MPDO tensor
\begin{equation}
\begin{aligned}
    \begin{array}{c}
        \begin{tikzpicture}[scale=1.,baseline={([yshift=-0.65ex] current bounding box.center)}]
		\draw (-0.75,0) node {$\alpha$};
		\draw (0.75,0) node {$\beta$};
		\draw (0,0.75) node {$i$};
        \draw (0,-0.75) node {$j$};
        \whaM{(0,0)}{2};
        \end{tikzpicture}
        \end{array}&=\sum_{I=1}^{13}[b(\omega)]_{ii'}\Phi_{i'j}(e_I)\Psi_{\alpha\beta}(e^I)\\
        &=\sum_{IJ}(R^{-1})^{IJ}[b(\omega)]_{ii'}\Phi_{i'j}(e_I)\Phi_{\alpha\beta}(e_J).
\end{aligned}
\end{equation}
Using the representation $\Phi$ and transformation $R$ in \cref{eqn:eqnR}, the part without $b(\omega)$ is 
\begin{equation}
\begin{aligned}
&\begin{array}{c}
\begin{tikzpicture}[scale=\rof]
\blackdot{(0,0)}{1}{1}{1}{1};
\end{tikzpicture}
\end{array}=
\begin{array}{c}
\begin{tikzpicture}[scale=\rof]
\blackdot{(0,0)}{3}{3}{1}{2};
\end{tikzpicture}
\end{array}=
\begin{array}{c}
\begin{tikzpicture}[scale=\rof]
\blackdot{(0,0)}{4}{4}{2}{1};
\end{tikzpicture}
\end{array}=
\begin{array}{c}
\begin{tikzpicture}[scale=\rof]
\blackdot{(0,0)}{2}{2}{2}{2};
\end{tikzpicture}
\end{array}\\
=&\begin{array}{c}
\begin{tikzpicture}[scale=\rof]
\blackdot{(0,0)}{5}{5}{2}{2};
\end{tikzpicture}
\end{array}=
\begin{array}{c}
\begin{tikzpicture}[scale=\rof]
\blackdot{(0,0)}{2}{1}{3}{3};
\end{tikzpicture}
\end{array}=\begin{array}{c}
\begin{tikzpicture}[scale=\rof]
\blackdot{(0,0)}{4}{3}{3}{4};
\end{tikzpicture}
\end{array}=
\begin{array}{c}
\begin{tikzpicture}[scale=\rof]
\blackdot{(0,0)}{5}{3}{3}{5};
\end{tikzpicture}
\end{array}\\
=&\begin{array}{c}
\begin{tikzpicture}[scale=\rof]
\blackdot{(0,0)}{1}{2}{4}{4};
\end{tikzpicture}
\end{array}=
\begin{array}{c}
\begin{tikzpicture}[scale=\rof]
\blackdot{(0,0)}{4}{5}{5}{4};
\end{tikzpicture}
\end{array}=
\begin{array}{c}
\begin{tikzpicture}[scale=\rof]
\blackdot{(0,0)}{2}{2}{5}{5};
\end{tikzpicture}
\end{array}=1,\\
&\begin{array}{c}
\begin{tikzpicture}[scale=\rof]
\blackdot{(0,0)}{3}{5}{4}{5};
\end{tikzpicture}
\end{array}=
\begin{array}{c}
\begin{tikzpicture}[scale=\rof]
\blackdot{(0,0)}{5}{4}{5}{3};
\end{tikzpicture}
\end{array}=\zeta,\\
&\begin{array}{c}
\begin{tikzpicture}[scale=\rof]
\blackdot{(0,0)}{3}{4}{4}{3};
\end{tikzpicture}
\end{array}=-
\begin{array}{c}
\begin{tikzpicture}[scale=\rof]
\blackdot{(0,0)}{5}{5}{5}{5};
\end{tikzpicture}
\end{array}=\zeta^2,
\end{aligned}
\end{equation}
which agrees with the explicit components given in \cite{ruizdealarcon2024matrix}. 

Using the explicit expression, one can also compute the transfer matrix 
\begin{equation}
    E=
\begin{array}{c}
        \begin{tikzpicture}[scale=1.,baseline={([yshift=-0.75ex] current bounding box.center)}
        ]
         \whaMtrace{(0,0)}{$M$};
        \end{tikzpicture}
        \end{array}=\begin{pmatrix}
        \frac{2}{5+\sqrt{5}} & \frac{1}{\sqrt{5}} & 0 & 0 & 0\\
        \frac{1}{\sqrt{5}} & \frac{1}{10}(5+\sqrt{5}) & 0 & 0 & 0\\
        0 & 0 & 0 & 0 & 0\\
        0 & 0 & 0 & 0 & 0\\
        0 & 0 & 0 & 0 & 0
    \end{pmatrix},
\end{equation}
and verifies that $E^2=E$ and $E=\Psi(\omega)$. The eigenvalues of $E$ are $1,0,0,0,0$, which indicates that $\rho(M)$ is a short-range correlated state, similar to the boundary state of the toric code. 

\subsection{Parent Lindbladian construction: vertical canonical form}
\label{sec:Fibvertical}
To construct the parent Lindbladian, we need the renormalization channels $\mT$ and $\mS$. In this section, we construct these two channels from the vertical canonical form \cref{eqn:Mvertical,eqn:Kvertical} and \cref{thm:non-simple}.

The one-site tensor $M_{(\alpha\beta)}$, taking the form of a direct sum, is already in the vertical canonical form. We write explicitly
\begin{equation}
    M_{(\alpha\beta)}=\mu_1 M_{(\alpha\beta),1}\oplus \mu_2 M_{(\alpha\beta),2},
\end{equation}
where $M_{(\alpha\beta),1}$ is the upper-left $2\times 2$ block of $M_{(\alpha\beta)}$ and $M_{(\alpha\beta),2}$ is the lower-right $3\times 3$ block of $M_{(\alpha\beta)}$ with proper normalization constants. Namely, there are two BNT elements. The corresponding multiplicity matrices are constants $\mu_1 = \frac{1}{\sqrt{5}}$ and $\mu_2=\frac{5+\sqrt{5}}{10}$, which takes care of the normalization (viewing $M_{(\alpha\beta),\lambda}$ as matrix product vector). 

To obtain the vertical canonical form for two-site MPDO $K_{(\alpha\beta)}$, we need a basis rotation $U$ on the physical indices. The result is
\begin{equation}
    \begin{aligned}
        U^\dg K_{(\alpha\beta)} U &= \begin{array}{c}
        \begin{tikzpicture}[scale=.4,baseline={([yshift=-0.75ex] current bounding box.center)}
        ]
        \uTensor{(\singledx*0.5,1.5)}{1.5}{0.5}{\singledx*0.5}{1}{\small $U^\dg$};
        \uTensor{(\singledx*0.5,-1.5)}{1.5}{0.5}{\singledx*0.5}{1}{\small $U$};
       \MTensor{(0,0)}{1}{.6}{\small $M$};
       \MTensor{(\singledx,0)}{1}{.6}{\small $M$};
       \draw(-1.5,0) node {\small $\alpha$};
       \draw(1.5+\singledx,0) node {\small $\beta$};
        \end{tikzpicture}
        \end{array}\\
        &=\left[\bigoplus_{\kappa=1}^2 \nu_\kappa\otimes M_{(\alpha\beta),\kappa}\right]\oplus \mathbf{0}_{12},\\
        \mathrm{where}\hspace{0.5em} \nu_1&=\begin{pmatrix}
        \frac{\sqrt{5}+1}{10} & 0\\
        0 & \frac{\sqrt{5}-1}{10}
    \end{pmatrix}, \nu_2=\begin{pmatrix}
        \frac{1}{5} & 0 & 0\\
        0 & \frac{1}{5} & 0\\
        0 & 0 & \frac{\sqrt{5}+1}{10}
    \end{pmatrix},
    \end{aligned}
\end{equation}
and $\mathbf{0}_{12}$ denotes a $12\times 12$ zero matrix. The explicit form of $U$ will be specified at the end of the section. Indeed, the vertical BNT of $K$ is still $\lbrace M_{(\alpha\beta),\lambda}\rbrace$, as stated in \cref{thm:non-simple}. The trace relation is also satisfied: $m_1=\mu_1=\tr(\nu_1)=\frac{1}{\sqrt{5}}$ and $m_2=\mu_2=\tr(\nu_2)=\frac{5+\sqrt{5}}{10}$. 

Using the vertical canonical form of $M$ and $K$, one can construct the channel $\mT$ and $\mS$ as follows. The one-site MPDO to two-site MPDO channel $\mT$ is
\begin{equation}
\begin{aligned}
     &\mT(X) \\
     &= U\left[(\frac{\nu_1}{m_1}\otimes P_1 X P_1^\dg) \oplus (\frac{\nu_2}{m_2}\otimes P_2 X P_2^\dg)\oplus\mathbf{0}_{12} \right] U^\dg,
\end{aligned}
\end{equation}
where 
\[
P_1=\begin{pmatrix}
    1 & 0 & 0 & 0 & 0\\
    0 & 1 & 0 & 0 & 0
\end{pmatrix},\quad
P_2=\begin{pmatrix}
    0 & 0 & 1 & 0 & 0\\
    0 & 0 & 0 & 1 & 0\\
    0 & 0 & 0 & 0 & 1
\end{pmatrix}
\]
are the isometries projecting to the first and second BNT, respectively. By direct inspection, one can see $\mT(M_{(\alpha\beta)})=K_{(\alpha\beta)}$. 

The two-site MPDO to one-site MPDO channel $\mS$ is 
\begin{equation}
\begin{aligned}
    \mS &= \mS^{(\Eone)}+\mS^{(\Etwo)}\\
    \mS^{(\Eone)}(X)&=\tr_1(P_1' U^\dg X U (P_1')^\dg) \oplus  \tr_1(P_2' U^\dg X U (P_2')^\dg) \\
    \mS^{(\Etwo)}(X)&=\tr(P_0' U^\dg X U (P_0')^\dg) \rho_0
\end{aligned}
\end{equation}
where we use $\mu_\lambda=m_\lambda$ in this example. The isometries $P_1',P_2',P_0'$ project to the blocks corresponding to the first BNT, the second BNT, and the zero matrix, respectively. Their corresponding dimensions are $4\times 25,9\times 25$ and $12\times 25$. $\tr_1$ denotes tracing over the multiplicity degrees of freedom. The term $\mS^{(\Etwo)}$ is added to make the map trace-preserving, and $\rho_0$ can be any density matrix supported on one site. By direct inspection one can see $\mS(K_{(\alpha\beta)})=M_{(\alpha\beta)}$. 

Combining $\mT$ and $\mS$, the local quantum channel $\mathcal{E}$ is
\begin{equation}
    \begin{aligned}
        \mathcal{E}&=\mathcal{E}^{(\Eone)}+\mathcal{E}^{(\Etwo)}\\
        \mathcal{E}^{(\Eone)}(X)&=U\left[\frac{\nu_1}{m_1}\otimes\tr_1(P_1' U^\dg X U (P_1')^\dg)\right.\\
        &\quad \left.\oplus \frac{\nu_2}{m_2}\otimes\tr_1(P_2' U^\dg X U (P_2')^\dg)\oplus \mathbf{0}_{12}\right]U^\dg\\
        \mathcal{E}^{(\Etwo)}(X)&=U\left[\frac{\nu_1}{m_1}\otimes (P_1 \rho_0 P_1^\dg)\tr(P_0' U^\dg X U (P_0')^\dg)\right.\\
         \oplus\frac{\nu_2}{m_2}&\left.\otimes (P_2 \rho_0 P_2^\dg)\tr(P_0' U^\dg X U (P_0')^\dg)\oplus \mathbf{0}_{12}\right]U^\dg.
    \end{aligned}
    \label{eqn:FibEv}
\end{equation}
By direct inspection, one can see the fixed-point space of channel $\mathcal{E}$ is the fixed-point space of $\mathcal{E}^{(\Eone)}$, 
\begin{equation}
    \mathcal{F}_{\mathcal{E}}=U\left[\left(\bigoplus_{\lambda=1}^2 \nu_\lambda\otimes \mathcal{M}_{d_\lambda}\right)\oplus \mathbf{0}_{12}\right] U^\dg.
\end{equation}
where $\mathcal{M}_d$ is the full $d\times d$ dimensional matrix algebra and $d_{\lambda=1}=2,d_{\lambda=2}=3$. The dimension of the fixed point space is therefore $d_{\lambda=1}^2+d_{\lambda=2}^2=13$. This finishes the explicit construction of parent Lindbladian. We note again that the construction is based on the fact that at RFP, $M$ and $K$ share the same vertical BNT; the coarse-graining and fine-graining channels only act on the multiplicity degrees of freedom, without changing the BNT part. 

Although this construction is intuitive, proving frustration-freeness and finding global steady states is difficult. In the next section, we present another construction using $C^*$-weak Hopf algebra, which is more convenient for finding the global steady states. Numerically, we verify that these two constructions give the same CP maps $\mathcal{E}^{(\Eone)}$. 

Finally, we provide the explicit form of the transformation matrix $U$. The unitary matrix $U$ is the concatenation the following 25 column vectors (where $|v_i\rangle_j$ is the $j$-the component of the vector $|v_i\rangle$),
\begin{equation}
\begin{aligned}
    &|v_1\rangle_{14}=1,\quad |v_2\rangle_{18}=\frac{\zeta}{\sqrt{\zeta^2+1}},\quad |v_2\rangle_{25}=\frac{1}{\sqrt{\zeta^2+1}}\\
    &|v_3\rangle_1=1,\quad |v_4\rangle_7=1\\
    &|v_5\rangle_3=1,\quad |v_6\rangle_9=1,\quad |v_7\rangle_{10}=1\\
    &|v_8\rangle_{12}=1,\quad |v_9\rangle_{16}=1,\quad |v_{10}\rangle_{22}=1\\
    &|v_{11}\rangle_{15}=1,\quad |v_{12}\rangle_{24}=1,\\
    &|v_{13}\rangle_{18}=\frac{1}{\sqrt{\zeta^2+1}},\quad |v_{13}\rangle_{25}=\frac{-\zeta}{\sqrt{\zeta^2+1}},
\end{aligned}
\end{equation}
with all other components equal to zero, 
and $|v_{14}\rangle$ to $|v_{25}\rangle$ being arbitrary basis vectors in the orthogonal space of $|v_1\rangle$ to $|v_{13}\rangle$, 

\subsection{Parent Lindbladian construction: \texorpdfstring{$\boldsymbol{C^*}$}{C*}-weak Hopf algebras}
In terms of $C^*$-weak Hopf algebras, the quantum channels $\mT$ and $\mS$ are given in \cref{eqn:wha-channelT,eqn:wha-channelS1,eqn:wha-channelS2}, from which we can write down $\mathcal{E}=\mT\circ\mS$. The first piece is
\begin{equation}
\begin{aligned}
    \mE^{(\Eone)}(X)&=\mT\circ\mS^{(\Eone)}(X)\\
    &=b(\omega)^{\otimes 2}(\Phi^{\otimes 2}\circ\Delta)(\Omega_{(1)})\\
    &\hspace{3.5em}\cdot\tr((\Phi^{\otimes 2}\circ\Delta\circ J)(\xi^{\frac{1}{2}}\Omega_{(2)}\xi^{\frac{1}{2}}) X),
\end{aligned}
\label{eqn:FibEh}
\end{equation}
where we use \cref{eqn:omega-identity}, 
and the second piece is
\begin{equation}
    \begin{aligned}
        \mE^{(\Etwo)}(X)&=\mT\circ\mS^{(\Etwo)}(X)\\
        &=b(\omega)^{\otimes 2}(\Phi^{\otimes 2}\circ\Delta)(\Omega_{(1)})\\
        &\hspace{3.5em}\cdot\tr((\Phi\circ J)(\xi^{\frac{1}{2}}\Omega_{(2)}\xi^{\frac{1}{2}})\rho_0)\tr(P^\perp X),
    \end{aligned}
\end{equation}
where $P^\perp=\bo-(\Phi^{\otimes 2}\circ\Delta)(1)$ and $\rho_0$ is an arbitrary density matrix. 

To explicitly construct $\mE$, we need the explicit form of $1,\Omega,\xi$ and the map $J$. The comultiplication $\Delta$ is given explicitly in \cite{ruizdealarcon2024matrix}. The explicit form of $1$ comes from its definition as the unit element in $A$,
\begin{equation}
    \begin{aligned}
        1&=e_{1,11}+e_{1,22}+e_{2,11}+e_{2,22}+e_{2,33},
    \end{aligned}
\end{equation}
and the explicit form of $\Omega$ is
\begin{equation}
    \begin{aligned}
        \Omega & = \frac{1}{\mathcal{D}^2}(d_1 \mathrm{x}_1+d_\tau \mathrm{x}_\tau)\\
        &=\frac{1}{\mathcal{D}^2}(e_{1,11}+\frac{1}{\zeta^2}e_{1,12}+\frac{1}{\zeta^2}e_{1,21}+\frac{1}{\zeta^4}e_{1,22})
    \end{aligned}
\end{equation}
where $d_1=1, d_\tau=\zeta^{-2}, \mathcal{D}^2=1+\zeta^{-4}$. The form of $\xi$ can be determined using $\xi^{-1}=\omega(\Omega_{(1)})\Omega_{(2)}$, and the result reads
\begin{equation}
    \xi=\mathcal{D}^4(e_{1,11}+\zeta^4 e_{1,22}+ \zeta^2 e_{2,11} + \zeta^2 e_{2,22} + \zeta^4 e_{2,33}).
\end{equation}
Finally, the linear map $J$ can be obtained from the pulling-through identity \cref{eqn:pull-through}, which after applying $J$ to the first site, becomes $(J(x)\xi^{\frac{1}{2}}\otimes 1)\cdot\Delta(\Omega)=(\xi^{\frac{1}{2}}\otimes x)\cdot\Delta(\Omega)$. The result is
\begin{equation}
\begin{aligned}
    &J(e_{1,11})=e_{1,11},\tb J(e_{1,12})=e_{1,21},\tb J(e_{1,21})=e_{1,12}\\
    &J(e_{1,22})=e_{1,22},\tb J(e_{2,11})=e_{2,22},\tb J(e_{2,12})=e_{2,12}\\
    &J(e_{2,13})=e_{2,32},\tb J(e_{2,21})=e_{2,21},\tb J(e_{2,22})=e_{2,11}\\
    &J(e_{2,23})=e_{2,31},\tb J(e_{2,31})=e_{2,23},\tb J(e_{2,32})=e_{2,13}\\
    &J(e_{2,33})=e_{2,33}.
\end{aligned}
\end{equation}
Using the above data, one can explicitly construct the quantum channel $\mathcal{E}$. 

By direct inspection, one can see the fixed-point space of channel $\mathcal{E}$ is the fixed-point space of $\mathcal{E}^{(\Eone)}$, which is the span
\begin{equation}
    \mathcal{F}_\mE=\left\lbrace \sum_{I=1}^{13} c_I b(\omega)^{\otimes 2}(\Phi^{\otimes 2}\circ\Delta) (e_I)\Big|c_I\in\mathbb{C}\right\rbrace.
\end{equation}
The dimension of the fixed-point space is 13. 

\subsection{A hint of duality}
Using two explicit parent Lindbladian constructions, we numerically find that $\mathcal{E}^{(\Eone)}$ from \cref{eqn:FibEv} and \cref{eqn:FibEh} are identical. This result is intriguing because the first construction is derived from the vertical canonical form of the MPDO, while the second stems from the algebraic structure, which implicitly relies on the horizontal canonical form. This observation suggests a potential duality between the vertical and horizontal canonical forms. One might speculate whether this duality arises from the self-duality of the Fibonacci $C^*$-weak Hopf algebra. To explore this further, we examine a non-self-dual case, the MPDO RFP constructed from the $\mathbb{C}S_3$ Hopf algebra. Interestingly, we find that $\mathcal{E}$ from the two constructions remains identical. A deeper understanding of this apparent duality is left for future work.

\section{Example of $\rho_{CZX}$}
\label{app:CZX}

In this appendix, we provide details of the vertical canonical form analysis of $\rho_{CZX}$ and study the properties of its parent Lindbladian. This is an example of the non-simple MPDO beyond the $C^*$-weak Hopf algebra~\cite{liu2025trading}. 

\subsection{Vertical canonical form}
We now show explicitly that this MPDO is an RFP after blocking two sites by using the vertical canonical form. Define $M$ as the tensor obtained by blocking two sites of $\tilde{M}$. By a unitary transformation $U$ on the physical degrees of freedom,
\begin{equation}
\begin{aligned}
    M_{(\alpha \beta)}&\ra M^U_{(\alpha \beta)} = U M_{(\alpha \beta)} U^\dg \\ 
    U&=\begin{pmatrix}
        1 & 0 & 0 & 0\\
        0 & 0 & 0 & 1\\
        0 & 1 & 0 & 0\\
        0 & 0 & 1 & 0
    \end{pmatrix}
\end{aligned}
\end{equation}
one can bring $M$ to the vertical canonical form with two BNT elements, 
\begin{equation}
    M^U_{(\alpha \beta)} = \mu_1 M_{(\alpha\beta),1}\oplus \mu_2 M_{(\alpha\beta),2}, 
\end{equation}
with $\mu_1= \mu_2 = \frac{\sqrt{3}}{4}$,
\begin{equation}
\begin{aligned}
    M_{(11),1} & = \frac{1}{\sqrt{3}}\begin{pmatrix}
        1 & 0 \\
        0 & 1
    \end{pmatrix},\quad 
    M_{(22),1}=\frac{1}{\sqrt{3}}\begin{pmatrix}
        0 & 1\\
        0 & 0
    \end{pmatrix}\\
    M_{(23),1} &=\frac{1}{\sqrt{3}}\begin{pmatrix}
        0 & 1\\
        0 & 0
    \end{pmatrix},\quad
    M_{(32),1} = \frac{1}{\sqrt{3}}\begin{pmatrix}
        0 & 0\\
        -1 & 0
    \end{pmatrix}\\
    M_{(33),1}&=\frac{1}{\sqrt{3}}\begin{pmatrix}
        0 & 0\\
        1 & 0
    \end{pmatrix}
\end{aligned}
\end{equation}
and
\begin{equation}
\begin{aligned}
    M_{(11),2} & = \frac{1}{\sqrt{3}}\begin{pmatrix}
        1 & 0 \\
        0 & 1
    \end{pmatrix},\quad
    M_{(22),2}=\frac{1}{\sqrt{3}}\begin{pmatrix}
        0 & 1\\
        0 & 0
    \end{pmatrix}\\
    M_{(23),2} &=\frac{1}{\sqrt{3}}\begin{pmatrix}
        0 & -1\\
        0 & 0
    \end{pmatrix},\quad
    M_{(32),2} = \frac{1}{\sqrt{3}}\begin{pmatrix}
        0 & 0\\
        1 & 0
    \end{pmatrix}\\
    M_{(33),2}&=\frac{1}{\sqrt{3}}\begin{pmatrix}
        0 & 0\\
        1 & 0
    \end{pmatrix}.
\end{aligned}
\end{equation}
One can verify that these two BNT elements are indeed independent by computing the leading eigenvalue of the mixed transfer matrix. 

Next, consider $K$ as the tensor obtained by blocking two sites of $M$ (which is, four sites of $\tilde{M}$). By a unitary transformation $V$ on the physical degrees of freedom, 
\begin{equation}
    K_{(\alpha\beta)}\ra K_{(\alpha\beta)}^V = V K_{(\alpha\beta)} V^\dg,
\end{equation}
one can bring $K$ to the vertical canonical form with the same two BNT elements,
\begin{equation}
    K_{(\alpha\beta)}^V = \nu_1 M_{(\alpha\beta),1}\oplus \nu_2 M_{(\alpha\beta),2}
\end{equation}
with $\nu_1=\nu_2=\frac{\sqrt{3}}{16}\bo_4$. The coefficients verify $m_\lambda=\mu_\lambda=\tr[\nu_\lambda]=\frac{\sqrt{3}}{4}$. The explicit form of the unitary $V$ is
\begin{equation}
    \begin{aligned}
        &V_{1,1}=V_{2,16}=V_{3,3}=V_{4,14}=1,\\
        &V_{5,5}=V_{6,12}=V_{7,7}=-V_{8,10}=1,\\
        &V_{9,2}=V_{10,15}=V_{11,4}=-V_{12,13}=1,\\
        &V_{13,6}=V_{14,11}=V_{15,8}=V_{16,9}=1.
    \end{aligned}
\end{equation}

We would like to comment that the one-site MPDO tensor $\tilde{M}$ contains only one BNT element, and therefore $\tilde{M}$ itself is not an RFP tensor. On the other hand, after blocking two sites, we see that $M$ is an RFP tensor by showing that $M$ and $K$ share the same set of BNT. This is expected, since we require $N$ to be even, and the above MPDO representation is valid only for even $N$.

\subsection{Steady states of parent Lindbladian}

Using the vertical canonical form, the renormalization channels $\mT,\mS$ are constructed as in~\cref{ex:CZX-ex} and with them one constructs the parent Lindbladian. In this section, we analyze the properties of this parent Lindbladian. 

Following the general treatment, we build the local channel $\mathcal{E}$ (that acts on four original sites) and the Lindbladian
\begin{equation}
    \begin{aligned}
        \mathcal{E}&=\mathcal{T}\circ\mathcal{S},\quad \mathcal{E}_i=\tau_i(\mathcal{E})\\
        \mathcal{L}_i&=\mathcal{E}_i-\bo,\quad \mathcal{L}^{(N)}=\sum_{i=1}^{N/2} \mathcal{L}_{2i-1}.
    \end{aligned}
\end{equation}
Let's first examine the fixed-point space of the local channel $\mathcal{E}$,
\begin{equation}
    \mathcal{F}_{\mE}=V^\dg\left[
    \bigoplus_{\lambda=1}^2 \nu_\lambda \otimes \mathcal{M}_{d_{\lambda}}
    \right] V,
\end{equation}
where $\mathcal{M}_{d_{\lambda}}$ is the full $d_{\lambda}\times d_{\lambda}$ dimensional matrix algebra, and $d_{\lambda=1}=d_{\lambda=2}=2$. The dimension of the fixed point space is therefore $d_{\lambda=1}^2+d_{\lambda=2}^2=8$. Explicitly,
\begin{equation}
\label{eqn:CZX-fixed-point-space}
\begin{aligned}
    &\mathcal{F}_\mE = \mV\oplus\mW, \\
    &\mV  = \text{span} \{
    \bo_2\otimes \bo_2\otimes \bo_2\otimes \bo_2,\bo_2\otimes \bo_2\otimes \bo_2\otimes \sigma_z,\\
    &\quad\quad  \sigma_z\otimes \bo_2\otimes \bo_2\otimes \bo_2,\sigma_z\otimes \bo_2\otimes \bo_2\otimes \sigma_z\}\\
    &\mW  = \text{span}\{\\
    & \sigma_x^{\otimes 4} CZ^{(4)}(|m\rangle\langle m|\otimes \bo_2\otimes\bo_2\otimes|n\rangle\langle n|)|m,n=0,1\},
\end{aligned}
\end{equation}
where $CZ^{(4)}=\prod_{i=1}^4(CZ)_{i,i+1}$ and note that $\mV\cap \mW =\{0\}$.
In particular, the vector space $\mW$ allows a simple MPO representation same as $U_{CZX}$,
\begin{equation}
\begin{aligned}
    \mW &= G_{L=4}(O)\\
    O^{12}&=\begin{pmatrix}
        1 & 1\\
        0 & 0
    \end{pmatrix},\quad O^{21}=\begin{pmatrix}
        0 & 0\\
        1 & -1
    \end{pmatrix},
\end{aligned}
\end{equation}
where $G_{L}(O)$ is defined as in Eq.~\eqref{eqn:GLM}. As a side remark, we note that the vector space $G_{L=4}(\tilde{M})$ has dimension 5 and is a subspace of the fixed-point space,
\begin{equation}
    G_{L=4}(\tilde{M})\subseteq \mathcal{F}_{\mE}.
\end{equation}

To find the space of steady state of the global Lindbladian, we examine the commutation relation of local channels. We verify numerically that 
\begin{equation}
\begin{aligned}
     &[\mathcal{E}_i,\mathcal{E}_{i+2}]=0.\\
\end{aligned}
\end{equation}
Therefore, the parent Lindbladian is commuting and thus frustration-free. The global space of steady state is then the intersection of the fixed-point spaces $\cap_{i=1}^N \mathcal{F}_{\mE_i}$. Let's first consider
\begin{equation}
\begin{aligned}
    \mathcal{F}_{\mE_i} \cap \mathcal{F}_{\mE_{i+2}} &= (\mV_i\oplus \mW_i) \cap  (\mV_{i+2}\oplus \mW_{i+2}),
\end{aligned}
\end{equation}
where $\mV_i= \mV\otimes \mathcal{M}_2\otimes\mathcal{M}_2$, $\mV_{i+2}=\mathcal{M}_2\otimes\mathcal{M}_2\otimes \mV$, and similar for $\mW_i,\mW_{i+2}$. To simplify, we note that the following two intersections are empty: $V_i\cap W_{i+2}=\{0\}$ and $W_i\cap V_{i+2}=\{0\}$, which can be proven by considering a non-zero vector in this intersection and then tracing out site $i+2$, leading to a contradiction. Along with the conditions $V_i\cap W_i=\{0\}$ and $V_{i+2}\cap W_{i+2}=\{0\}$, the above expression can be simplified to 
\begin{equation}
    \mathcal{F}_{\mE_i} \cap \mathcal{F}_{\mE_{i+2}} = (\mV_i\cap \mV_{i+2})\oplus(\mW_i\cap \mW_{i+2}).
\end{equation}
The first term can be computed by considering a non-zero vector in the intersection $\bm{v}_1 = \bm{v}_2$, $\bm{v}_1\in \mV_i,\bm{v}_2\in \mV_{i+2}$, multiplying $\sigma_z$ on site $i+2$, and then tracing out site $i+2$. Since $\bm{v}_1$ can only take $\bo_2$ on site $i+2$, this procedure leads to $\tr_{i+2}[\bm{v}_2 (\sigma_z)_{i+2}]=0$, enforcing $\bm{v}_2$ to take $\bo_2$ on site $i+2$. One can apply similar procedures on other sites to obtain other constraints, and the final result is
\begin{equation}
\begin{aligned}
    \mV_{i,i+2}:&=\mV_i\cap \mV_{i+2}\\
    &=\text{span}\{ \bo_2\otimes \bo_2\otimes \bo_2\otimes \bo_2\otimes \bo_2\otimes \bo_2,\\
    &\quad \bo_2\otimes \bo_2\otimes \bo_2\otimes \bo_2\otimes \bo_2\otimes \sigma_z,\\
    &\quad \sigma_z\otimes \bo_2\otimes \bo_2\otimes \bo_2\otimes \bo_2\otimes \bo_2,\\
    &\quad \sigma_z\otimes \bo_2\otimes \bo_2\otimes \bo_2\otimes \bo_2\otimes \sigma_z\}.
\end{aligned}
\end{equation}
Next, we consider the second term $ \mW_i\cap \mW_{i+2}$. We note that by blocking two sites of MPO tensor $O^{ij}_{\alpha\beta}$, it becomes injective when viewed as a matrix product vector (MPV) (by grouping $i,j$ indices). One can then apply the intersection property of MPV~\cite{cirac2021matrix} and obtain
\begin{equation}
\begin{aligned}
    \mW_{i,i+2}:&=\mW_i\cap \mW_{i+2} \\
    &= (G_{L=4}(O) \otimes \mathcal{M}_4)\cap (\mathcal{M}_4\otimes G_{L=4}(O))\\&= G_{L=6}(O). 
\end{aligned}
\end{equation}
The above procedure can be performed iteratively until we reach the system size $N$, and for $\mW$, using the closure property when crossing the boundary. We conclude
\begin{equation}
\label{eqn:CZX-steady-state}
    \rho_\infty^{(N)}=\frac{1}{2^N}(\bo_2^{\otimes N}+ c U_{CZX}^{(N)}),\quad c\in[-1,1], 
\end{equation}
where $N$ is even. 

We note that the above-constructed parent Lindbladian is rapid-mixing, which is a direct result of commuting local terms. Nevertheless, it cannot drive a product state fast to $\rho_{CZX}$, as shown in the following evolution analysis. 

\subsubsection{Evolution}
Starting from a trivial state $\rho=\frac{1}{2^N}(\bo^{\otimes N}+\sigma_x^{\otimes N})$, we numerically verify that: when $N=4,6,8$, the resulting steady state is
\begin{equation}
    \rho_\infty^{(N)} = \frac{1}{2^N}(\bo_2^{\otimes N}+ \frac{1}{2^{N/2-1}} U_{CZX}^{(N)}).
\end{equation}
We conjecture that this relation holds for general $N$. In the following, we prove that the parent Lindbladian constructed above cannot drive a fully separable state to $\rho_{CZX}$. 

\begin{prop}
Consider any initial state as a fully separable state
\begin{equation}
    \rho_\init^{(N)} =\sum_s q_s (\rho_{1,s}\otimes \rho_{2,s}\otimes\cdots\otimes \rho_{N,s})
\end{equation}
where $N$ is an even integer and $\sum_s q_s =1$. The corresponding steady state $\rho_\infty^{(N)}$ under the evolution of CZX parent Lindbladian takes the form of~\cref{eqn:CZX-steady-state}, with
\begin{equation}
    |c|\leq \frac{1}{2^{N/2-1}}. 
\end{equation} 
\end{prop}

\begin{proof}
    To begin, note that the coefficient $c$ in~\cref{eqn:CZX-steady-state} can be computed via
    \begin{equation}
    \begin{aligned}
        c &= \tr[U_{CZX}^{(N)} \rho_\infty^{(N)}]=\tr[U_{CZX}^{(N)} \mathcal{E}_{\infty}(\rho_\init^{(N)})]\\
        & = \tr[\mathcal{E}_\infty^*(U_{CZX}^{(N)}) \rho_\init^{(N)}]
    \end{aligned}
    \end{equation}
    where $\mathcal{E}_{\infty}:=\lim_{t\ra\infty}e^{\mathcal{L}t}$, and $\mathcal{E}_\infty^*$ is its adjoint. Since the CZX parent Lindbladian $\mathcal{L}$ is a commuting Lindbladian, $\mathcal{E}_{\infty}=\prod_i \mathcal{E}_{i,\infty}$ where $\mathcal{E}_{i,\infty}:=\lim_{t\ra\infty} e^{\mathcal{L}_i t}$. Similarly, $\mathcal{E}_\infty^*$ can also be decomposed as $\mathcal{E}_\infty^*=\prod_i \mathcal{E}_{i,\infty}^*$ where $\mathcal{E}_{i,\infty}^*$ is the adjoint of $\mathcal{E}_{i,\infty}^*$. 

    The local infinite time evolution channel is a projector on the space of local fixed point (on four physical sites), 
    \begin{equation}
    \mathcal{E}_{i,\infty}(X)=\sum_{m=1}^8 \tr[O_m^\dg X] O_m,
    \end{equation}
    where $\tr[O_m^\dg O_n]=\delta_{mn}$, and $\{O_m\}$ are the basis operators of local fixed point space in~\cref{eqn:CZX-fixed-point-space} up to normalization, 
    \begin{equation}
        \begin{aligned}
            &O_1=\frac{1}{4}\bo_2\otimes \bo_2\otimes \bo_2\otimes \bo_2, O_2=\frac{1}{4}\bo_2\otimes \bo_2\otimes \bo_2\otimes \sigma_z,\\
    &O_3 = \frac{1}{4}\sigma_z\otimes \bo_2\otimes \bo_2\otimes \bo_2,O_4=\frac{1}{4}\sigma_z\otimes \bo_2\otimes \bo_2\otimes \sigma_z,\\
    &O_5 =\sigma_x^{\otimes 4} CZ^{(4)}(|0\rangle\langle 0|\otimes \bo_2\otimes\bo_2\otimes|0\rangle\langle 0|),\\
    &O_6 =\sigma_x^{\otimes 4} CZ^{(4)}(|0\rangle\langle 0|\otimes \bo_2\otimes\bo_2\otimes|1\rangle\langle 1|),\\
    &O_7 =\sigma_x^{\otimes 4} CZ^{(4)}(|1\rangle\langle 1|\otimes \bo_2\otimes\bo_2\otimes|0\rangle\langle 0|),\\
    &O_8 =\sigma_x^{\otimes 4} CZ^{(4)}(|1\rangle\langle 1|\otimes \bo_2\otimes\bo_2\otimes|1\rangle\langle 1|).
        \end{aligned}
    \end{equation}
    Subsequently, its adjoint takes the form of
    \begin{equation}
        \mathcal{E}_{i,\infty}^*(X) = \sum_{m=1}^8 \tr[O_m X] O_m^\dg.
    \end{equation}
    Note that $O_m=O_m^\dg$ for $m=1,\cdots,4$ and $O_5^\dg=O_8, O_6^\dg=O_7$, the explicit form of $\mathcal{E}_{i,\infty}^*$ is therefore the same as $\mathcal{E}_{i,\infty}$, and thus $U_{CZX}$ stays invariant under the global ajoint map,
    \begin{equation}
        \mathcal{E}_{\infty}^*(U_{CZX}):= \left(\prod_{i=1}^{N/2} \mathcal{E}_{2i-1,\infty}^*\right) U_{CZX}=U_{CZX}. 
    \end{equation}
    
    Now, consider a pure single-site initial state with no summation over $s$, $\rho_i=|v_i\rangle\langle v_i|$ where $|v_i\rangle=a_i|0\rangle+b_i|1\rangle$, $|a_i|^2+|b_i|^2=1$. The MPU representation of $U_{CZX}$ reads
    \begin{equation}
        U^{01}=\begin{pmatrix}
            1 & 1\\
            0 & 0
        \end{pmatrix},\quad U^{10}=\begin{pmatrix}
            0 & 0\\
            1 & -1
        \end{pmatrix},
    \end{equation}
    and we define a transfer matrix 
    \begin{equation}
        E_{U,i} :=\begin{array}{c}
        \begin{tikzpicture}[scale=1.,baseline={([yshift=-0.75ex] current bounding box.center)}
        ]
        \whaMtraceOp{(0,0)}{\small $U$}{\small $\rho_i$};
        \end{tikzpicture}
        \end{array} = \begin{pmatrix}
            a_i^* b_i & a_i^* b_i\\
            a_i b_i^* & -a_i b_i^*
        \end{pmatrix}.
    \end{equation}
    Then,
    \begin{equation}
        c = \tr[E_{U,1} E_{U,2}\cdots E_{U,N} ]
    \end{equation}
    and its absolute value can be bound by
    \begin{equation}
    \begin{aligned}
        |c|&\leq \|E_{U,1}\|_1 \|E_{U,2}\cdots E_{U,N} \|_\infty \\
        &\leq \|E_{U,1}\|_1 \prod_{i=2}^N \|E_{U,i}\|_\infty. 
    \end{aligned}
    \end{equation}
    Using the explicit form of $E_{U,i}$, one can show $\|E_{U,i}\|_\infty=\sqrt{2}|a_i||b_i|\leq 1/\sqrt{2}$, where the inequality comes from $|a_i|^2+|b_i|^2=1$; and similarly $\|E_{U,1}\|_1 = 2\sqrt{2}|a_1||b_1|\leq \sqrt{2}$. This leads to
    \begin{equation}
        |c|\leq \frac{1}{2^{N/2-1}}.
    \end{equation}

    Finally, we justify that the extreme value of $c$ must occur when there is no summation over $s$ and each $\rho_i$ in $\rho_\init$ is a pure state. To see this, consider 
    \begin{equation}
    \label{eqn:c+1}
    \begin{aligned}
         &c + 1\\
         &= \tr[(\bo_2^{\otimes N}+U^{(N)}_{CZX}) \rho_\infty^{(N)}]\\
         & = \tr[\mathcal{E}_\infty^*(\bo_2^{\otimes N}+U_{CZX}^{(N)}) \rho_\init^{(N)}]\\
         & =\sum_s q_s \sum_{\{m_{is}\}} p_{m_{1s}}\cdots p_{m_{Ns}}\tr[(\bo_2^{\otimes N}+U^{(N)}_{CZX})\\
         &\quad\cdot(|\psi_{m_{1s}}\rangle\langle\psi_{m_{1s}}|\otimes\cdots\otimes |\psi_{m_{Ns}}\rangle\langle\psi_{m_{Ns}}|)]
    \end{aligned}
    \end{equation}
    where we decompose each 
    \[
    \rho_{i,s}=\sum_{m_{is}}p_{m_{is}}|\psi_{m_{is}}\rangle\langle\psi_{m_{is}}|
    \]
    with $p_{m_{is}}\geq 0,\sum_{m_{is}}p_{m_{is}}=1$. Since $(\bo_2^{\otimes N}+U^{(N)}_{CZX})$ is positive semi-definite, each trace term in the last line of~\cref{eqn:c+1} must be non-negative. Therefore,
    \begin{equation}
        c \leq \max_{\{|\psi_{m_{is}}\rangle\}}\langle\psi_{m_{1s}}|\cdots\langle\psi_{m_{Ns}}|U_{CZX}^{(N)}|\psi_{m_{1s}}\rangle\cdots |\psi_{m_{Ns}}\rangle,
    \end{equation}
    and 
    \begin{equation}
        c \geq \min_{\{|\psi_{m_{is}}\rangle\}}\langle\psi_{m_{1s}}|\cdots\langle\psi_{m_{Ns}}|U_{CZX}^{(N)}|\psi_{m_{1s}}\rangle\cdots |\psi_{m_{Ns}}\rangle. 
    \end{equation}
    This finishes the proof. 
\end{proof}



The above proposition and proof can be generalized to the initial state after blocking a finite number of $l$ sites.
\begin{prop}
Consider any initial state as a fully separable state
\begin{equation}
    \rho_\init^{(N)} =\sum_s q_s (\rho_{1,s}\otimes \rho_{2,s}\otimes\cdots\otimes \rho_{N/l,s})
\end{equation}
where each $\rho_{i,s}$ supports on $l$ consecutive sites of qubits, $N$ is an even integer, $N/l$ is an integer and $\sum_s q_s =1$. The corresponding steady state $\rho_\infty^{(N)}$ under the evolution of CZX parent Lindbladian takes the form of~\cref{eqn:CZX-steady-state}, with
\begin{equation}
    |c|\leq \frac{1}{2^{N/2l-1}}.
\end{equation}
\end{prop}

\begin{proof}
    Most of the proof runs parallel to the previous one. In particular, we now consider the transfer matrix of a block of $l$ consecutive sites, 
    \begin{equation}
        E_{U^{(l)}}:= \begin{array}{c}
        \begin{tikzpicture}[scale=1.,baseline={([yshift=-0.75ex] current bounding box.center)}
        ]
        \whaMtraceOpmul{(0,0)}{\small $U$}{\small $\rho$};
        \end{tikzpicture}
        \end{array}
    \end{equation}
    where we omit the position label $i$, and we would like to show that $\|E_{U^{(l)}}\|_\infty\leq 1/\sqrt{2}$ and $\|E_{U^{(l)}}\|_1\leq \sqrt{2}$. 

    To start, denote $U^{(l)}$ as the horizontal blocking of $l$ sites of $U$, and apply a Hadamard gate $H$ on its right, $U^{(l)}\ra U^{(l)} H$. Note that $\|U^{(l)} H\|_p=\|U^{(l)}\|_p$, the Schatten norm is unchanged. 
    The explicit form of $U^{(l)}$ (after applying $H$) reads
    \begin{equation}
        \begin{aligned}
            U^{(l)}&=\sqrt{2}|0\rangle\langle 0|\otimes \sigma_1^+\otimes O^{(1)}\otimes \sigma_l^+\\
            &\quad+  \sqrt{2}|0\rangle\langle 1|\otimes \sigma_1^+\otimes O^{(2)}\otimes \sigma_l^-\\
            &\quad +\sqrt{2}|1\rangle\langle 0|\otimes \sigma_1^-\otimes O^{(3)}\otimes \sigma_l^+ \\
            &\quad +  \sqrt{2}|1\rangle\langle 1|\otimes \sigma_1^-\otimes O^{(4)}\otimes \sigma_l^-
        \end{aligned}
    \end{equation}
    where the first register is the virtual bond degree of freedom, and $O^{(r)}$ is a summation of strings of $\sigma^{\pm}$, 
    \begin{equation}
        \begin{aligned}
            O^{(r)}=\sum_k O_k f_r(k)
        \end{aligned}
    \end{equation}
    where $k=(k_2,k_3,\cdots,k_{n-1})$ with $k_i=0,1$ and denoting $\sigma^1:=\sigma^-,\sigma^0:=\sigma^+$,
    \begin{equation}
    \begin{aligned}
         O_k&= \sigma_2^{k_2}\otimes \sigma_3^{k_3}\otimes \cdots\otimes \sigma_{l-1}^{k_{l-1}}\\
         f_1(k)&=(-1)^{k_2 k_3+k_3 k_4+\cdots+k_{l-2} k_{l-1}}\\
         f_2(k)&=(-1)^{k_2 k_3+k_3 k_4+\cdots+k_{l-2} k_{l-1}+ k_{l-1}} \\
         f_3(k)&=(-1)^{k_2 k_3+k_3 k_4+\cdots+k_{l-2} k_{l-1}+ k_{2}} \\
         f_4(k)&=(-1)^{k_2 k_3+k_3 k_4+\cdots+k_{l-2} k_{l-1}+ k_2+ k_{l-1}}. 
    \end{aligned}
    \end{equation}

    Now, consider the pure state $\rho=|\psi\rangle\langle\psi|$ and write
    \begin{equation}
        |\psi\rangle=\sum_{i,j=0}^1 |i\rangle_1\otimes |\psi_{i,j}\rangle\otimes |j\rangle_l
    \end{equation}
    where we single out the first and last physical qubit as above. Thus, we can write
    \begin{equation}
        E_{U^{l}}=\sum_k E_{U^{l},k}
    \end{equation}
    where $E_{U^{l},k}$ are $2\times 2$ matrices with components
    \begin{equation}
    \begin{aligned}
        E_{U^{l},k}(0,0)&=\sqrt{2} \langle \psi_{0,0}|O_k|\psi_{1,1}\rangle f_1(k)\\
        E_{U^{l},k}(0,1)&=\sqrt{2} \langle \psi_{0,1}|O_k|\psi_{1,0}\rangle f_2(k)\\
        E_{U^{l},k}(1,0)&=\sqrt{2} \langle \psi_{1,0}|O_{\bar{k}}|\psi_{0,1}\rangle f_3(\bar{k})\\
        E_{U^{l},k}(1,1)&=\sqrt{2} \langle \psi_{1,1}|O_{\bar{k}}|\psi_{0,0}\rangle f_4(\bar{k})
    \end{aligned}
    \end{equation}
    Note that in the third and fourth lines, we have changed the index of summation and replaced $k\ra\bar{k}=(1,1,\cdots,1)-k$. Expand $|\psi_{i,j}\rangle$ in the computational basis
    \begin{equation}
        |\psi_{i,j}\rangle=\sum_k a^{ij}_k |k\rangle= \sum_k b^{ij}_k |\bar{k}\rangle.
    \end{equation}
    Note that $\langle k''|O_k|k'\rangle=\delta_{k',\overline{k''}}\delta_{k,k''}$, the matrix $E_{U^{l},k}$ is simplified to
    \begin{equation}
        E_{U^{l},k}=\sqrt{2} \begin{pmatrix}
            y_k f_1(k) & z_k f_2(k)\\
            \bar{z}_k f_3(\bar{k}) & \bar{y}_k f_4(\bar{k})
        \end{pmatrix}
    \end{equation}
    where $y_k=\bar{a}_k^{00} b_k^{11}$ and $z_k=\bar{a}_k^{01} b_k^{10}$. Using 
    \begin{equation}
        \|E_{U^{l}}\|_p \leq \sum_k \|E_{U^{l},k}\|_p 
    \end{equation}
    we need to compute the singular values of $E_{U^{l},k}$. When $f_1(k) f_3(\bar{k})+f_2(k) f_4(\bar{k})=0$, the singular values are degenerate, 
    \begin{equation}
    \begin{aligned}
        \lambda_1=\lambda_2&=\sqrt{2}\sqrt{|y_k|^2+|z_k|^2}\\
        &\leq \frac{1}{\sqrt{2}}(|a_k^{00}|^2+|b_k^{11}|^2+|a_k^{01}|^2+|b_k^{10}|^2);
    \end{aligned}
    \end{equation}
    and when $f_1(k) f_3(\bar{k})+f_2(k) f_4(\bar{k})=\pm 2$, the singular values are
    \begin{equation}
    \begin{aligned}
        \lambda_1 &= \sqrt{2}(|y_k|+|z_k|)\\
        \lambda_2 &= \sqrt{2}| |y_k|-|z_k| |\\
        \lambda_{1,2}&\leq  \frac{1}{\sqrt{2}}(|a_k^{00}|^2+|b_k^{11}|^2+|a_k^{01}|^2+|b_k^{10}|^2). 
    \end{aligned}
    \end{equation}
    Thus,
    \begin{equation}
    \begin{aligned}
        \|E_{U^{l}}\|_\infty&\leq \frac{1}{\sqrt{2}}\sum_k (|a_k^{00}|^2+|b_k^{11}|^2+|a_k^{01}|^2+|b_k^{10}|^2)\\
        &=\frac{1}{\sqrt{2}}\langle\psi|\psi\rangle=\frac{1}{\sqrt{2}}.
    \end{aligned}
    \end{equation}
    Similarly, 
    \begin{equation}
        \|E_{U^{l}}\|_1\leq \sqrt{2}. 
    \end{equation}
    This finishes the proof. 
\end{proof}

\end{document}